\documentclass[12pt]{article}
\usepackage[letterpaper,top=2cm,bottom=2cm,left=3cm,right=3cm,marginparwidth=1.75cm,margin=1in ]{geometry}
\usepackage{amsmath}
\usepackage{graphicx} % Required for inserting images
\usepackage{amssymb}
\usepackage{comment}
\usepackage{xcolor}
\usepackage{graphicx}
\usepackage{graphbox}
\usepackage{natbib}
\usepackage{float}
\usepackage{amsthm}
\usepackage{multirow}
\usepackage{centernot}
\usepackage{url} 
\usepackage{hyperref}
\usepackage{authblk}
\DeclareMathOperator*{\argmin}{arg\,min}
\DeclareMathOperator*{\argmax}{arg\,max}
\newtheorem{theorem}{Theorem}[section]
\newtheorem{prop}[theorem]{Proposition}
\newtheorem{lem}[theorem]{Lemma}
\newtheorem*{remark}{Remark}
\title{A Functional Approach to Curve Alignment and Shape Analysis}

\author[1,2]{Issam-Ali Moindjié$^*$}
\author[1]{Cédric Beaulac}
\author[1]{Marie-Hélène Descary}
\affil[1]{Department of mathematics, Université du Québec à Montréal, Canada}
\affil[2]{LAMPS, Université de Perpignan Via Domitia, France}
\affil[*]{Corresponding author: issam-ali.moindjie@univ-perp.fr}

\date{}
\newcommand{\norm}[1]{\left\lVert#1\right\rVert}

\begin{document}

\maketitle

\begin{abstract}{ In many image analysis problems, the contours of objects carry important statistical information about shape. Such contours are typically affected by deformation variables including scaling, translation, rotation, and reparametrization. }{   
Previous studies in statistical shape analysis have mainly focused on analyzing contours and shapes through discrete observations. While this approach might offer computational advantages, it overlooks the continuous nature of these objects and their underlying geometric structure.} It also ignores potential dependencies between the deformation variables and their effect on the shape, which may result in a loss of statistical information and reduced interpretability. In this paper, we introduce a novel framework for analyzing shapes within the context of Functional Data Analysis (FDA). Basis expansion techniques are employed to derive analytic solutions for the estimation of deformation variables, namely scaling, translation, rotation, and reparametrization, thereby achieving curve alignment. A generative model for random {  contours} is then developed using principal component analysis techniques. Numerical experiments on simulated data and the \textit{MPEG-7} database demonstrate that our method successfully identifies deformation parameters and captures the underlying distribution of { random contours} in settings where traditional FDA methods fail.
\end{abstract}

\section{Introduction}
The analysis of images has become increasingly important with advances in acquisition and storage technologies. While representing images as pixel matrices has enabled powerful learning methods, this representation does not explicitly capture the fundamental components of images: the objects they depict, characterized by their shapes, colors, and textures. In this work, we focus on the statistical analysis of object shapes.
\par   
\begin{figure}[ht]
    \centering
\begin{tabular}{ c c }
  (a) &  (b) \\ 
     \includegraphics[align=c, width=.35\textwidth]{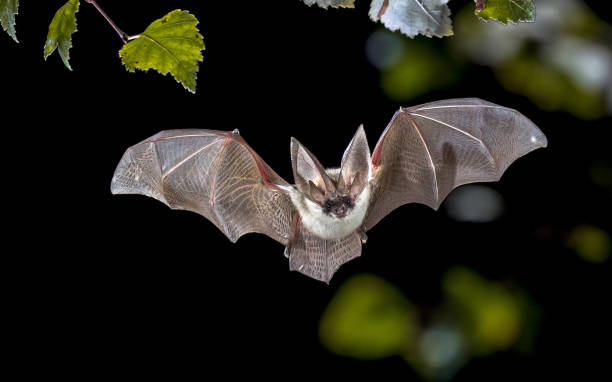}  & \includegraphics[align=c, width=.35\textwidth]{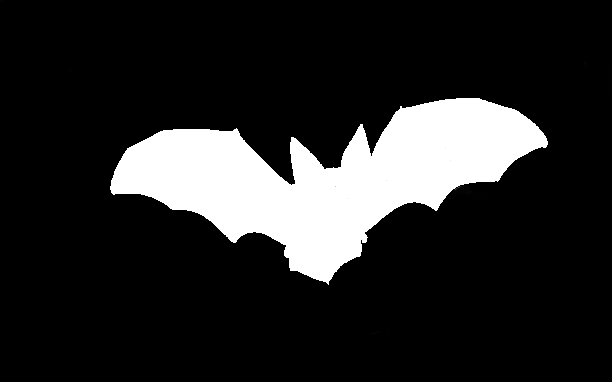} 
     \\  \\
     (c) &  (d) \\ 
     \includegraphics[align=c, width=.40\textwidth]{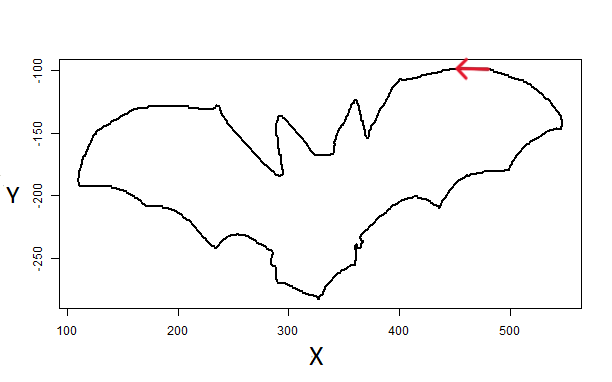} & 
     \begin{tabular}{c}
      \includegraphics[align=c, width=.25\textwidth]{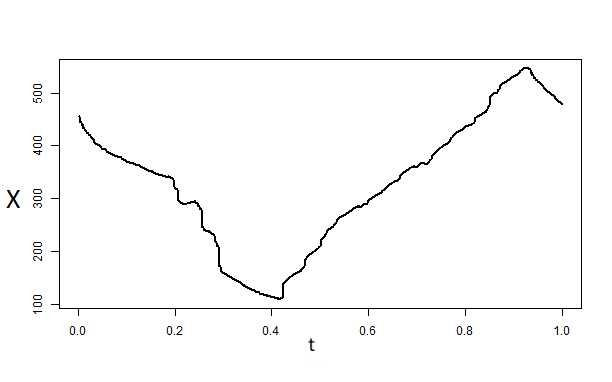} \\ \includegraphics[align=c, width=.25\textwidth]{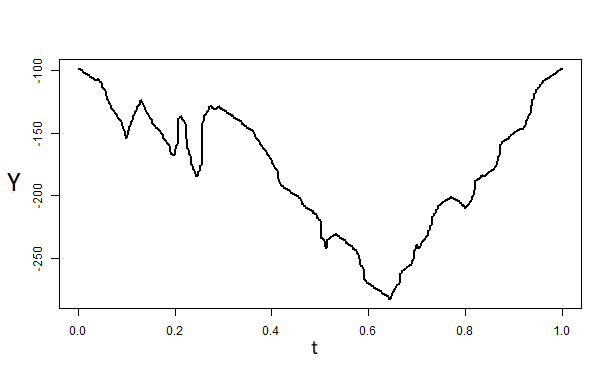}
     \end{tabular}
\end{tabular}
    \caption{From an image to the planar curve $\mathbf{C}$: (a) original image  \citep{bat}, (b) binarized image, (c) contour of the object in the image, and (d) the coordinate functions $X(t)$ and $Y(t)$. The red arrow in (c) indicates the starting point and the orientation of the traveling used to obtain the coordinate functions in (d).}
    \label{pres1}
\end{figure}
\par 
If we are interested in analyzing shapes in images, a natural representation is through binary images, as illustrated in Figure \ref{pres1}(b). These binary images, often referred to as masks or { silhouettes}, typically result from image segmentation and consist of pixels taking two values: one indicating regions outside the object of interest and the other indicating regions inside, thereby representing shapes as regions rather than explicit boundaries ({ or contours}). \cite{lelivre} identify other representations of shape, such as point clouds and ordered samples of points, and argue in favor of a functional representation; see \cite{lelivre} for a detailed discussion. We also adopt a curve-based representation of shapes in images. Compared to binary images, { the curve representation, which gives the contour of the silhouette image, leads to a substantial reduction in dimensionality while preserving the geometric structure. }%  of the shape
Moreover, if contours are extracted as collections of pixels, images of different resolutions yield contours with varying numbers of points. These issues can be addressed by smoothing the collection of points and projecting the resulting contours onto a common functional basis, thereby enabling a coherent statistical analysis of shapes.
\par
We focus on the { variable} $\mathbf{C}$, which represents { the contour} of the main object in the image. This variable can be viewed as a random parametric planar curve:
$$
\mathbf{C}(t)=\begin{pmatrix}
    X(t)\\  Y(t)
\end{pmatrix}, 
$$  
where $t\in [0, 1]$ represents the proportion of the curve that has been traversed, from the starting point ($t=0$) to the end point ($t=1$), and $X$, $Y$ are the coordinate functions. {  We restrict our attention to the case where the contour is closed, that is, $\mathbf{C}(0)=\mathbf{C}(1)$}. Figure \ref{pres1} illustrates the pipeline for extracting the planar curve $\mathbf{C}$ from an image.
\par
The { contour} $\mathbf{C}$ is assumed to be a deformed version of a latent variable $\mathbf{\tilde C}$:
\begin{equation}
\textbf{C}(t)= \rho\mathbf{O} \mathbf{\tilde C} \circ \gamma(t) + \mathbf{T}, \ t \in [0, 1],
\label{general_mod}    
\end{equation}
where $(\rho, \mathbf{O}, \mathbf{T}, \gamma)$ are deformation variables and $\mathbf{\tilde C}$ is the \textit{shape} of $\mathbf{C}$ in the sense of \cite{kendall}, i.e., what remains when all deformations are removed. The deformation variables act through scaling ($\rho\in \mathbb{R}^+$), translation ($\mathbf{T}\in \mathbb{R}^2$), rotation ($\mathbf{O}\in SO(2)$), and reparametrization ($\gamma \in \Gamma$, where $\Gamma$ denotes a set of increasing functions from $[0,1]$ to $[0,1]$).
\par
The analysis of $\mathbf{C}$ has been widely studied in the shape analysis literature. Early approaches relied on a finite set of points along the curve, referred to as landmarks, which are assumed to correspond across objects \citep{dryden1998}. More recent work considers $\mathbf{C}$ as a continuous planar curve (see, e.g., \cite{younes1998}) and introduces elastic shape analysis frameworks in which reparametrization functions $\gamma$ are modeled as diffeomorphic transformations and estimated jointly with rotation. In practice, these approaches rely on discretizations of the curve and numerical procedures, such as dynamic programming, and do not admit closed-form solutions. Moreover, elastic reparametrizations introduce an additional layer of transformation that must be estimated and interpreted.
\par
{  This work addresses the specific challenges arising when contours are automatically extracted from images, where phase variability due to arbitrary starting points is the dominant source of misalignment.} In this context, we restrict the class of admissible transformations by focusing on global deformations, namely scaling, translation, rotation, and a simple reparametrization of the starting point. This choice leads to alignment procedures that are fully compatible with a functional representation of the curve and can be carried out without additional discretization beyond the chosen basis representation. While more flexible approaches may be incorporated at a later stage, they are not required to capture the main geometric variability induced by image-based contour extraction.
\par
Functional Data Analysis (FDA) provides a natural framework for modeling planar curves as functional observations \cite{ramsay2008}. While several works have considered shape data within this framework (see, e.g., \cite{stocker2023, dai2018, dai2022}), the problem of jointly estimating deformation variables and modeling the resulting aligned curves has not been explicitly addressed.
\par

In this paper, we adopt a functional data analysis framework to study the random planar curve defined in \eqref{general_mod}. In particular, $\mathbf{C}$ is viewed as a bivariate functional variable and its realizations  are therefore treated as functional data.

\par
In summary, the main contributions of this work are as follows:  
\begin{itemize}
	\item A functional data-based procedure to estimate the deformation variables from a bi-variate random planar curve $\mathbf{C}$, focusing on phase variability induced by different starting points.
	\item A novel model for $\mathbf{C}$ that explicitly accounts for the deformation variables through two separate principal component analysis (PCA): one is performed on the aligned functional variable ($\rho \mathbf{\tilde C}$) and the other on the remaining deformation variables. 
\end{itemize}

The remainder of the paper is organized as follows. Section \ref{sec_shape_align} focuses on estimating the deformation variables, while Section \ref{gen_mod} introduces a generative model for $\mathbf{C}$. Numerical experiments on simulated data are presented in Section \ref{sims}, followed by an application to the \textit{MPEG-7} database in Section \ref{app}. The paper concludes with a discussion in Section \ref{disc}.

\section{Alignment of planar closed  curves}
\label{sec_shape_align}
Before introducing our alignment procedure, we first illustrate why alignment is a necessary step for the statistical analysis of planar curves. 
\par
For illustration purposes, Figure \ref{ex2} depicts realizations of $\mathbf{C}$ and their sample mean. The left column displays planar curves, while the corresponding coordinate functions are shown in the middle and right columns. The first row presents a sample of $n=10$ "heart-shaped" planar curves exhibiting variations in rotation and reparametrization, while the second row displays the estimate of the mean function $\mathbb{E}(\mathbf{C})$, defined 

coordinate-wise as:
$$
\mathbb{E}(\mathbf{C}(t))=\begin{pmatrix}
\mathbb{E}(X(t)) \\ \mathbb{E}(Y(t))
\end{pmatrix}, \quad t\in [0,1].
$$
As illustrated in Figure~\ref{ex2}, this definition does not preserve the geometric structure of the shapes, yielding an average that is not representative of the underlying objects.
%As illustrated in Figure \ref{ex2}, this definition fails to preserve the geometric structure of the shape, producing an average that no longer resembles the original objects. 
This issue arises from the presence of deformation variables, such as rotation and reparametrization, which introduce non-linear variability that cannot be captured by linear functional methods.

\begin{figure}[H]
    \centering
    \begin{tabular}{ c c c c}
    	%& (1) & (2) & (3) \\
       &\includegraphics[align=c, width=0.2\linewidth]{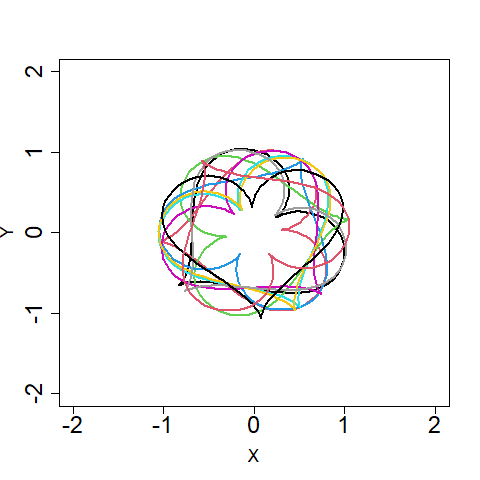} &     \includegraphics[align=c, width=0.2\linewidth]{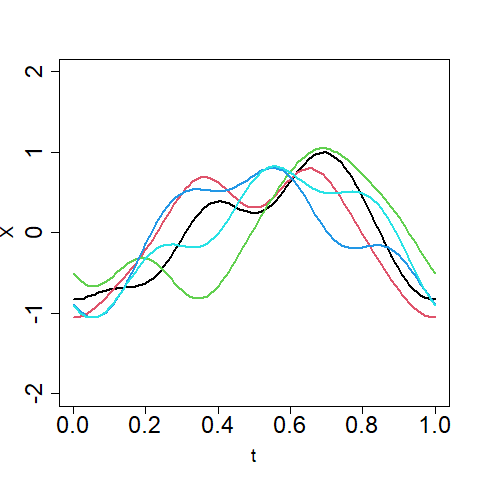} &  \includegraphics[align=c, width=0.2\linewidth]{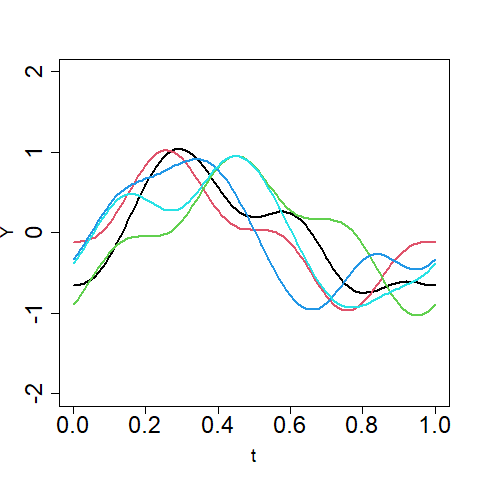} \\ 
         &\includegraphics[align=c, width=0.2\linewidth]{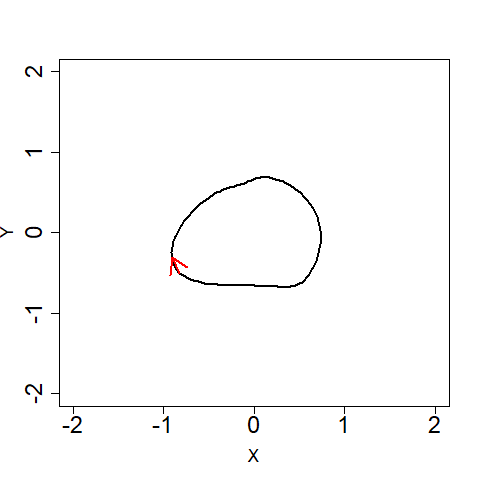} &     \includegraphics[align=c, width=0.2\linewidth]{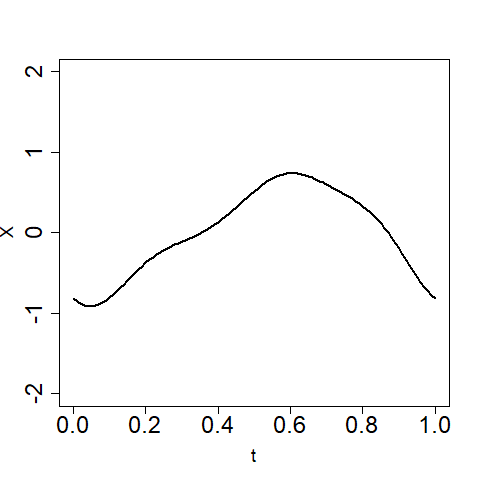} &  \includegraphics[align=c, width=0.2\linewidth]{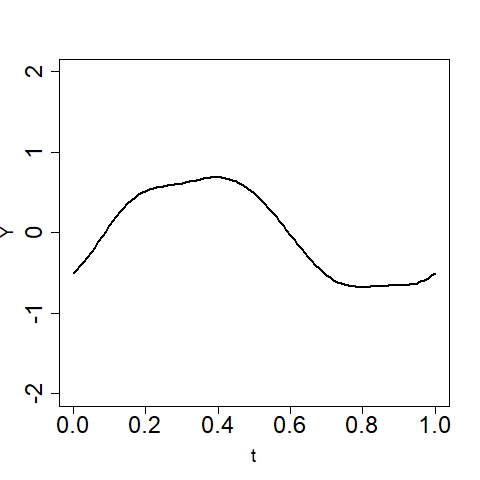}
   \end{tabular}
    \caption{Examples of observations of $\mathbf{C}$ (first row) and their sample mean (second row). Columns present the planar curves, the first and the second coordinate functions, respectively. }
    \label{ex2}
\end{figure}

A similar problem occurs when alignment is not addressed for a univariate functional variable $X$ taking values in $L^2([0,1],\mathbb{R})$. In this case, $X$ can be expressed as $X=\tilde{X}\circ \gamma$, where $\gamma$ is a warping function and $\tilde{X}$ represents the curve without phase variation. Hence, obtaining the shape $\mathbf{\tilde C}$ from $\mathbf{C}$ can be viewed as a generalization of the alignment procedure for univariate functional variables to the bivariate case. In the univariate case, this problem has been the source of many developments in FDA (see Chap.7 in \cite{ramsay2008} and \cite{marron2015} for a review). 
For example, in \cite{FPCA_amp} a flexible generative model of $X$ is proposed. 
This method relies on a joint functional principal component analysis of bijective transformations of $\tilde{X}$ and $\gamma$, inspired by the square-root velocity framework. Using a transformation of $\gamma\in \Gamma$ is mandatory due to the complex structure of $\Gamma$, defined in this case as a space of diffeomorphism functions. 
\cite{happ_am} discuss several choices of transformations, comparing their relative advantages and drawbacks.
They found that square-root velocity inverse transformations might lead to warping functions with atypical structures, making them hard to interpret.  
\par 
To solve this issue, in this manuscript we consider the following space of reparametrization functions 
\begin{equation} \label{def_Gamma}
\Gamma=\left\{ \gamma_\delta : [0,1]\rightarrow [0,1], \ \gamma_\delta(t)= \text{mod}(t-\delta, 1), \ \delta \in [0,1] \right\},
%\Gamma=\left\{ \gamma_\delta(t)= \text{mod}(t-\delta, 1), \ t\in [0,1], \ \delta \in [0,1] \right\}, 
\end{equation} where $\text{mod}(\cdot,1)$ is the modulo $1$ function. It is thus assumed that each reparametrization function is uniquely characterized by its "starting point" $\delta \in [0,1]$. The motivation for choosing this particular space is discussed in Section \ref{reparm}.

\subsection{Model and parametrization}

We recall the model introduced in \eqref{general_mod}:
\begin{equation*}
\textbf{C}(t)= \rho\mathbf{O} \mathbf{\tilde C} \circ \gamma(t) + \mathbf{T}, \ t \in [0, 1],
\end{equation*}
where $\mathbf{C}=(X,Y)^\top$ is a bivariate functional variable taking values in the space of square-integrable functions $L^2([0,1],\mathbb{R}^2)$. Since $L^2([0,1],\mathbb{R}^2)$ is isomorphic to $L^2([0,1],\mathbb{R}) \times L^2([0,1],\mathbb{R})$, we equivalently view $\mathbf{C}$ as an element of the Hilbert space $\mathcal{H} = L^2([0,1],\mathbb{R}) \times L^2([0,1],\mathbb{R})$. Moreover, since $\mathbf{C}$ represents a closed curve, we assume that $\mathbf{C}(0)= \mathbf{C}(1)$. The space $\mathcal{H}$ is equipped with the inner product $\langle \cdot , \cdot\rangle_\mathcal{H}$ defined as: 
$$
\langle \boldsymbol{f},  \boldsymbol{g} \rangle_\mathcal{H}=  \int_{0}^1 \left\{f_{1}(t)g_{1}(t) + f_{2}(t)g_{2}(t)\right\}dt,\ \boldsymbol{f}=(f_1,f_2)^\top, \boldsymbol{g}=(g_1,g_2)^\top \in \mathcal{H}. 
$$

The first step in analyzing a dataset of planar closed curves is to align the data, meaning to extract the shape $\mathbf{\tilde C}$ from $\mathbf{C}$ or, equivalently, to identify the deformation variables in model \eqref{general_mod}. We formulate the following hypothesis regarding these deformation variables:
%For the identification of the unobserved deformation variables, we assume that: 
\begin{itemize}
    \item $\rho$ is a positive random scalar variable;
    \item $\mathbf{T}=(T_1  ,T_2)^\top$ is a random vector in $\mathbb{R}^2$;
    \item The rotation matrix  
		 \begin{equation}
\mathbf{O}:=\mathbf{O}_\theta=\begin{pmatrix}
			\cos(\theta) & -\sin(\theta) \\ 
			\sin(\theta) & \cos(\theta)
		\end{pmatrix}
		\label{rot}\end{equation} depends on a random angle $\theta$ in $[0,2\pi]$. Note that the case of the reflection matrix $\mathbf{R}_\theta$ can be considered using the propriety: 
$$ 
\mathbf{R}_\theta = \mathbf{O}_\theta \begin{pmatrix}
    1 & 0 \\ 
    0& -1
\end{pmatrix}.$$
This reduces to analyzing $
    (X,-Y)^\top$ in our framework instead of $\mathbf{C} = (X,Y)^\top$.
    \item The reparametrization function $\gamma := \gamma_\delta$ takes values in $\Gamma$, where $\Gamma$ is defined as in \eqref{def_Gamma}.   
    \item The random function $\tilde{\mathbf{C}}=({\tilde X} ,{\tilde Y} )^\top $ takes values in $\mathbf{S}^\infty=\{ \boldsymbol{f} \in \mathcal{H},\ \norm{\boldsymbol f}_\mathcal{H}=1 \}$, where $\norm{\ \cdot\ }_\mathcal{H}$ is the norm induced by the inner product $\langle \cdot, \cdot \rangle_\mathcal{H}$. For identification purposes, we will also assume that $\mathbf{\tilde C}$ is centered, i.e. $$\int_0^1\tilde{X}(t)dt=\int_0^1\tilde{Y}(t)dt=0. $$  
\end{itemize}

Before presenting our alignment method, we first provide the intuition behind the definition of the function space $\Gamma$.

\subsection{The space of reparametrization functions}
\label{reparm}
{
Recall that we define the space of reparametrization functions $\Gamma$ in \eqref{def_Gamma} as $$
\Gamma=\left\{ \gamma_\delta : [0,1]\rightarrow [0,1], \ \gamma_\delta(t)= \text{mod}(t-\delta, 1), \ \delta \in [0,1] \right\},$$
with $\text{mod}(x,1)=x-\lfloor x \rfloor, \forall x\in \mathbb{R}$. Functions that belong to this space can be seen as a generalization of the well-known time-shift deformation warping functions for univariate functional data (see e.g. \cite{marron2015} for details)  to the case of planar closed curves.} The modulo function allows considering the cyclic nature of the closed curve $\mathbf{C}$.

Figure \ref{ex_gamma} exhibits the effect of $\gamma_\delta \in \Gamma$ for different values of $\delta$ on a realization $\mathbf{\tilde c}=(\tilde x, \tilde y)^\top$ of $\mathbf{\tilde C}$. We see the "rearrangement effect" of the reparametrization function: $\delta=0$ defines the reference shape and the starting point of the coordinate functions of $\mathbf{\tilde c} \circ \gamma_\delta$ coincides with the point $t=\delta$ of the reference shape.

\begin{figure}[H]
    \centering
    \begin{tabular}{c c c c c }
    &   $\gamma_\delta$ & $\mathbf{\tilde c} \circ \gamma_\delta$ & ${\tilde x} \circ \gamma_\delta$ & ${\tilde y} \circ \gamma_\delta$ \\ 
        $\delta=0$&\includegraphics[align=c, width=0.15\textwidth]{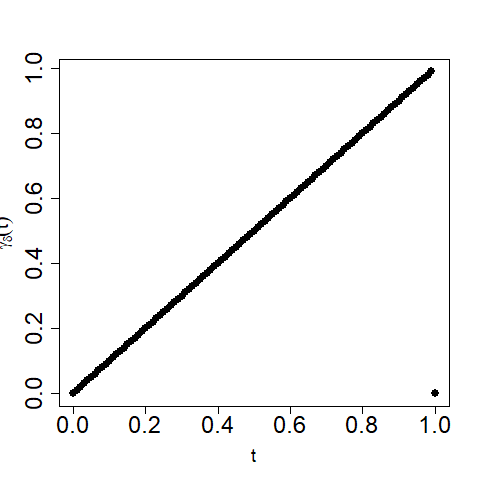} &\includegraphics[align=c, width=0.15\textwidth]{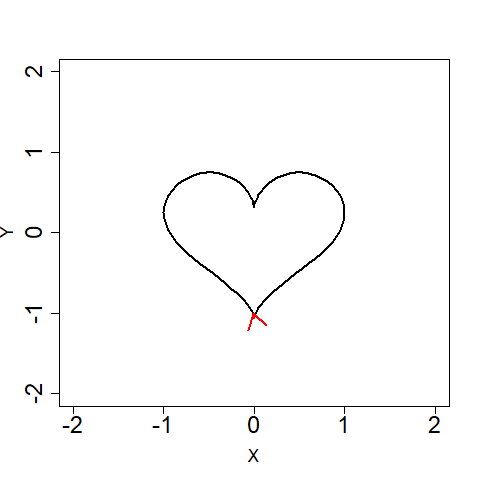}& \includegraphics[align=c, width=0.15\textwidth]{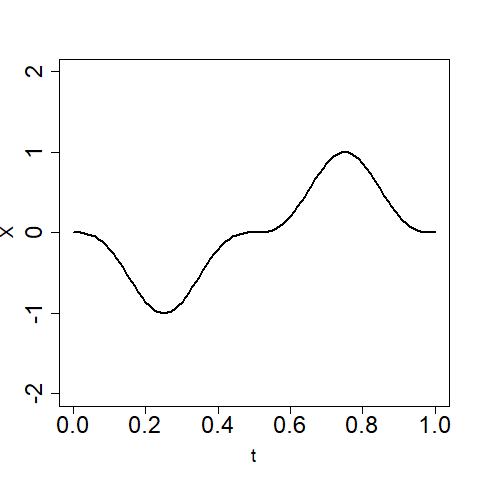} & \includegraphics[align=c, width=0.15\textwidth]{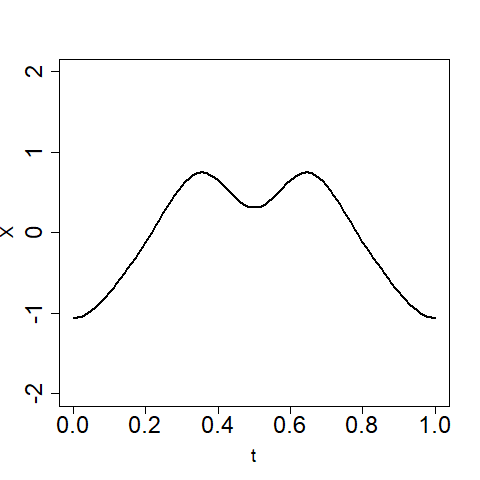}    \\
         
        $\delta=0.5$&\includegraphics[align=c, width=0.15\textwidth]{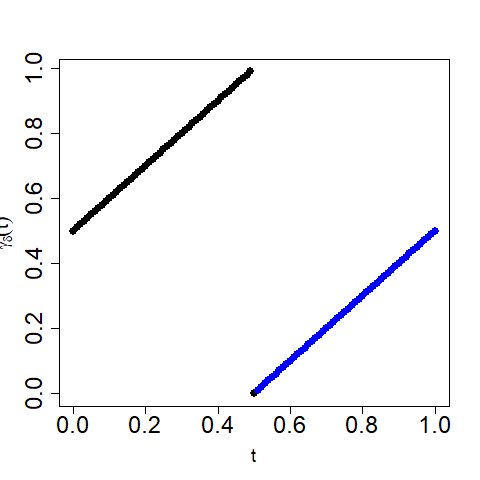} &\includegraphics[align=c, width=0.15\textwidth]{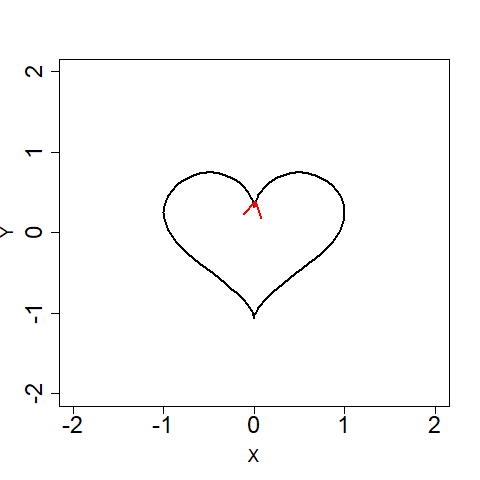}& \includegraphics[align=c, width=0.15\textwidth]{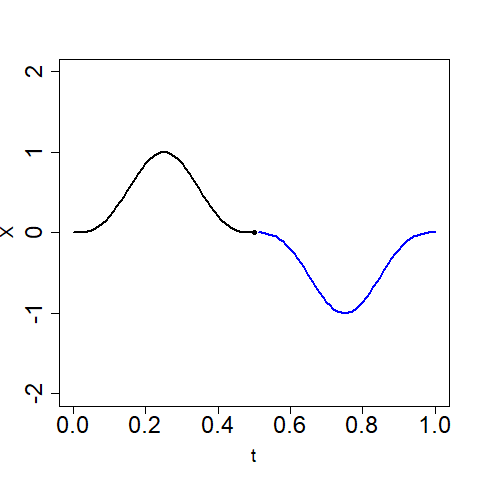} & \includegraphics[align=c, width=0.15\textwidth]{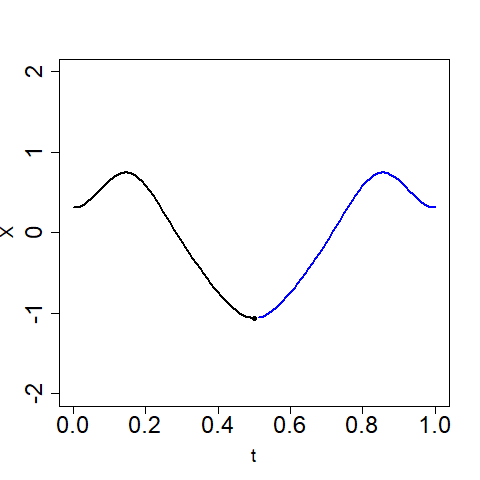}    \\
        
        $\delta=0.155$& \includegraphics[align=c, width=0.15\textwidth]{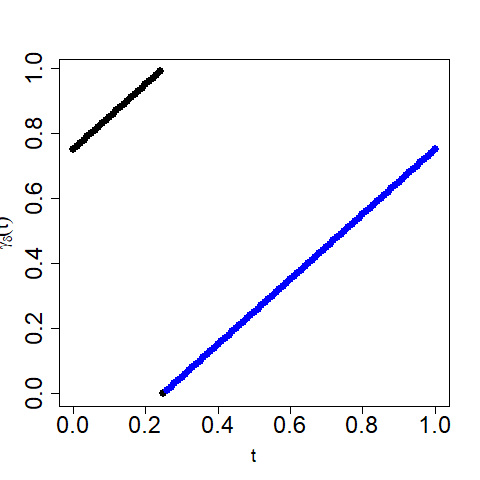} &\includegraphics[align=c, width=0.15\textwidth]{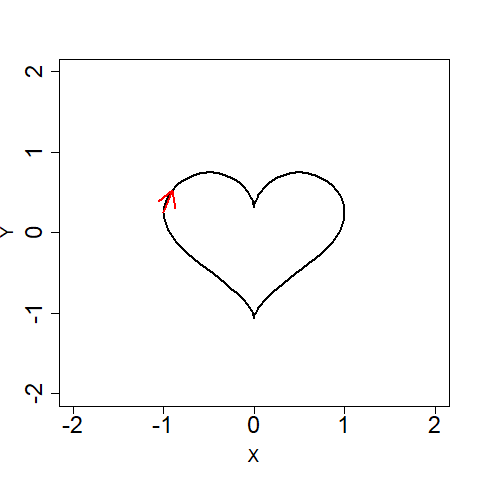}& \includegraphics[align=c, width=0.15\textwidth]{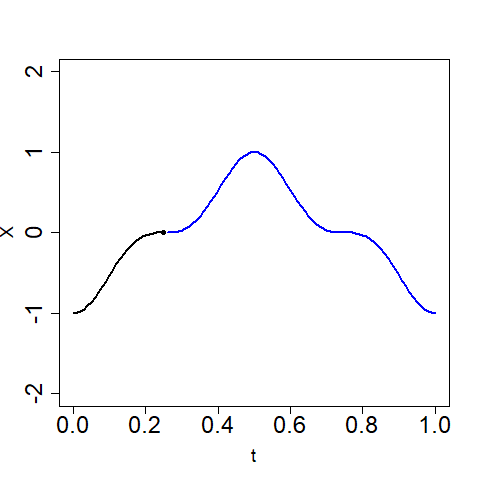} & \includegraphics[align=c, width=0.15\textwidth]{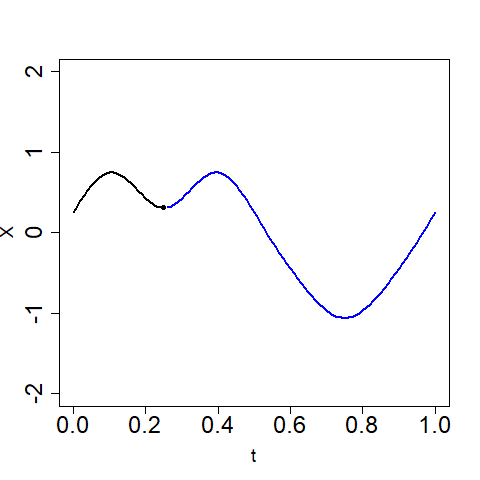}    \\
    \end{tabular}
    \caption{Plots of the functions $\gamma_\delta$ and of its composition with the shape $\mathbf{\tilde c}$ for different values of $\delta$, resulting in different starting points (represented by the red arrows) and coordinate functions.}
    \label{ex_gamma}
\end{figure}
{

In contour datasets, the variation of $\delta$ arises from edge detection algorithms, which select the starting point as the argument of the coordinate functions associated with the smallest/largest distance from the origin. The issue with this approach is that the starting value depends on the transformation variables, particularly rotation. Figure \ref{param_theta} illustrates this effect when the smallest distance from the origin is considered: each row shows the same shape subjected to different rotations, leading to entirely different coordinate functions. Therefore, the challenge is to obtain $\delta$ jointly with the other transformation variables. We will elaborate on this point in Section \ref{estimation}. 
\begin{figure}[H]
    \centering
    \begin{tabular}{ c c c c }
         %$\mathbf{O}_{\theta_1}\tilde{\mathbf{x} }\circ\gamma_1$
         &\includegraphics[align=c, width=0.15\linewidth]{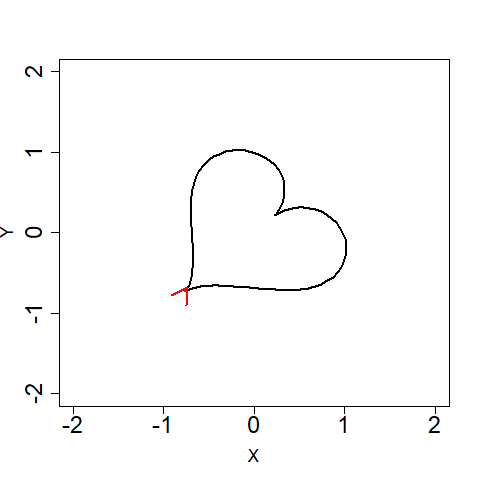}
   &  \includegraphics[align=c, width=0.15\linewidth]{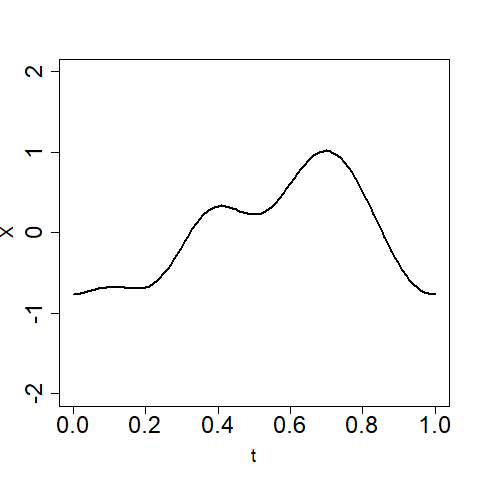} & \includegraphics[align=c, width=0.15\linewidth]{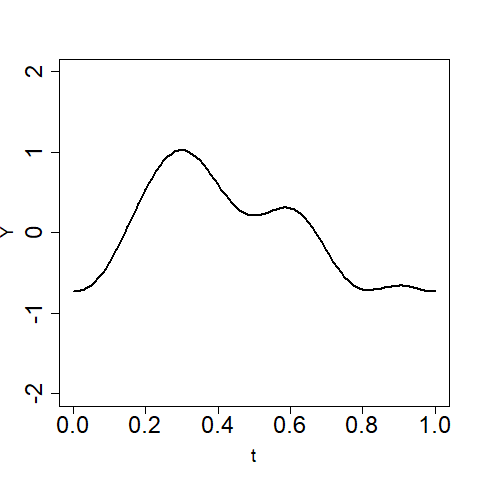}\\
   & \includegraphics[align=c, width=0.15\linewidth]{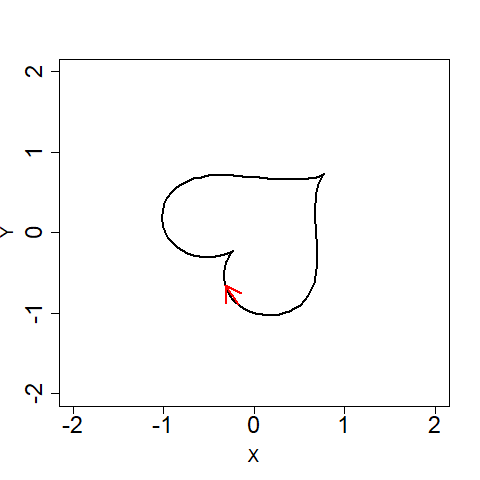}
         &  \includegraphics[align=c, width=0.15\linewidth]{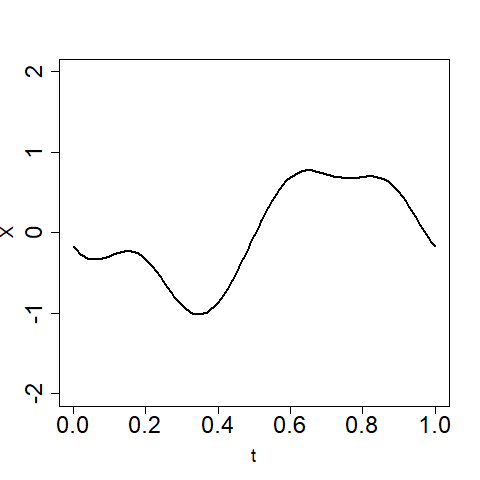} & \includegraphics[align=c, width=0.15\linewidth]{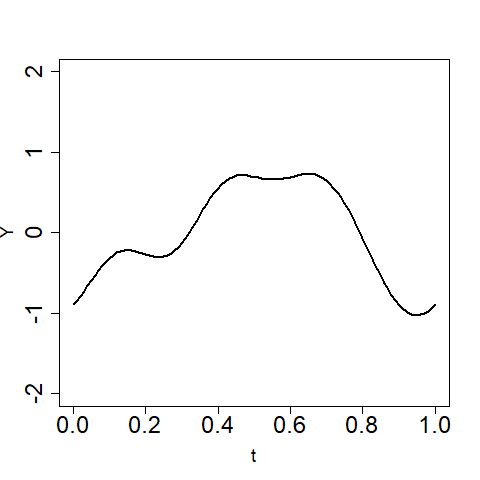} \\ 
   \end{tabular}
    \caption{Illustration of the effect of rotation on a given shape and its coordinate functions when the starting value (illustrated by the red arrow on the shape) is defined as the point associated with the lowest values of the coordinate functions. }
    \label{param_theta}
\end{figure}
In addition to being well-suited for contour data, the space $\Gamma$ has interesting mathematical properties, as demonstrated by the following result.
\begin{prop} $(\Gamma, \circ)$ is a group and the isometry property
$
\norm{\boldsymbol{f} \circ \gamma_\delta}_\mathcal{H}= \norm{\boldsymbol{f}}_\mathcal{H} 
$
holds for $\boldsymbol{f} \in \mathcal{H}$ and $\gamma_\delta \in \Gamma$. 
\label{prop_1}
\end{prop}

While the class of reparametrization functions $\Gamma$ defined in \eqref{def_Gamma} focuses on variability induced by different starting points, it can be extended to incorporate more flexible transformations. In particular, one may consider reparametrizations of the form $\gamma \in \Gamma_0$, where $\Gamma_0$ is the space of diffeomorphism functions from $[0,1] \to [0,1]$, leading to elastic alignment frameworks similar to those studied in the literature (see, e.g., \cite{FPCA_amp}). This can be done by modifying $\Gamma$ as follows: 
$$
\tilde\Gamma=\left\{ \tilde \gamma_{\gamma,\delta}:[0,1]\rightarrow [0,1], \ \tilde \gamma_{\gamma,\delta}(t)=\text{mod}(\gamma(t) -\delta, 1), \ \delta \in [0,1] \text{ and } \gamma \in \Gamma_0\right\}. 
$$

However, such extensions introduce additional complexity both in modeling and estimation, and are not required to capture the dominant source of phase variability arising in contour data extracted from images. For this reason, we restrict our attention to the simpler class $\Gamma$, which provides an interpretable and tractable framework consistent with the data generation process. This perspective allows us to bridge practical contour extraction procedures with a principled statistical framework.

}

With this in mind, the rest of this section presents the process of determining the deformation variables $(\rho, \mathbf{O}_\theta, \mathbf{T}, \gamma_\delta)$. 

\subsection{Estimation of the deformation variables}
\label{st}

Although our ultimate goal is to analyze a sample of planar closed curves—that is, to work with $n$ realizations of the random function $\mathbf{C}$—the alignment procedure is carried out on a curve-by-curve basis. Specifically, for each curve, the translation and scaling variables can be directly determined based on the assumptions made on the shape $\mathbf{\tilde{C}}$, as will be discussed in Section~\ref{T_and_rho}. As we will see in Section~\ref{rot-par}, the rotation and reparametrization variables are defined relative to a template, which may depend on other curves in the sample. However, once this template has been specified, the alignment procedure can be applied independently to each curve. Before presenting the estimation procedure, we first introduce the assumed basis representation of $\mathbf{C}$.

\subsubsection{Representation of the contour using a basis expansion}
\label{estimation}

Due to technological limitations, functional variables are observed on discrete grids in practice. In our setting, this discretization is induced by the image resolution. A crucial first step in functional data analysis is therefore smoothing, which consists of reconstructing a continuous representation of each discretely observed curve. Working with such a continuous representation, rather than directly with pixel-level data, provides several advantages. In particular, it allows for the analysis of images with different resolutions and yields a representation of the shape that is significantly more parsimonious than the original binary image.

A common approach to smoothing is to use a basis function expansion \cite[Chap.~3]{ramsay2008}, i.e., to assume that $\mathbf{C}$ can be expressed as
\begin{equation} \label{C_fourrier}
\mathbf{C}(t)= \begin{pmatrix} X(t) \\ Y(t) \end{pmatrix}
=\sum_{k=0}^M  \begin{pmatrix} a_{k1} \\ a_{k2}\end{pmatrix} \psi_k(t), \quad t\in [0,1],
\end{equation}
where $\psi_0,\ldots,\psi_M \in L^2([0,1])$ are known basis functions, and $\boldsymbol{a}_{k}=(a_{k1},a_{k2})^\top\in \mathbb{R}^2$ are the corresponding coefficients. Reconstructing the full function $\mathbf{C}$ from discrete observations thus reduces to estimating the coefficients $\boldsymbol{a}_{k}$, which can be achieved, for instance, by minimizing a least-squares criterion.

The choice of basis functions $\psi_k$ depends on the application. In FDA, B-splines and Fourier basis functions are among the most commonly used options: B-splines are typically preferred for non-periodic data, whereas Fourier basis functions are well suited for periodic data \citep{ramsay2008}. In this work, we use Fourier basis functions due to the inherent periodicity of $\mathbf{C}$. Recall that the first $M+1$ Fourier basis functions are defined as follows. 
\begin{equation}\label{fourier}
\psi_0(t)=1, \ \psi_{k}(t)=\left\{ \begin{array}{ll}
\sqrt{2}\sin((k+1)\pi t) & \textrm{if $k$ odd,}   \\
\sqrt{2}\cos(k\pi t) & \textrm{if $k$ even,} \end{array} \right. 
\end{equation}
for $t\in [0,1]$, $k=1,\ldots,M$, with $M$ even. These functions are orthonormal and satisfy $\int_0^1 \psi_k(t)\,dt=0$ for all $k\ge 1$. For the remainder of the paper, $\psi_k$ denotes the $(k+1)$-th Fourier basis function defined in \eqref{fourier}.

A key advantage of this choice is the well-known \emph{time-shift property} of the Fourier basis, which, as we will show below, greatly simplifies the estimation of the reparametrization variable. This property is formalized in the following lemma, which describes the effect of reparametrization on the Fourier basis functions.

\begin{lem} Let $\boldsymbol{\psi}=(\psi_1,\ldots,\psi_M)^\top$ and define $
\boldsymbol{\psi}_{\delta}= (
    \psi_1 \circ \gamma_\delta, \ldots , \psi_M \circ \gamma_\delta 
)^\top,
$
with $\gamma_\delta \in \Gamma$. Then we have 
\begin{equation}
     \boldsymbol{\psi} = \boldsymbol{\beta}(\delta) \boldsymbol{\psi}_\delta,
\end{equation}
where 
$$
\boldsymbol{\beta}(\delta)=\begin{pmatrix}
    \mathbf{O}_{2\pi \delta}& {\bf 0} & \ldots & {\bf 0} \\ 
    {\bf 0} & \mathbf{O}_{4\pi \delta}& \ldots & {\bf 0} \\ 
    \vdots & \vdots & \ldots & \vdots \\ 
    {\bf 0} & {\bf 0} & \ldots & \mathbf{O}_{M\pi \delta}
\end{pmatrix} \in \mathbb{R}^{M \times M},$$
is a sparse-orthogonal matrix and $\mathbf{O}_{k\pi \delta}, k\in \{2,4,\ldots,M\}$ is the rotation matrix defined in \eqref{rot}. 
\label{lem_f_gamma}
\end{lem}
This result shows that reparametrization reduces to a simple linear transformation of the basis functions, where $\boldsymbol{\beta}(\delta)$ acts as a transfer matrix. Note that if $\delta \in \{0,1\}$, then $\boldsymbol{\beta}(\delta)$ is the identity matrix, and the effect of reparametrization vanishes.

\subsubsection{The translation and scaling transformations}\label{T_and_rho}

To determine the translation and scaling variables, we rely on the definition of the shape variable $\mathbf{ \tilde C}$ which is centered in $(0, 0)^\top$ and has unit norm. Direct computations, combined with the properties of the Fourier basis, yield
 \begin{eqnarray*}
&&\mathbf{T} = \int_0^1{\bf C}(t)dt=\int_0^1 \sum_{k=0}^M\boldsymbol{a}_{k} \psi_k(t) dt= \boldsymbol{a}_{0}, \\ 
&&\rho = ||{\bf C}-{\bf T}||_\mathcal{H} = \left \|\sum_{k=1}^M\boldsymbol{a}_{k} \psi_k\right\|_\mathcal{H}=\sqrt{\sum_{k=1}^M\norm{\boldsymbol{a}_k}_2^2},
\end{eqnarray*}
with $\|\cdot\|_2$ the $\ell^2$-vector norm.
These variables have a natural interpretation for analyzing an object's position within an image : the parameter $\mathbf{T}$ captures the spatial position of the object, while $\rho$ reflects its scale, which is related to its distance from the camera.

\par  

Let $\mathbf{C}^*$ be the variable obtained after removing the effects of translation and scaling from $\mathbf{C}$:

\begin{equation}\label{defXstar}
\displaystyle {\bf C}^*(t)= \frac{1}{\rho } \left(\mathbf{C}(t)- \mathbf{T}\right)=\sum_{k=1}^M\boldsymbol{\alpha}_{k} \psi_k(t)=\boldsymbol{\alpha}\boldsymbol{\psi}(t), \ t \in [0,1]
\end{equation}
with $\boldsymbol{\alpha}_k= \boldsymbol{a}_k/\rho \in \mathbb{R}^2$, for $k=1,\ldots,M$, and where $\boldsymbol{\alpha}=(\boldsymbol{\alpha}_1 , \boldsymbol{\alpha}_2 , \ldots , \boldsymbol{\alpha}_M) \in \mathbb{R}^{2\times M}$ and $\boldsymbol{\psi}=( \psi_1 , \psi_2 , \ldots , \psi_M)^\top \in (L^2([0,1]))^{M}$. This normalized representation will be used in the next section to estimate the rotation matrix $\mathbf{O}_\theta$ and the reparametrization function $\gamma_\delta$.
    
\begin{remark}By construction, $\mathbf{C}^*$ belongs to the same space as $\mathbf{\tilde{C}}$, namely $\mathbf{S}^\infty$.
In statistical shape analysis, $\mathbf{S}^\infty$ is called the "pre-shape space" (\cite{dryden1998}, \cite{larticle}) and it contains shapes that are possibly deformed by rotation and reparametrization. 
The "shape space", which could be denoted by $\mathbf{S}^\infty/ \sim$, represents the quotient space of $\mathbf{S}^\infty$ with the equivalence relation $\sim$ defined as 
\begin{equation}
    \boldsymbol{f} \sim \boldsymbol{g} \iff \exists (\theta, \gamma) \in [0,2\pi ]\times \Gamma \text{ such that } \boldsymbol{f}= \mathbf{O}_\theta \boldsymbol{g}\circ \gamma,
    \label{eq_class}
\end{equation}
where $\boldsymbol{f}, \boldsymbol{g} \in \mathbf{S}^\infty$. In other words, this relation says that $\boldsymbol{f} \sim \boldsymbol{g}$ if and only if $\boldsymbol{f}$ and $\boldsymbol{g}$ share the same shape.
\end{remark}

\subsubsection{The rotation and reparametrization transformations}
\label{rot-par}

Our strategy for estimating the rotation and reparametrization parameters consists in aligning the function $\mathbf{C}^*$ to a reference function $\boldsymbol{\mu} \in \mathbf{S}^\infty$. As in the univariate curve alignment setting, there is no ground truth for the choice of $\boldsymbol{\mu}$ \citep{marron2015}. This reference can either be specified a priori or estimated from the data.
A common choice, which we also adopt here, is to define $\boldsymbol{\mu}$ as the Fréchet mean of the shape variable, i.e. to define $\boldsymbol{\mu}$ as
\begin{equation} \label{def_frechet_mean}
\boldsymbol{\mu}= \mathbb{E}_d(\mathbf{\tilde C}) = \arg \min_{\boldsymbol{f}\in \mathbf{S}^\infty }\mathbb{E}[d(\boldsymbol{f},\mathbf{\tilde C})^2], 
\end{equation}
where $d$ denotes a suitable distance on $\mathbf{S}^\infty$. In \cite{FPCA_amp}, the authors consider an elastic distance $d$ and propose an iterative procedure to approximate the corresponding Fréchet mean in the univariate setting. Such an approach is necessary since direct computation of the Fréchet mean is not feasible in practice, as it depends on the (unknown) aligned shapes $\tilde{\mathbf{C}}$. In Section~\ref{estim_frechet_mean}, we propose a similar iterative strategy to estimate the Fréchet mean, using a distance $d$ that is invariant to rotation and reparametrization. For the theoretical developments that follow, we assume that $\boldsymbol{\mu}$ is known and can be expressed as
$$
\boldsymbol{\mu}(t)= \sum_{k=1}^M \boldsymbol{u}_k \psi_k(t) =  \boldsymbol{u} \boldsymbol{\psi}(t),
$$
with $\boldsymbol{u}=(\boldsymbol{u}_1 , \boldsymbol{u}_2 , \ldots , \boldsymbol{u}_M)\in \mathbb{R}^{2\times M}$.

\par 
By replacing $\mathbf{C}(t)$ in \eqref{defXstar} with its representation in \eqref{general_mod}, we obtain that $\mathbf{C}^*$ can be written as
\begin{equation}
    \mathbf{C}^*(t) = \mathbf{O}_\theta \mathbf{\tilde C}\circ \gamma_\delta(t)
    \label{x-star}, \ t\in [0,1].
\end{equation} 
Determining the rotation matrix $\mathbf{O}_\theta$ and the reparametrization function $\gamma_\delta$ that align $\mathbf{{C}}^*$ to $\boldsymbol{\mu}$ is equivalent to finding the parameters $\theta$ and $\delta$ that minimize the distance, measured in terms of the $\norm{\cdot}_\mathcal{H}$ norm, between $\mathbf{C}^* = \mathbf{O}_\theta \mathbf{\tilde C}\circ \gamma_\delta$ and $\boldsymbol{\mu}$. Equivalently, one might think of this problem as aligning the template $\boldsymbol{\mu}$ to the pre-shape $\mathbf{C}^*$, and thus finding the parameters $\theta$ and $\delta$ that minimize the distance between $\mathbf{O}_\theta\boldsymbol{\mu} \circ \gamma_\delta$ and $\mathbf{C}^*$. This leads to the optimisation problem 

\begin{equation}
\min_{(\theta,\delta) \in [0, 2\pi]\times [0,1]}\norm{\mathbf{O}_\theta\boldsymbol{\mu} \circ \gamma_\delta - \mathbf{C}^* }_\mathcal{H}^2.  
\label{princ}
\end{equation}

Using the fact that $\mathbf{O}_\theta^{-1}=\mathbf{O}_\theta^\top$ (since rotation matrices are orthogonal) and that the inverse of $\gamma_\delta \in \Gamma$ is $\gamma_{1-\delta}$ (as shown in the proof of Proposition \ref{prop_1}), we can recover the shape variable $\mathbf{\tilde C}$ once $(\hat{\theta}, \hat{\delta})$ have been obtained, via $\mathbf{\tilde C} = \mathbf{O}_{\hat{\theta}}^\top \mathbf{C}^* \circ \hat{\gamma}_{1-\hat{\delta} }$.

Note that from the isometry propriety of $\Gamma$ (Proposition \ref{prop_1}), it follows that
$$\norm{\boldsymbol{f} \circ\gamma-\boldsymbol{g}\circ\gamma}_\mathcal{H} =\norm{\boldsymbol{f}-\boldsymbol{g}}_\mathcal{H}, \ \text{ for } \boldsymbol{f}, \boldsymbol{g}\in \mathcal{H} \text{ and } \gamma \in \Gamma.$$
As in the univariate case \citep{FPCA_amp}, this invariance ensures that our estimators remain unchanged under a common reparametrization applied to both $\mathbf{C}$ and $\boldsymbol{\mu}$.

The next result builds on Lemma \ref{lem_f_gamma} to reformulate the objective function in \eqref{princ}, allowing it to be expressed in matrix form rather than in functional form.
\begin{prop} The optimization problem \eqref{princ} is equivalent to 
\begin{equation}
   \min_{(\theta,\delta)\in [0,2\pi] \times  [0,1]}\norm{\mathbf{O}_\theta\boldsymbol{u}\boldsymbol{\beta}(-\delta) - \boldsymbol{\alpha}}_F^2
    \label{princ_2}
\end{equation}
where $\norm{\cdot}_F$ is the Frobenius norm matrix and $\boldsymbol{\beta}(\cdot)$ the transfer matrix defined in Lemma \ref{lem_f_gamma}. 
\label{chang_b}
\end{prop} 
Problems \eqref{princ} and \eqref{princ_2} do not admit closed-form solutions. To address this issue, we propose a novel alternating optimization algorithm inspired by the iterative closest point (ICP) algorithm \citep{ICP-review}. The method operates directly in the functional space, without discretizing the curves, and is therefore well suited to the present framework. We refer to this procedure as the \textit{Iterative Closest Function} (ICF) algorithm, which is described in the next section. 

\subsubsection{The ICF algorithm}
\label{ficp}

Starting from an initial value of $\delta$, the ICF algorithm alternates between the following two steps until convergence.

\begin{enumerate}
    \item[(i)] \textbf{Estimation of $\theta$ for a given $\delta$} \\ 
       In this step, the estimator of ${\theta}$ is obtained by solving a Procrustes orthogonal problem: 
\begin{equation}
\hat \theta = \argmin_{\theta \in [0, 2\pi] }\norm{\mathbf{O}_\theta { \boldsymbol{\mu}}\circ {\gamma}_{{\delta}}-\mathbf{C}^*}_\mathcal{H}^2.
	\label{pbl_1}
\end{equation}
The multivariate version of this problem has been extensively studied (see, e.g., \cite{procrustes}). In the functional setting, by arguments analogous to those in the classical case, one can show that $\hat{\theta}$ belongs to the set $\{\theta_1, \theta_2\}$, where $\theta_1$ and $\theta_2$ are the two solutions to the following equation:
$$
\tan({\theta_k})=\langle \mathbf{C}^*, \boldsymbol{\mu}\circ \gamma_{{\delta}} \rangle_{\mathcal{H}}^{-1} \left\{ \langle X^*,\mu_1\circ \gamma_{\delta} \rangle_{L^2} - \langle Y^*,\mu_2\circ \gamma_{\delta} \rangle_{L^2}\right\}, k=1,2,
$$
where $ \boldsymbol{\mu}=(\mu_1,\mu_2)^\top$ and $\mathbf{ C}^*=(X^*,Y^*)^\top$. Once $\theta_1$ and $\theta_2$ have been computed, we evaluate the objective function in \eqref{pbl_1} at both values and select $\hat{\theta}$ as the one yielding the smallest value.

\item[(ii)] \textbf{Estimation of $\delta$ for a given $\theta$} \\
When the rotation angle $\theta$ is known, problem \eqref{princ_2} reduces to the following optimization problem:
\begin{equation}
    \hat{\delta}= \argmin_{\delta\in [0, 1]}\norm{\mathbf{O}_\theta\boldsymbol{u}\boldsymbol{\beta}(-\delta) - \boldsymbol{\alpha}}_F^2.
	\label{pb2}
\end{equation}
Unlike step (i), this problem does not generally admit a closed-form solution. Instead, we solve it numerically using the result stated in the following proposition.
\begin{prop} \label{prop} 
The solution $\hat{\delta}$ of problem \eqref{pb2} satisfies the equation  
\begin{equation}
    \sum_{k\in \{1, 3, \ldots, M-1\} }w_{k}^{1,\theta} \sin( (k+1) \pi \hat{\delta} ) = \sum_{k\in \{1, 3, \ldots, M-1\} }w_{k}^{2,\theta} \cos( (k+1) \pi \hat{\delta}), 
    \label{optim_prob}
\end{equation}
where the coefficients $w_{ k}^{1,\theta}$ and $w_{ k}^{2,\theta}$  are derived from the matrix  $(\mathbf{O}_\theta\boldsymbol{u})^\top \boldsymbol{\alpha}$.  
\end{prop}
For clarity, the exact expressions of the coefficients $(w_{k}^{l, \theta})_{k,l}$ and the proof of the proposition are provided in Appendix \ref{delta_sol}.

We solve equation \eqref{optim_prob} with respect to $\hat{\delta}$ using a bisection method. This approach yields accurate results, as demonstrated in the numerical experiments of Section \ref{sims}.
\end{enumerate}

\begin{remark}
The performance of the proposed ICF algorithm depends on the initialization of $\delta$ in step (i). We therefore recommend running the algorithm with multiple initial values and selecting the pair $(\hat{\theta}, \hat{\delta})$ that minimizes the objective function in \eqref{princ}.
\end{remark}

\subsection{Estimation of the Fréchet mean} \label{estim_frechet_mean}

The Fréchet mean introduced in \eqref{def_frechet_mean} depends on the choice of a distance on $\mathbf{S}^\infty$. In this work, motivated by \eqref{princ} and Proposition \ref{prop_1}, we define the function $d:\mathbf{S}^\infty\times\mathbf{S}^\infty\rightarrow \mathbb{R}^+$ as
$$
d(\boldsymbol{f}, \boldsymbol{g})=\min_{\delta\in [0,1], \ \theta \in [0, 2\pi]}\norm{\boldsymbol{f}\circ \gamma_\delta- \mathbf{O}_\theta \boldsymbol{g}}_\mathcal{H}.
$$
This quantity can be interpreted as a Procrustes-type distance \citep{dryden1998}, as it is invariant to rotation and reparametrization:
$$
d(\boldsymbol{f}, \boldsymbol{g})=d(\mathbf{O}_{\theta_1}\boldsymbol{f}\circ\gamma_{\delta_1}, \mathbf{O}_{\theta_2}\boldsymbol{g}\circ\gamma_{\delta_2}), \quad \forall \delta_1,\delta_2 \in [0,1], \forall \theta_1,\theta_2  \in [0, 2\pi].
$$ 
In particular, this implies that $d(\mathbf{C}^*,\mathbf{\tilde C})= 0$, and more generally that $d$ is not a proper distance on $\mathbf{S}^\infty$, since the separability property does not hold:
$$ 
d(\boldsymbol{f}, \boldsymbol{g})= 0 \centernot\implies \boldsymbol{f}= \boldsymbol{g}, \quad \text{for } \boldsymbol{f},\boldsymbol{g} \in \mathbf{S}^\infty.
$$
However, $d$ defines a proper distance when considered on the shape space $\mathbf{S}^\infty/\sim$ rather than on the pre-shape space $\mathbf{S}^\infty$ { (see Lemma~\ref{lem} in appendix)}.
%\textcolor{red}{est-ce que le texte en vert est vrai?}
Plugging the distance $d$ into \eqref{def_frechet_mean}, one can show that
$$ 
\boldsymbol{\mu}= \mathbb{E}_d(\mathbf{\tilde C}) = \argmax_{\boldsymbol{f}\in \mathbf{S}^\infty }\mathbb{E}\left[\langle \mathbf{\tilde C}, \boldsymbol{f} \rangle_{\mathcal{H}}\right]
$$

%}
We now describe how to estimate $\boldsymbol{\mu}$ from a sample of $n$ shapes $\mathbf{\tilde C}_1, \ldots, \mathbf{\tilde C}_n$. We assume that each shape admits a Fourier basis expansion:
$$
\mathbf{\tilde C}_i(t)= \sum_{k=1}^M \tilde{\boldsymbol{a}}_{i,k} \psi_k(t)=  \mathbf{\tilde A}_i \boldsymbol{\psi}(t), \quad  t \in [0,1], \ i=1,\ldots,n,
$$
where $ \mathbf{\tilde A}_i\in \mathbb{R}^{2\times M}$ is the corresponding coefficient matrix. We define the estimator of the Fréchet mean as $\hat{\boldsymbol{\mu}}= \hat{\boldsymbol{u}} \boldsymbol{\psi},$
where 
\begin{equation}
    \hat{\boldsymbol{u}}= \argmax_{\boldsymbol{u}\in \mathbb{R}^{2\times M}} \frac{1}{n} \sum_{i=1}^n \langle \mathbf{\tilde A}_i, \boldsymbol{u}\rangle_F ,
\label{frech_mean}
\end{equation}
and $\langle \cdot, \cdot \rangle_F$ denotes the Frobenius inner product.

The optimization problem in \eqref{frech_mean} is non-convex and does not admit a closed-form solution, and must therefore be solved numerically. This problem is structurally similar to the Fréchet mean estimation problem encountered in classical statistical shape analysis, for which a numerical solution is proposed in \citet{R-shapes}. We rely on their implementation to solve \eqref{frech_mean}, using the general-purpose nonlinear optimization function \textit{nlm} from the \texttt{stats} package in R \citep{R-stats}. Additional details on the implementation are provided in our publicly available code (\url{https://github.com/imoindjie/Shape-FDA}).

\subsection{Iterative estimation of shapes and of the Fréchet mean}
\label{sec_iter_algo}

From the previous sections on alignment using the ICF algorithm and on Fréchet mean estimation, it follows that there is a strong interdependence between shape estimation and mean estimation: the Fréchet mean is required to align the shapes, while aligned shapes are in turn needed to estimate the Fréchet mean. Consequently, poor estimates at either stage may adversely affect the overall procedure. To address this issue, we propose an iterative algorithm that alternates between shape alignment and Fréchet mean estimation.

The complete estimation pipeline, from raw data to aligned shapes and the Fréchet mean, is summarized below.

\begin{itemize}
\item \textbf{Input:}
 \begin{itemize} 
 \item $\{(\mathbf{C}_i(t_{i1}),\ldots,\mathbf{C}_i(t_{iJ_i}))\}_{i=1}^n$: observed curves on a discrete grid
 \item $M$: number of Fourier basis functions
 \item $\xi$: convergence threshold
 \end{itemize}

\item \textbf{Algorithm:}
\begin{enumerate}
    \item \textbf{Smoothing} \\
    For each $i\in \{1,\ldots,n\}$, estimate the coefficient vectors $\boldsymbol{a}_{i,k}$, $k=0,\ldots,M$, defined in \eqref{C_fourrier}, by minimizing a least-squares criterion based on the discretized observations $(\mathbf{C}_i(t_{i1}),\ldots,\mathbf{C}_i(t_{iJ_i}))$. The resulting smoothed curve is
    $$
    \hat{\mathbf{C}}_i(t)=\sum_{k=0}^M \hat{\boldsymbol{a}}_{i,k} \psi_k(t), \quad t\in[0,1].
    $$

    \item \textbf{Centering and normalization} \\
    For each $i\in \{1,\ldots,n\}$, define the estimated pre-shape as
    $$
    \hat{\mathbf{C}}^*_i(t) = \frac{1}{\hat \rho_i}\left(\hat{\mathbf{C}}_i(t)-\hat{\mathbf{T}}_i \right)
    =\sum_{k=1}^M \hat{\boldsymbol{\alpha}}_{i,k} \psi_k(t),
    $$
    where $\hat{\mathbf{T}}_i=\hat{\boldsymbol{a}}_{i,0}$, $\hat{\rho}_i=\sqrt{\sum_{k=1}^M \|\hat{\boldsymbol{a}}_{i,k}\|_2^2}$, and $\hat{\boldsymbol{\alpha}}_{i,k}= \hat{\boldsymbol{a}}_{i,k} / \hat{\rho}_i, k=1,\ldots,M$.

    \item \textbf{Iterative alignment and Fréchet mean estimation} \\
    Initialize $\hat{\boldsymbol{\mu}}$ by randomly selecting one pre-shape from $\{\hat{\mathbf{C}}^*_1,\ldots,\hat{\mathbf{C}}^*_n\}$. Then repeat the following steps until convergence, i.e., until $\eta \leq \xi$:
    
    \begin{itemize} 
        \item[a.] For each $i\in \{1,\ldots,n\}$, estimate the rotation angle $\theta_i$ and the reparametrization parameter $\delta_i$ using the ICF algorithm, described in Section \ref{ficp}, with $\hat{\boldsymbol{\mu}}$ as the template. Define the estimated shape as
        $$
        \hat{\tilde{\mathbf{C}}}_i =  \mathbf{O}_{\hat{\theta}_i}^\top \hat{\mathbf{C}}^*_i \circ \gamma_{1-\hat{\delta}_i}.
        $$

        \item[b.] Using the shapes $\hat{\tilde{\mathbf{C}}}_1,\ldots, \hat{\tilde{\mathbf{C}}}_n$, compute an estimate $\hat{\boldsymbol{\mu}}$  of the Fréchet mean as described in Section~\ref{estim_frechet_mean}. Then compute the distance-based variance
        $$
        \eta=\frac{1}{n} \sum_{i=1}^n \left(\cos^{-1}\left(\langle \hat{\tilde{\mathbf{C}}}_i, \hat{\boldsymbol{\mu}}\rangle_{\mathcal{H}} \right)\right)^2.
        $$
    \end{itemize}
\end{enumerate}
\end{itemize}

\color{black}

\section{Modeling contours with a PCA-based approach}
\label{gen_mod}

In the previous section, we introduced a procedure to estimate the deformation parameters associated with a planar curve and to recover its underlying shape. As a result, each curve $\mathbf{C}$ can be decomposed into a shape component and a set of deformation variables, namely $(\tilde{\mathbf{C}}, \rho, \theta, \delta, \mathbf{T})$.

We now address the problem of modeling the variability of these quantities. In order to preserve interpretability, we propose to model the shape and the deformation variables separately. This strategy is motivated by the fact that these two sources of variability play fundamentally different roles and may be of independent interest in applications.

Following this idea, and inspired by phase–amplitude modeling approaches of \cite{FPCA_amp}, we introduce the two random variables
$$
{\bf Z}_1 = \rho \tilde{\mathbf{C}} \in \mathcal{H}, 
\quad \text{and} \quad 
{\bf Z}_2 = \left( 
\tan\left(\frac{\pi}{2} \left(\delta - \frac{1}{2}\right)\right),
\tan\left(\frac{1}{4} \left(\frac{\theta}{2} - \pi\right)\right),
\mathbf{T}^\top
\right)^\top \in \mathbb{R}^4.
$$
The variable ${\bf Z}_1$ is a linear transformation of $(\tilde{\mathbf{C}}, \rho) \in \mathbf{S}^\infty \times \mathbb{R}^+$ that preserves all the information while allowing us to work in the Hilbert space $\mathcal{H}$. This is convenient since the pre-shape space $\mathbf{S}^\infty$ does not naturally admit well-defined notions of mean or covariance. Similarly, we work with ${\bf Z}_2 \in \mathbb{R}^4$ instead of $(\delta, \theta, \mathbf{T})$, as the chosen transformations map bounded or periodic variables into an unconstrained Euclidean space, making standard multivariate techniques applicable. Note that both transformations are invertible.

\medskip

Since ${\bf Z}_1$ is a bivariate functional variable, we model it using multivariate functional principal component analysis (MFPCA) \citep{jacques2014model}, which yields
\begin{equation} 
\label{eq_Z1}
\mathbf{Z}_1(t) \simeq \mathbb{E}[\mathbf{Z}_1(t)] + \sum_{k=1}^{M_1} \xi_k^{(1)} \boldsymbol{\phi}_k(t), 
\quad t \in [0,1],
\end{equation}
where $\boldsymbol{\phi}_1, \ldots, \boldsymbol{\phi}_{M_1} \in \mathcal{H}$ are the eigenfunctions of the covariance operator of $\mathbf{Z}_1$. For ${\bf Z}_2$, we use standard multivariate PCA, leading to
\begin{equation}
\label{eq_Z2}
\mathbf{Z}_2\simeq \mathbb{E}[\mathbf{Z}_2] + \sum_{k=1}^{M_2} \xi_k^{(2)} \mathbf{U}_k,
\end{equation}
where $\mathbf{U}_1, \ldots, \mathbf{U}_{M_2} \in \mathbb{R}^4$ are the eigenvectors of the covariance matrix of $\mathbf{Z}_2$. For $j=1,2$, the scores $\{\xi_k^{(j)}\}$ are centered and uncorrelated, with
$$
\mathbb{E}(\xi_k^{(j)} \xi_l^{(j)}) =
\begin{cases}
\lambda_k^{(j)} & \text{if } k = l, \\
0 & \text{otherwise},
\end{cases}
$$
where $\lambda_k^{(j)}$ denotes the $k$-th eigenvalue of the covariance operator of $\mathbf{Z}_j$.

\medskip

A possible alternative would be to perform a single PCA on the joint variable $({\bf Z}_1, {\bf Z}_2)$ using methods for hybrid data (see, e.g., Chap.~10 of \cite{ramsay2008}). While such an approach may lead to a more compact representation, it obscures the interpretation of the principal components. In particular, it becomes difficult to disentangle whether a given mode of variation is driven by shape or by deformation effects. By contrast, the proposed approach preserves this separation, thereby providing a more interpretable representation. This is especially advantageous in applications such as anomaly detection or classification, where identifying the source of variability is of primary importance.

\subsection{A generative model for contours}
\label{gen-sec}

The representations \eqref{eq_Z1}--\eqref{eq_Z2} provide a finite-dimensional parametrization of both shape and deformation variability through the score vectors $\boldsymbol{\xi}^{(1)}$ and $\boldsymbol{\xi}^{(2)}$. To obtain a fully probabilistic model, it remains to specify a distribution for these scores.

A natural and widely used choice is to assume that the concatenated score vector
$$
\boldsymbol{\xi} = (\xi_1^{(1)}, \ldots, \xi_{M_1}^{(1)}, \xi_1^{(2)}, \ldots, \xi_{M_2}^{(2)}) \in \mathbb{R}^{M_1 + M_2}
$$
follows a multivariate Gaussian distribution with mean zero and covariance matrix
$$
\boldsymbol{\Sigma} =
\begin{pmatrix}
\boldsymbol{\Sigma}_1 & \boldsymbol{\Sigma}_{1,2} \\
\boldsymbol{\Sigma}_{1,2}^\top & \boldsymbol{\Sigma}_2
\end{pmatrix},
$$
where $\boldsymbol{\Sigma}_j = \mathrm{diag}(\lambda_1^{(j)}, \ldots, \lambda_{M_j}^{(j)})$ and $\boldsymbol{\Sigma}_{1,2}$ captures possible dependencies between shape and deformation scores.

\medskip

This specification defines a generative model for planar curves through the following procedure:
\begin{enumerate}
    \item Sample $\boldsymbol{\xi} \sim \mathcal{N}(0, \boldsymbol{\Sigma})$;
    \item Construct ${\bf Z}_1$ and ${\bf Z}_2$ using \eqref{eq_Z1} and \eqref{eq_Z2};
    \item Construct $\mathbf{C}$ from ${\bf Z}_1$ and ${\bf Z}_2$ via
$$
\mathbf{C}= \mathbf{O}_{\,8\tan^{-1}(Z_{22}) + 2\pi}\,
\mathbf{Z}_1 \circ \gamma_{\frac{2}{\pi}\tan^{-1}(Z_{21}) + \frac{1}{2}}
+ (Z_{23}, Z_{24})^\top.
$$
\end{enumerate}

It is important to note that the construction step 2 relies on the mean functions and the eigenstructures of the covariance operators of $\mathbf{Z}_1$ and $\mathbf{Z}_2$. In practice, these quantities are unknown and are replaced by their empirical estimates obtained from the sample. The resulting model is therefore a plug-in generative model, where population quantities are approximated by their empirical counterparts.

This construction provides a flexible and interpretable model for random planar curves, in which shape and deformation variability are explicitly represented while allowing for dependence between them through the joint distribution of the scores.

\section{Simulation study}%\textit{[find a competitor?]}}
\label{sims}

In this section, we present a simulation study to assess the performance of the proposed alignment methodology. Since translation and scaling can be directly recovered from the Fourier coefficients, we focus on the more challenging task of estimating the rotation and reparametrization parameters, namely $\theta$ and $\gamma$. In particular, we evaluate the ability of the alignment procedures to recover the underlying shapes from their associated pre-shapes.

To this end, we compare the proposed ICF algorithm with two widely used alignment methods: the iterative closest point algorithm (ICP; see \cite{ICP-review}) and the elastic alignment approach of \cite{larticle}, based on the square-root velocity function (SRVF) framework. A key distinction between these methods lies in the representation of the curves. Both ICP and SRVF-based approaches operate directly on discretely observed curves, whereas ICF first reconstructs a continuous representation using a Fourier basis and performs alignment in the corresponding functional space. 

Since our goal is to compare the alignment procedures themselves, we use a common fixed template throughout this section instead of the Fréchet mean, whose estimation would introduce an additional source of variability.

\subsection{Data simulation}

We generate a sample of pre-shapes $\mathbf{C}_i^*, \ i=1,\ldots,n$, according to
$$
\mathbf{C}_i^*(t) = \mathbf{O}_{\theta_i} \, \mathbf{\tilde C}_i \circ \gamma_{\delta_i}(t), \quad t \in [0,1],
$$
where the reparametrization parameters are sampled as $\delta_i \stackrel{i.i.d}{\sim} \mathcal{U}(0,1)$, and the rotation angles are defined by $\theta_i = 2\pi \delta_i$.

Inspired by \cite{heart}, the shape variables are constructed as $\mathbf{\tilde C}_i = {\bf H}_i/\| {\bf H}_i\|_{\mathcal{H}}$, where ${\bf H}_i$ is a heart-shaped parametric curve defined by
\begin{equation}
{\bf H}_i(t) =
\begin{pmatrix}
b_{i0} \sin^3(\pi(2t-1)) \\
\sum_{k=1}^4 b_{ik} \cos(k\pi(2t-1))
\end{pmatrix}, \quad t \in [0,1],
\label{hearth}
\end{equation}
with $(b_{i0}, \ldots, b_{i4}) \stackrel{i.i.d}{\sim} \mathcal{N}((16,13,-5,-2,-1)^\top/16, \sigma^2 \mathbf{I}_5)$. Sampling these coefficients induces variability in the shapes $\mathbf{\tilde C}_i$, while $\sigma$ controls the magnitude of this variability. We consider two levels of shape variability, $\sigma \in \{0.01, 0.1\}$.

We define the template as the normalized mean curve,
$$
\boldsymbol{\mu}=\mathbb{E}({\bf H}_i) /\|\mathbb{E}({\bf H}_i)\|_{\mathcal{H}},
$$
with a reference parametrization such that $\delta=0$ corresponds to the pointed end of the heart.

The first row of Figure~\ref{mean_heart_sig} displays the template $\boldsymbol{\mu}$ and its coordinate functions. The second and third rows show observations $\mathbf{\tilde c}_i$ of simulated shapes $\mathbf{\tilde C}_i$ for $\sigma=0.01$ and $\sigma=0.1$, respectively, illustrating the impact of $\sigma$ on shape variability.

\begin{figure}[H]
    \centering
    \begin{tabular}{c c c}
    \includegraphics[width=0.16\linewidth]{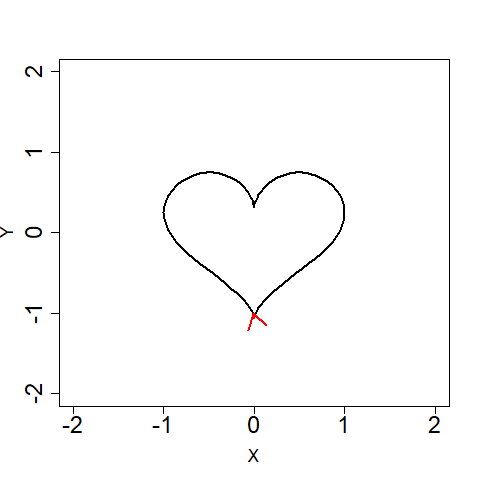}&
    \includegraphics[width=0.16\linewidth]{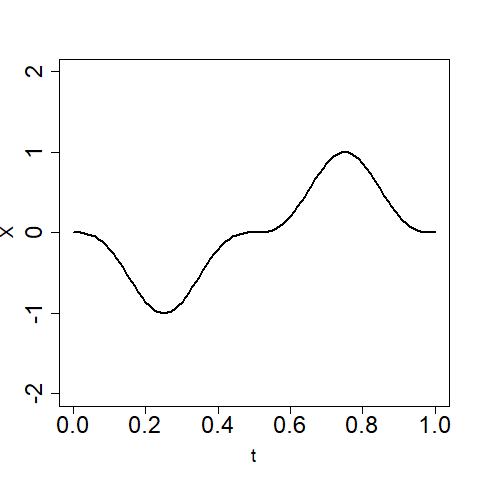} & 
    \includegraphics[width=0.16\linewidth]{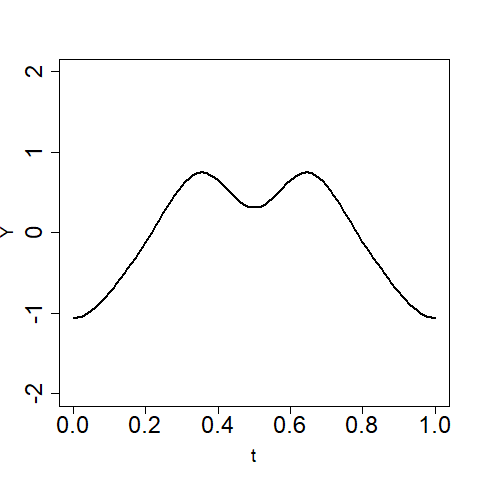} \\
    \includegraphics[scale=.16, align=c]{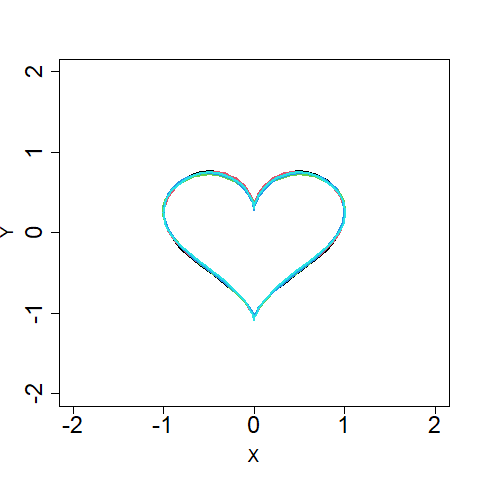} &  \includegraphics[scale=.16, align=c]{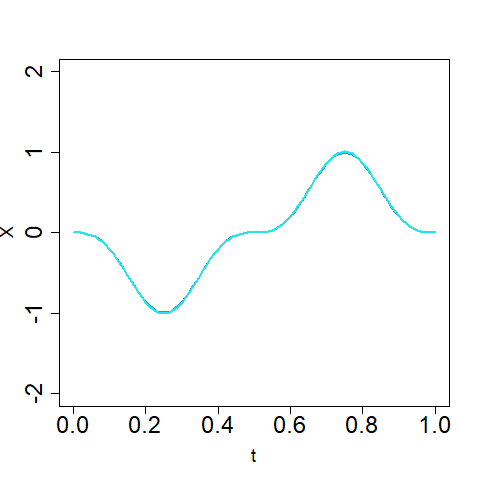} & \includegraphics[scale=.16, align=c]{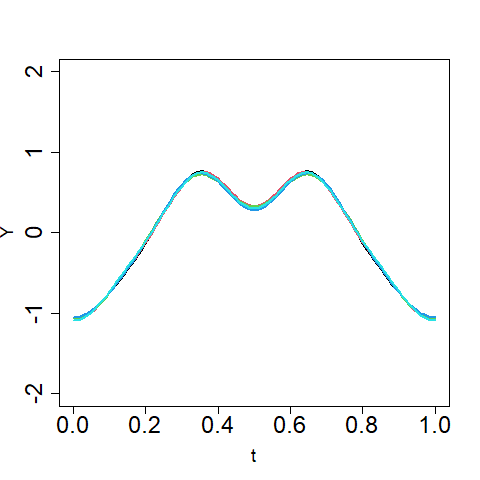} \\
          \includegraphics[scale=.16, align=c]{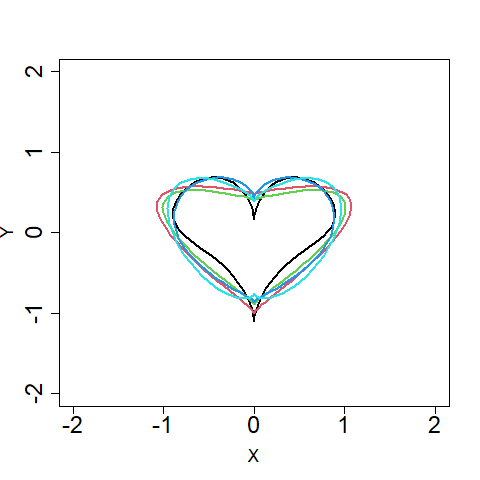} &  \includegraphics[scale=.16, align=c]{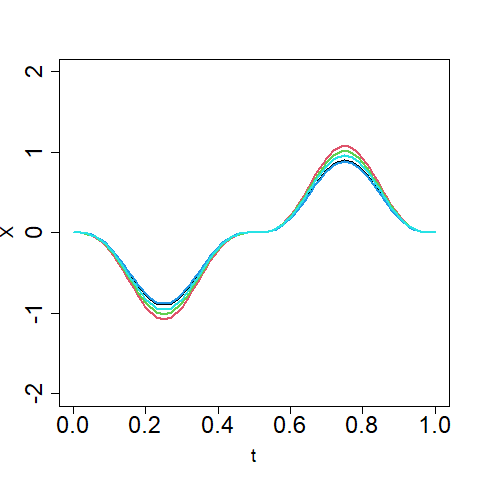} & \includegraphics[scale=.16, align=c]{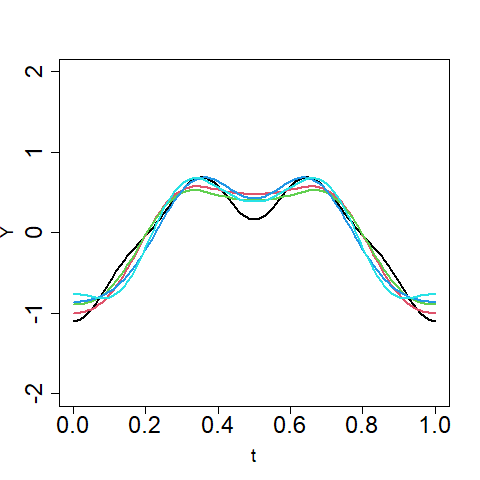} 
    \end{tabular}
    \caption{Template $\boldsymbol{\mu}$ and its coordinate functions (first row), and observations of simulated shapes and their coordinate functions for $\sigma=0.01$ (second row) and $\sigma=0.1$ (third row)
    %Plot of the template, ($\rho\mathbf{\tilde C}$ defined as equation \ref{hearth}  with coefficient set at their mean value)  and of its coordinate functions (first row), plot of $n=5$ realizations of $\mathbf{\tilde C}$ and of their coordinate functions, simulated with $\sigma=0.01$ (second row) and $\sigma=0.1$ (third row).
    }
    \label{mean_heart_sig}
\end{figure}

Figure~\ref{x-star-x} displays representative simulated pre-shapes $\mathbf{c}_i^*$ for $\sigma=0.1$, together with their coordinate functions.

% Figure~\ref{x-star-x} displays representative realizations of the deformed curves $\mathbf{C}_i^*$ and their associated coordinate functions for $\sigma=0.1$.

\begin{figure}[H]
    \centering
    \begin{tabular}{c c c c}
    \includegraphics[scale=0.16]{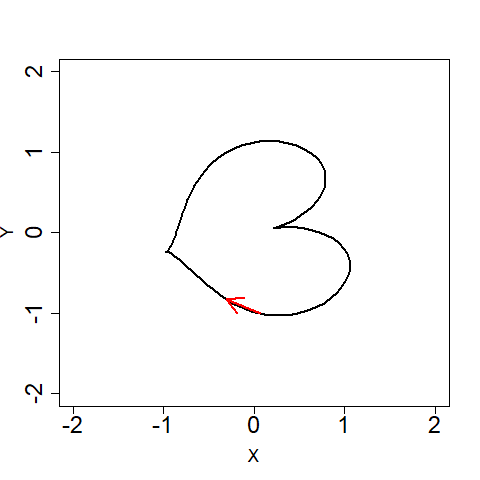}& 
     \includegraphics[scale=0.16]{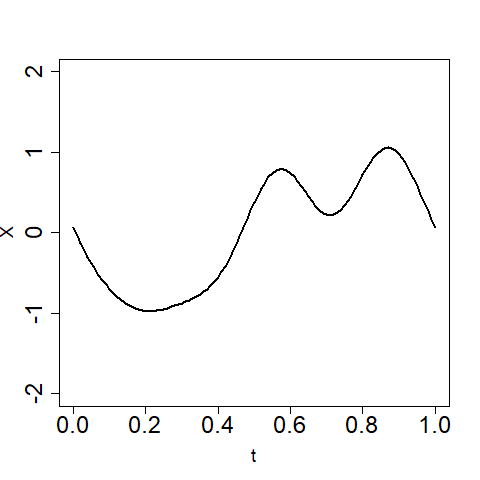} & 
     \includegraphics[scale=0.16]{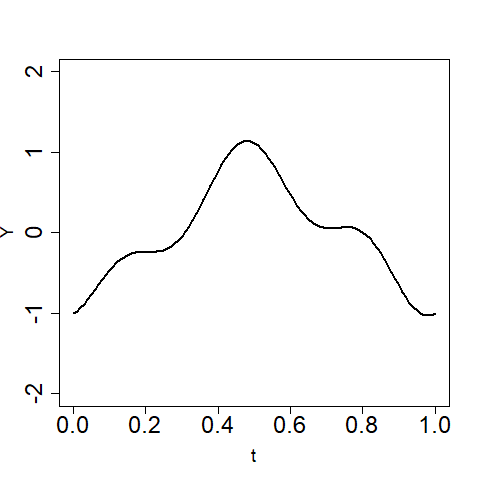}  \\
    \includegraphics[scale=0.16]{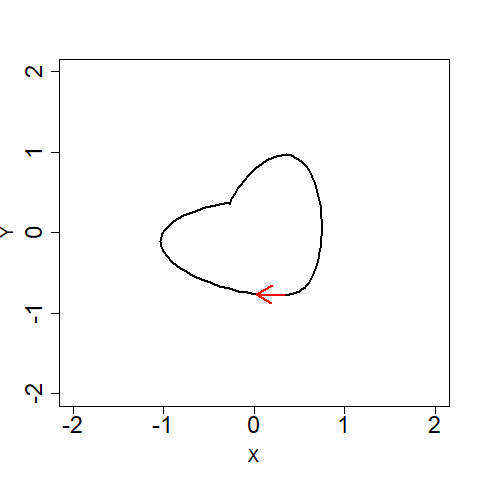}& 
     \includegraphics[scale=0.16]{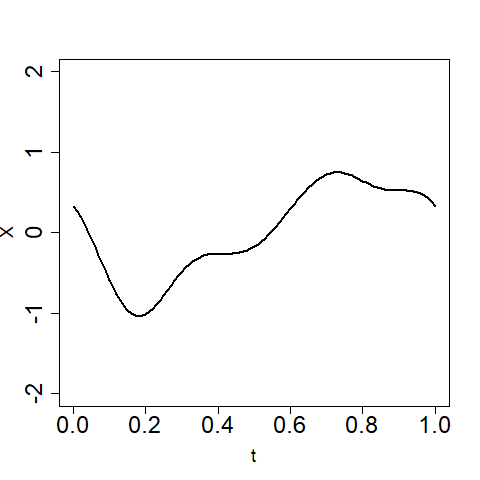} & 
     \includegraphics[scale=0.16]{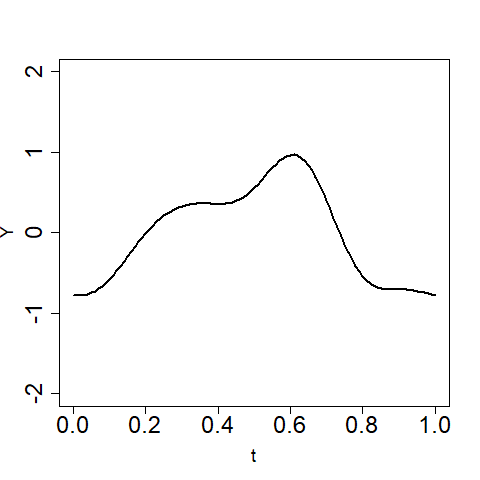}  \\ 
       \includegraphics[scale=0.16]{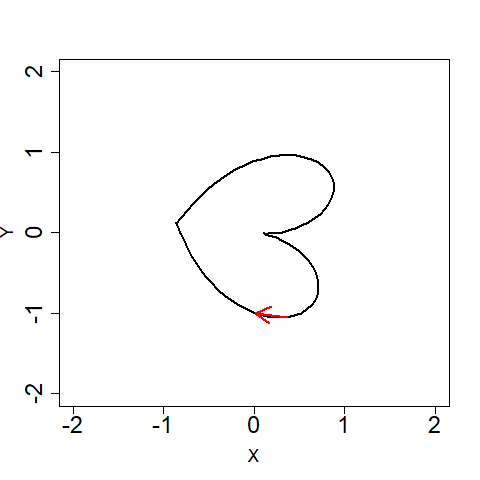}& 
     \includegraphics[scale=0.16]{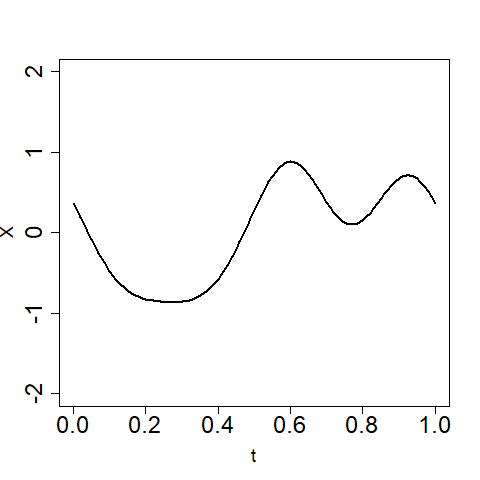} & 
     \includegraphics[scale=0.16]{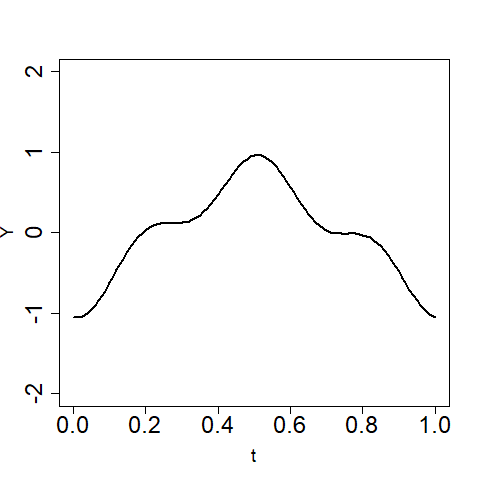}  
    \end{tabular}    
\caption{Observations of simulated pre-shapes and their coordinate functions for $\sigma=0.1$.
}
    \label{x-star-x}
\end{figure}

For each value of $\sigma$, we generate a dataset $\mathbf{c}_1^*, \ldots, \mathbf{c}_n^*$ of size $n=200$. Each pre-shape $\mathbf{c}_i^*$ is then discretized by sampling $J_i=J$ points $\{\mathbf{c}_i^*(t_{i,j})\}_{j=1}^{J}$, where the sampling locations $t_{i,j}$ are independently drawn from a discrete uniform distribution over $\{0.01, \ldots, 0.99\}$, with endpoints fixed at $t_{i,1}=0$ and $t_{i,J}=1$. To assess the impact of discretization, we consider $J \in \{15, 50, 90\}$.

These simulation settings allow us to evaluate the effect of both shape variability and sampling resolution on the performance of the alignment methods.

\subsection{Numerical results}

We apply the three alignment procedures, namely ICF, ICP, and SRVF-based alignment, to each simulated dataset. As discussed previously, the ICF method relies on a functional representation of the curves. Consequently, the discretized observations are first smoothed according to the reconstruction procedure described in Section~\ref{sec_iter_algo}. In all experiments, this smoothing step is performed using $M=10$ Fourier basis functions. We use in-house implementations for both the ICF and ICP algorithms, while the SRVF-based alignment is carried out using the \texttt{fdasrvf} package \cite{fdasrvf}, provided by the original authors of the method.

Figure~\ref{fig:comp} illustrates alignment results for one pre-shape observation $\mathbf{c}^*$ of the dataset simulated with $\sigma=0.1$. On top of the figure we display the template function, the observation $\mathbf{c}^*$ and the true reparametrization and shape functions, $\gamma$ and $\tilde{\mathbf{c}}$, associated with it. Results are shown for two sampling resolutions $J$. For each resolution, we display the discretized pre-shape, and the estimated reparametrization and shape functions, $\hat{\gamma}$ and $\hat{\tilde{\mathbf{c}}}$, produced by each method. The objective is to recover the correct starting point of the curve and align the resulting shape with the reference template $\boldsymbol{\mu}$.

\begin{figure}
    \centering
    \begin{tabular}{c c | c c  c }
    $\boldsymbol{\mu}$ &   \vline& $\mathbf{c}^*$ &$\gamma$ & $\tilde{\mathbf{c}}$  \\
    \includegraphics[align=c, width=0.15\linewidth]{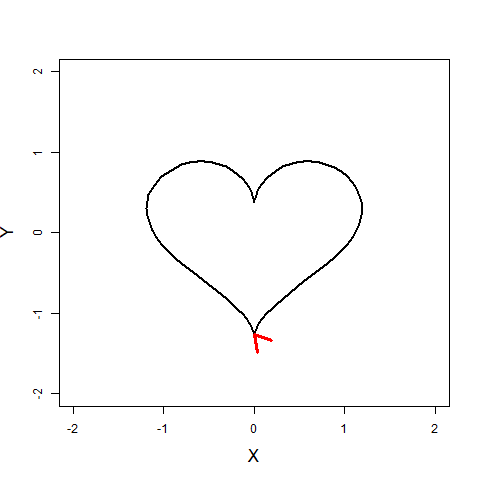} & \vline &  \includegraphics[align=c, width=0.15\linewidth]{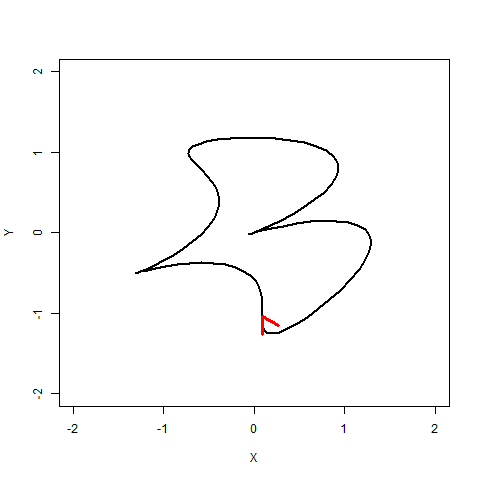} &\includegraphics[align=c, width=0.15\linewidth]{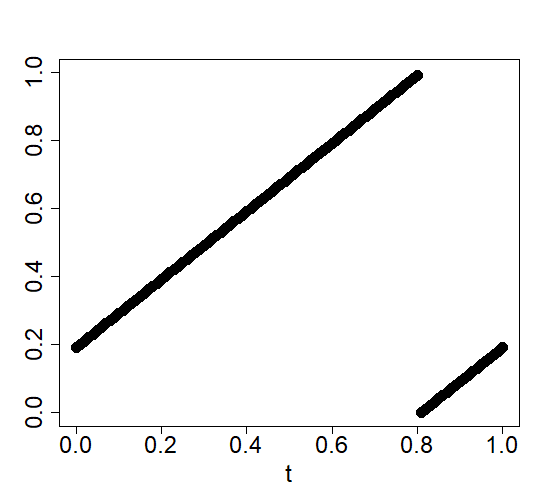} &  \includegraphics[align=c, width=0.15\linewidth]{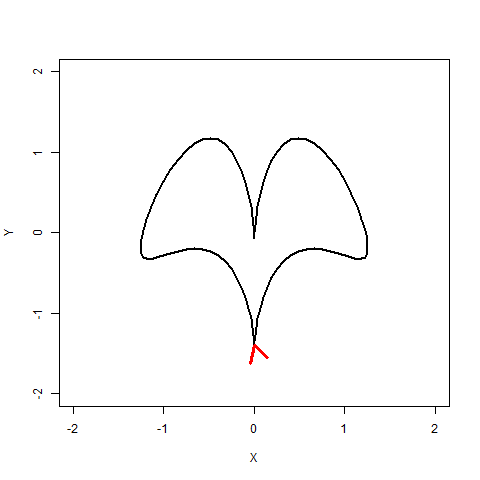} 
    \end{tabular}
    \begin{tabular}{c  c c |  c c }
    \hline
    \hline
& \multicolumn{2}{c }{$\{\mathbf{c}^*(t_{j})\}_{j=1}^{J}, \ J=15$} & \multicolumn{2}{c }{$\{\mathbf{c}^*(t_{j})\}_{j=1}^{J}, \ J=90$} \\
               &\multicolumn{2}{c }{\includegraphics[align=c, width=0.15\linewidth]{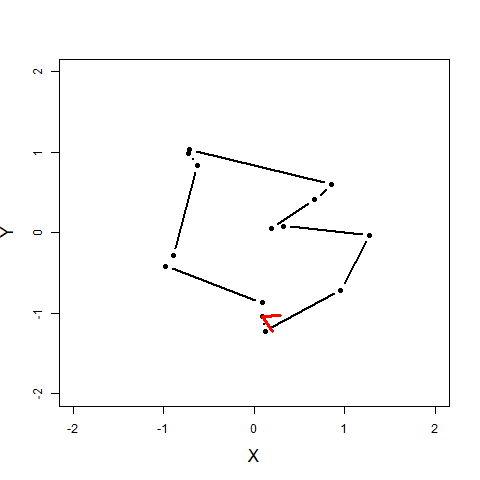}} &  \multicolumn{2}{c }{\includegraphics[align=c, width=0.15\linewidth]{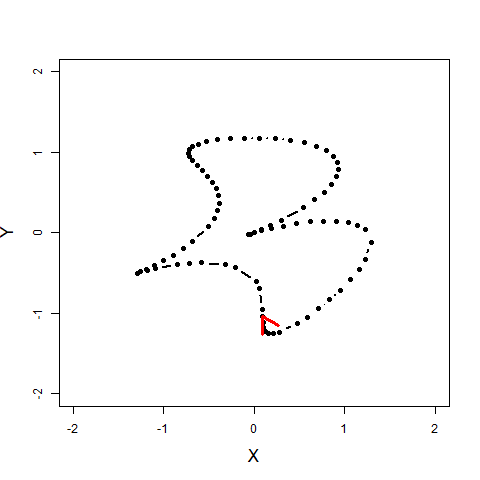}}  \\
             & $\hat{\gamma}$ &  $\hat{\tilde{\mathbf{c}}}$ & $\hat{\gamma}$ &  $\hat{\tilde{\mathbf{c}}}$ \\
            \text{ICF}  & \includegraphics[align=c, width=0.15\linewidth]{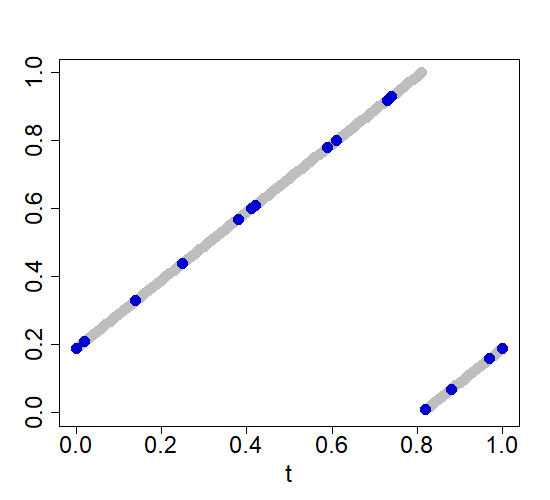} &  \includegraphics[align=c, width=0.15\linewidth]{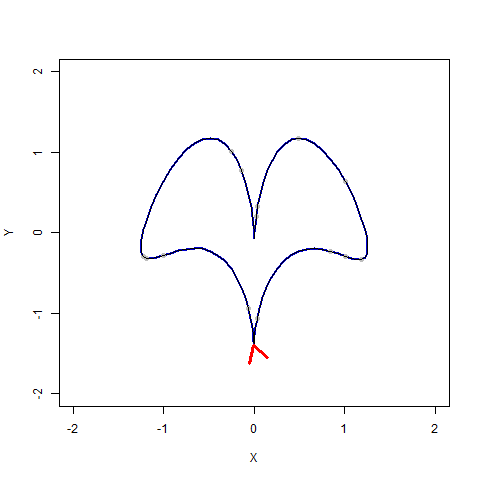} & \includegraphics[align=c, width=0.15\linewidth]{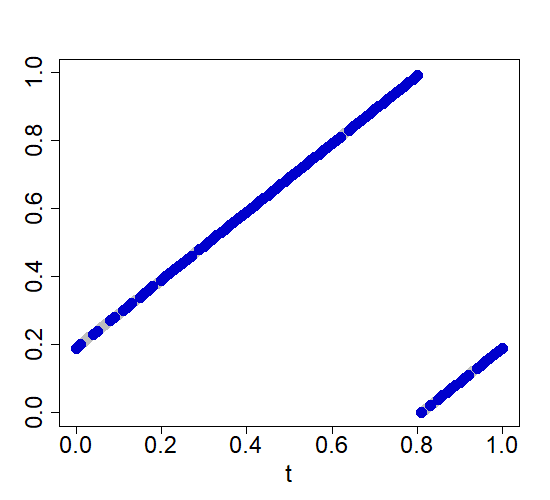} &\includegraphics[align=c, width=0.15\linewidth]{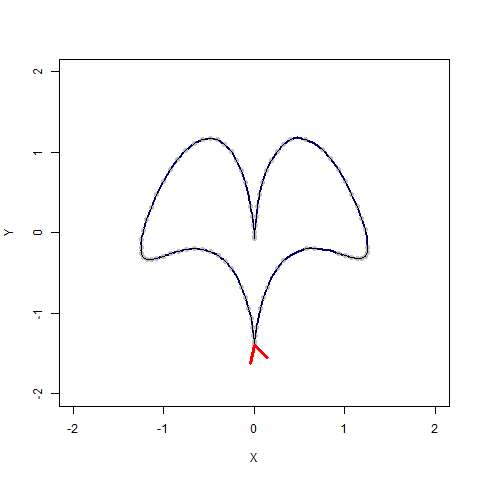} \\
            \text{ICP} &  

            \includegraphics[align=c, width=0.15\linewidth]{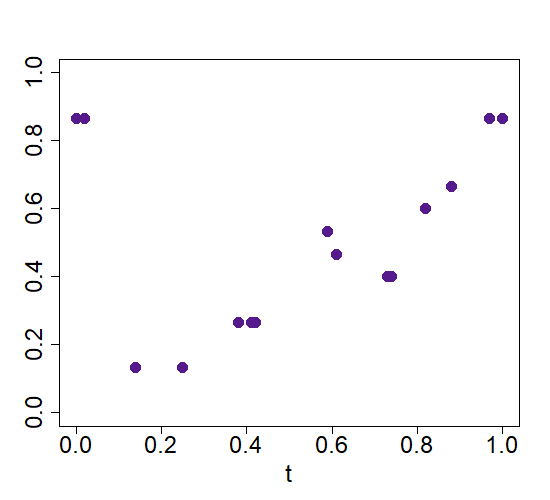} &  \includegraphics[align=c, width=0.15\linewidth]{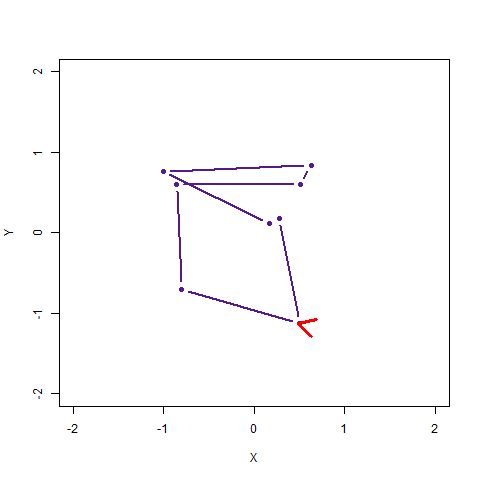} &  \includegraphics[align=c, width=0.15\linewidth]{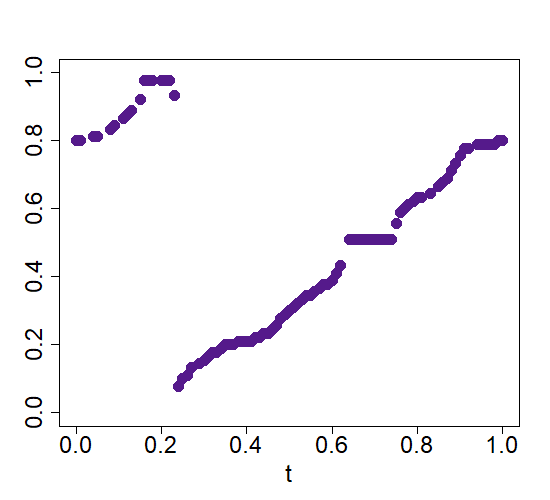} & \includegraphics[align=c, width=0.15\linewidth]{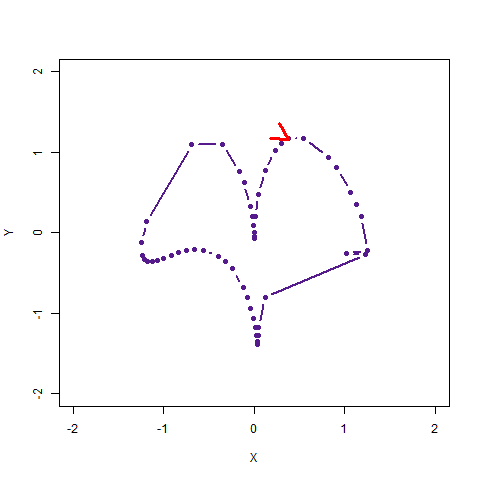}          
            \\
            \text{SRVF}  & 

        \includegraphics[align=c, width=0.15\linewidth]{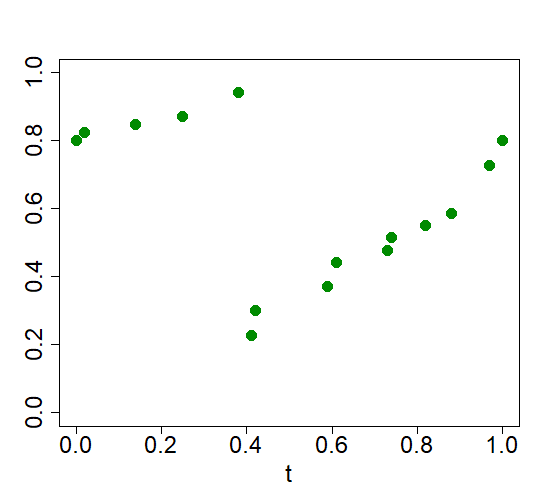} & \includegraphics[align=c, width=0.15\linewidth]{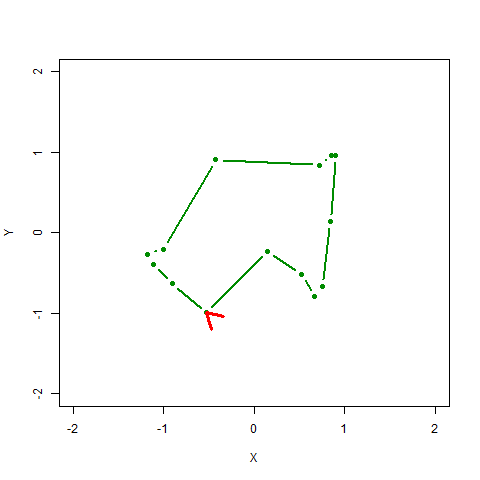}&  \includegraphics[align=c, width=0.15\linewidth]{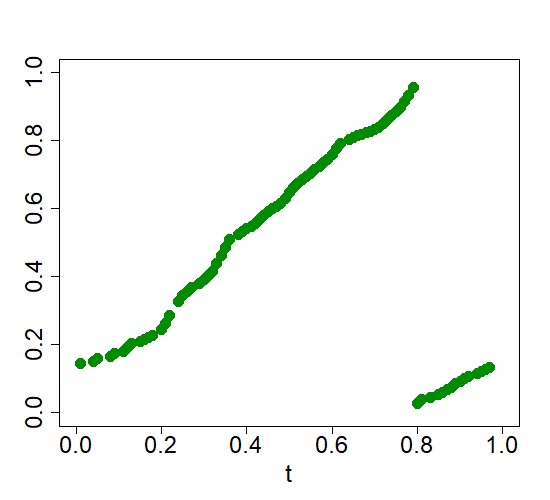} & \includegraphics[align=c, width=0.15\linewidth]{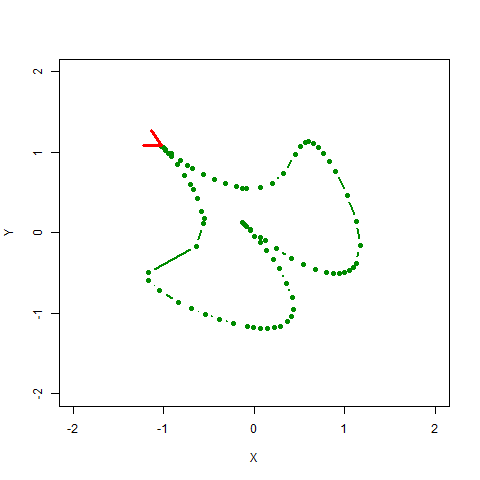}  
            
    \end{tabular}

\caption{Example of a simulated pre-shape curve $\mathbf{c}^*$ together with its associated reparametrization function $\gamma$ and shape function $\tilde{\mathbf{c}}$ (top row) under $\sigma=0.1$. Results are shown for two sampling resolutions $J$. For each alignment method, we display the estimated reparametrization function $\hat{\gamma}$ and the estimated shape $\hat{\tilde{\mathbf{c}}}$ obtained after aligning the discretized pre-shape observations $\{\mathbf{c}^*(t_j)\}_{j=1}^{J}$ (second row) to the template $\boldsymbol{\mu}$ (top left).}
    \label{fig:comp}
\end{figure}

%Figure~\ref{fig:comp} highlights the key differences between ICF and competing alignment methods. 

Several important differences between the methods can be observed from Figure~\ref{fig:comp}. First, ICP is designed for unordered point clouds and therefore does not preserve the original parametrization of the curves. During the alignment process, point correspondences may be reassigned arbitrarily, which can substantially distort the underlying shape. This phenomenon is particularly visible when the number of sampled points is small ($J=15$).

%A limitation of the ICP algorithm is that it is designed for unordered point clouds and therefore does not preserve the original parametrization of the curves. Consequently, the correspondence between points can be arbitrarily reassigned during the alignment procedure, which may distort the underlying shape. This behavior is clearly visible when $n_s=15$ in Figure~\ref{fig:comp}.

The figure also highlights the additional flexibility provided of elastic alignment methods. Although this flexibility can be advantageous in general settings, it may lead the SRVF-based procedure to favor elastic deformations instead of correctly identifying the starting point $\delta$. In such cases, the shape can become artificially distorted rather than properly registered to the template. This effect is again more pronounced for small values of $J$.

Finally, we observe that the SRVF-based method may also struggle to accurately estimate the rotation parameter. This effect is present in both scenarios displayed in Figure~\ref{fig:comp}. While the exact cause of this behavior is not fully understood, it may be related to the joint estimation of rotation and reparametrization in elastic alignment frameworks, where these two sources of variability can partially compensate for one another.

To quantitatively assess the accuracy of the estimated deformation parameters, we compute the following mean squared errors (MSE). Since rotations are periodic, the error associated with $\theta$ is measured through its representation on the unit circle: 

%while accounting for the periodic nature of rotations, we compute the following \textit{mean squared errors (MSE)}:
$$ 
\text{MSE}_\theta=\frac{1}{n} \sum_{i=1}^n \norm{  \begin{pmatrix}
    \cos(\theta_i)\\ 
    \sin(\theta_i)
\end{pmatrix}-\begin{pmatrix}
    \cos(\hat{\theta}_i)\\ 
    \sin(\hat{\theta}_i)
\end{pmatrix} }_2^2,
$$
while the error associated with the reparametrization function is defined as
$$ 
\text{MSE}_{\gamma} = \frac{1}{n} \sum_{i=1}^{n} \frac{1}{J}\sum_{j=1}^{J} \left( {\gamma}_{\delta_{i}}(t_{i,j})-\hat{\gamma}_{i}(t_{i,j})\right)^2,
$$
where $\hat{\gamma}_i$ denotes the estimated reparametrization function for the $i$-th curve. 

\begin{table}[ht]
    \centering
      \begin{tabular}{c c | c c c | c c c c }
    $\sigma$ & $J$ & \multicolumn{3}{c}{\text{MSE}$_\gamma$ } & \multicolumn{3}{c}{\text{MSE}$_\theta$ }\\
    & & ICF & ICP& SRVF & ICF & ICP& SRVF  \\ \hline
 0.01 & 15 & 0.03 & 0.15 & 0.14 & 0.40 & 1.15 & 1.97 \\ 
  0.01 & 50 & 0.03 & 0.18 & 0.04 & 0.31 & 0.59 & 2.01 \\ 
  0.01 & 90 & 0.04 & 0.17 & 0.02 & 0.46 & 0.31 & 2.04 \\ 
 \hline 
   0.10 & 15 & 0.03 & 0.15 & 0.15 & 0.32 & 1.28 & 1.75 \\ 
   0.10 & 50 & 0.02 & 0.18 & 0.05 & 0.26 & 0.36 & 1.65 \\ 
   0.10 & 90 & 0.03 & 0.18 & 0.04 & 0.35 & 0.37 & 1.58 
    \end{tabular}
    \caption{Mean squared errors associated with the estimation of the reparametrization function $\gamma$ and the rotation parameter $\theta$ for different values of $\sigma$ and $J$, obtained with the three alignment methods.}
    \label{res-sim}
\end{table}

%Table~\ref{res-sim} reports the MSE values for the estimation of the reparametrization function $\gamma$ and the rotation parameter $\theta$ across all simulation settings.

Table~\ref{res-sim} summarizes the MSE values obtained under all simulation scenarios. A clear difference between ICF and ICP can be observed across all settings. As already mentioned, while ICF explicitly preserves the ordering of points along the curve, ICP operates on unordered point clouds and may alter point correspondences during alignment. This lack of structural constraint leads to substantially larger errors for ICP, particularly for the estimation of $\gamma$ and, when $J$ is small, for the estimation of $\theta$ as well.
%significantly larger errors for ICP in the estimation { of $\gamma$ for general $n_s$, and $\theta$ when $n_s$ is small}.

The comparison between ICF and the SRVF-based alignment is more nuanced for the estimation of $\gamma$. When the sampling resolution is low ($J=15$), ICF clearly outperforms SRVF. As the number of sampled points increases, however, the performance of the SRVF method improves substantially and becomes comparable to, or slightly better than, that of ICF.

In contrast, a noticeably different behavior is observed for the estimation of the rotation parameter $\theta$. Across all scenarios, the SRVF-based method exhibits significantly larger errors than both ICF and ICP, with MSE values remaining consistently high regardless of the sampling resolution. This observation supports the qualitative behavior already noted in Figure~\ref{fig:comp}. As mentioned previously, this effect may stem from the joint estimation of rotation and reparametrization in elastic alignment frameworks, where the two deformation components can partially compensate for one another. However, a more thorough investigation would be required to precisely identify the origin of this effect.

%This behavior may be related to the joint estimation of rotation and reparametrization in elastic alignment frameworks, where these two sources of variability can partially compensate for each other. 

\section{Real data analysis} 
\label{app}

In this section, we analyze datasets from the \textit{MPEG-7} database\footnote{\url{https://dabi.temple.edu/external/shape/MPEG7/dataset.html}}, which has been widely used in statistical shape analysis (see, for example, \citet{larticle}, \citet{lelivre}, \citet{datasets}). This database contains binarized images of $101$ different object classes, each with approximately 20 instances. 

In our study, we focus on five representative classes: butterflies, forks, bats, horseshoes, and spoons. For clarity of presentation, we report detailed results for the butterfly and fork datasets, while the results for the remaining classes are provided in Appendix~\ref{appendix}. Figure~\ref{binary} displays sample images for these two classes. As can be observed, the images exhibit variability in rotation, scaling, and translation, leading to differences in orientation, size, and positioning.

\begin{figure}[H]
	\centering 
\begin{tabular}{ c c c c c }
	\includegraphics[align=c, width=2cm]{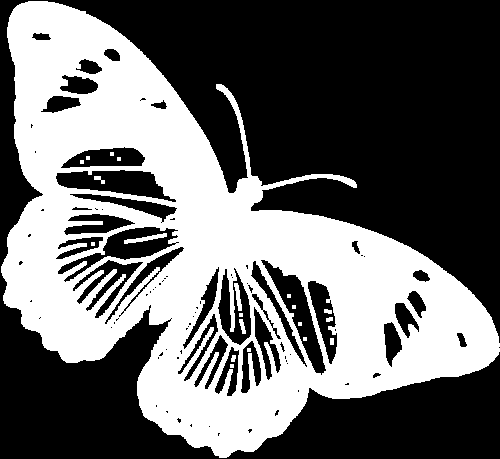}& \includegraphics[align=c, width=2cm]{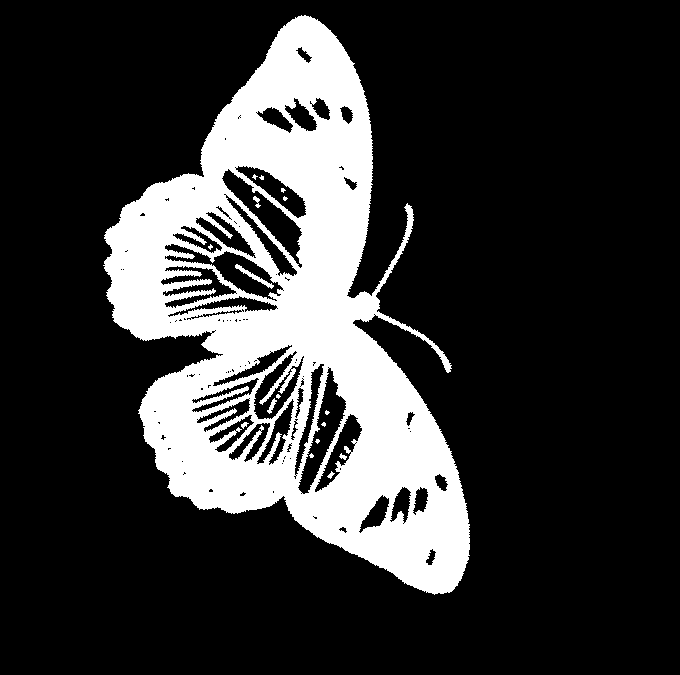} & \includegraphics[align=c, width=2cm]{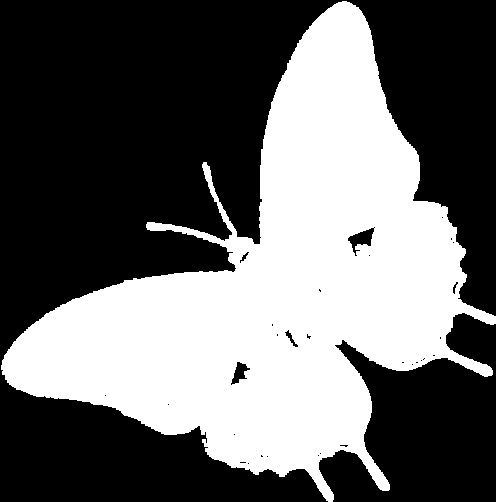}& \includegraphics[align=c, width=2cm]{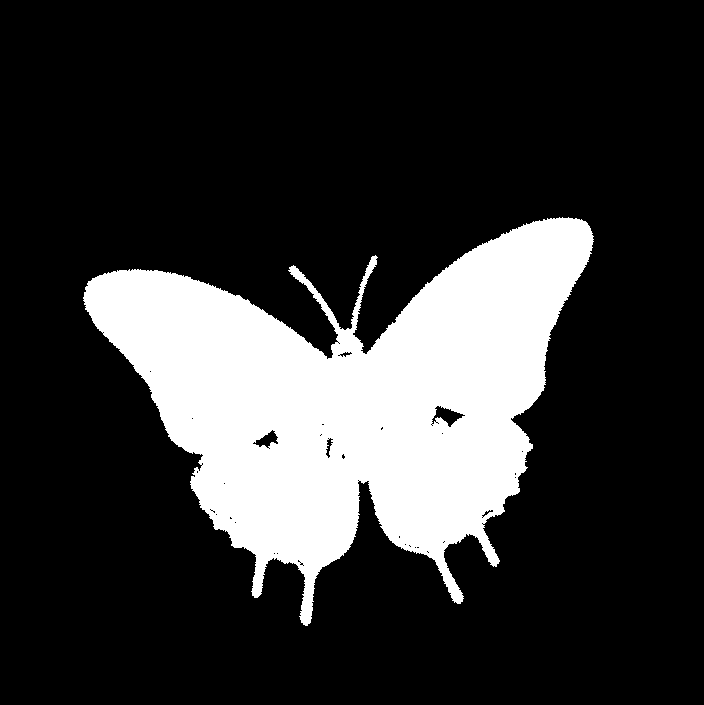}& \includegraphics[align=c, width=2cm]{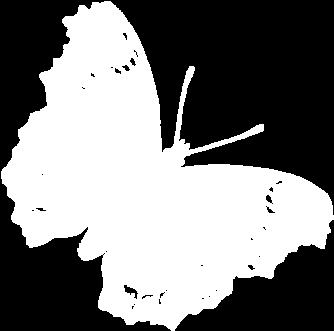} \\ 
	 \includegraphics[align=c, height=2cm]{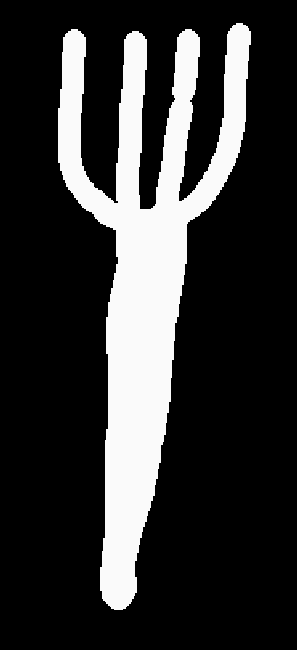}& \includegraphics[align=c, width=2cm]{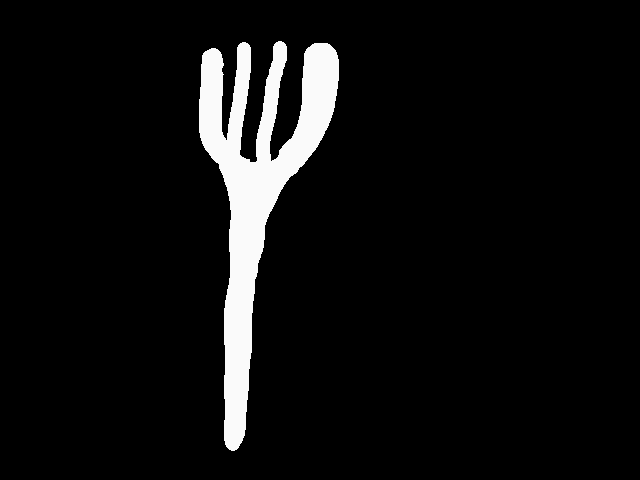} & \includegraphics[align=c, width=2cm]{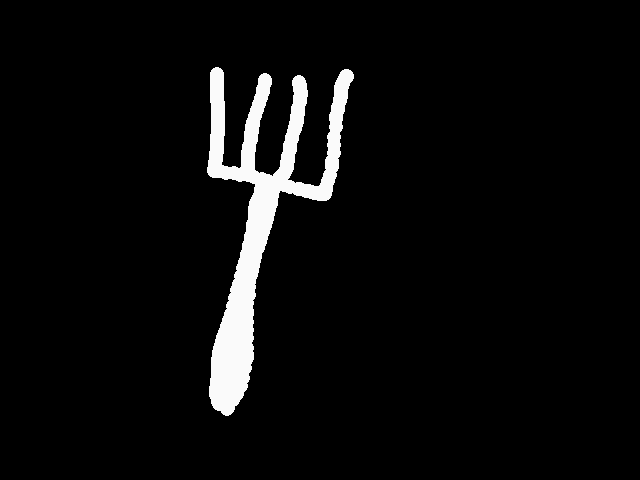}& \includegraphics[align=c, width=2cm]{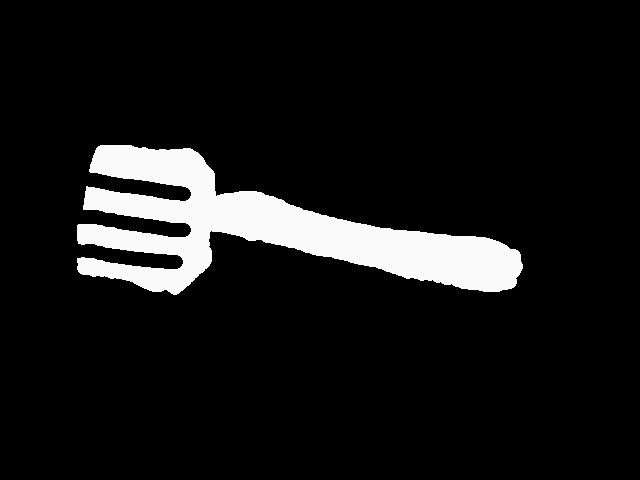}& \includegraphics[align=c, width=2cm]{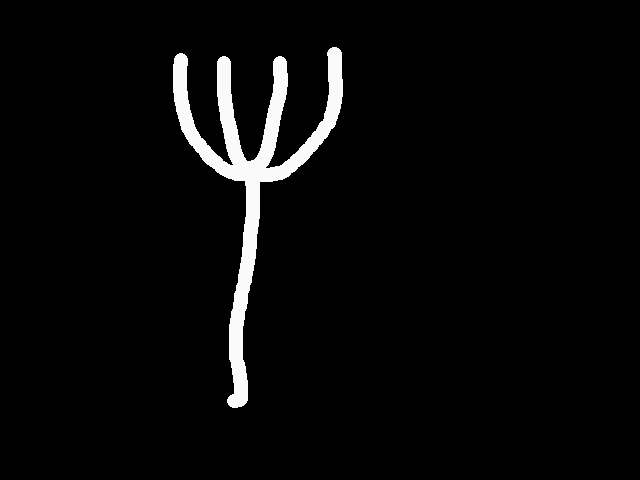}
\end{tabular}
\caption{Examples of images from the database for the butterfly and fork objects.}
\label{binary}
\end{figure}

Before applying our methodology, planar curves must be extracted from the binarized images. We perform this step using the OpenCV library in Python \citep{opencv_library}. Figure~\ref{curves} illustrates the resulting closed planar curves corresponding to the images in Figure~\ref{binary}. 

The extraction process introduces variability in the parametrization of the curves, as the starting points differ across observations, particularly for the fork shapes. In addition, the distribution of deformation variables appears to vary across object classes. For instance, butterfly shapes exhibit substantial variability in orientation, whereas fork shapes are mostly aligned in a similar direction. This suggests that the distribution of certain deformation parameters, such as rotation, may depend on the object class. More generally, this highlights that the deformations applied to a shape are not independent of the shape itself.

\begin{figure}[H]
	\centering 
   % \begin{tabular}{c}
\begin{tabular}{c}
\begin{tabular}{ c c c c c  }
	\includegraphics[align=c, height=2.5cm]{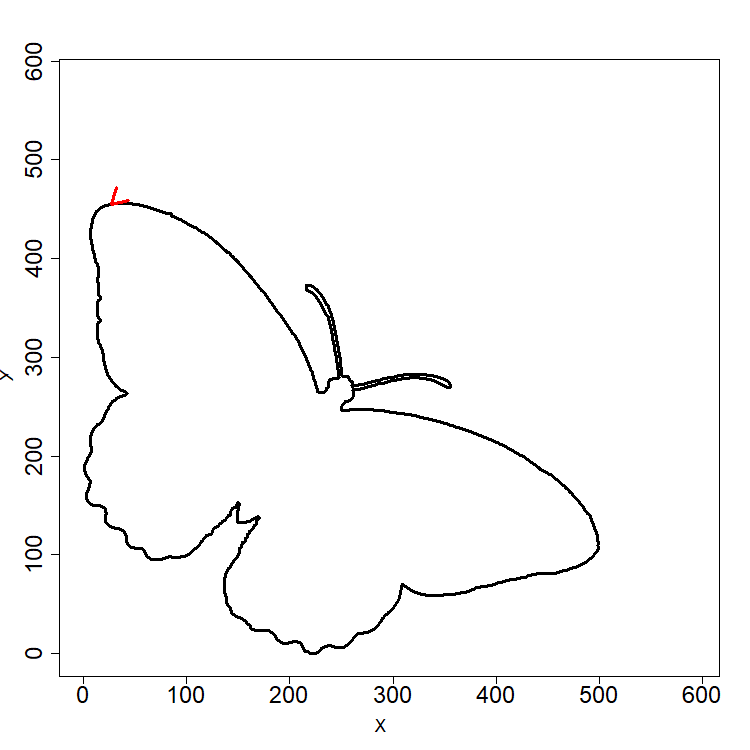}& \includegraphics[align=c, height=2.5cm]{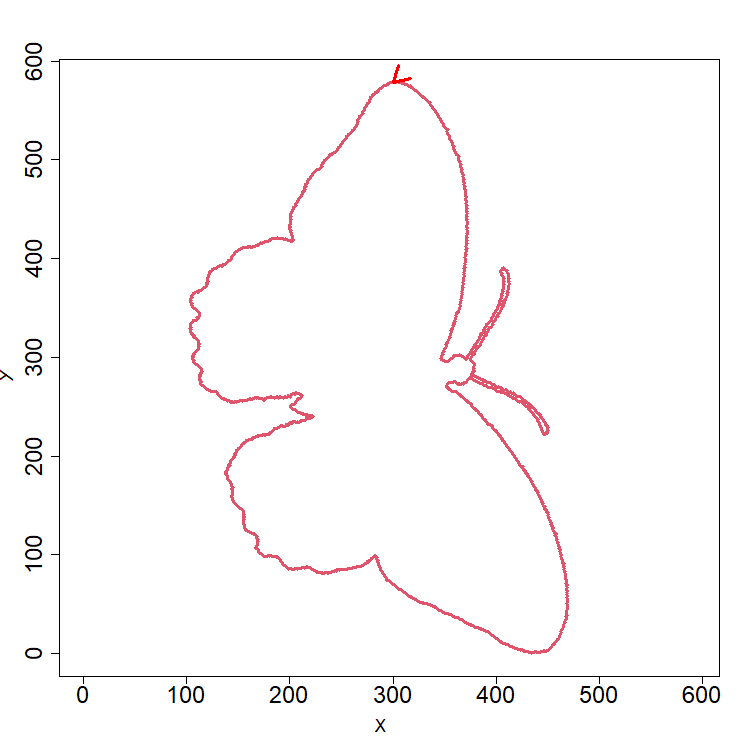} & \includegraphics[align=c, height=2.5cm]{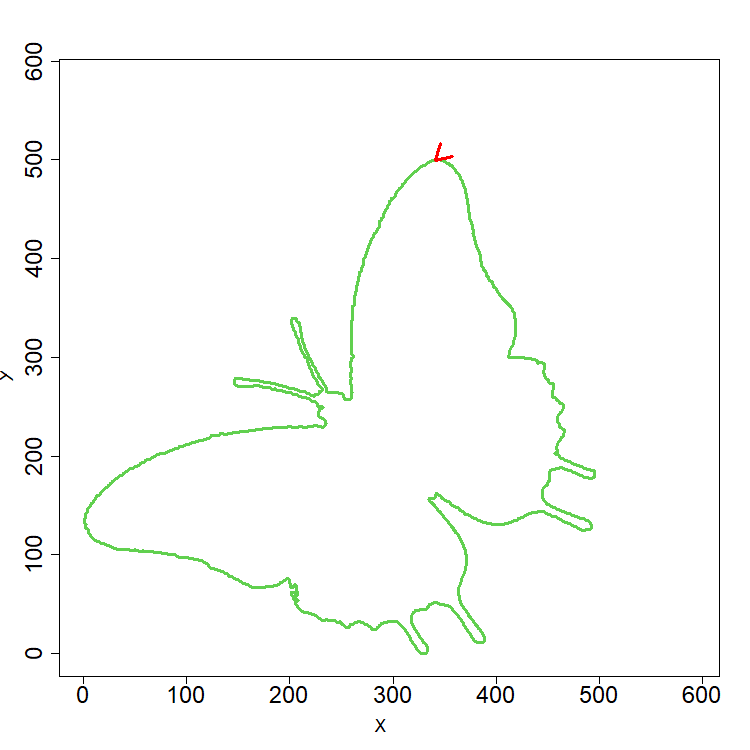}& 
    \includegraphics[align=c, height=2.5cm]{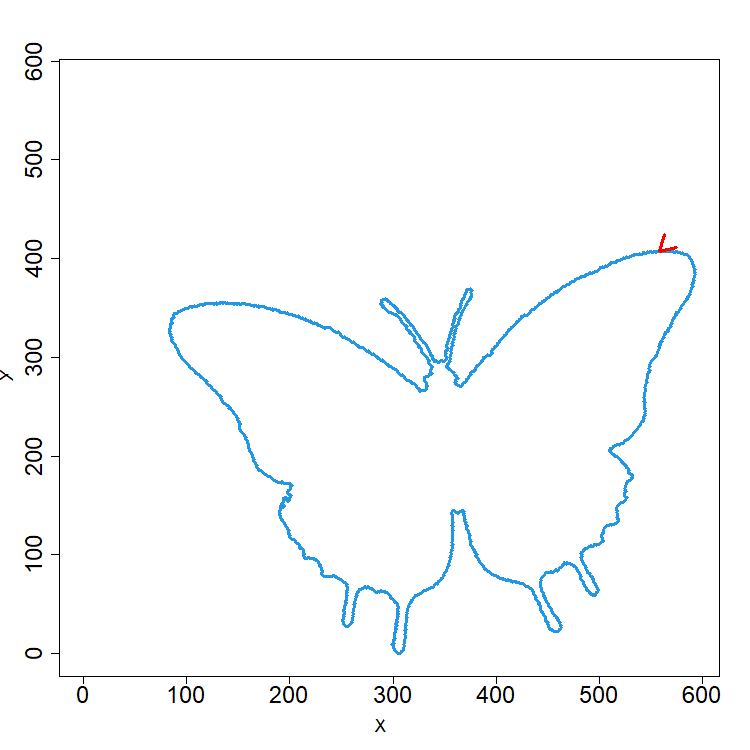}& \includegraphics[align=c, height=2.5cm]{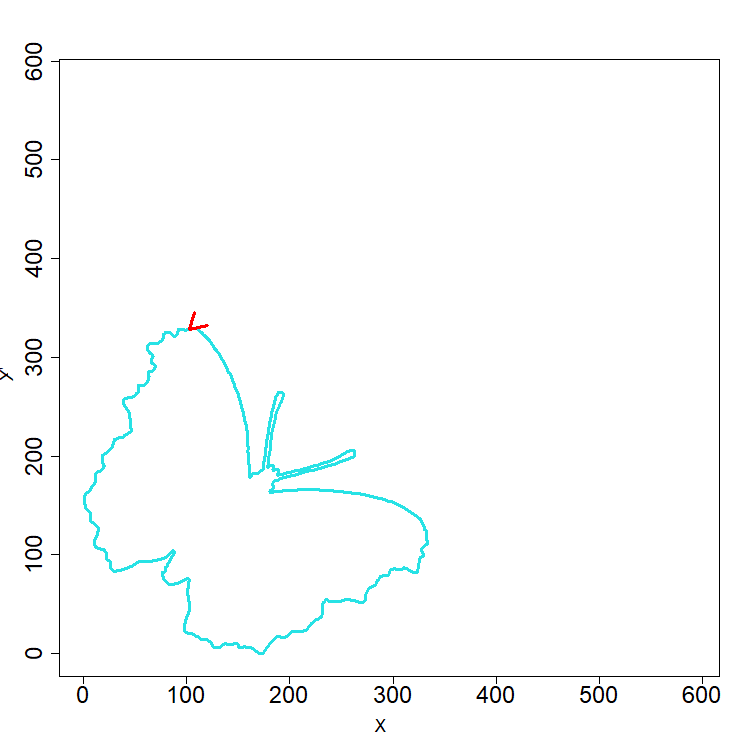} \\
\end{tabular} \\ 
\begin{tabular}{c c}
    \includegraphics[align=c, height=3cm]{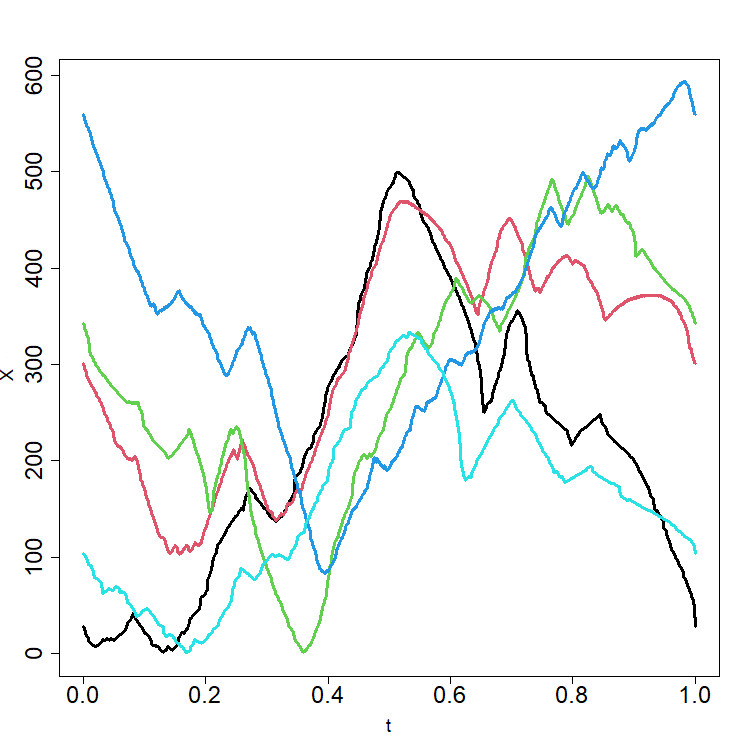} &
     \includegraphics[align=c, height=3cm]{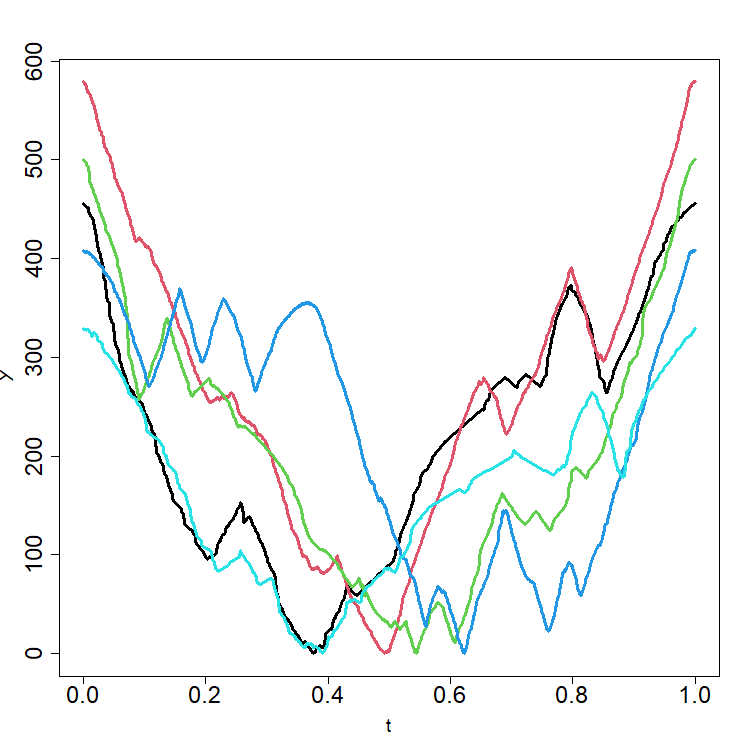}  
\end{tabular}
\\  \hline  
\begin{tabular}{ c c c c c  }
	\includegraphics[align=c, height=2.5cm]{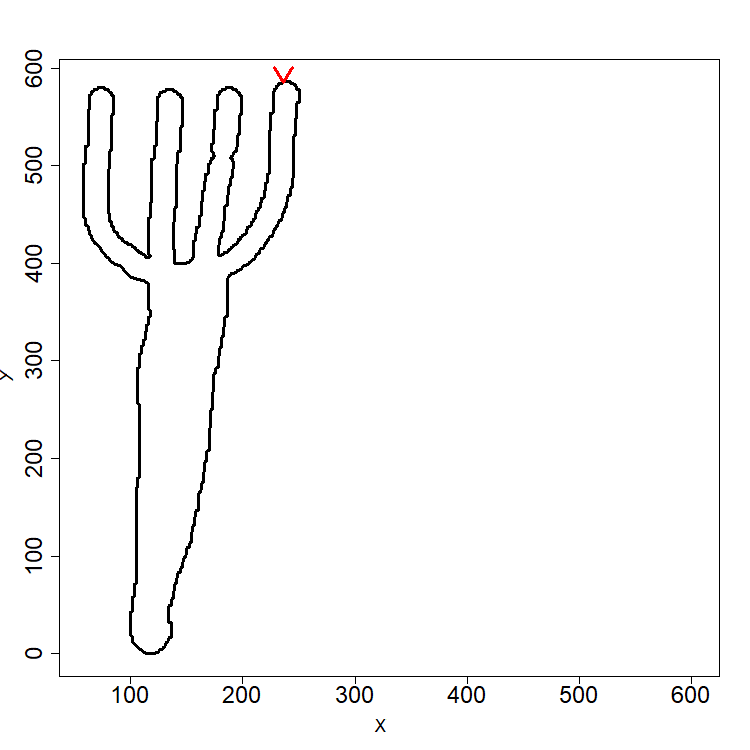}& \includegraphics[align=c, height=2.5cm]{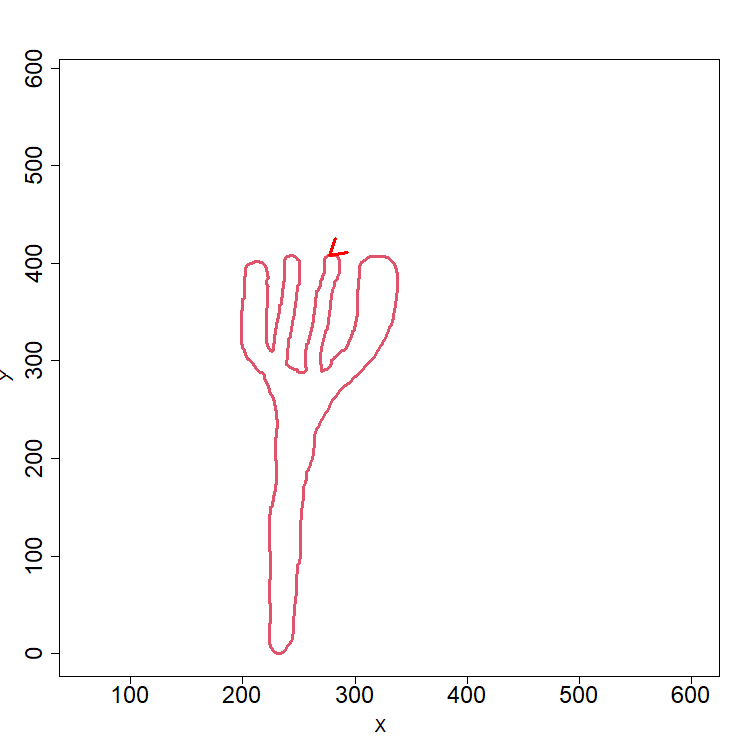} & \includegraphics[align=c, height=2.5cm]{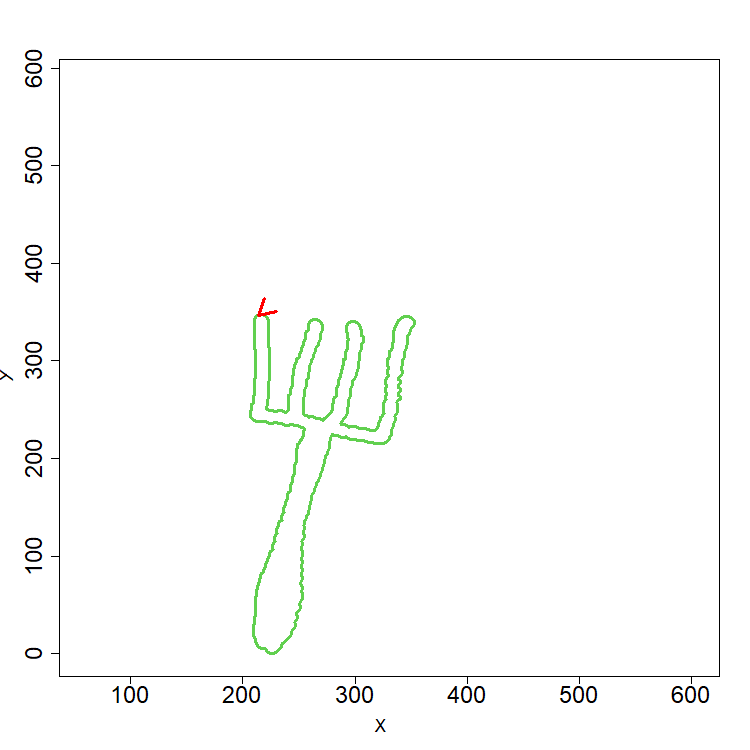}&   
    \includegraphics[align=c, height=2.5cm]{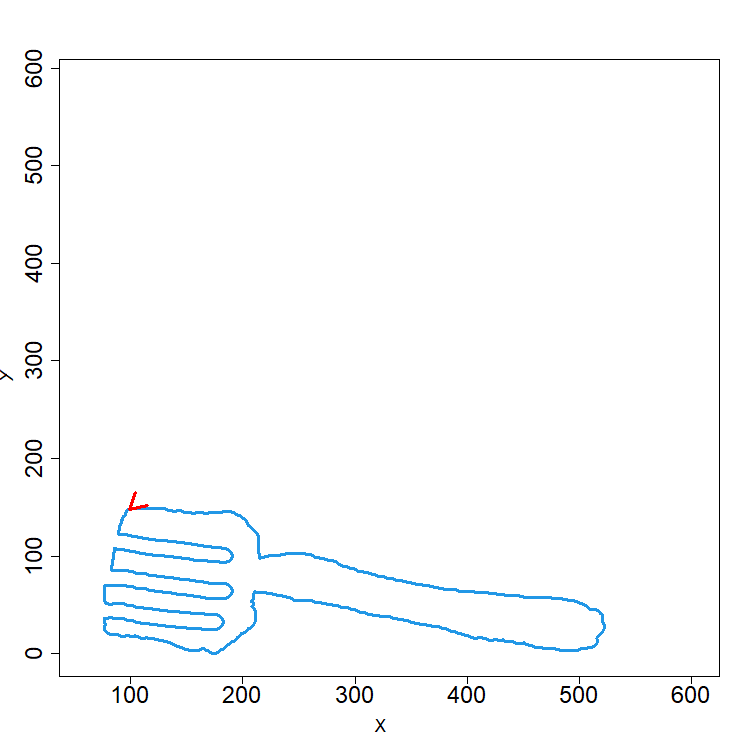}& \includegraphics[align=c, height=2.5cm]{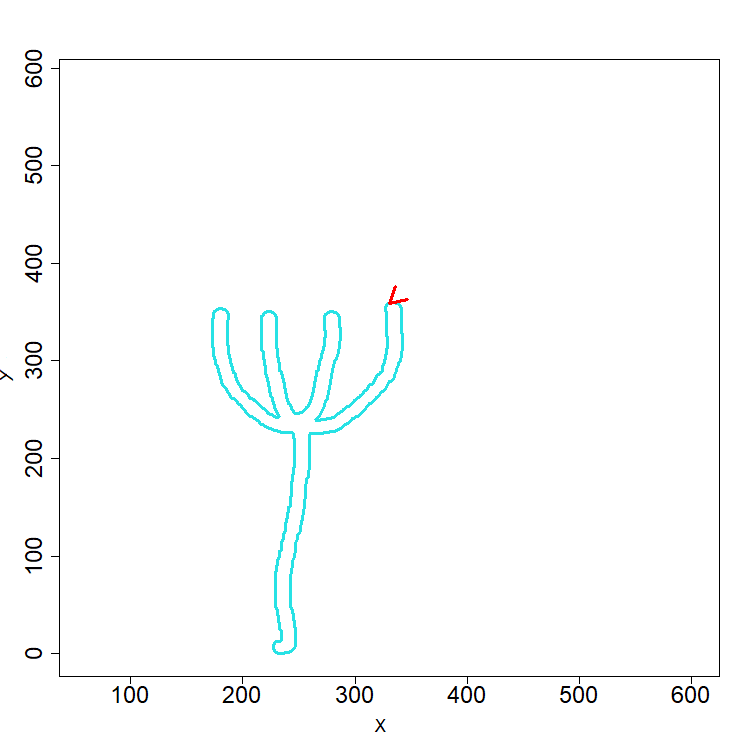} \end{tabular} \\ 
    \begin{tabular}{c c}
    \includegraphics[align=c, height=3cm]{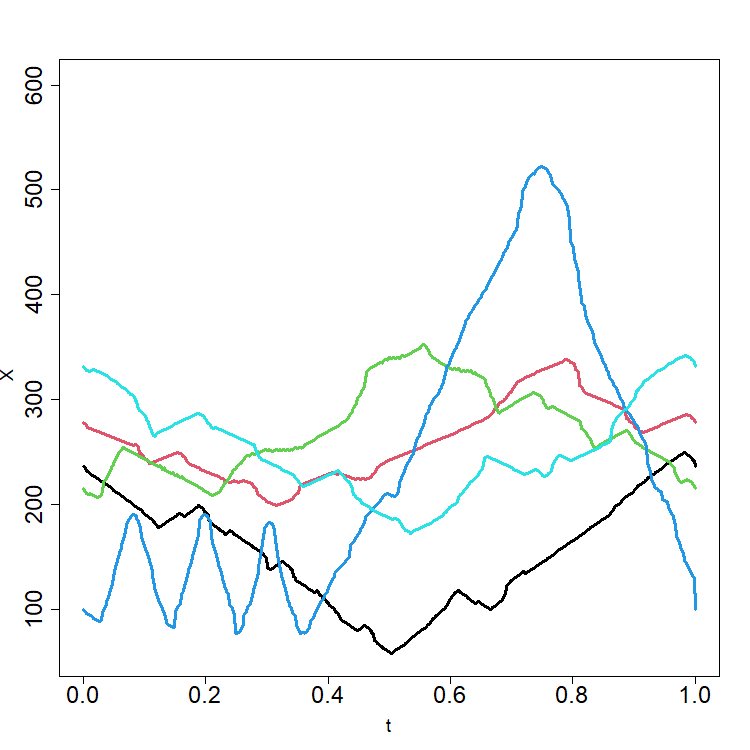}
     &  \includegraphics[align=c, height=3cm]{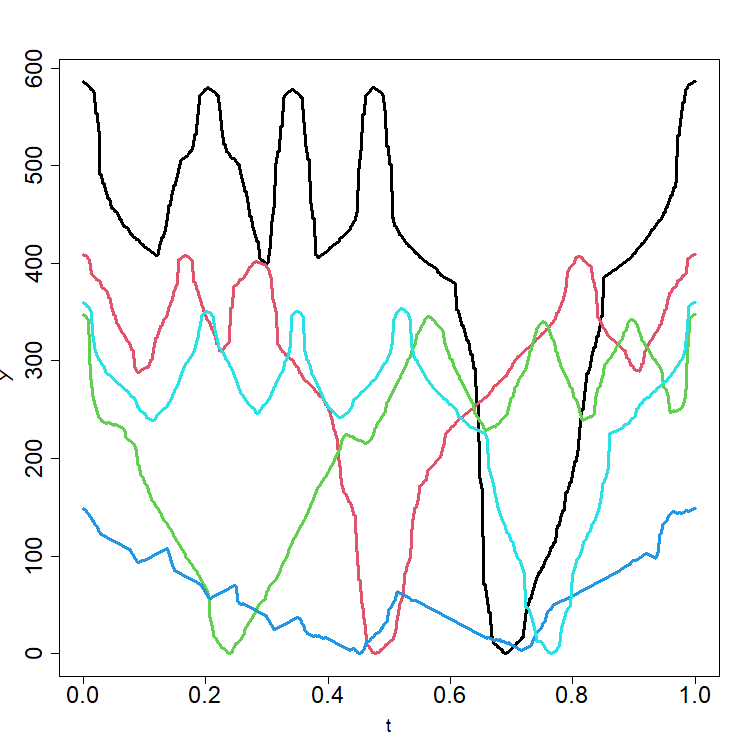} 
    \end{tabular}
\end{tabular}
\caption{Planar closed curves extracted from the binarized images of Figure \ref{binary} (first row of each block) with their associate coordinate functions (second row of each block).}
\label{curves}
\end{figure}

\subsection{Results from the alignment procedure}
\label{Resu_align}

Before aligning the planar curves, we first recover their functional form using a basis expansion with $M=50$ Fourier functions. Let $\hat{\mathbf{c}}^s_1,\ldots,\hat{\mathbf{c}}^s_{n_s}$ denote the resulting functions, where the index $s$ refers to the dataset: $s=b$ for butterflies and $s=f$ for forks. We then estimate the translation and scaling parameters for each curve, yielding the standardized shapes $\hat{\mathbf{c}}^{*s}_i, i=1,\ldots,n_s$, illustrated in Figure~\ref{alignement}.

To estimate the rotation and reparametrization parameters, we define a template $\boldsymbol{\mu}^s$. While a reference curve could be selected from the dataset, we instead use the Fréchet mean introduced in Section~\ref{estim_frechet_mean}. The full iterative procedure described in Section~\ref{sec_iter_algo} is then applied to jointly estimate the Fréchet mean and align the curves. This results in the estimated shape curves $\hat{\tilde{\mathbf{c}}}_i^s, i=1,\ldots,n_s$, also shown in Figure~\ref{alignement}. The coordinate functions appear well aligned, indicating that the proposed alignment procedure performs effectively. Moreover, despite the relatively small sample size, the estimated Fréchet means yield coherent and interpretable representative shapes. This is a non-trivial task in shape analysis  \cite{lelivre}, as it requires jointly accounting for geometric variability and deformation effects, and further illustrates the stability of the proposed estimation procedure.

A key feature of our approach is that it provides explicit estimates of the deformation parameters for each curve, enabling a direct statistical analysis of these quantities. The estimated deformation parameters for both datasets are displayed in Figure~\ref{est_def}. We observe clear differences in their distributions across object classes. For instance, in the butterfly dataset, the translation components exhibit a strong linear relationship, whereas in the fork dataset, the translation vectors are more tightly clustered around a central location. Similarly, the rotation parameters $\hat{\theta}_i$ and the reparametrization parameters show substantially greater variability in the butterfly dataset than in the fork dataset.

These observations are consistent with the visual variability seen in Figure~\ref{curves} and highlight that deformation patterns are strongly shape-dependent. Unlike approaches that treat deformations as nuisance parameters, our framework allows them to be explicitly modeled and interpreted. 

\begin{figure}[H]
    \centering
   \begin{tabular}{ c c c}
   $\hat{\boldsymbol{\mu}}^b$& \includegraphics[align=c, width=0.15\linewidth]{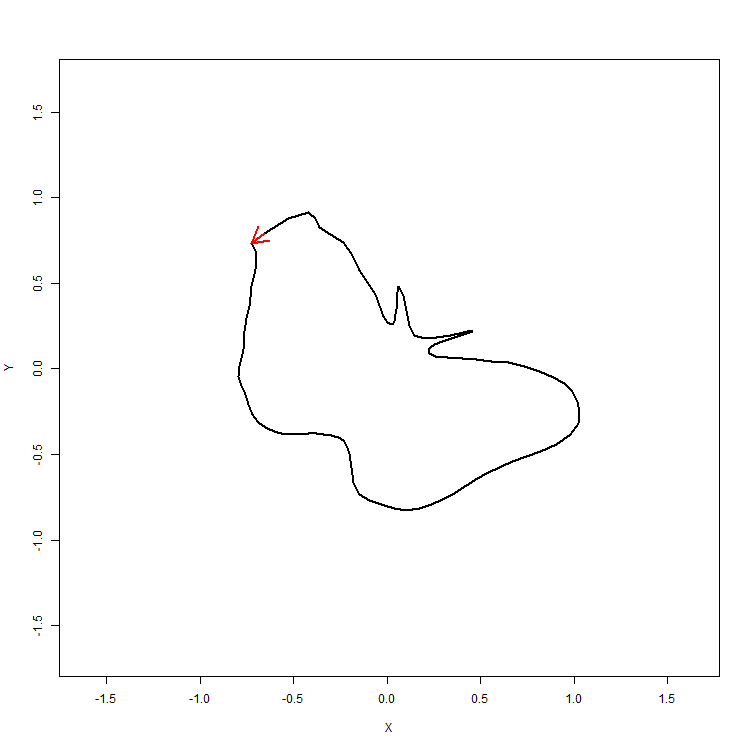}    &  \begin{tabular}{c c c c } 
    
    ${\hat{\mathbf{c}}^{*b}}$& 
       \includegraphics[align=c, width=0.15\linewidth]{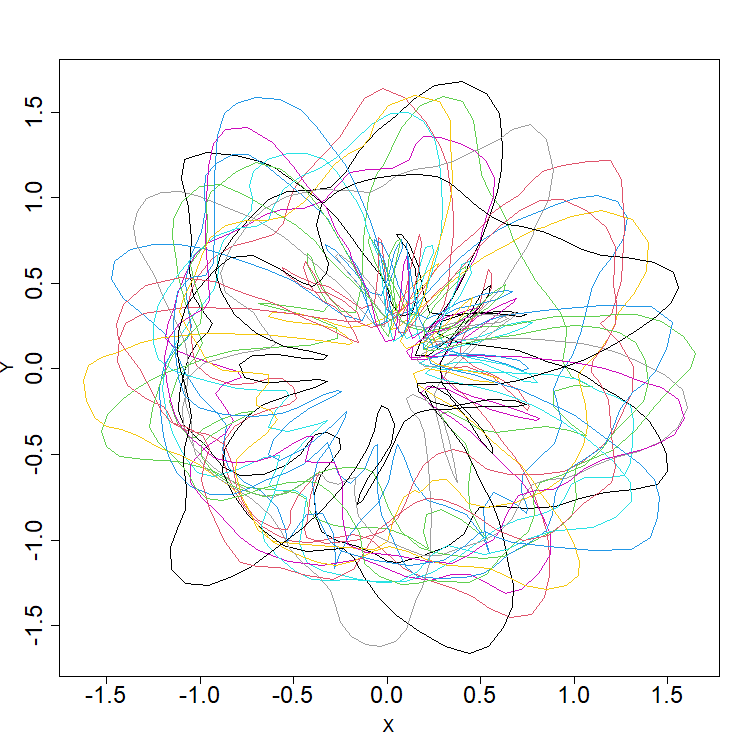} & 
          \includegraphics[align=c, width=0.15\linewidth]{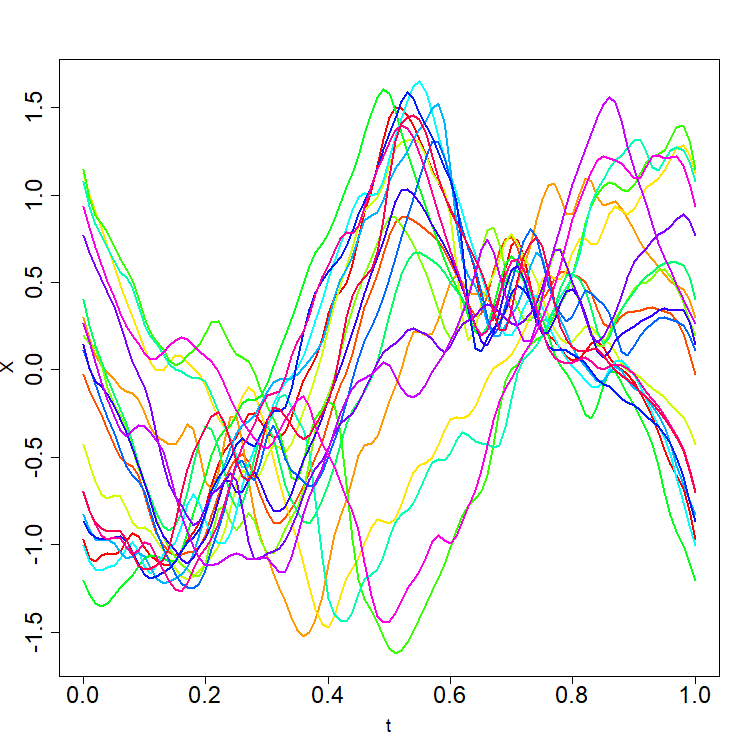} &
          \includegraphics[align=c, width=0.15\linewidth]{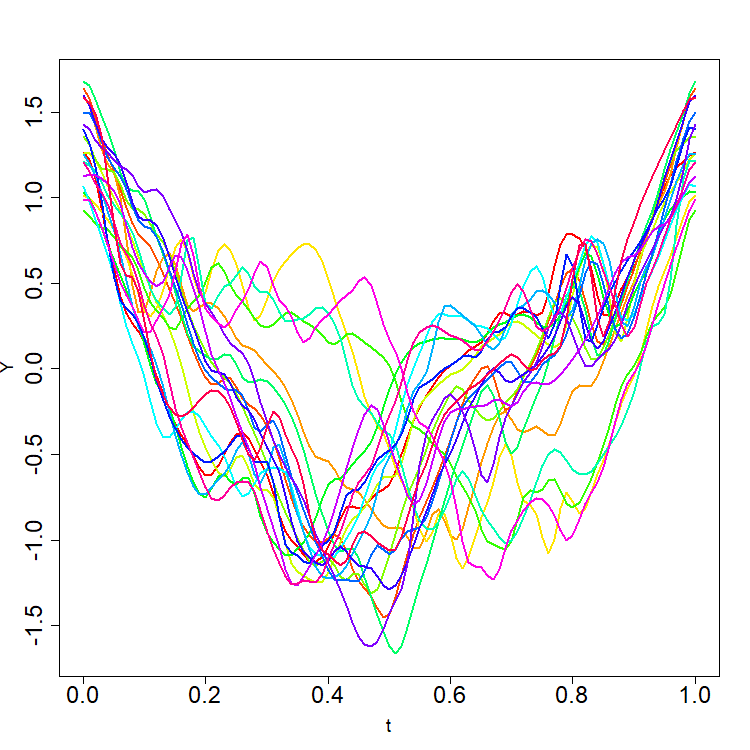} \\
      $\hat{\tilde {\mathbf{c}}}^b$& \includegraphics[align=c, width=0.15\linewidth]{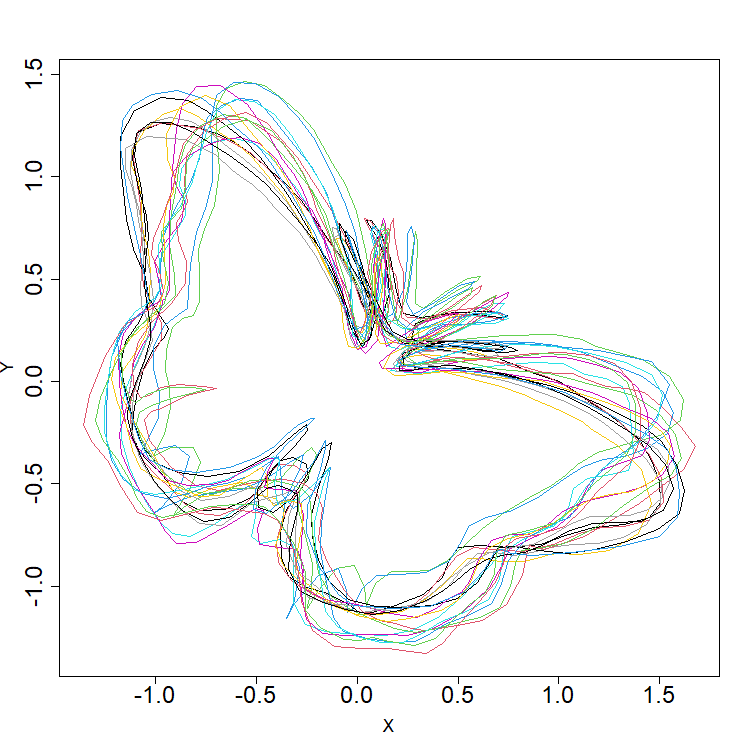} & 
          \includegraphics[align=c, width=0.15\linewidth]{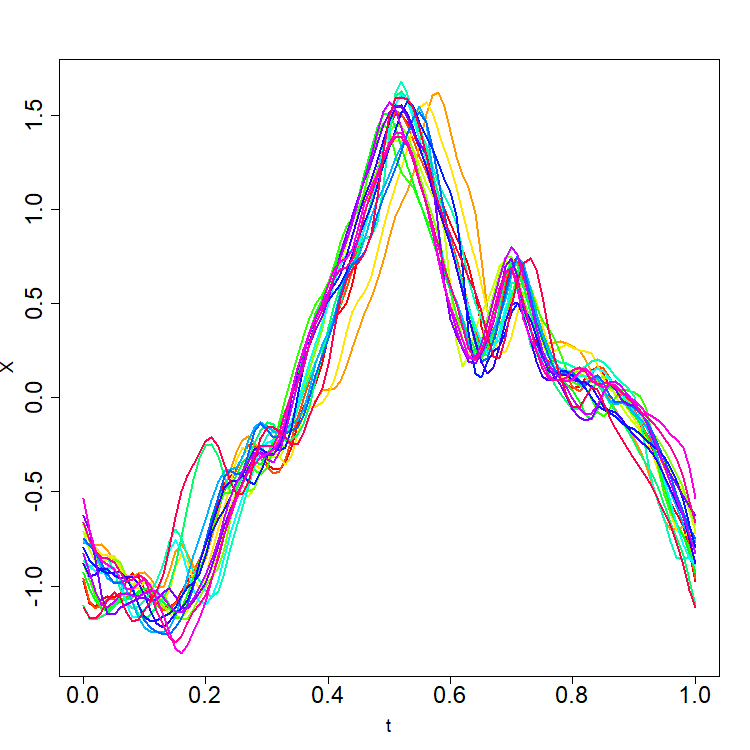} &
          \includegraphics[align=c, width=0.15\linewidth]{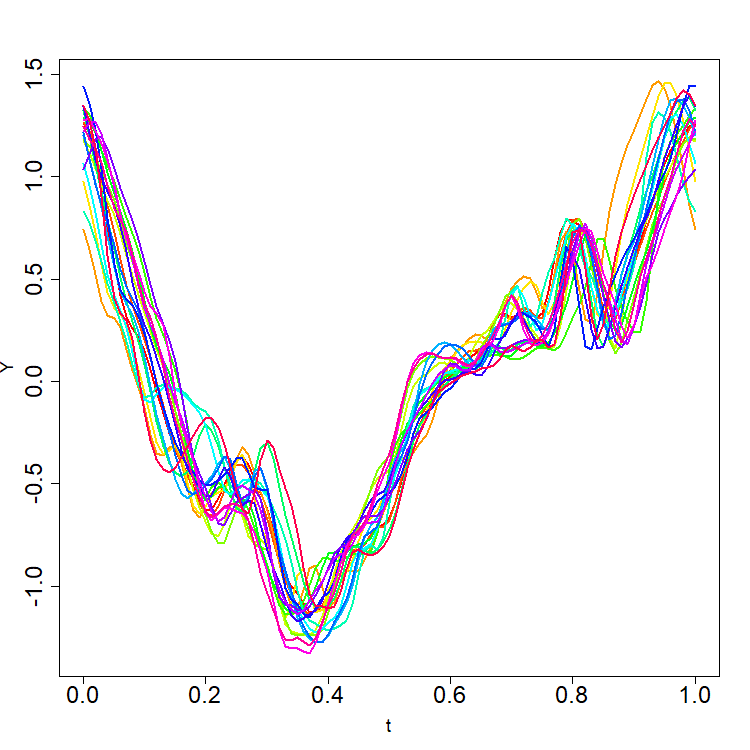}    
   \end{tabular} \\ 
   \hline \\ 
   
   $\hat{\boldsymbol{\mu}}^f$& \includegraphics[align=c, width=0.15\linewidth]{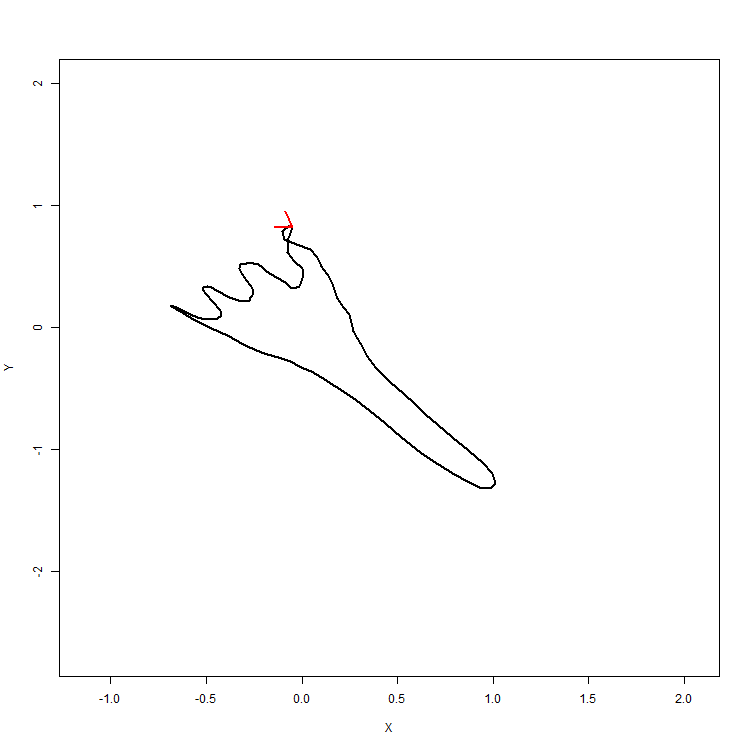}    &  \begin{tabular}{c c c c } 
    ${\hat{\mathbf{c}}^{*f}}$&  \includegraphics[align=c, width=0.15\linewidth]{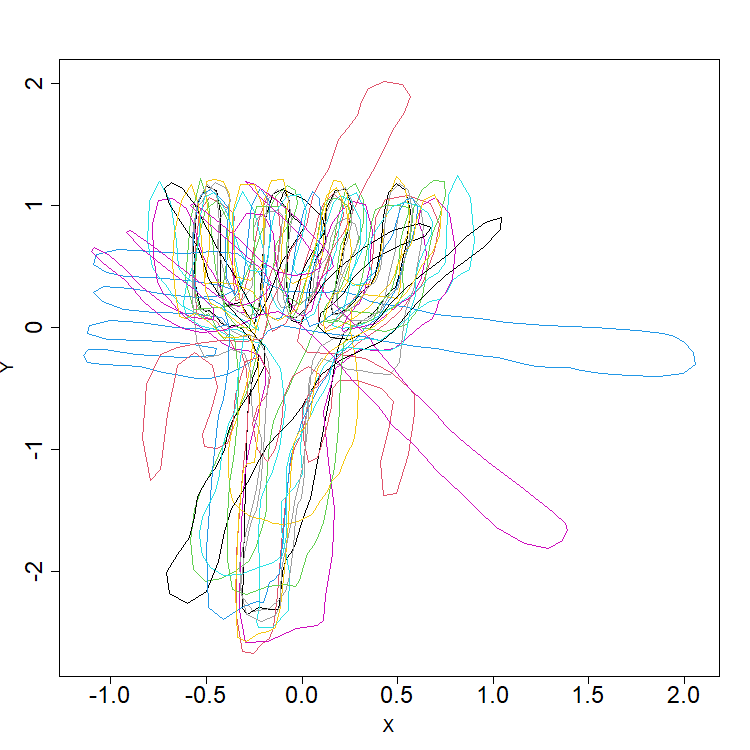} & 
          \includegraphics[align=c, width=0.15\linewidth]{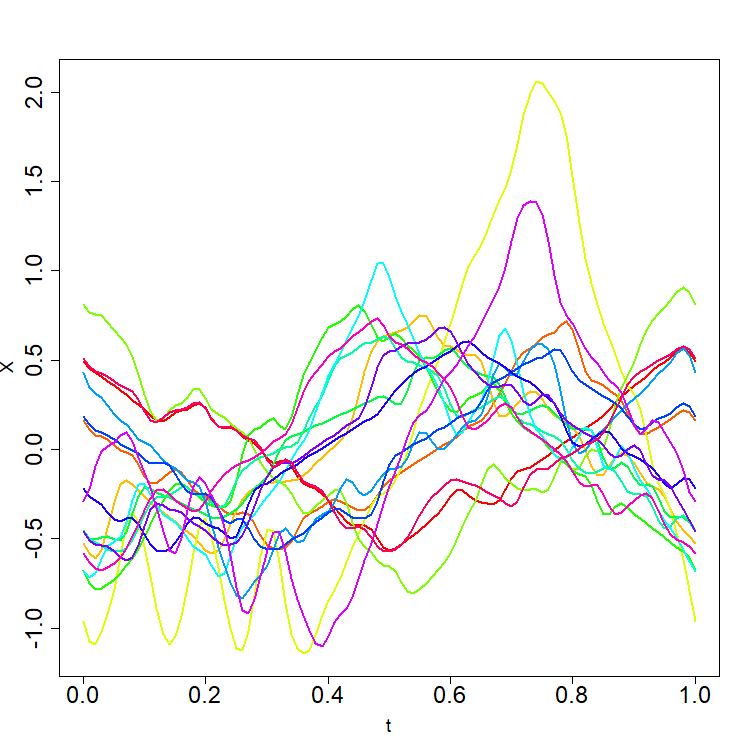} &
          \includegraphics[align=c, width=0.15\linewidth]{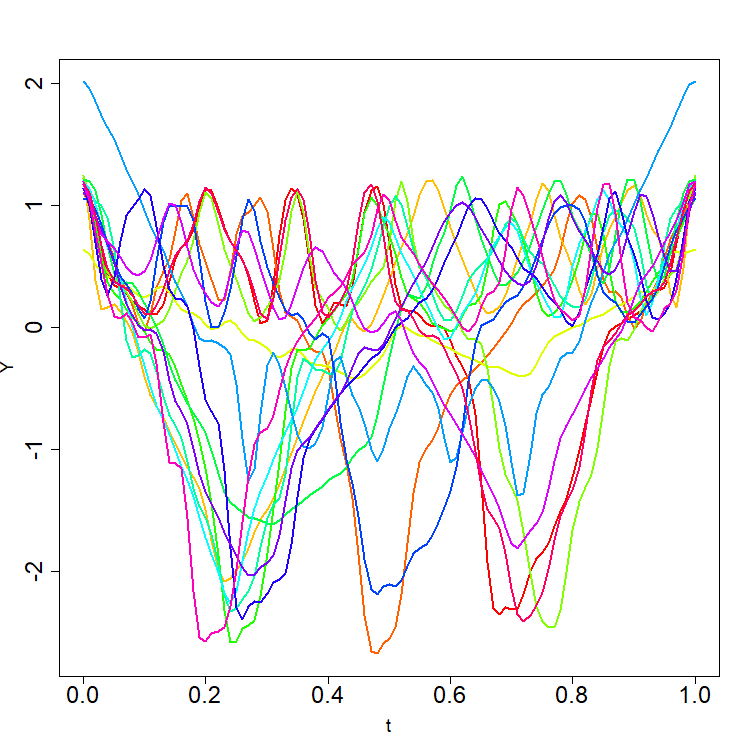} \\
      $\hat{\tilde {\mathbf{c}}}^f$& \includegraphics[align=c, width=0.15\linewidth]{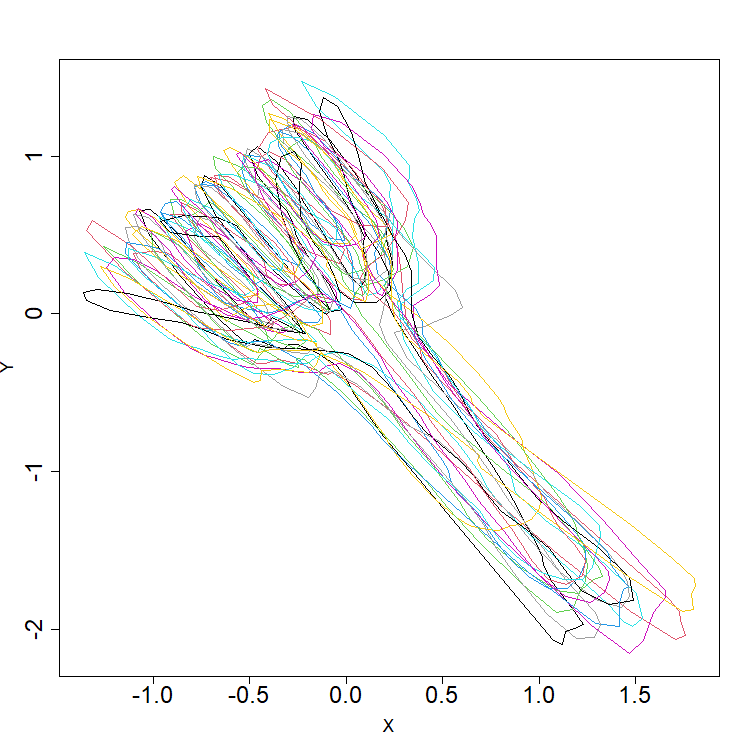} & 
          \includegraphics[align=c, width=0.15\linewidth]{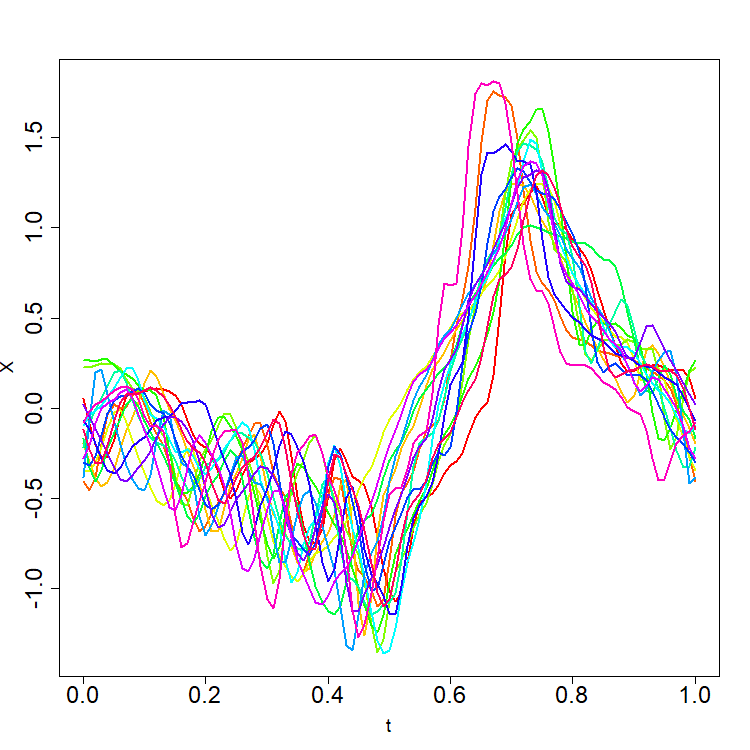} &
          \includegraphics[align=c, width=0.15\linewidth]{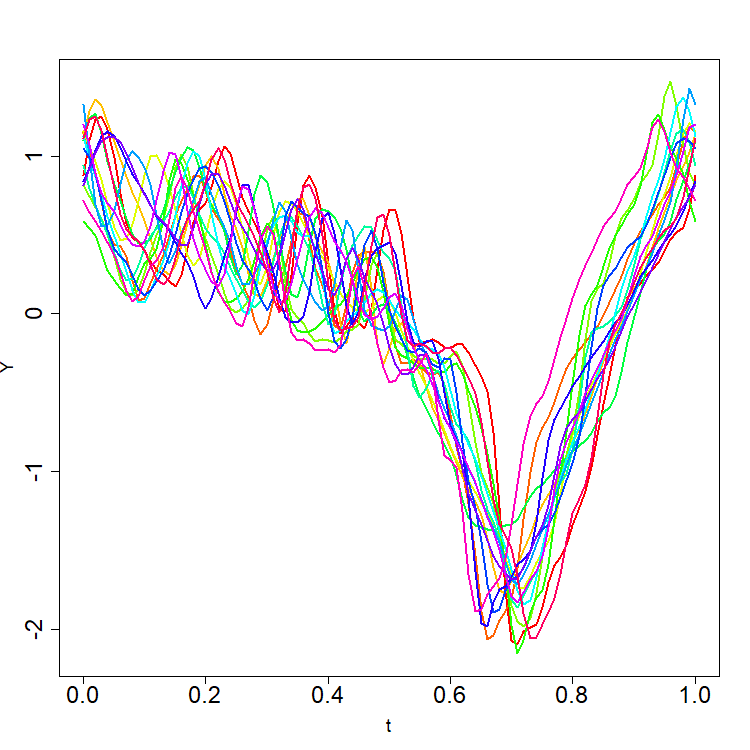}  
   \end{tabular}
    \end{tabular}
    \caption{%Illustration of estimated pre-shapes in the first row and of the shapes obtained from aligning those to the Fréchet mean ($\boldsymbol{\mu}$), in the first column. This procedure was apply to the butterfly (first block) and fork (second block) datasets.
    Illustration of the alignment procedure for the butterfly dataset (top block) and the fork dataset (bottom block). In each block, the left column displays the estimated Fréchet mean. The first row shows estimated pre-shapes together with their coordinate functions, while the second row presents the corresponding aligned shapes.}
    \label{alignement}
\end{figure}

\begin{figure}[H]
    \centering
    \begin{tabular}{c c c c c  }
    & $\hat{\mathbf{T}}_i$ & $\hat{\rho}_i$ & $\hat{\theta}_i$ &   $2\pi {\hat\delta}_i$ \\ 
         Butterfly &  \includegraphics[width=.15\linewidth, align=c]{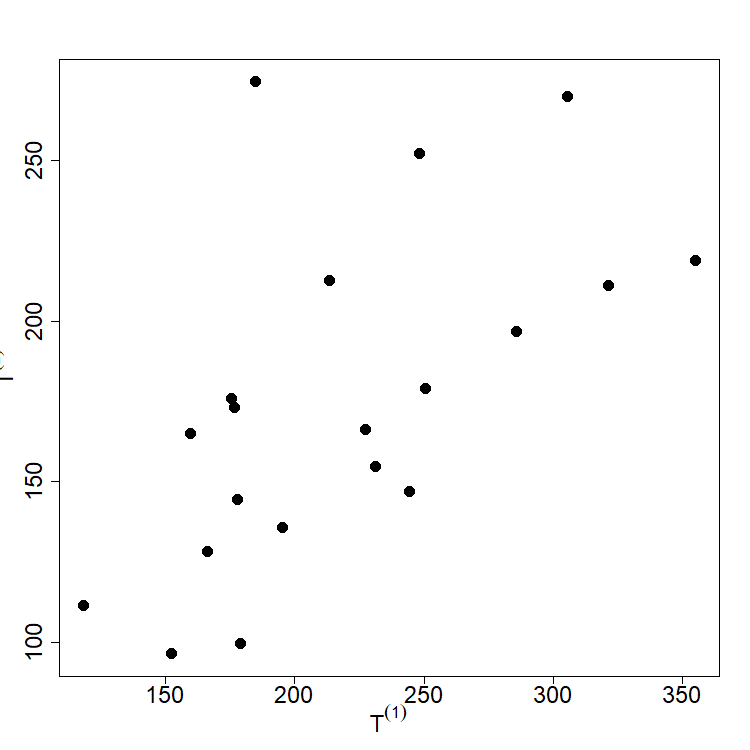} & \includegraphics[width=.15\linewidth, align=c]{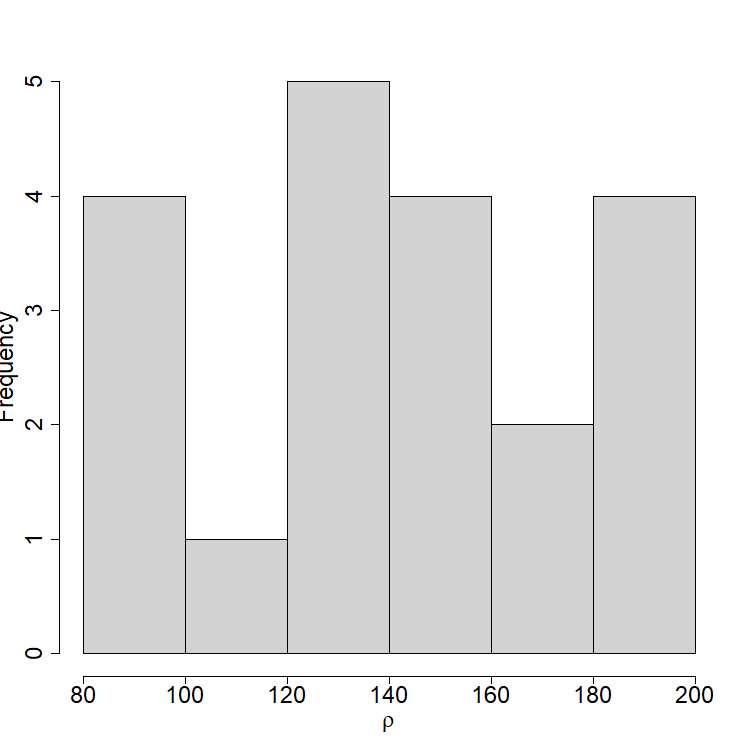} &
          \includegraphics[width=.15\linewidth, align=c]{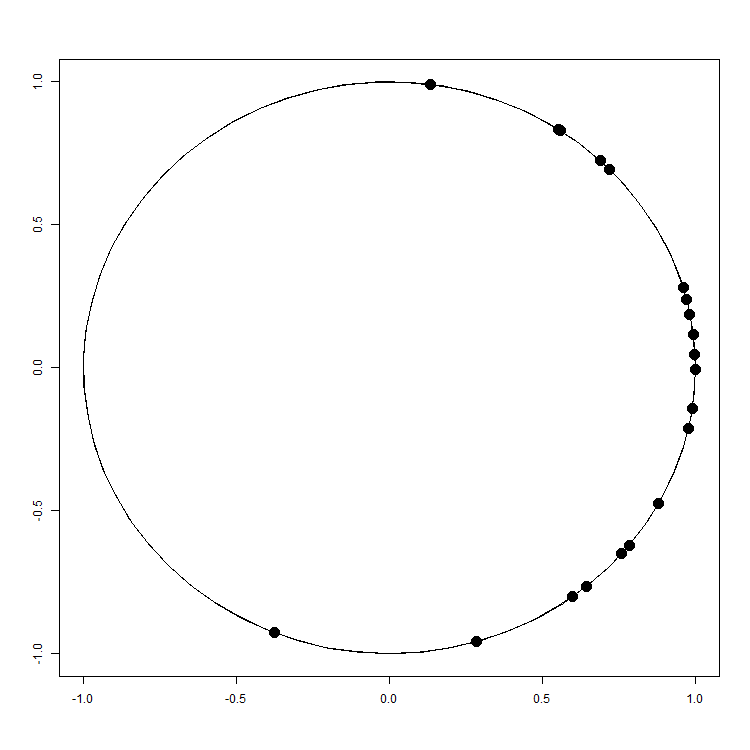} &    
\includegraphics[width=.15\linewidth, align=c]{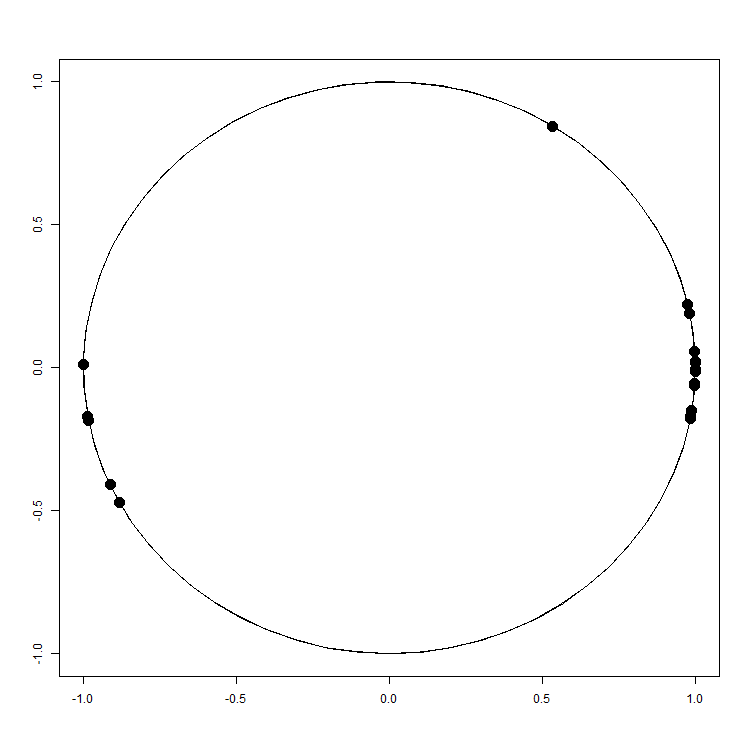}  \\
         Fork & 
         \includegraphics[width=.15\linewidth, align=c]{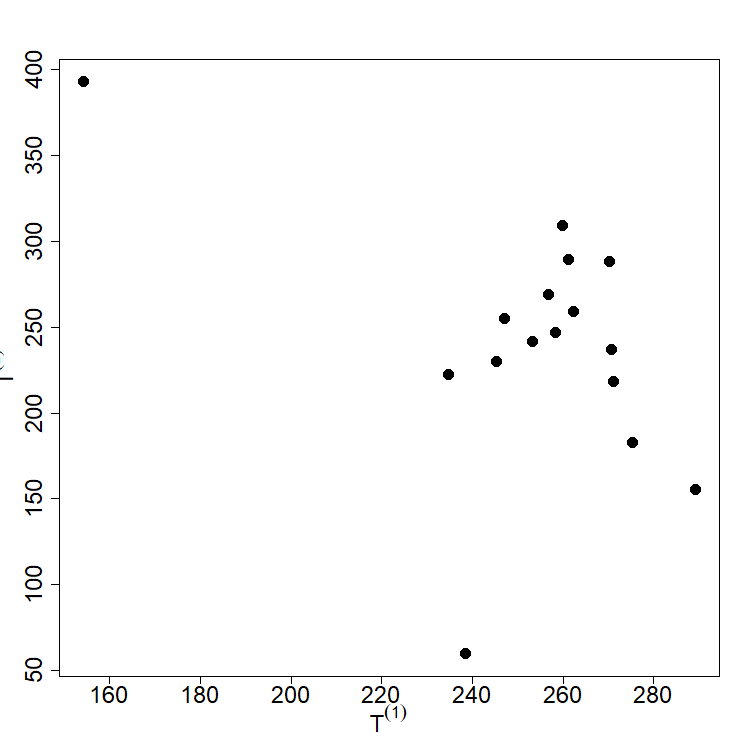} & \includegraphics[width=.15\linewidth, align=c]{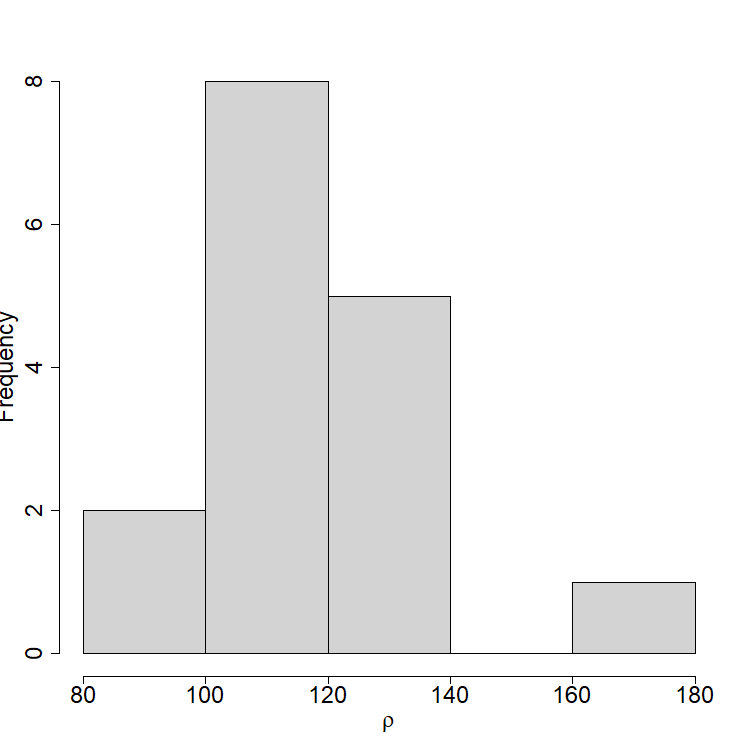} &
         \includegraphics[width=.15\linewidth, align=c]{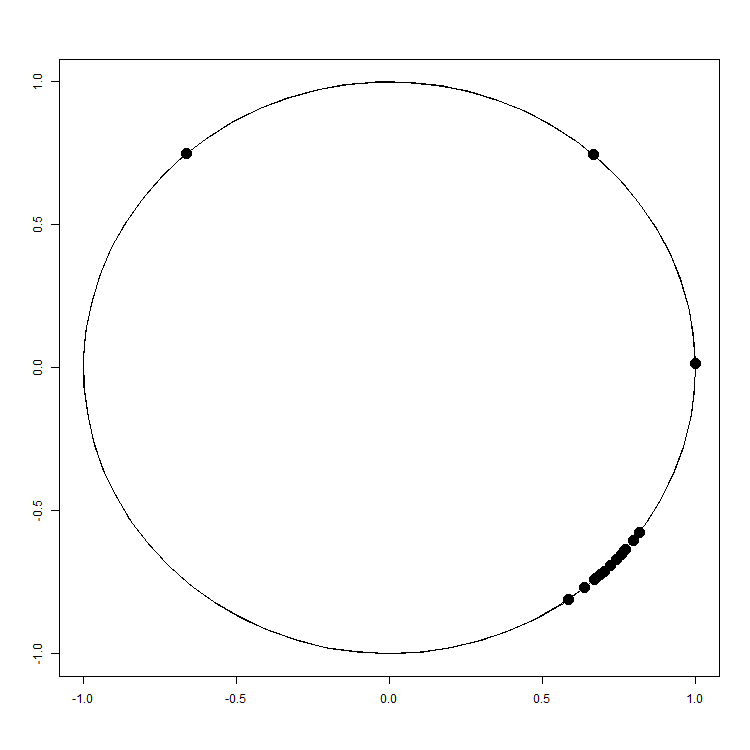} &
         \includegraphics[width=.15\linewidth, align=c]{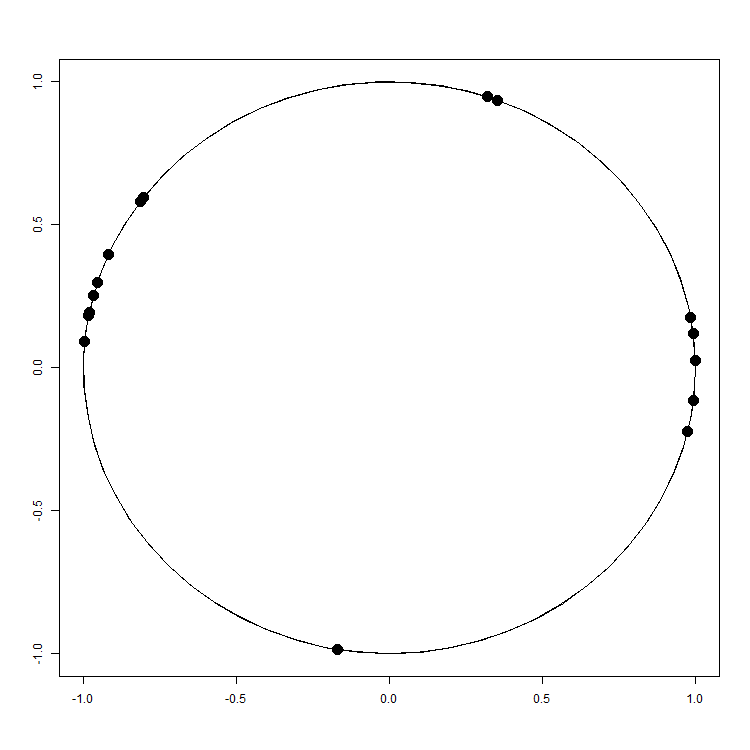}  
    \end{tabular}
    \caption{Plots of the estimated deformation parameters associated with each curve in the butterfly (first row) and fork (second row) datasets}
    \label{est_def}
\end{figure}

\subsection{Results from modeling contours}
\label{Resu_mod}
We now apply the PCA framework introduced in Section~\ref{gen_mod} to analyze the butterfly and fork datasets. For each dataset, we construct the samples $\mathbf{z}_1=(\mathbf{z}_{11},\ldots,\mathbf{z}_{n_s1})$ and $\mathbf{z}_2=(\mathbf{z}_{12},\ldots,\mathbf{z}_{n_s2})$ where
$$\mathbf{z}_{i1} = \hat\rho_i\hat{\tilde{\mathbf{c}}}_i \textrm{ and } \mathbf{z}_{i2}=\left( 
    \tan\left(\frac{\pi}{2} \left(\hat\delta_i-\frac{1}{2}\right)  \right), \tan\left(\frac{1}{4} \left(\frac{\hat\theta_i}{2}-\pi\right) \right)  ,\hat{\mathbf{T}}_i^\top
\right)^\top,
 $$
for $i=1,\ldots,n_s$. The variables $\mathbf{z}_1$ and $\mathbf{z}_2$ are then analyzed separately using the PCA-based approach described in Section~\ref{gen_mod}.

The first row of Figure~\ref{fpca_1} displays the estimated mean shape (in black) for the butterfly dataset, along with perturbations obtained by adding (in red) or subtracting (in blue) a multiple of the first three eigenfunctions. The second row presents the same visualization for the fork dataset. This type of representation, inspired by \cite[Chap.~8]{ramsay2008}, provides insight into the main modes of variation captured by the model.

For the butterfly dataset, the dominant sources of variability are concentrated in the wings. The first eigenfunction primarily captures variations in wing size relative to body length, essentially accounting for the scaling effect $\rho$. The second and third eigenfunctions reflect more subtle changes in wing geometry, corresponding to intrinsic shape variability. This is particularly noticeable for the third eigenfunction, where large positive scores lead to short, pointed wings (red), while large negative scores produce elongated, rounded wings (blue).

For the fork dataset, the modes of variation are less immediately interpretable, although a similar pattern can still be observed. The first eigenfunction mainly captures variations in the overall size of the fork, again reflecting the scaling effect $\rho$. The second eigenfunction captures a key structural difference in the dataset by distinguishing between three-pronged and four-pronged forks, with large positive scores (red) associated with three-pronged shapes.

Overall, these results demonstrate that the proposed PCA framework captures meaningful and interpretable modes of shape variability, even in the presence of deformation effects.

\begin{figure}[H]
    \centering
    \begin{tabular}{c c c c c}PC1 ($57\%$) & PC2 ($22\%$)& PC3 ($9\%$) \\  

        \includegraphics[width=.2\linewidth]{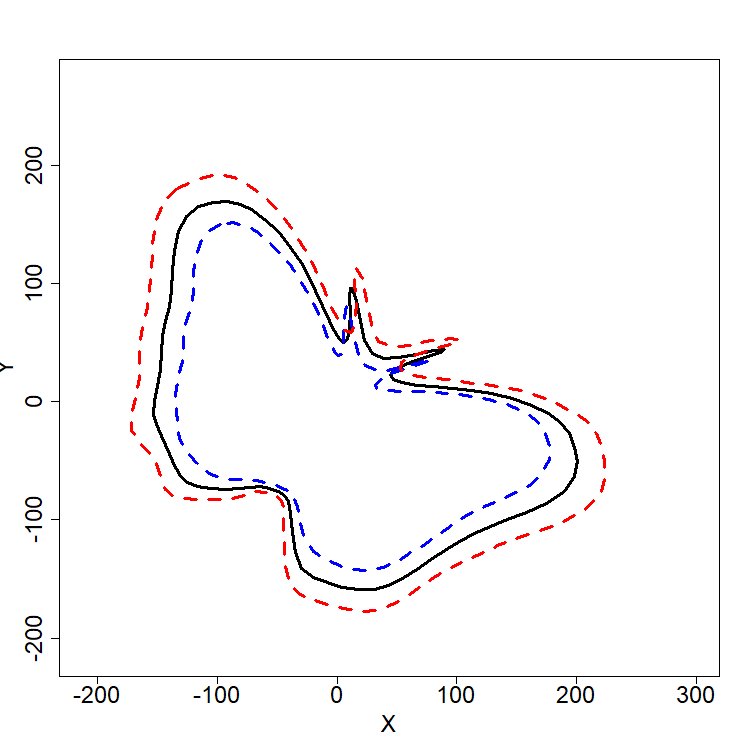} &  \includegraphics[width=.2\linewidth]{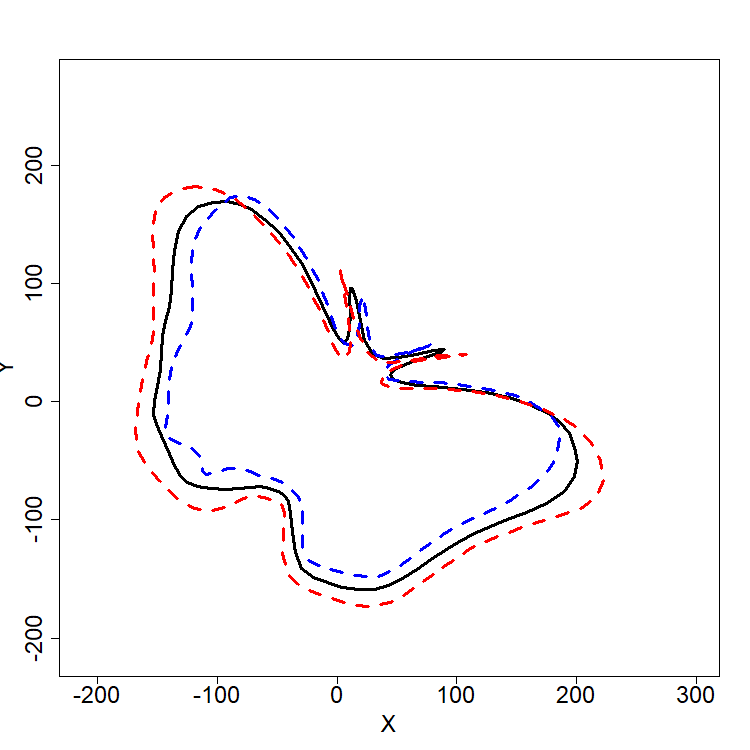} & \includegraphics[width=.2\linewidth]{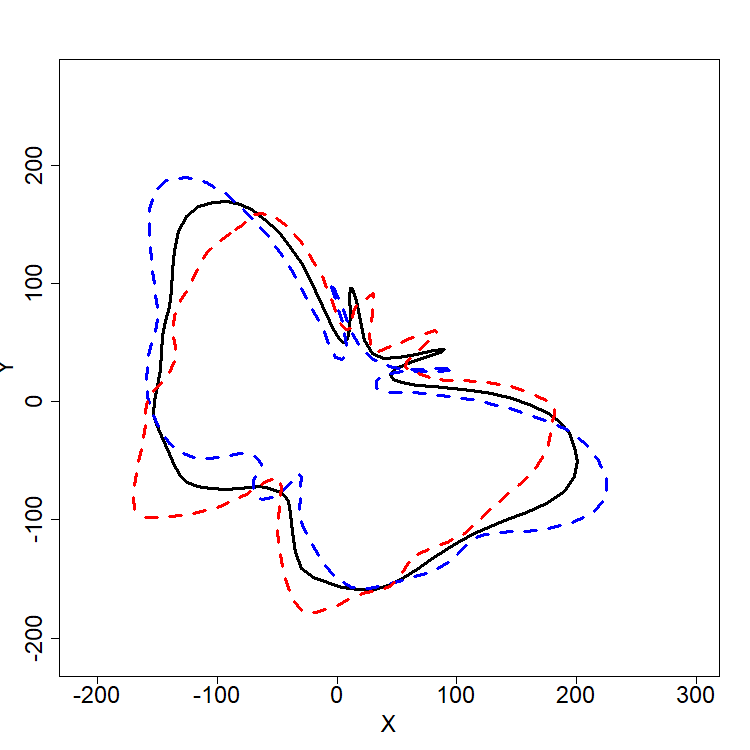}   \\
        PC1 ($42\%$) & PC2 ($19\% $) & PC3 ($10\% $ ) \\

        \includegraphics[width=.2\linewidth]{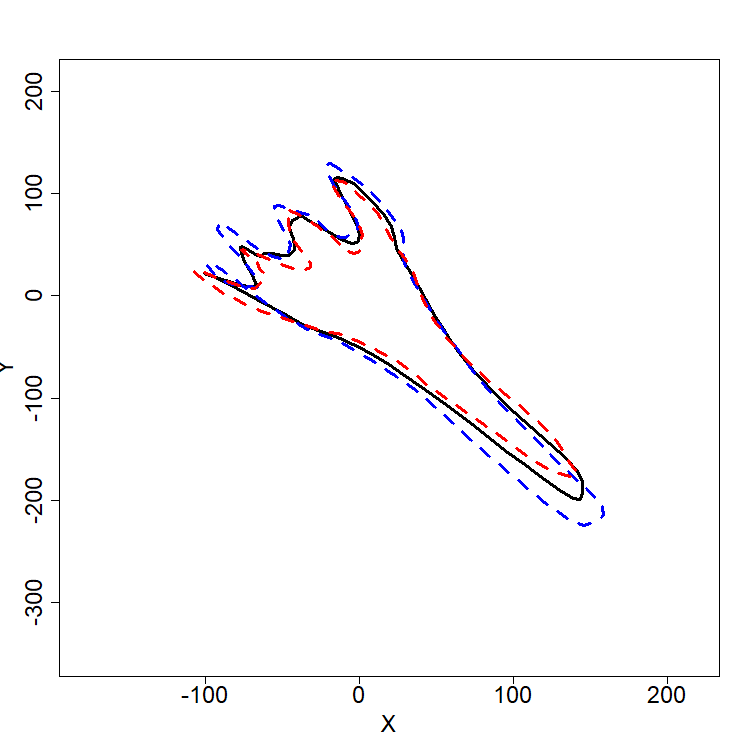} &  \includegraphics[width=.2\linewidth]{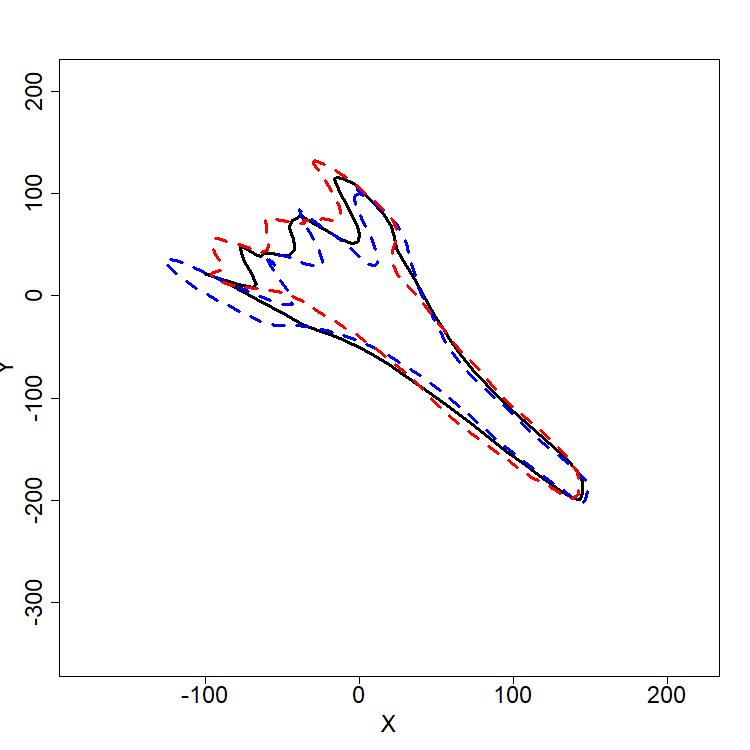} & \includegraphics[width=.2\linewidth]{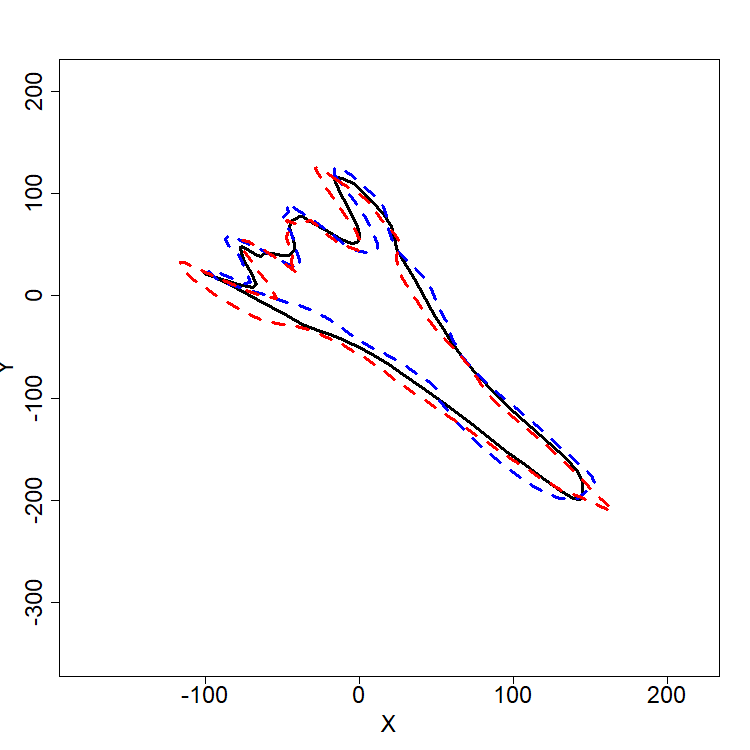}   \\
    \end{tabular}
    \caption{Plots of the estimated mean function $\bar{\mathbf{z}} = \sum_{i}\mathbf{z}_{i1}$ in black, of $\bar{\mathbf{z}} - 20\hat{\boldsymbol{\phi}}_k$ in blue and of $\bar{\mathbf{z}} + 20\hat{\boldsymbol{\phi}}_k$ in red, for $k=1$ (first column), $k=2$ (second column) and $k=3$ (third column). The first row corresponds to the butterfly dataset, while the second row corresponds to the fork dataset. The number in parentheses in each plot title indicates the percentage of total variation explained by the corresponding eigenfunction.  }
    \label{fpca_1}
\end{figure}

To assess the ability of the proposed PCA model to capture the variability of planar curves, we generate new samples from the estimated model. The underlying idea is that a well-fitted model should produce synthetic shapes that are both realistic and representative of the original datasets.

We define our generative model as described in Section~\ref{gen-sec}. In particular, we simulate the score vectors $\boldsymbol{\xi} \in \mathbb{R}^{M_1+M_2}$ from the estimated joint Gaussian model. The number of retained principal components, $M_1$ for $\mathbf{Z}_1$ and $M_2$ for $\mathbf{Z}_2$, is chosen such that the cumulative proportion of explained variance reaches at least $90\%$. This leads to $M_1=4$ and $M_2=4$ for the butterfly dataset, and $M_1=7$ and $M_2=4$ for the fork dataset.

For comparison, we also consider a baseline approach based on multivariate functional PCA (MFPCA), applied directly to the smoothed curves $\hat{\mathbf{c}}^s_i$, $i=1,\ldots,n_s$, with $s \in \{b,f\}$, following \citet{happ2018}. As before, we retain a number of principal components, denoted $M$, such that at least $90\%$ of the variance is explained. This results in $M=3$ for the butterfly dataset and $M=4$ for the fork dataset.

This approach serves as a natural baseline to assess the impact of the alignment step. Since geometric variability and deformation effects are not explicitly disentangled, we expect these sources of variation to be confounded, which may limit the ability of the model to generate coherent and realistic { contours}.

Figures~\ref{gen_1} and \ref{gen_2} illustrate examples of generated butterfly and fork curves, respectively. Each figure consists of three rows, each displaying five curves: row (a) shows curves generated using our approach without deformations, row (b) shows curves generated using our approach with deformations, and row (c) displays curves generated using the MFPCA baseline.

To quantitatively assess how well the generated curves resemble those in the original datasets, we introduce two similarity metrics based on the distance $d$ defined in Section~\ref{estim_frechet_mean}:
$$
D_1(\mathbf{y}^{s})=\frac{1}{n_{s}}\sum_{i=1}^{n_s} d\left(\mathbf{y}^{*s},\hat{\mathbf{c}}_i^{*s}\right) \  \textrm{ and } \ D_2(\mathbf{y}^{s})=\min_{1\le i \le n_s } d\left(\mathbf{y}^{*s},\hat{\mathbf{c}}_i^{*s}\right),
$$
where $\mathbf{y}^{s}$ denotes a generated curve and $\mathbf{y}^{*s}$ its standardized shape, with $s \in \{b,f\}$.

A low value of $D_1(\mathbf{y}^{s})$ indicates that the generated curve is, on average, close to the shapes in the dataset, while a low value of $D_2(\mathbf{y}^{s})$ indicates that the generated curve is close to at least one observed shape. The values of these two metrics are reported above each generated curve in Figures~\ref{gen_1} and \ref{gen_2}.

\begin{figure}[H]
    \centering
    \begin{tabular}{c c c  c c c c c c c c }
 $D_1$: &0.06 & 0.10 & 0.06 & 0.06 & 0.05 \\ 
 $D_2$: & 0.02 & 0.01 & 0.02 & 0.02 & 0.02 \\ 
(a) &\includegraphics[align=c, width=.13\textwidth]{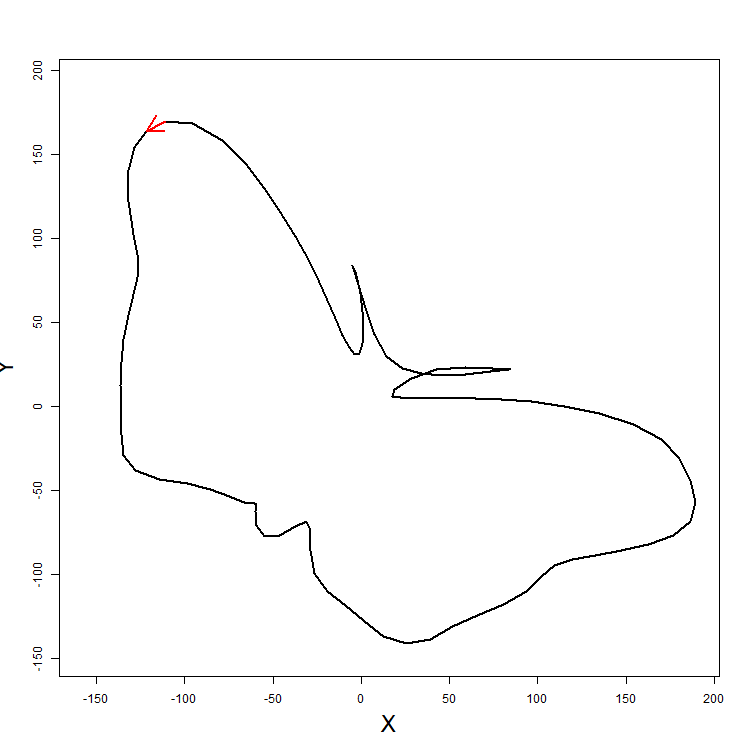} &\includegraphics[align=c, width=.13\textwidth]{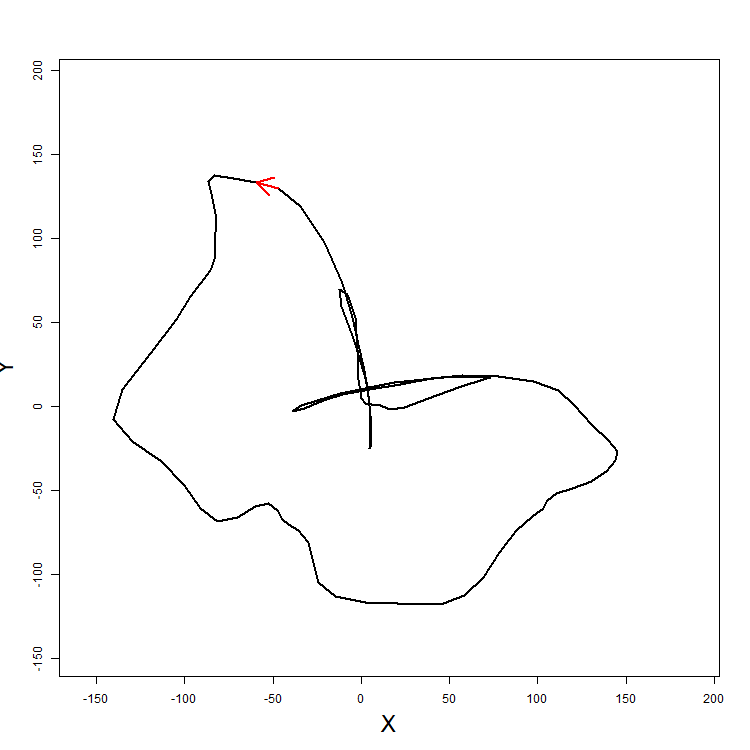}& \includegraphics[align=c, width=.13\textwidth]{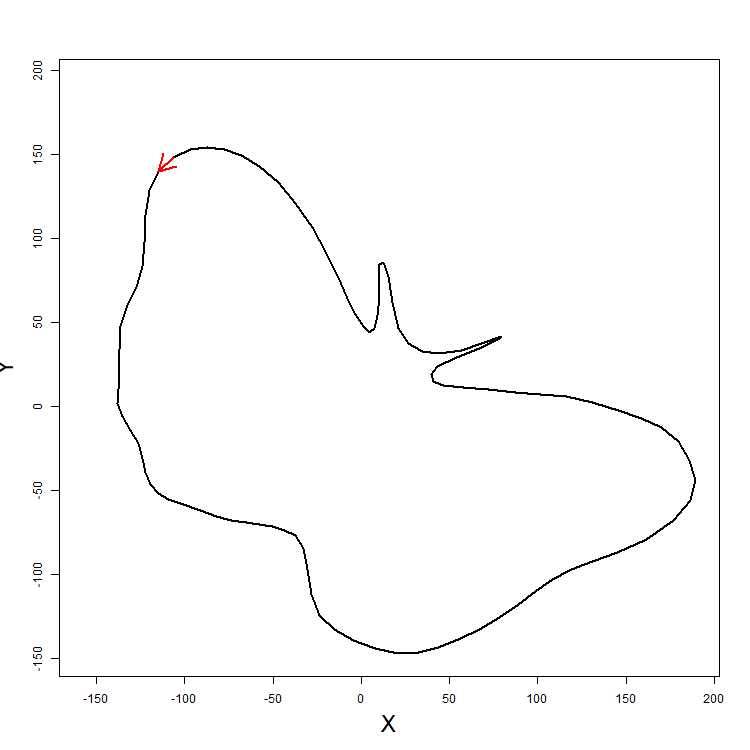}& \includegraphics[align=c, width=.13\textwidth]{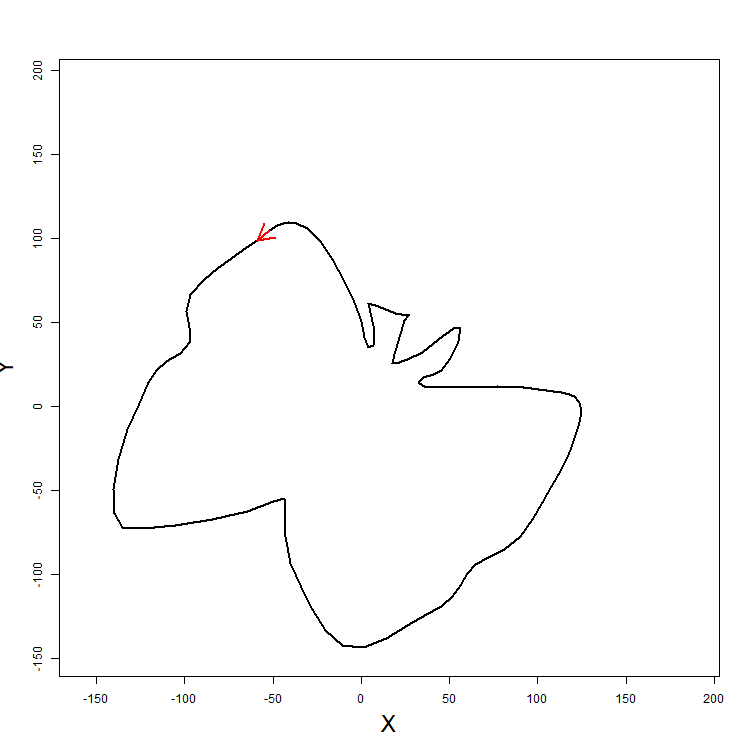} & \includegraphics[align=c, width=.13\textwidth]{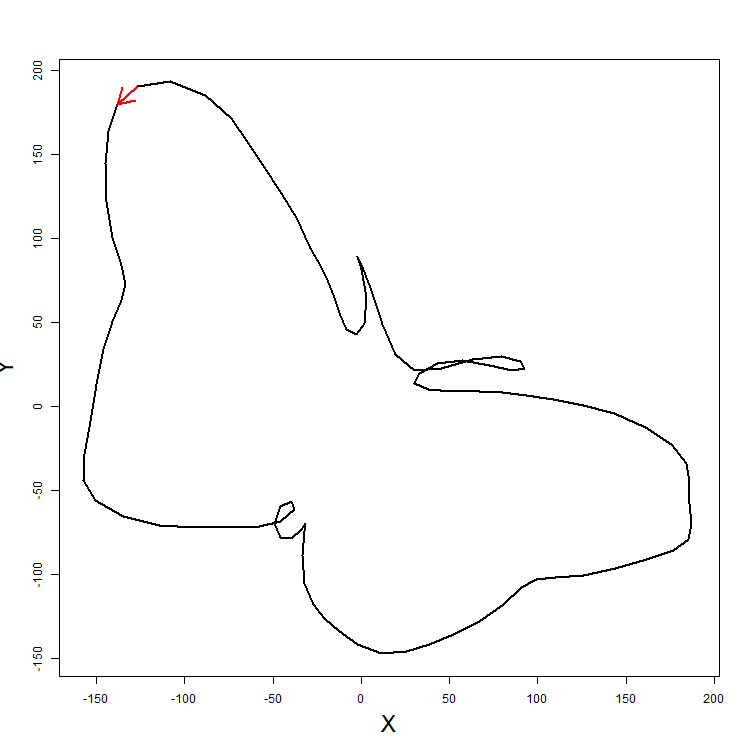}\\ 
\hline
 $D_1$: & 0.07 & 0.10 & 0.07 & 0.07 & 0.05 \\ 
  $D_2$: & 0.02 & 0.02 & 0.04 & 0.02 & 0.02 \\ 
 (b)  & \includegraphics[align=c, width=.13\textwidth]{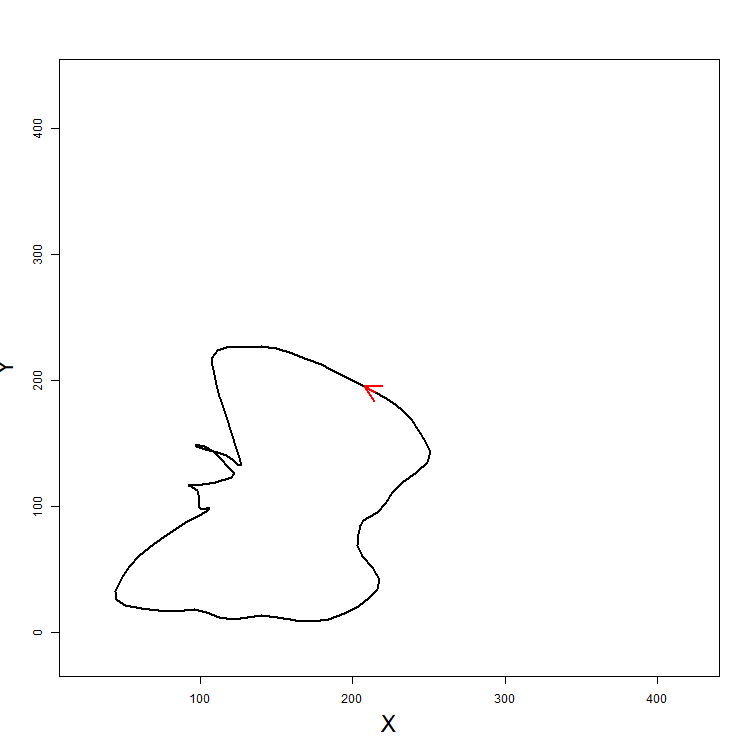} &\includegraphics[align=c, width=.13\textwidth]{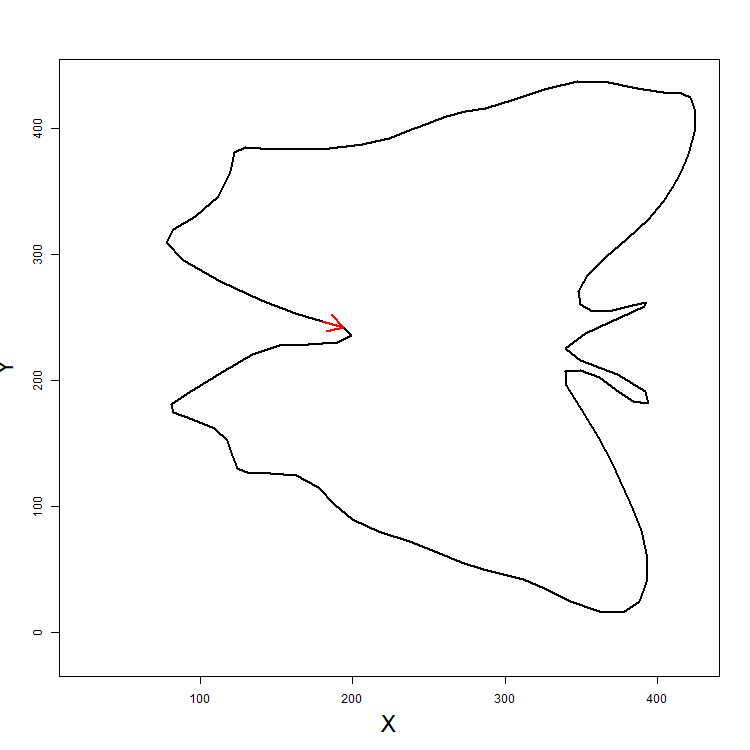}& \includegraphics[align=c, width=.13\textwidth]{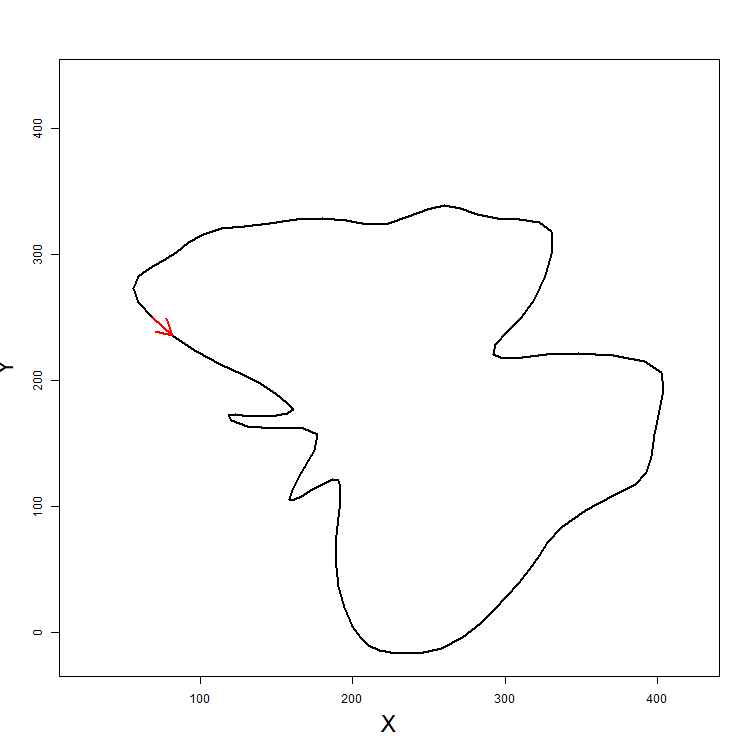}& \includegraphics[align=c, width=.13\textwidth]{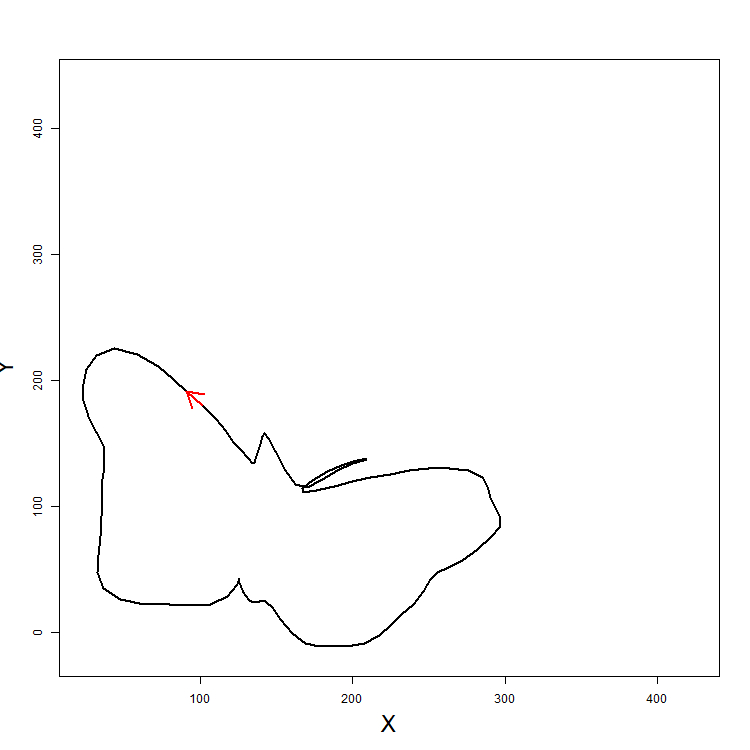} & \includegraphics[align=c, width=.13\textwidth]{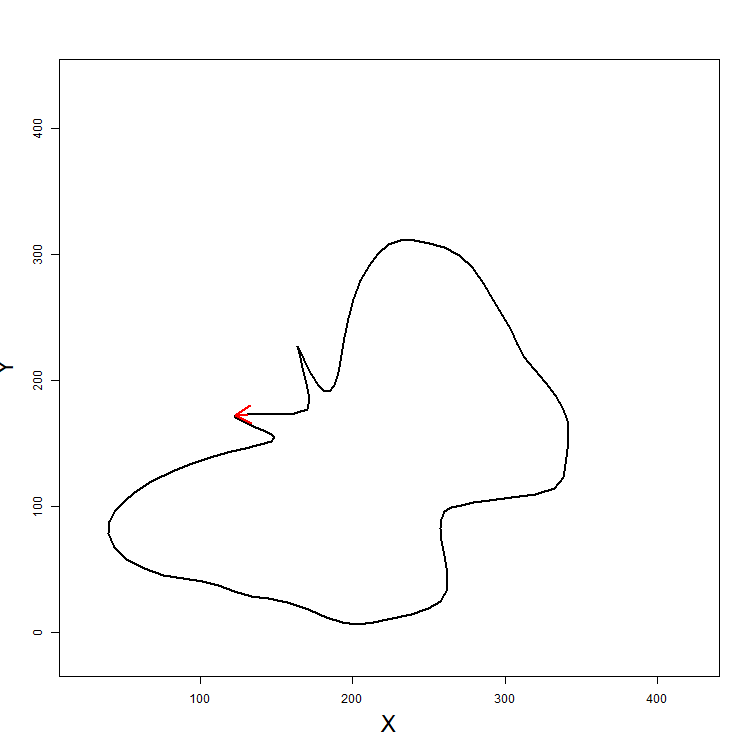}\\ 
 \hline
  $D_1$: & 0.07 & 0.07 & 0.07 & 0.05 & 0.14 \\ 
  $D_2$: & 0.04 & 0.04 & 0.03 & 0.02 & 0.08 \\ 
  (c)  &\includegraphics[align=c, width=.13\textwidth]{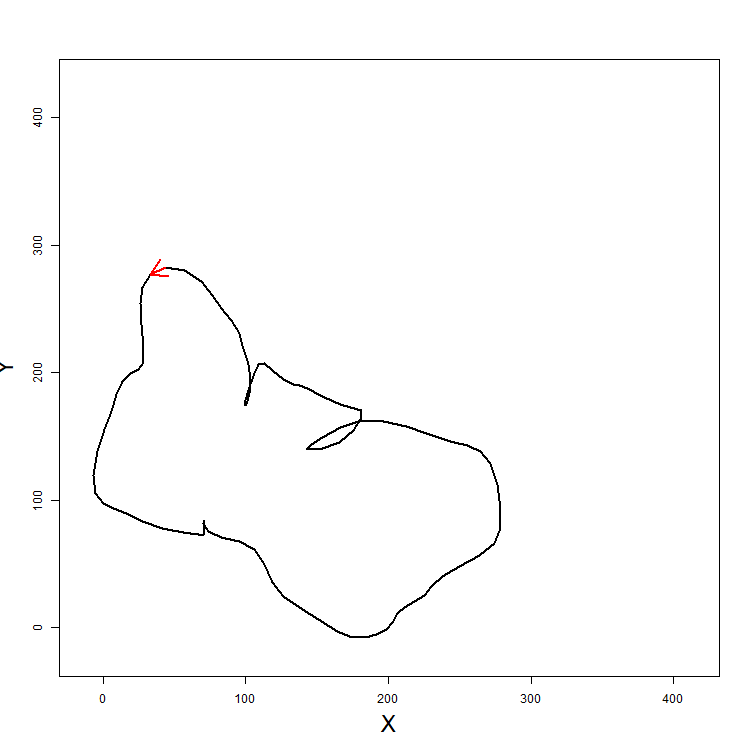} &\includegraphics[align=c, width=.13\textwidth]{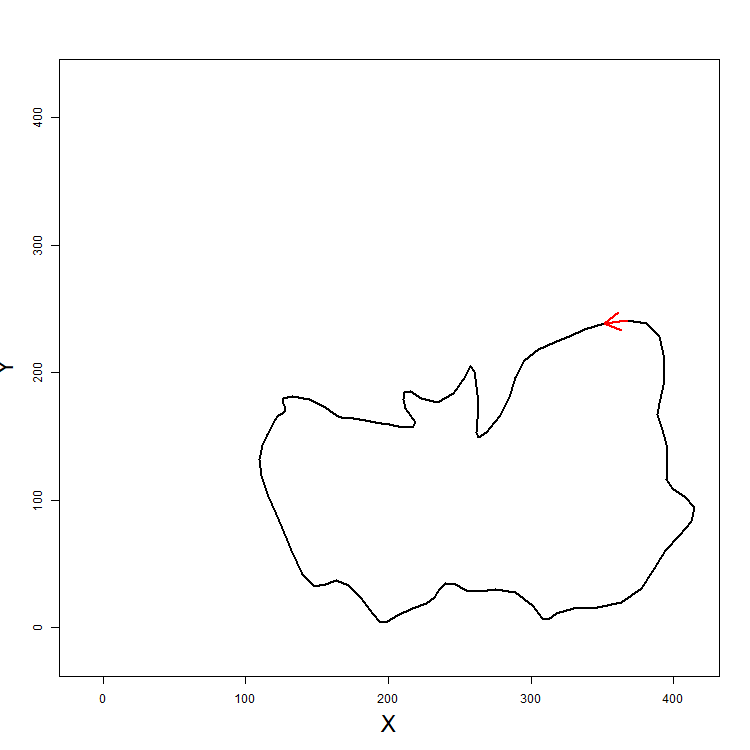}& \includegraphics[align=c, width=.13\textwidth]{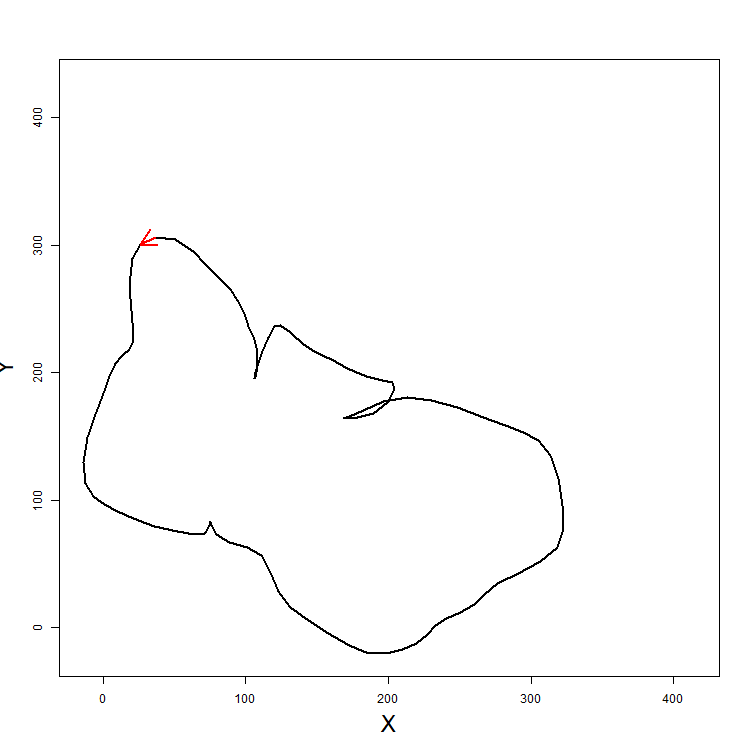}& \includegraphics[align=c, width=.13\textwidth]{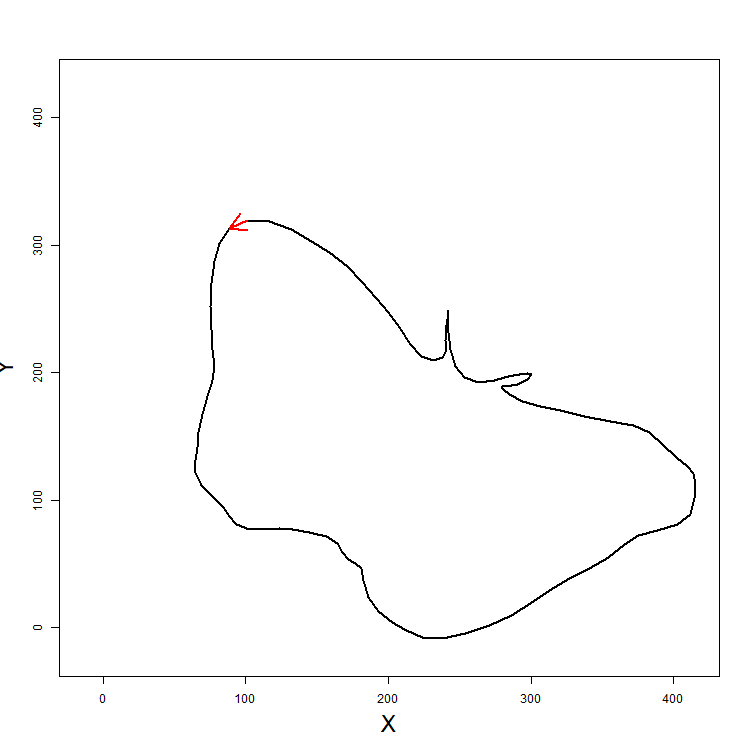} & \includegraphics[align=c, width=.13\textwidth]{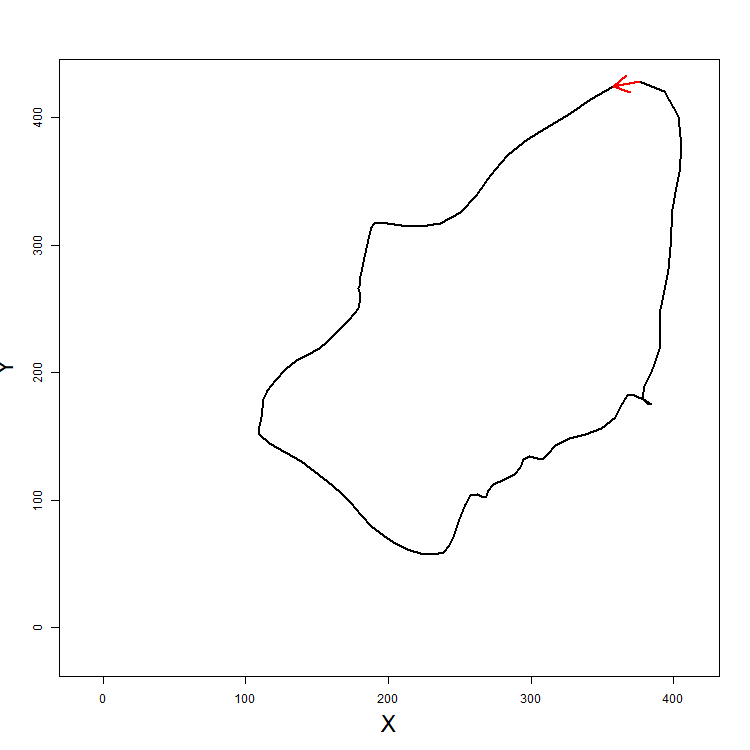}\\
\end{tabular}
\caption{Butterfly contours generated with (a) our approach without the deformation parameters, (b) our approach with the deformation parameters, and (c) with MFPCA.}
    \label{gen_1}
\end{figure}
\begin{figure}[H]
\centering

\begin{tabular}{c c c c c c c  c c c c c c}
$D_1$: & 0.03 & 0.03 & 0.06 & 0.05 & 0.03 \\ 
  $D_2$: & 0.01 & 0.01 & 0.01 & 0.02 & 0.01 \\ 
  (a)& \includegraphics[align=c, width=.13\textwidth]{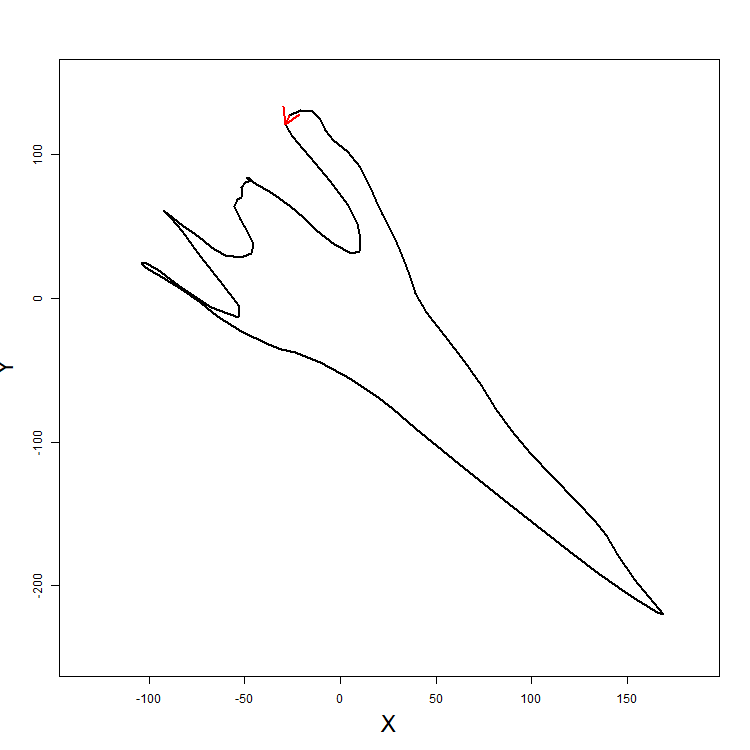} &\includegraphics[align=c, width=.13\textwidth]{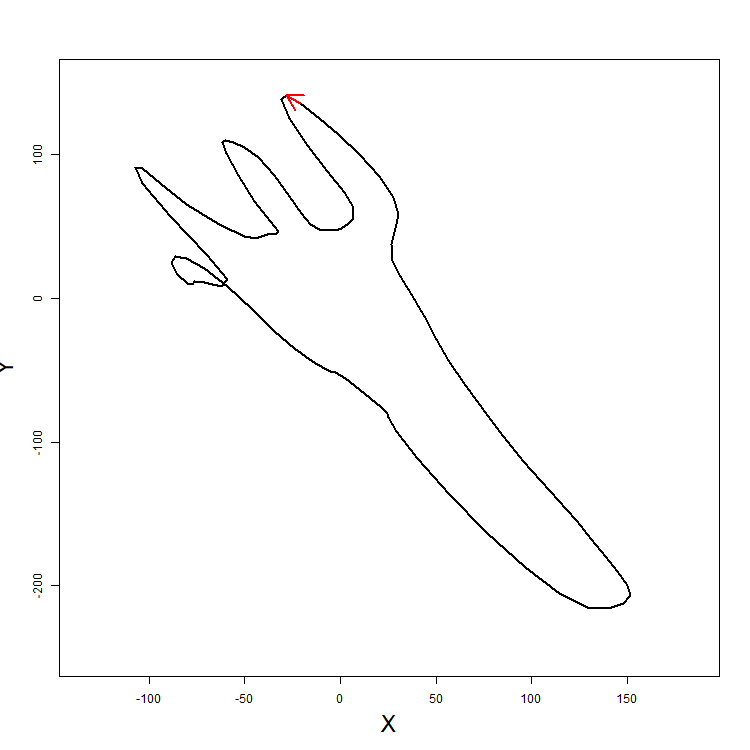}& \includegraphics[align=c, width=.13\textwidth]{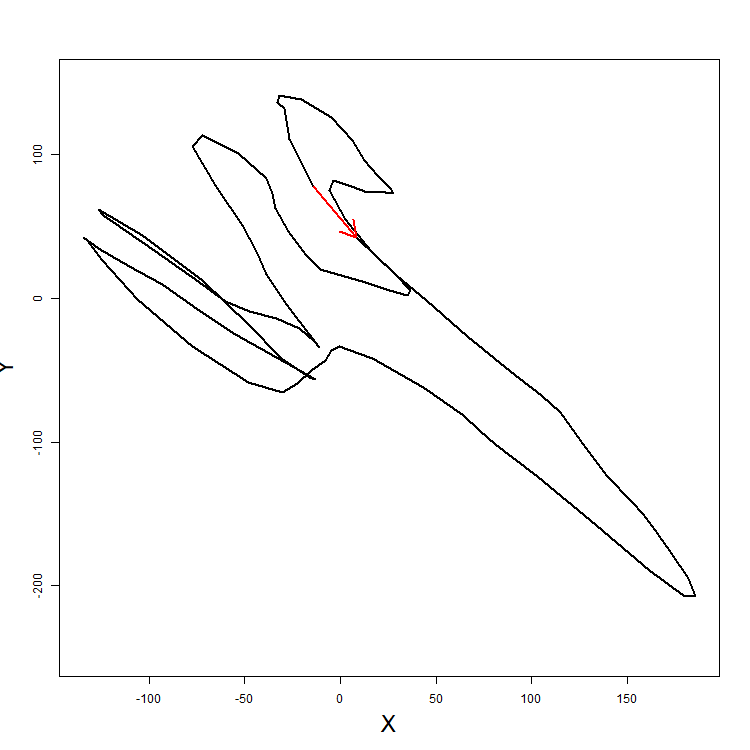} & \includegraphics[align=c, width=.13\textwidth]{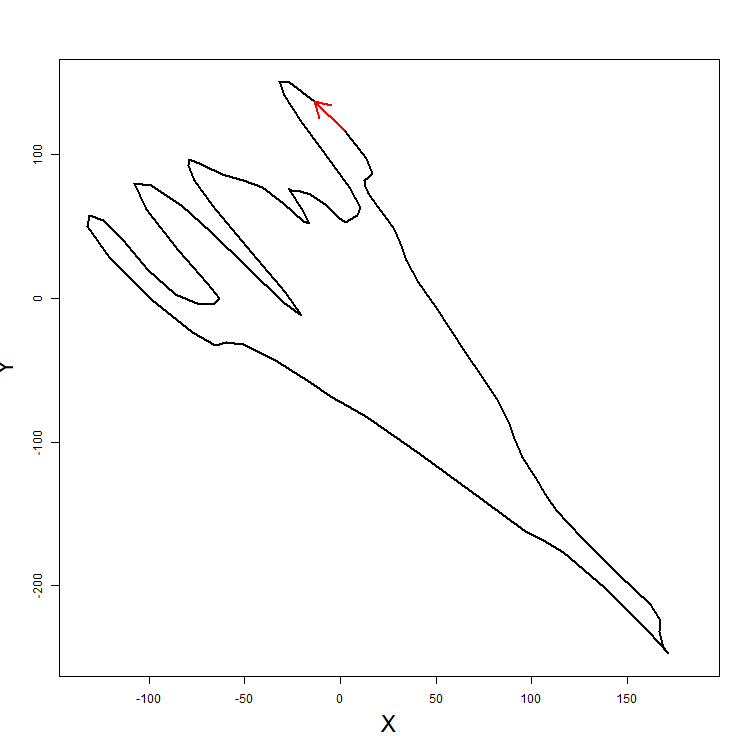} & \includegraphics[align=c, width=.13\textwidth]{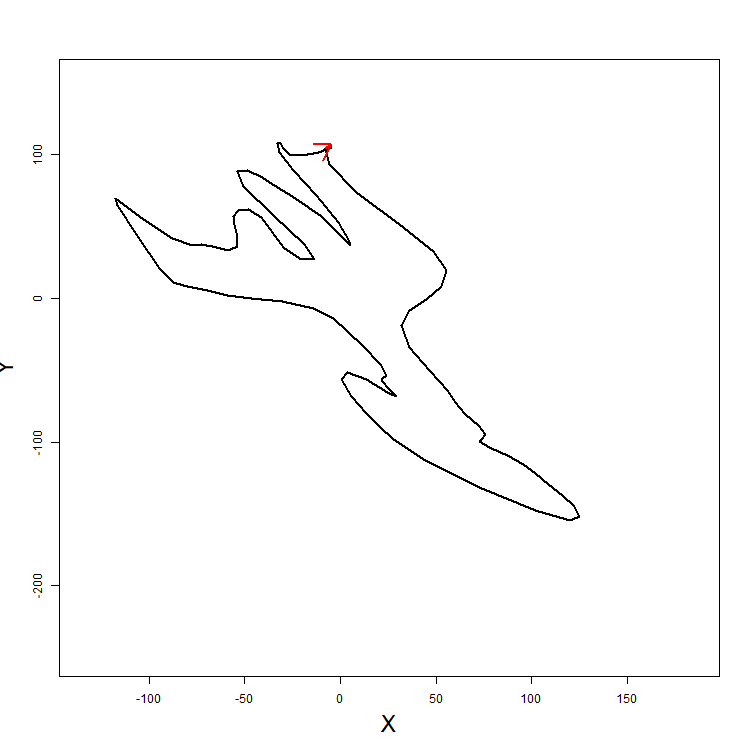}\\
 \hline
  $D_1$: & 0.03 & 0.05 & 0.07 & 0.06 & 0.03 \\ 
  $D_2$: & 0.01 & 0.03 & 0.02 & 0.02 & 0.01 \\
   (b)& \includegraphics[align=c, width=.13\textwidth]{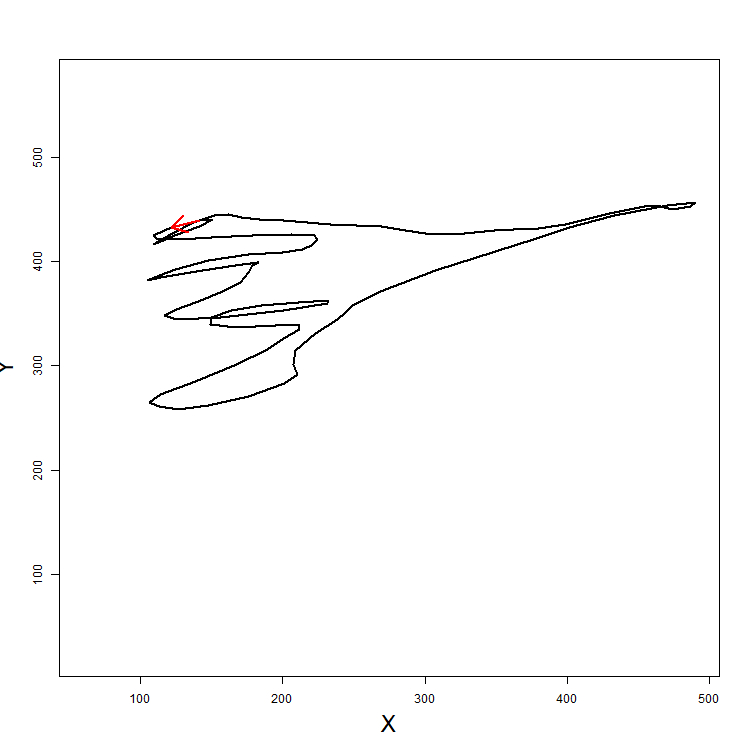} &\includegraphics[align=c, width=.13\textwidth]{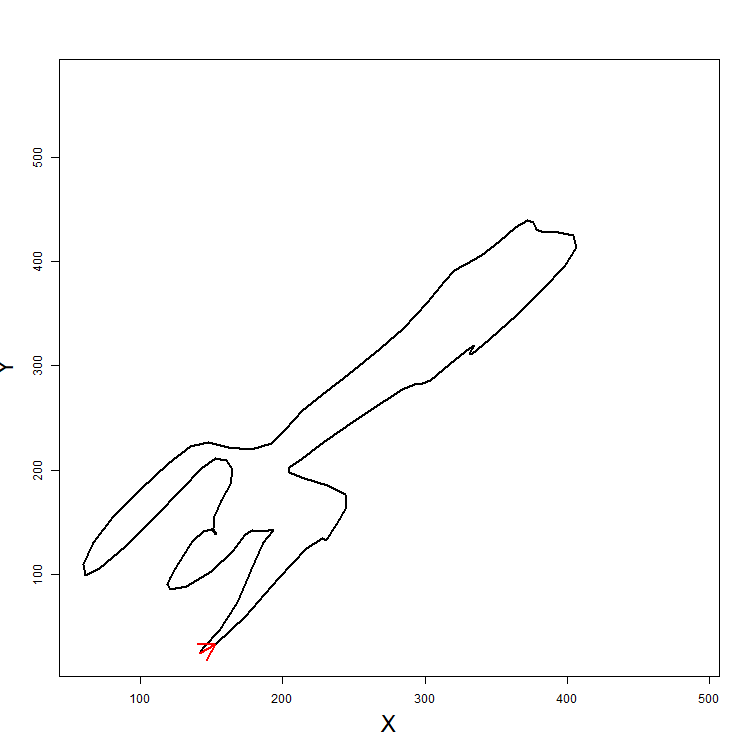}& \includegraphics[align=c, width=.13\textwidth]{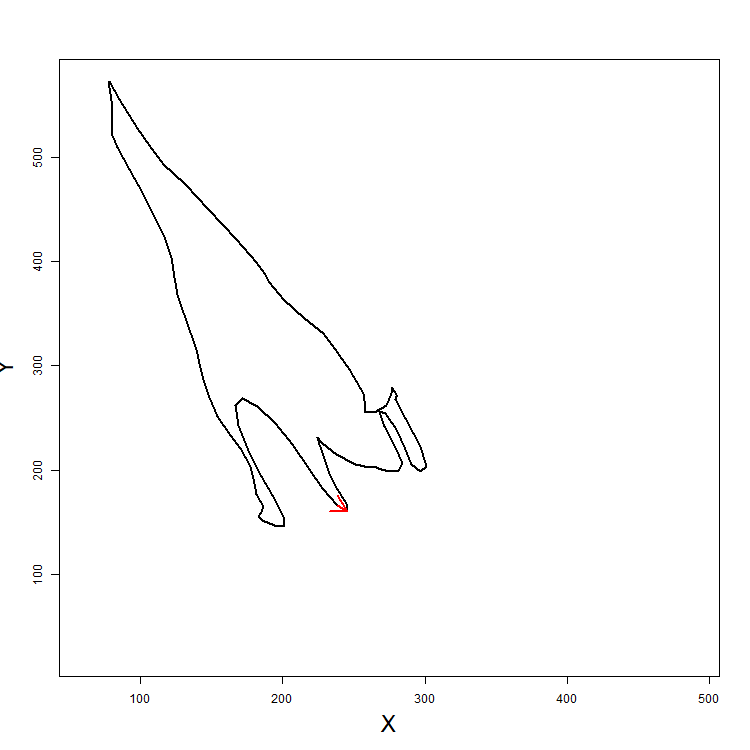} & \includegraphics[align=c, width=.13\textwidth]{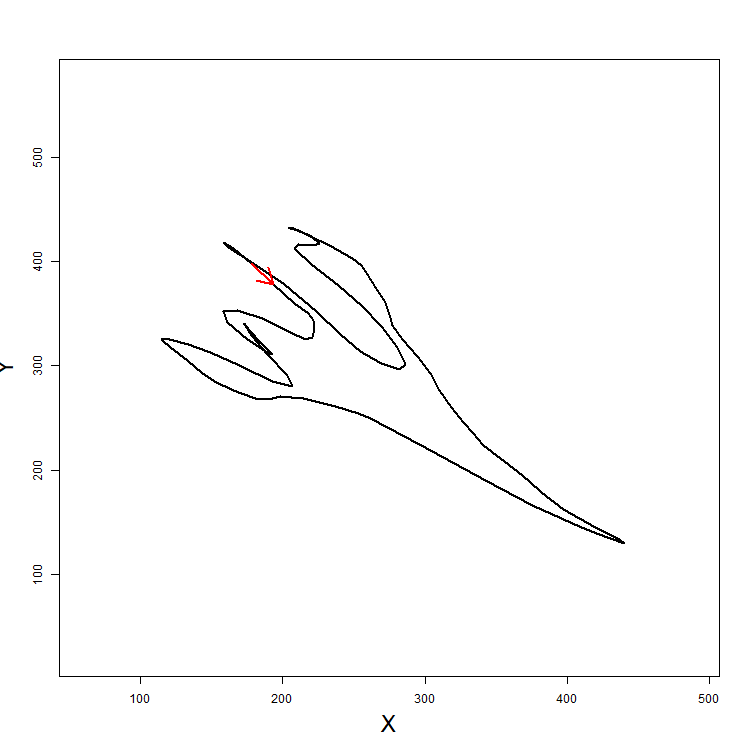} & \includegraphics[align=c, width=.13\textwidth]{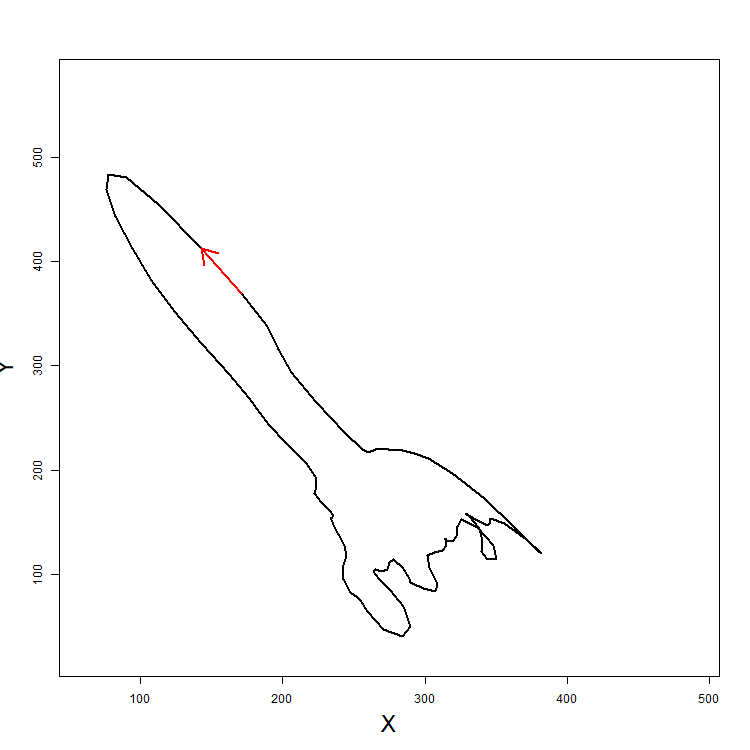}\\
   \hline
 $D_1$: & 0.80 & 0.06 & 0.16 & 0.12 & 0.13 \\ 
  $D_2$: & 0.68 & 0.03 & 0.10 & 0.08 & 0.07 \\ 
    (c)&\includegraphics[align=c, width=.13\textwidth]{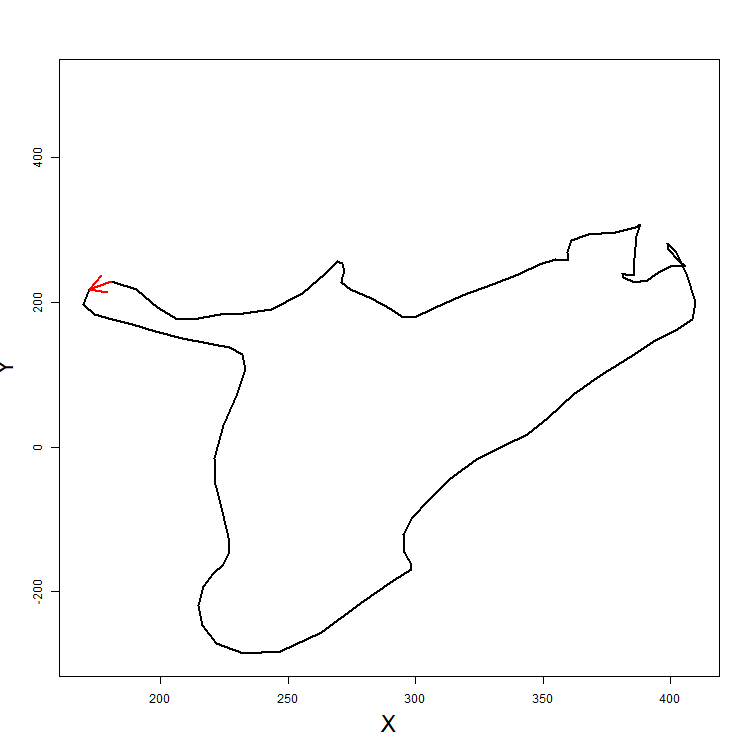} &\includegraphics[align=c, width=.13\textwidth]{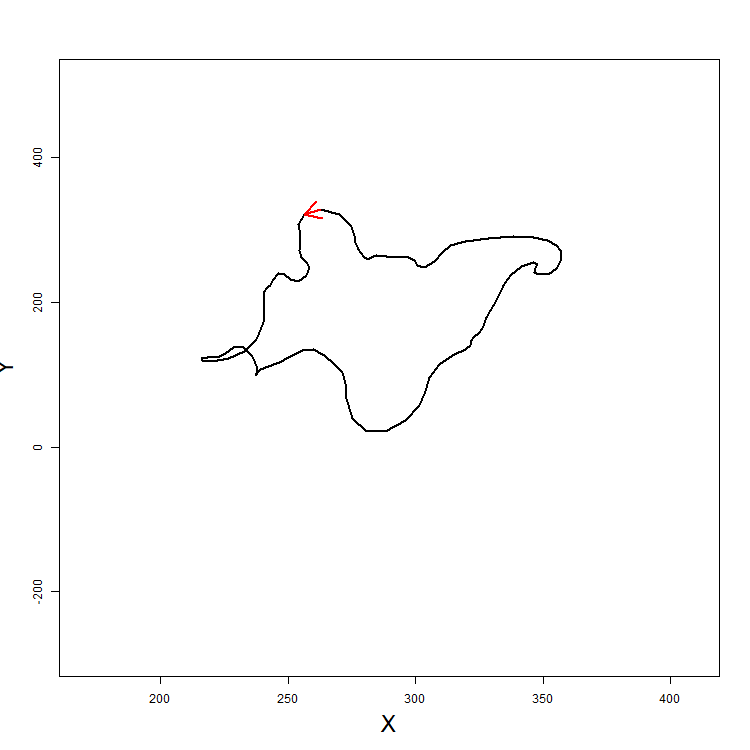}& \includegraphics[align=c, width=.13\textwidth]{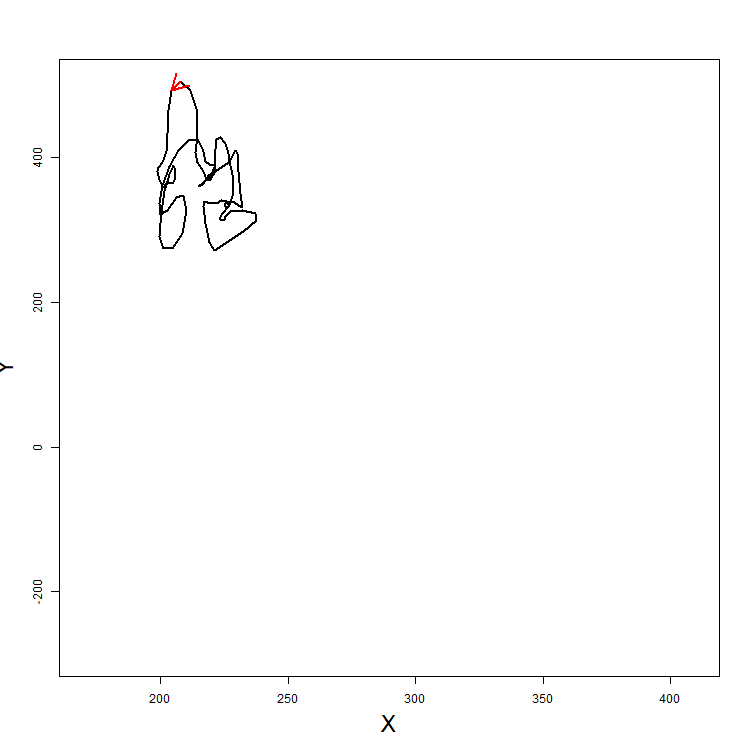} & \includegraphics[align=c, width=.13\textwidth]{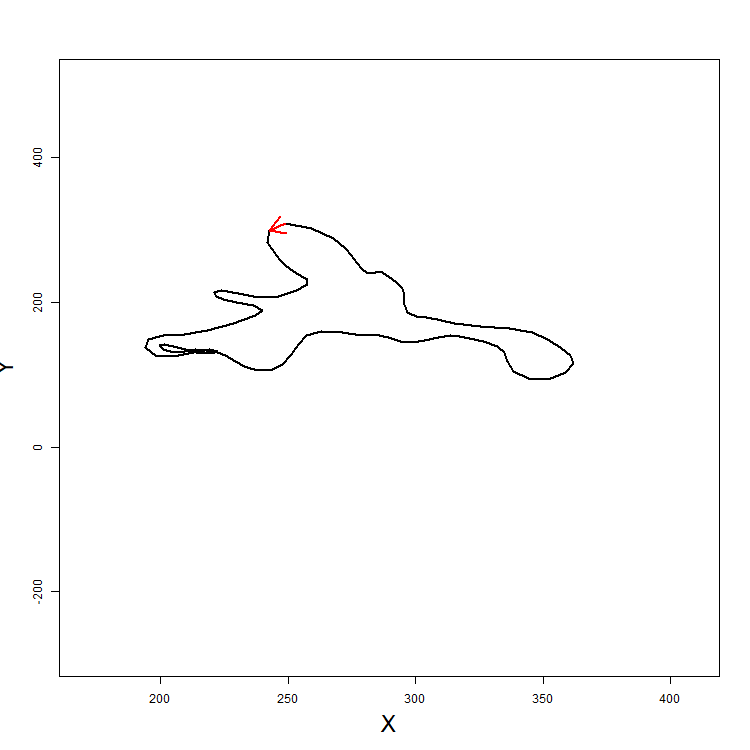} & \includegraphics[align=c, width=.13\textwidth]{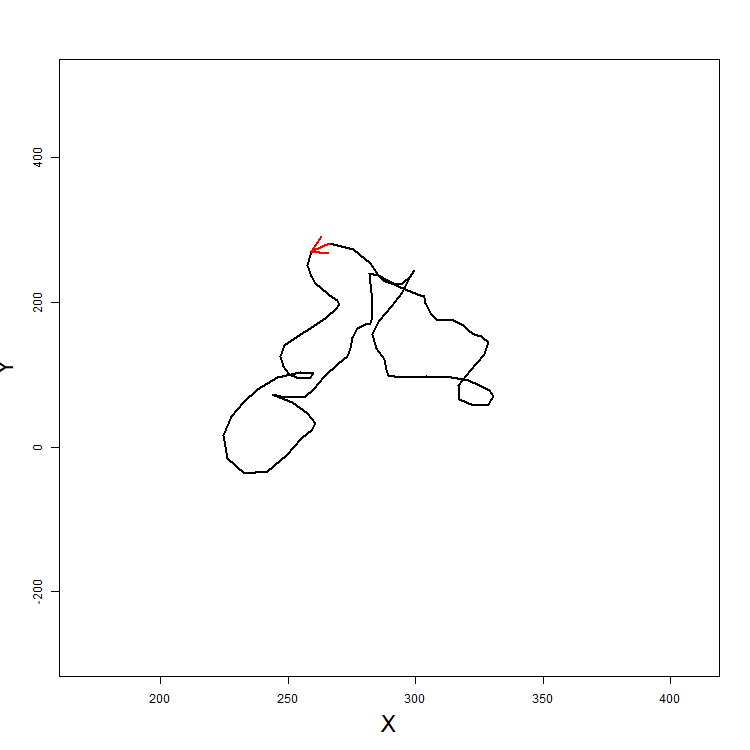} \\
\end{tabular}
\caption{Fork contours generated with (a) our approach without the deformation parameters, (b) our approach with the deformation parameters, and (c) with MFPCA}
    \label{gen_2}
\end{figure}

For the butterfly dataset, illustrated in Figure~\ref{gen_1}, both the proposed approach and the MFPCA baseline yield comparable performance according to the proposed metrics. However, visual inspection favors our approach. In particular, the generated shapes better preserve fine structural features such as the head and antennae, which are often poorly represented or entirely absent in curves generated using MFPCA.

For the fork dataset, illustrated in Figure~\ref{gen_2}, the metrics clearly favor the proposed approach. The shapes generated using MFPCA show little resemblance to actual forks, reflecting the limitations of modeling unaligned curves. In contrast, the proposed method produces shapes that remain consistent with the overall fork geometry. Nevertheless, some generated curves still exhibit irregular or distorted features. This highlights the intrinsic difficulty of modeling complex shapes from a limited sample size ($n_f = 20$).

\section{Discussion}
\label{disc}

This paper studies random planar curves, with a particular focus on contours automatically extracted from images. The main contribution of this work lies in the proposed alignment framework. We introduce a functional approach that explicitly models deformation variables and focuses on the dominant source of phase variability in contour data, namely the arbitrary choice of the starting point. This leads to a tractable and fully functional alignment procedure, referred to as the ICF algorithm, which avoids additional discretization and yields interpretable estimates of the deformation parameters.

In contrast to fully elastic alignment frameworks, which rely on highly flexible reparametrizations, our approach deliberately restricts the class of admissible transformations. This choice is motivated by the nature of image-derived contours, where the primary source of misalignment arises from the parametrization rather than from complex local deformations. By focusing on this setting, the proposed method provides a simple, identifiable, and computationally efficient alternative that is particularly well suited for automatically extracted shapes.

Beyond alignment, we propose a generative model for random planar curves within a functional data analysis framework. This model is based on two separate principal component analyses: one performed on the aligned functional variable and the other on the estimated deformation parameters. %This model is based on %{ two principal component analysis on the aligned functional variable and the estimated deformation parameters}. 
%\textcolor{red}{on enlève le mot "joint" dans la phrase suivante ?}
A key advantage of this formulation is that it enables the analysis of shape variability and deformation effects, rather than treating the latter as nuisance quantities.

The effectiveness of the proposed framework is demonstrated through both simulation studies and real data experiments. The results of the simulation study highlight the accuracy of the alignment procedure and its robustness to discretization. The real data analysis further emphasizes the importance of alignment: the estimated shapes are coherent, and the proposed model generates realistic curves that better reflect the geometric structure of the data than approaches based on unaligned functional representations.

A key feature of our approach is the explicit inclusion of deformation variables in the modeling framework, in particular the scaling parameter. This allows us to work within a linear functional space of closed curves. In contrast, treating scaling as a nuisance parameter would require working on the unit sphere $\mathbf{S}^\infty$, which is a non-linear manifold, and adapting tools from shape analysis such as tangent space approximations \citep{dryden1998, dai2018}. 

Overall, this work highlights the importance of tailoring alignment procedures to the structure of the data. By focusing on the specific characteristics of contours extracted from images, we obtain a simple yet effective framework that bridges functional data analysis, shape analysis and image analysis.

 Our work also has some limitations. First, we focus on single-object contours, which may be restrictive in applications involving multiple objects. Extending the framework to handle multiple interacting curves is a natural direction for future research. Second, the class of reparametrization functions considered in this paper is restricted to shifts of the starting point. As already mentioned, our framework is intentionally designed for contours extracted from images, where the dominant source of phase variability arises from the arbitrary choice of the starting point. While this setting motivates the use of simple reparametrization functions, the proposed approach could be extended to incorporate more flexible transformations when required. In particular, integrating richer classes of reparametrizations, such as diffeomorphic transformations, constitutes a natural direction for future work.

\newpage

\bibliography{refs}

@Manual{fdasrvf,
    title = {fdasrvf: Elastic Functional Data Analysis},
    author = {J. Derek Tucker},
    year = {2026},
    note = {R package version 2.4.3},
    url = {https://CRAN.R-project.org/package=fdasrvf},
  }

@article{jacques2014model,
  title={Model-based clustering for multivariate functional data},
  author={Jacques, Julien and Preda, Cristian},
  journal={Computational Statistics \& Data Analysis},
  volume={71},
  pages={92--106},
  year={2014},
  publisher={Elsevier}
}

@article{happ2018,
  title={Multivariate functional principal component analysis for data observed on different (dimensional) domains},
  author={Happ, Clara and Greven, Sonja},
  journal={Journal of the American Statistical Association},
  volume={113},
  number={522},
  pages={649--659},
  year={2018},
  publisher={Taylor \& Francis}
}

@article{FPCA_amp,
  title={Generative models for functional data using phase and amplitude separation},
  author={Tucker, J Derek and Wu, Wei and Srivastava, Anuj},
  journal={Computational Statistics \& Data Analysis},
  volume={61},
  pages={50--66},
  year={2013},
  publisher={Elsevier}
}

@article{happ_am,
  title={A general framework for multivariate functional principal component analysis of amplitude and phase variation},
  author={Happ, Clara and Scheipl, Fabian and Gabriel, Alice-Agnes and Greven, Sonja},
  journal={Stat},
  volume={8},
  number={1},
  pages={e220},
  year={2019},
  publisher={Wiley Online Library}
}

@book{lelivre,
  title={Functional and shape data analysis},
  author={Srivastava, Anuj and Klassen, Eric P},
  volume={1},
  year={2016},
  publisher={Springer}
}

@article{larticle,
  title={Shape analysis of elastic curves in euclidean spaces},
  author={Srivastava, Anuj and Klassen, Eric and Joshi, Shantanu H and Jermyn, Ian H},
  journal={IEEE transactions on pattern analysis and machine intelligence},
  volume={33},
  number={7},
  pages={1415--1428},
  year={2010},
  publisher={IEEE}
}

@Book{ramsay2008
,	title	= {Functional Data Analysis}
,	edition	= {2nd}
,	publisher	= {Springer New York}
,	year	= {2005}
,	author	= {Ramsay, J O AND Silverman, B W}
}

@article{procrustes,
  title={A generalized solution of the orthogonal procrustes problem},
  author={Sch{\"o}nemann, Peter H},
  journal={Psychometrika},
  volume={31},
  number={1},
  pages={1--10},
  year={1966},
  publisher={Springer}
}

@article{datasets,
  title={The 2d shape structure dataset: A user annotated open access database},
  author={Carlier, Axel and Leonard, Kathryn and Hahmann, Stefanie and Morin, Geraldine and Collins, Misha},
  journal={Computers \& Graphics},
  volume={58},
  pages={23--30},
  year={2016},
  publisher={Elsevier}
}

@article{marron2015,
  title={Functional data analysis of amplitude and phase variation},
  author={Marron, James Stephen and Ramsay, James O and Sangalli, Laura M and Srivastava, Anuj},
  journal={Statistical Science},
  pages={468--484},
  year={2015},
  publisher={JSTOR}
}

@article{kendall,
  title={A survey of the statistical theory of shape},
  author={Kendall, David G},
  journal={Statistical Science},
  volume={4},
  number={2},
  pages={87--99},
  year={1989},
  publisher={Institute of Mathematical Statistics}
}

@book{dryden1998,
  title={Statistical analysis of shape},
  author={Dryden, Ian L and Mardia, Kantilal Varichand},
  year={1998},
  publisher={Wiley}
}

@article{younes1998,
  title={Computable elastic distances between shapes},
  author={Younes, Laurent},
  journal={SIAM Journal on Applied Mathematics},
  volume={58},
  number={2},
  pages={565--586},
  year={1998},
  publisher={SIAM}
}

@article{stocker2023,
  title={Functional additive models on manifolds of planar shapes and forms},
  author={St{\"o}cker, Almond and Steyer, Lisa and Greven, Sonja},
  journal={Journal of Computational and Graphical Statistics},
  volume={32},
  number={4},
  pages={1600--1612},
  year={2023},
  publisher={Taylor \& Francis}
}

@article{dai2018,
  title={Principal component analysis for functional data on Riemannian manifolds and spheres},
  author={Dai, Xiongtao and M{\"u}ller, Hans-Georg},
journal={The Annals of Statistics},
  volume={46},
  number={6B},
  pages={3334--3362},
  year={2018}
}

@article{dai2022,
  title={Statistical inference on the Hilbert sphere with application to random densities},
  author={Dai, Xiongtao},
  journal={Electronic Journal of Statistics},
  volume={16},
  number={1},
  pages={700--736},
  year={2022},
  publisher={The Institute of Mathematical Statistics and the Bernoulli Society}
}

@article{opencv_library,
		author = {Bradski, G.},
		journal = {Dr. Dobb's Journal of Software Tools},
		title = {{The OpenCV Library}},
		year = {2000}
	}

@article{heart,
  title={Heart curve},
  author={Weisstein, Eric W},
  journal={https://mathworld. wolfram. com/},
  year={2003},
  publisher={Wolfram Research, Inc.}
}

@article{bat,  
	title={Flying Grey long eared bat isolated on white background stock photo}, 
	author={CreativeNature nl}, 
	journal={https://www.istockphoto.com/photo/}, 
	year={2020},
	publisher={iStock, Getty Images} 
}

@incollection{ICP-review,
  title={Iterative closest point (ICP)},
  author={Zhang, Zhengyou},
  booktitle={Computer vision: a reference guide},
  pages={718--720},
  year={2021},
  publisher={Springer}
}

@manual{R-shapes,
 title = {shapes: Statistical Shape Analysis},
 author = {Ian L. Dryden},
 year = {2023},
 note = {R package version 1.2.7},
 url = {https://CRAN.R-project.org/package=shapes}
}

@manual{R-stats,
  title = {The R Stats Package},
  author = {{R Core Team}},
  year = {2023},
  note = {R package version 4.3.0},
  url = {https://CRAN.R-project.org/package=stats}
}
\bibliographystyle{apalike}

\newpage 
\appendix 
\section{Additional figures}
\label{appendix}

This section presents additional results obtained from the analysis of the bat, horseshoe, and spoon datasets introduced in Section~\ref{app}. The same analysis pipeline as for the butterfly and fork datasets was applied, following the alignment and modeling procedures described in Sections~\ref{Resu_align} and \ref{Resu_mod}.

The alignment results are displayed in Figure~\ref{align-2}. Overall, the estimated shapes and their associated coordinate functions appear to be well aligned with the reference function $\hat{\boldsymbol{\mu}}$. The spoon dataset is particularly interesting, as the bowl of the spoon does not always appear on the same side of the contour. This illustrates a potential identifiability issue for certain shapes. Nevertheless, the corresponding coordinate functions remain visually well aligned.

The modeling results for $\mathbf{C}$ are presented in Figures~\ref{bat-gen}, \ref{horse-gen}, and \ref{spoon-gen}. Compared with the standard multivariate functional PCA (MFPCA) approach, the proposed method generally generates more realistic shapes, particularly for the bat and horseshoe datasets.

%This section presents the results obtained from the analysis of the bat, horseshoe, and spoon datasets mentioned in Section \ref{app}. We applied the same procedure as for the butterfly and fork datasets, as described in Sections \ref{Resu_align} and \ref{Resu_mod}.

%The results from the alignment procedure are illustrated in Figure \ref{align-2}. We observe that the obtained shapes and their corresponding coordinate functions are visually close to the function of reference $\boldsymbol{\mu}$. The case of the spoon dataset is particularly interesting: the bowl of the spoons is not always in the same position. This highlights an identifiability issue for certain shapes. However, their coordinate functions remain visually well-aligned.

%The results related to the modeling of $\mathbf{C}$ are shown in Figures \ref{bat-gen}, \ref{horse-gen} and \ref{spoon-gen}. Compared to the principal concurrent model (MFPCA), our proposed method generally produces more realistic shapes, especially for the bat and horseshoe datasets.
\begin{figure}[H]
    \centering
   \begin{tabular}{ c c c}
   $\hat{\boldsymbol{\mu}}$& \includegraphics[align=c, width=0.1\linewidth]{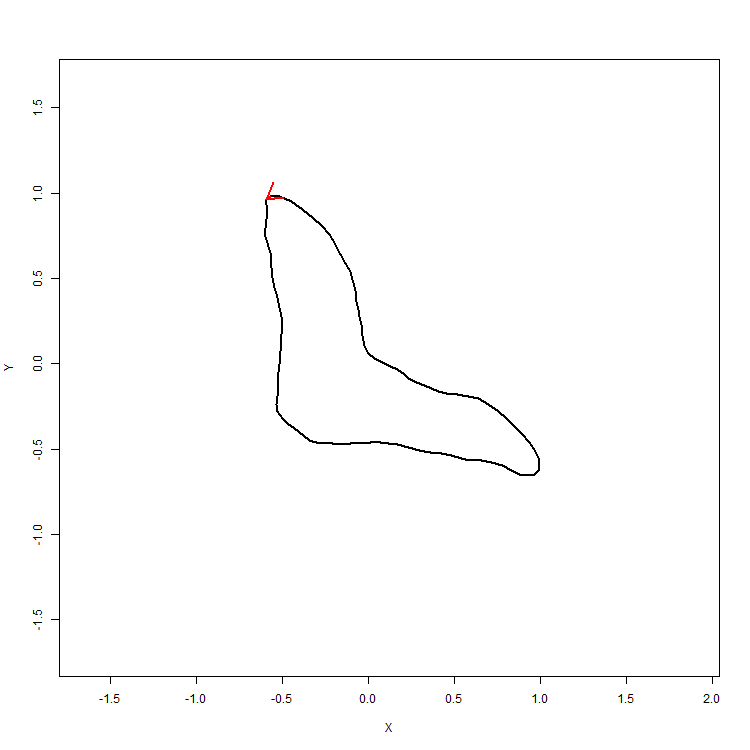}    &  \begin{tabular}{c c c c } 
    $\hat{\mathbf{c}}^*$& 
       \includegraphics[align=c, width=0.15\linewidth]{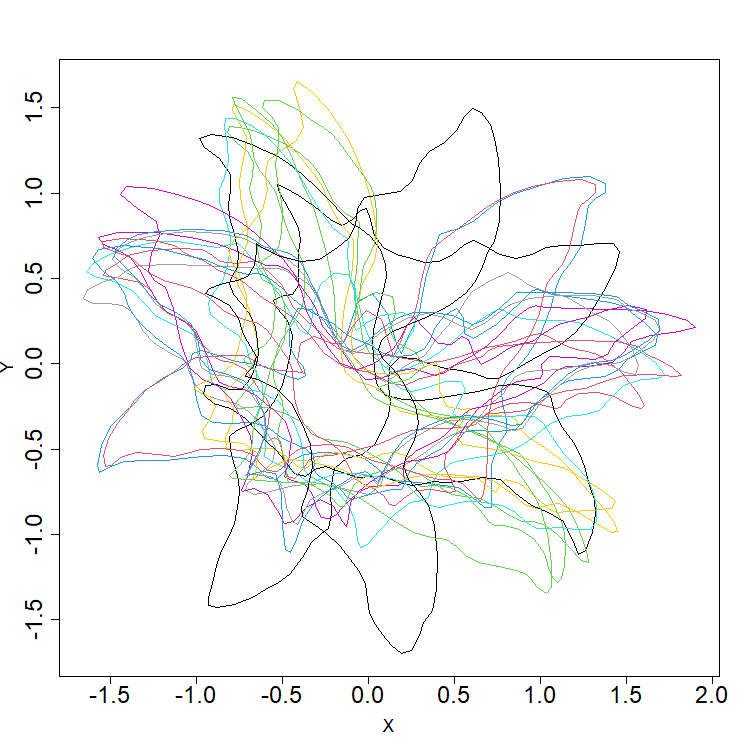} & 
          \includegraphics[align=c, width=0.15\linewidth]{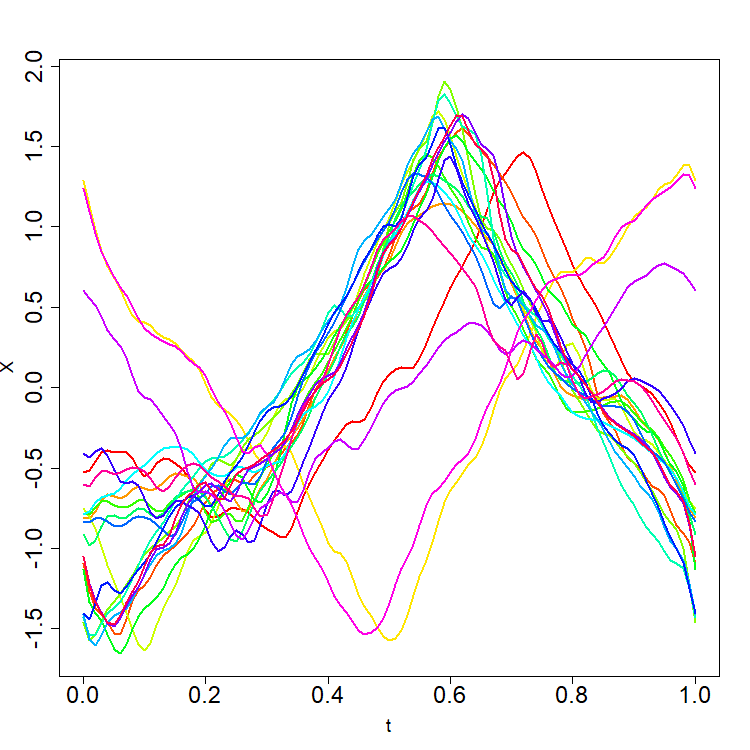} &
          \includegraphics[align=c, width=0.15\linewidth]{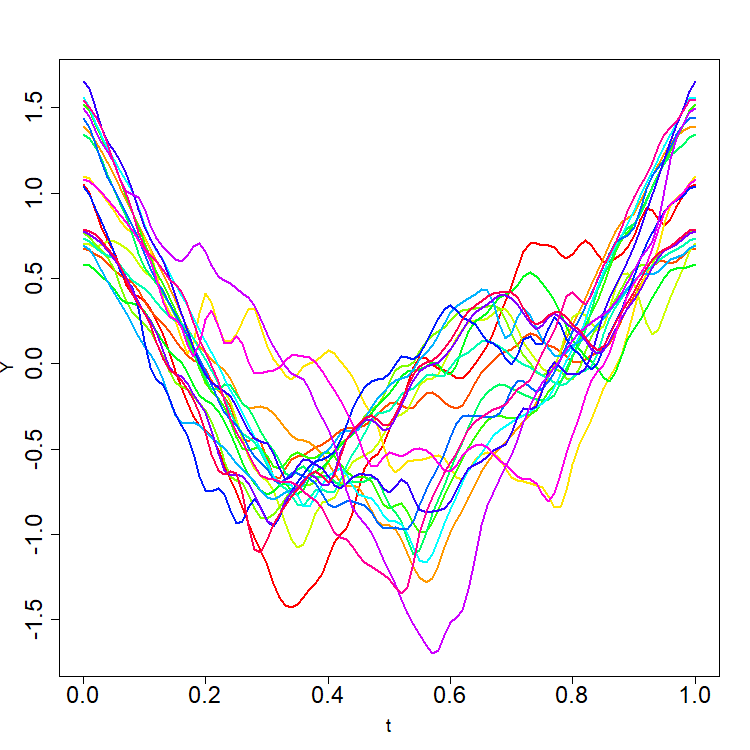} \\
      $\hat{\tilde{\mathbf{c}}}$& \includegraphics[align=c, width=0.15\linewidth]{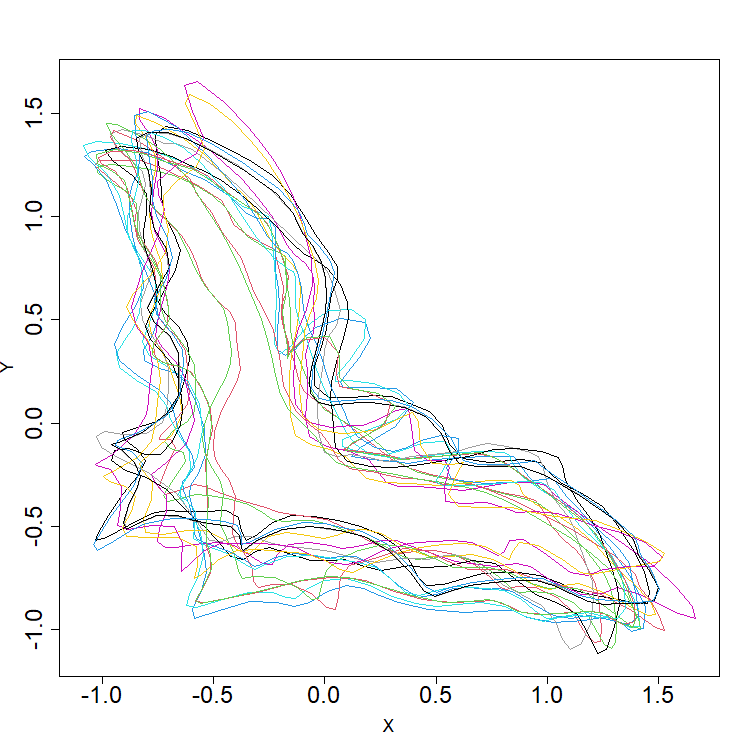} & 
          \includegraphics[align=c, width=0.15\linewidth]{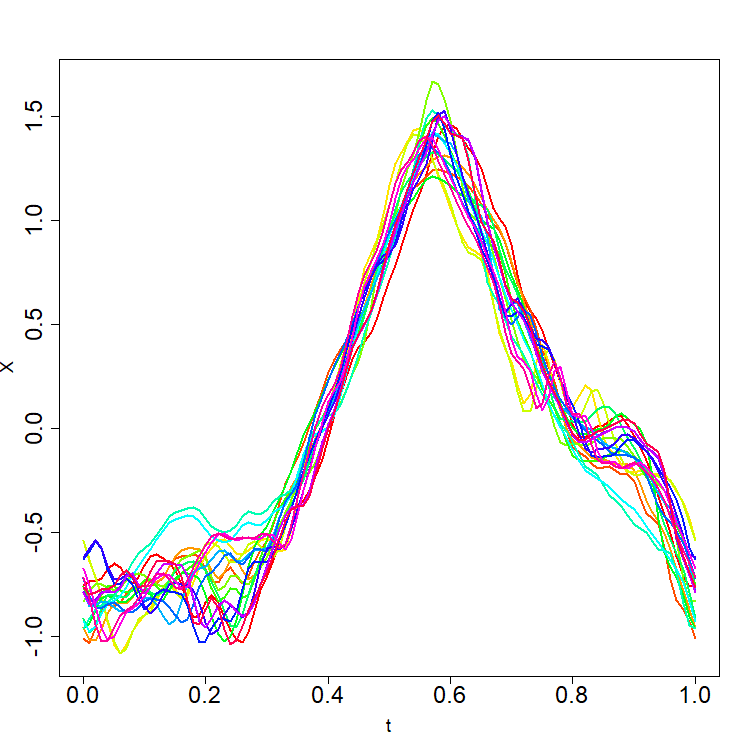} &
          \includegraphics[align=c, width=0.15\linewidth]{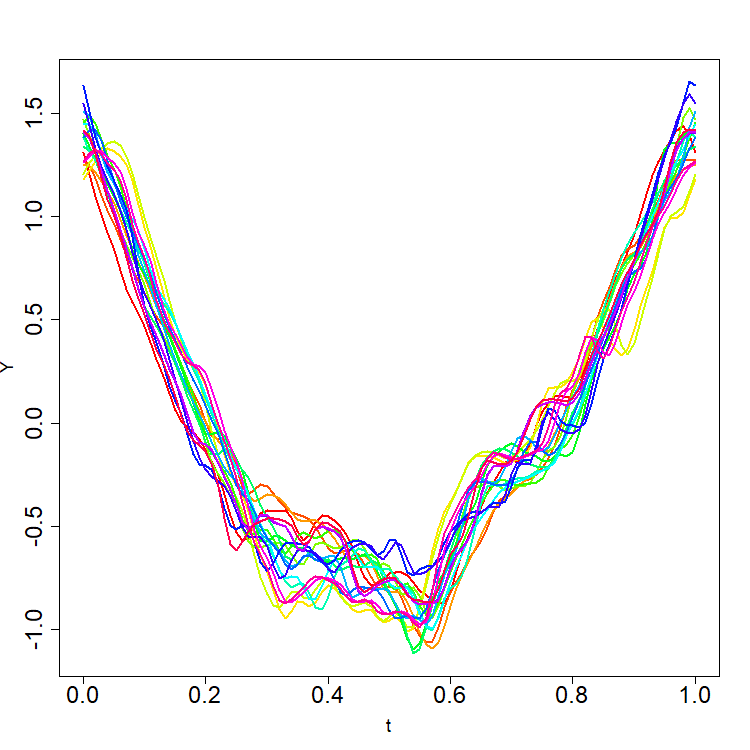}    
   \end{tabular} \\ 
   \hline \\ 
   
   $\hat{\boldsymbol{\mu}}$& \includegraphics[align=c, width=0.154\linewidth]{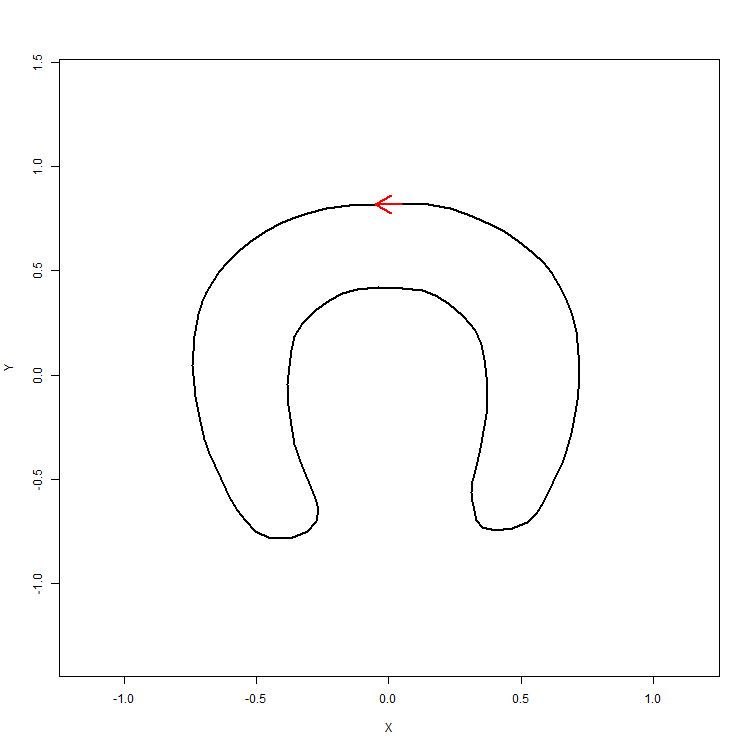}    &  \begin{tabular}{c c c c } 
    $\hat{\mathbf{c}}^*$& 
       \includegraphics[align=c, width=0.15\linewidth]{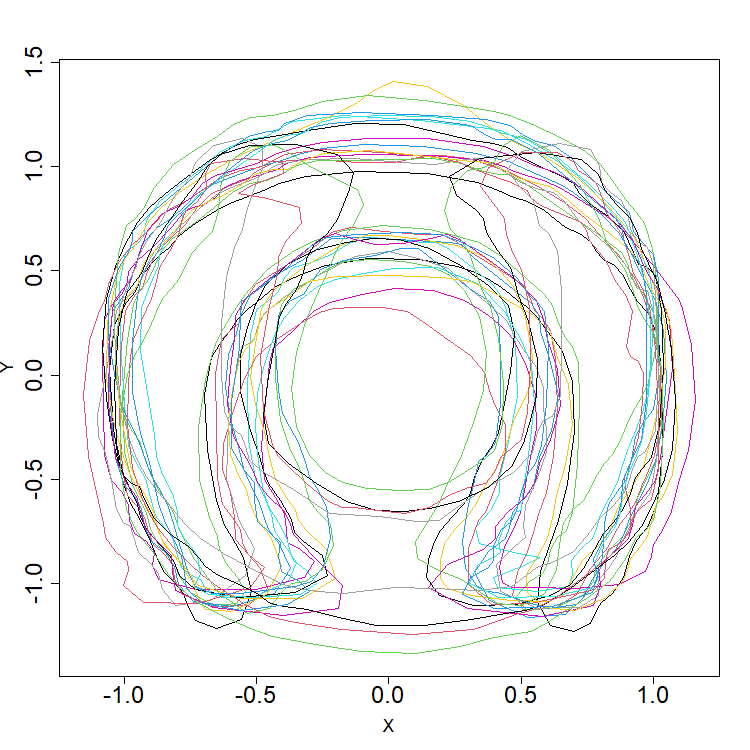} & 
          \includegraphics[align=c, width=0.15\linewidth]{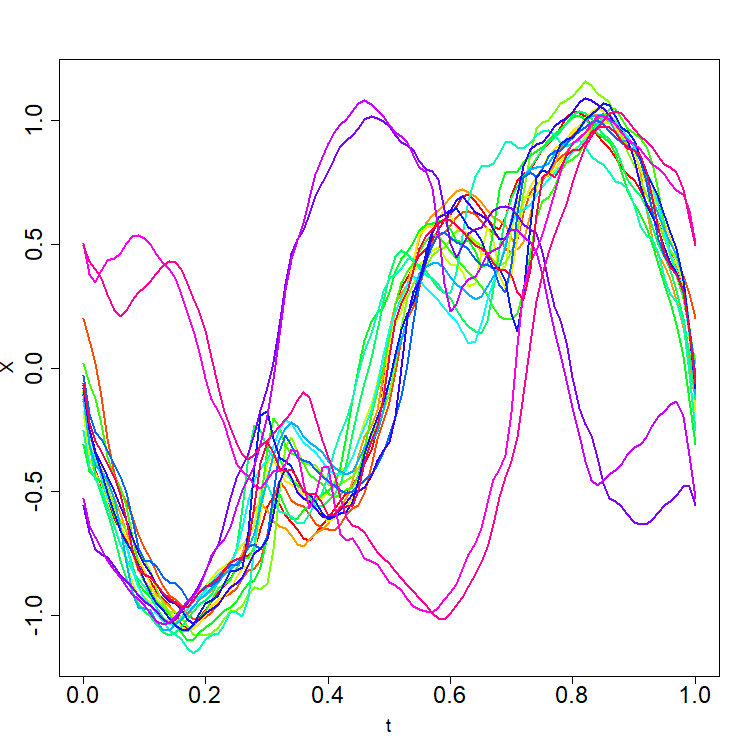} &
          \includegraphics[align=c, width=0.15\linewidth]{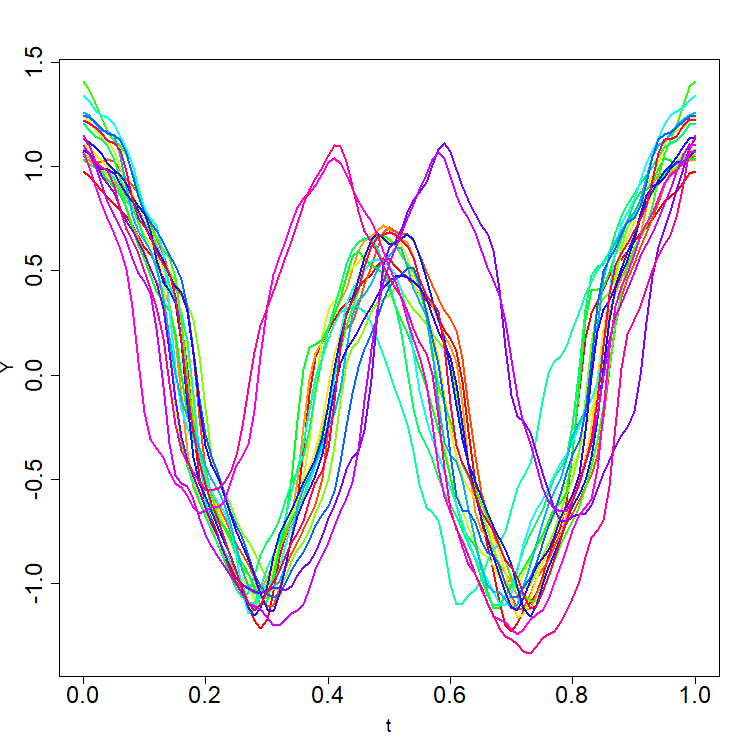} \\
      $\hat{\tilde{\mathbf{c}}}$& \includegraphics[align=c, width=0.15\linewidth]{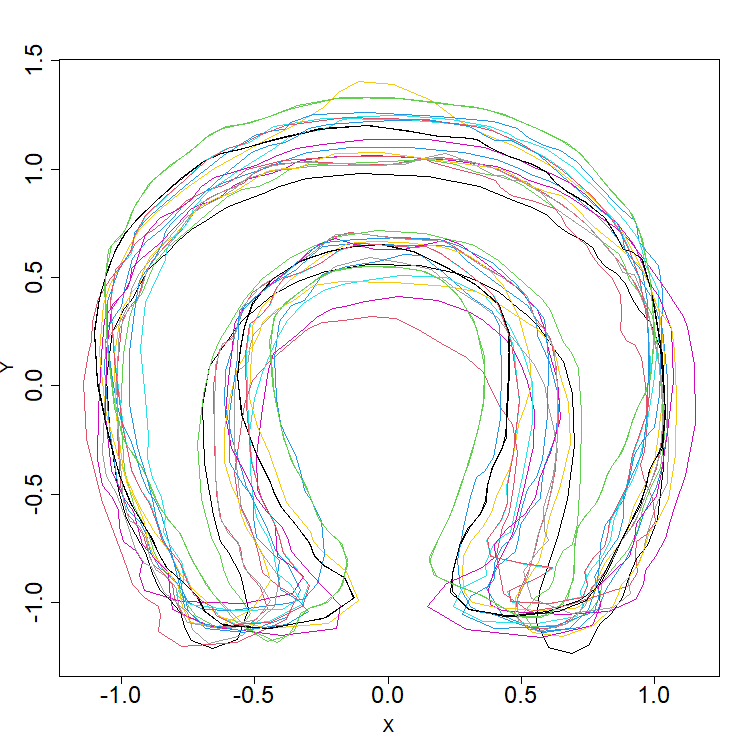} & 
          \includegraphics[align=c, width=0.15\linewidth]{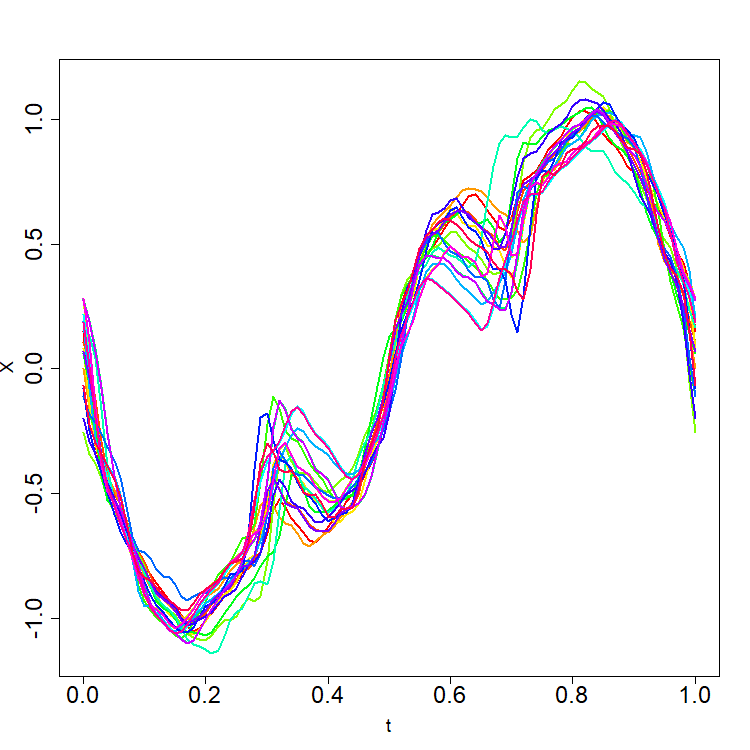} &
          \includegraphics[align=c, width=0.15\linewidth]{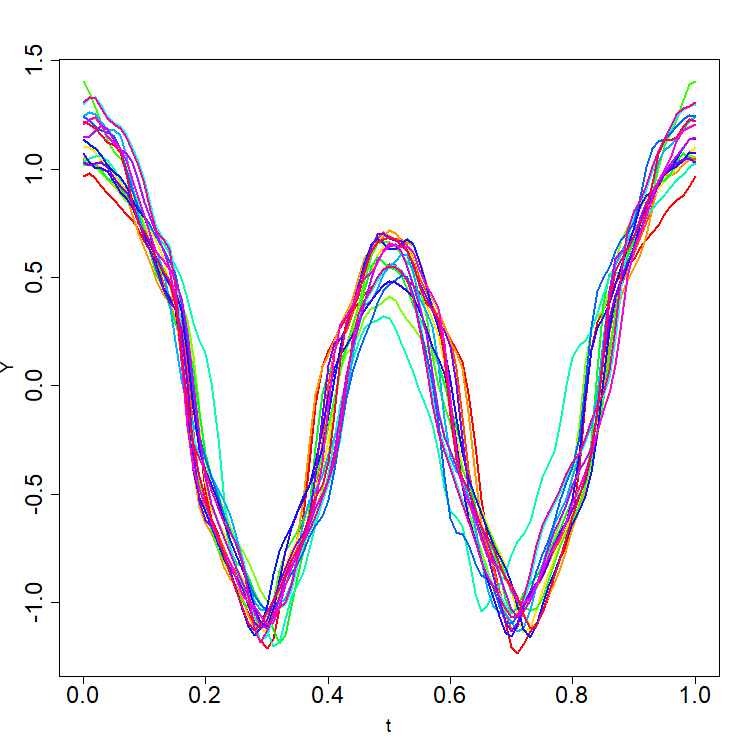}    
   \end{tabular} \\ 
        \hline \\ 
   $\hat{\boldsymbol{\mu}}$& \includegraphics[align=c, width=0.153\linewidth]{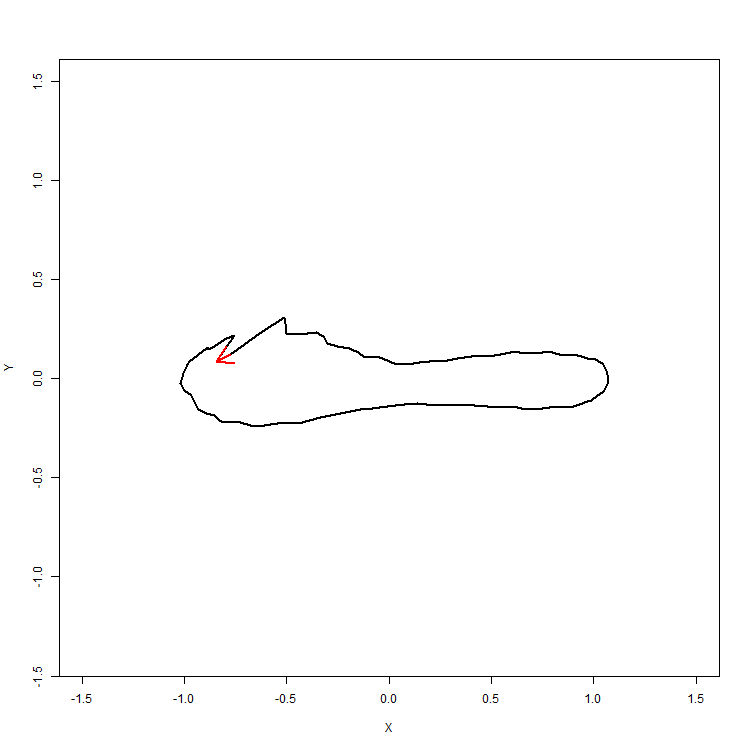} &   \begin{tabular}{c c c c } 
    $\hat{\mathbf{c}}^*$& 
       \includegraphics[align=c, width=0.15\linewidth]{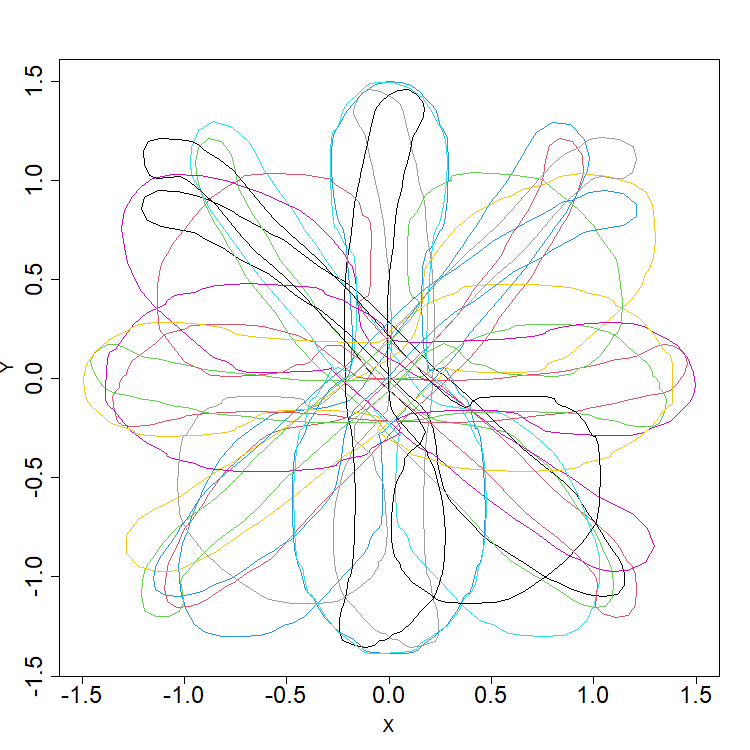} & 
          \includegraphics[align=c, width=0.15\linewidth]{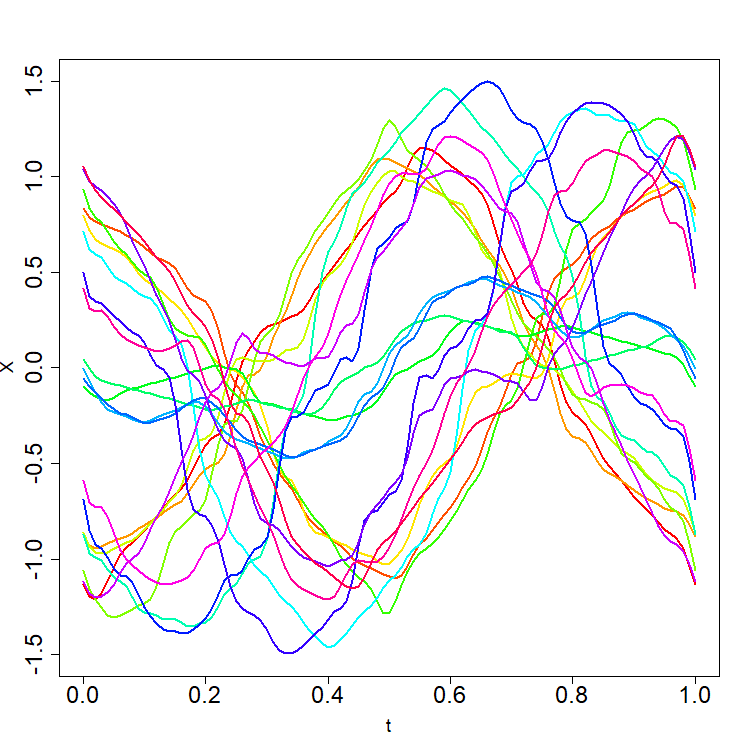} &
          \includegraphics[align=c, width=0.15\linewidth]{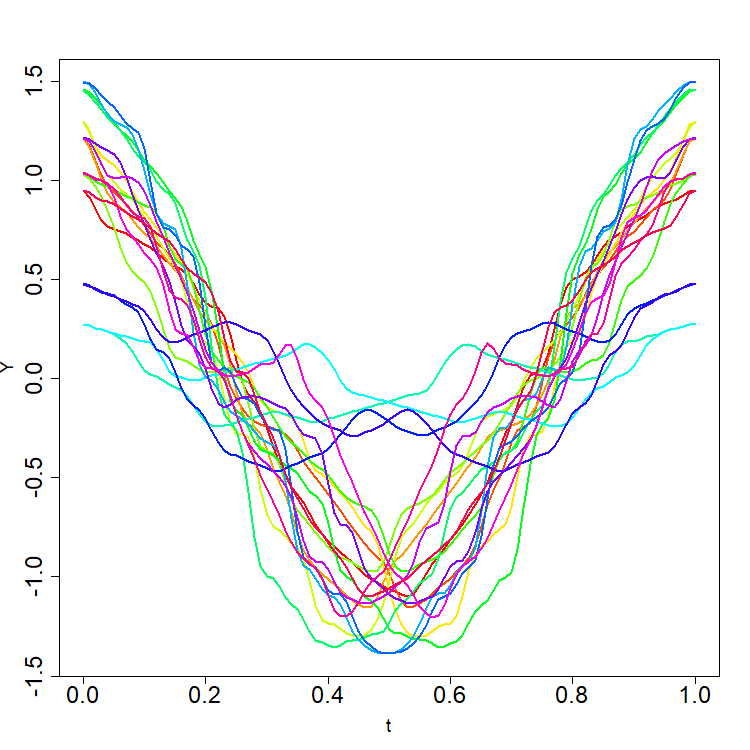} \\
      $\hat{\tilde{\mathbf{c}}}$& \includegraphics[align=c, width=0.15\linewidth]{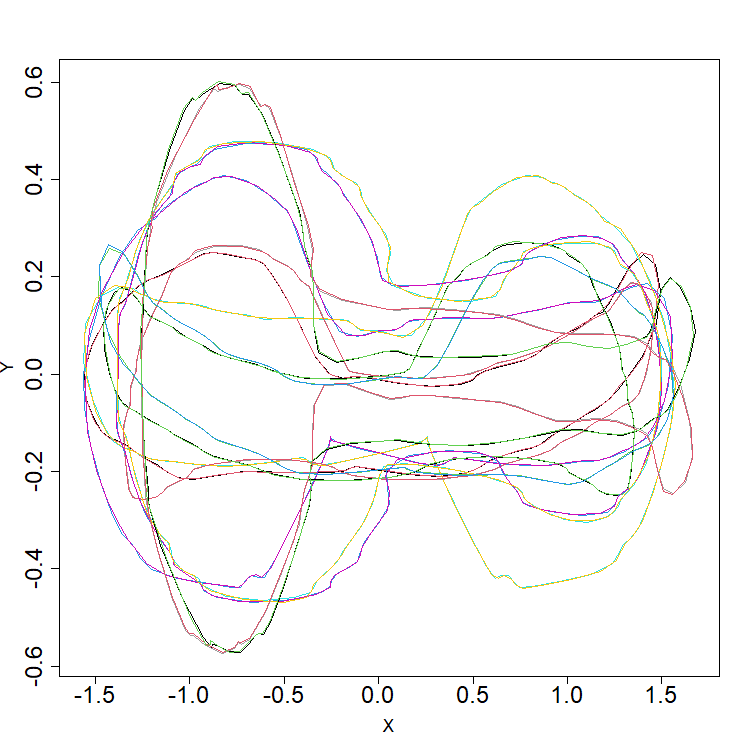} & 
          \includegraphics[align=c, width=0.15\linewidth]{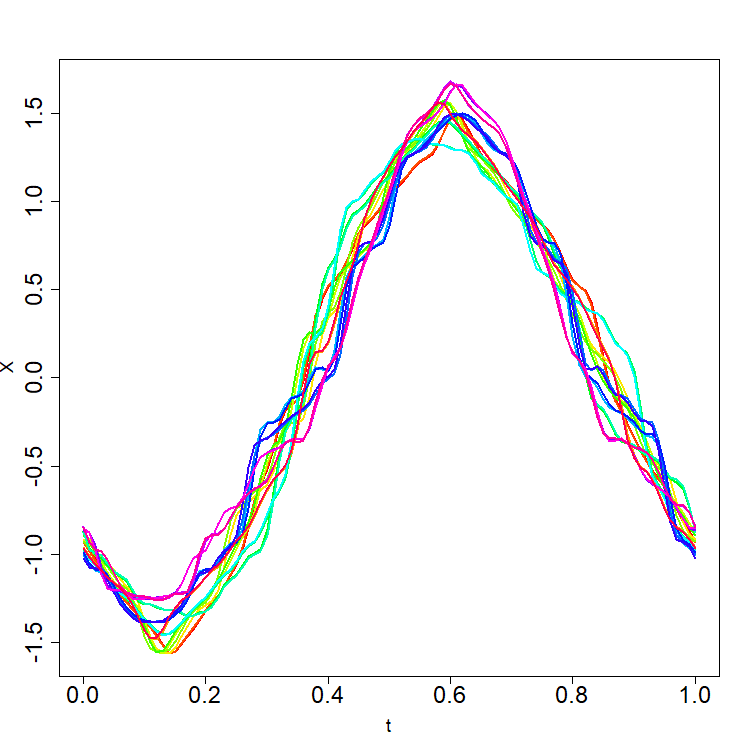} &
          \includegraphics[align=c, width=0.15\linewidth]{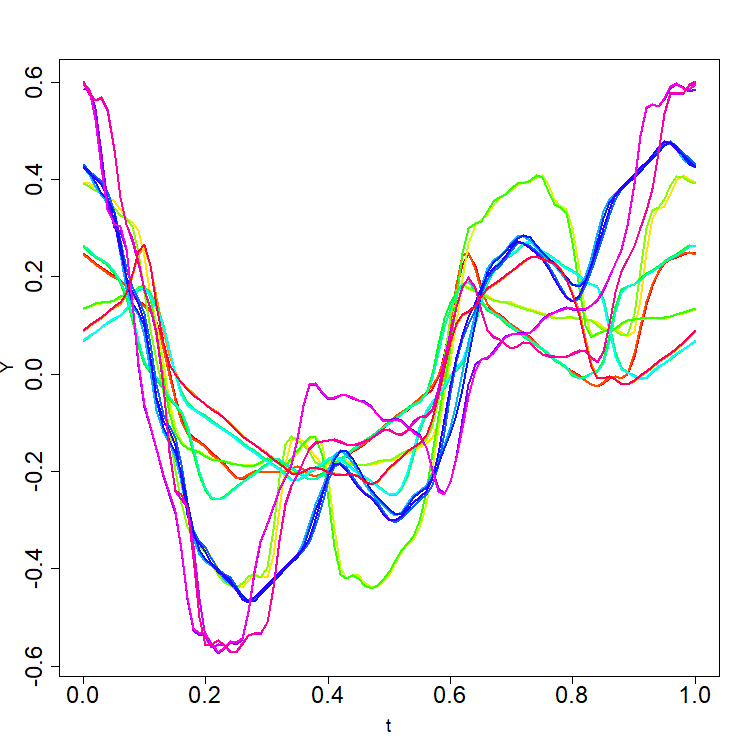}    
   \end{tabular} 
    \end{tabular}
    \caption{Illustration of the alignment procedure for the bat dataset (top block), the horseshoe dataset (second block) and the spoon dataset (bottom block). In each block, the left column displays the estimated Fréchet mean. The first row shows estimated pre-shapes together with their coordinate functions, while the second row presents the corresponding aligned shapes.}
    \label{align-2}
\end{figure}

\begin{figure}[H]
    \centering
\begin{tabular}{c c c c c c c  c c c c c c}
$D_1$: & 0.03 & 0.03 & 0.06 & 0.05 & 0.03 \\ 
  $D_2$: & 0.01 & 0.01 & 0.01 & 0.02 & 0.01 \\ 
 (a)& \includegraphics[align=c, width=.13\textwidth]{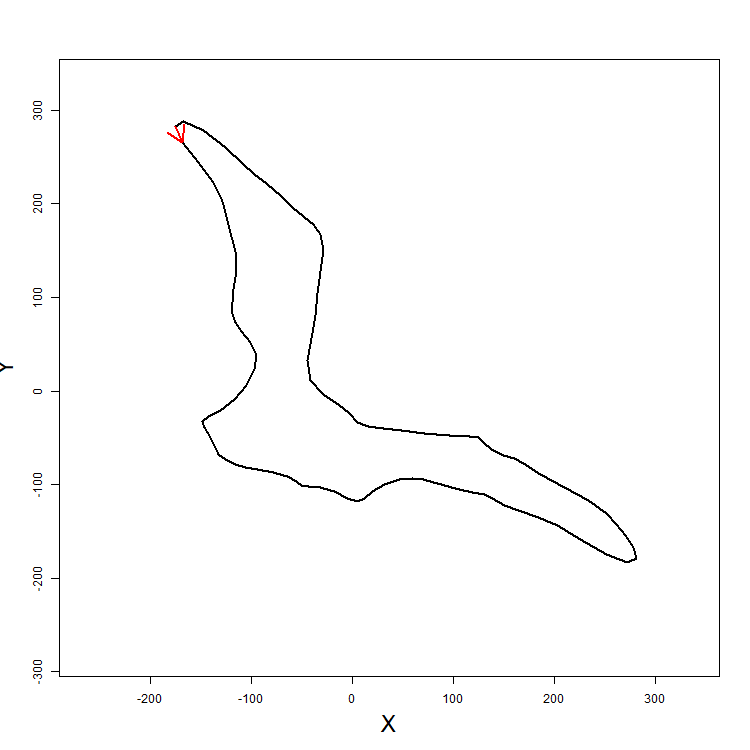} &\includegraphics[align=c, width=.13\textwidth]{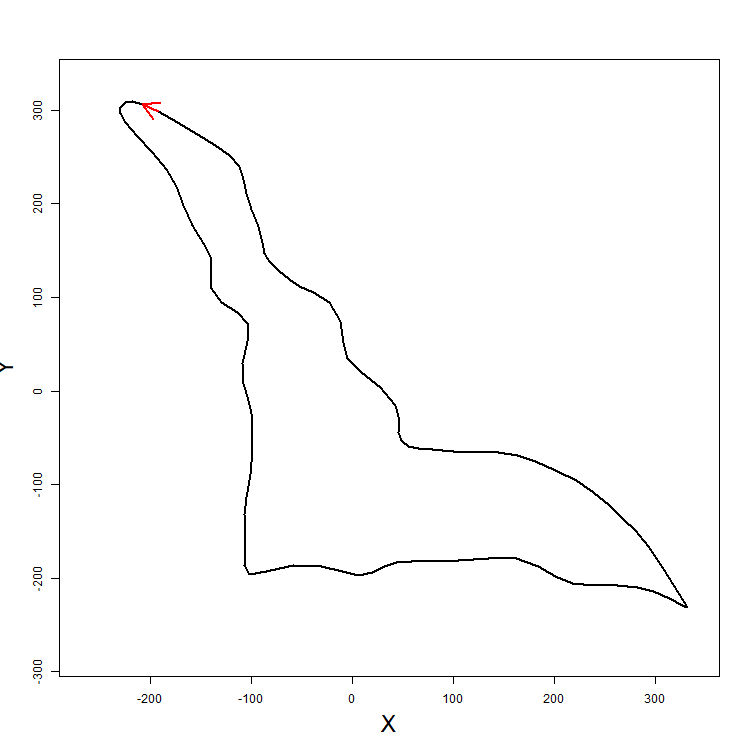}& \includegraphics[align=c, width=.13\textwidth]{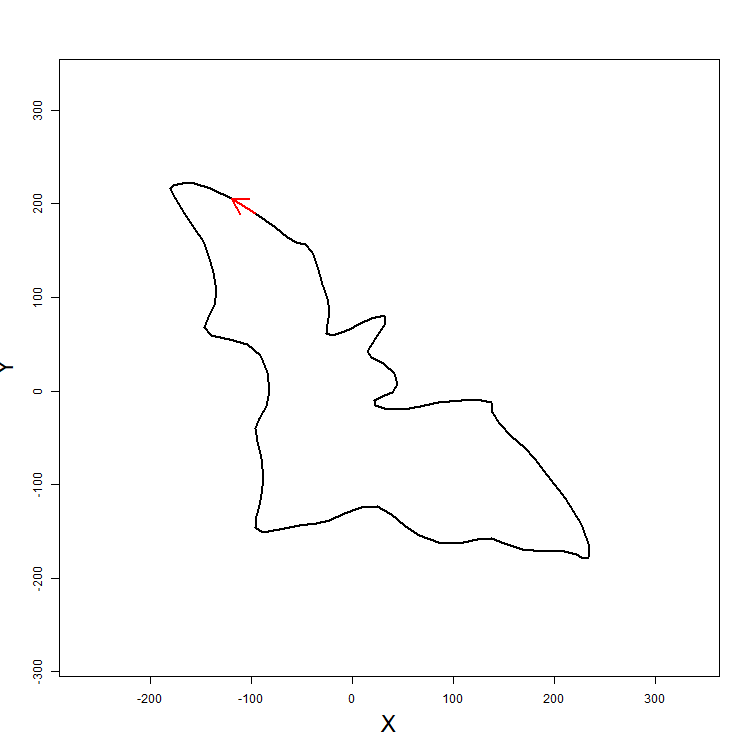} & \includegraphics[align=c, width=.13\textwidth]{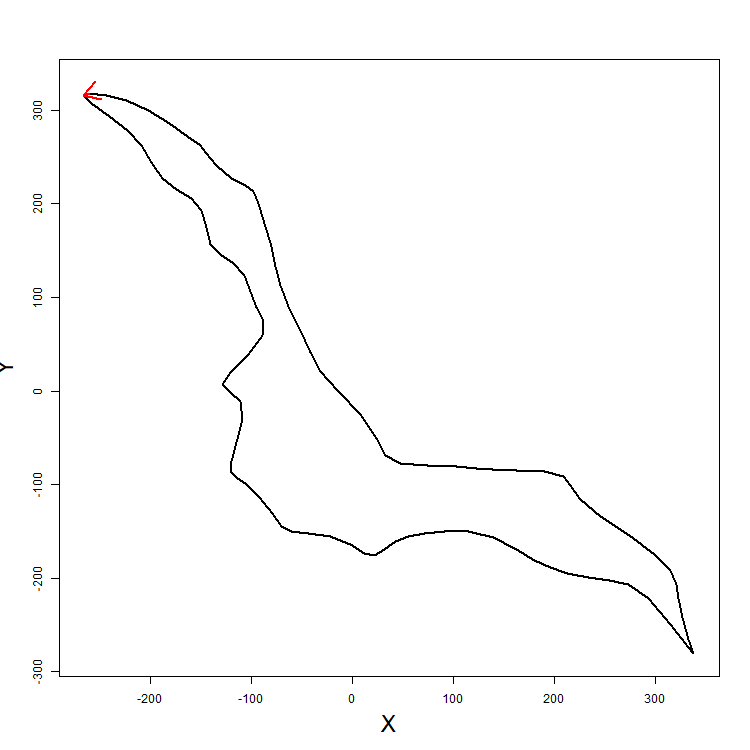} & \includegraphics[align=c, width=.13\textwidth]{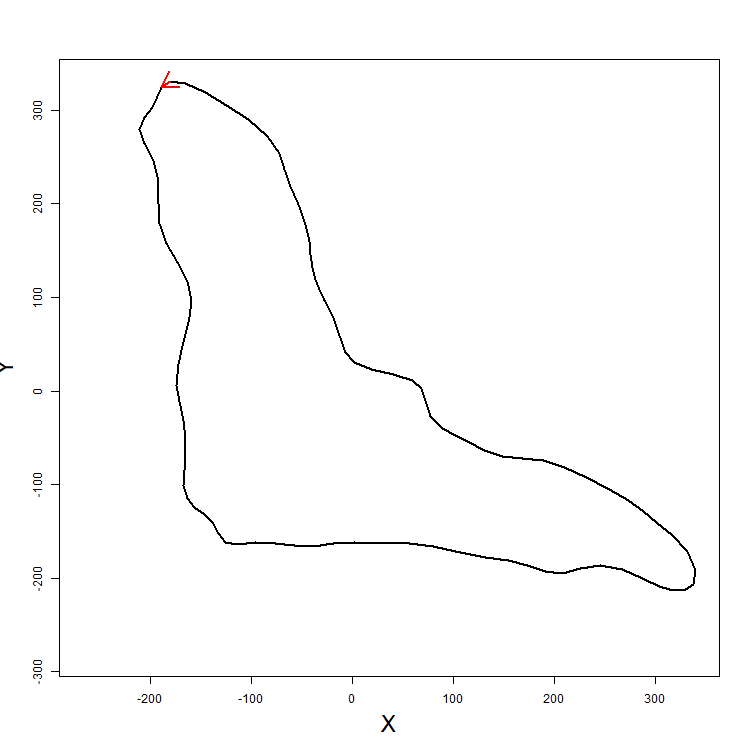}\\
 \hline
  $D_1$: & 0.03 & 0.05 & 0.07 & 0.06 & 0.03 \\ 
  $D_2$: & 0.01 & 0.03 & 0.02 & 0.02 & 0.01 \\ 
   (b)& \includegraphics[align=c, width=.13\textwidth]{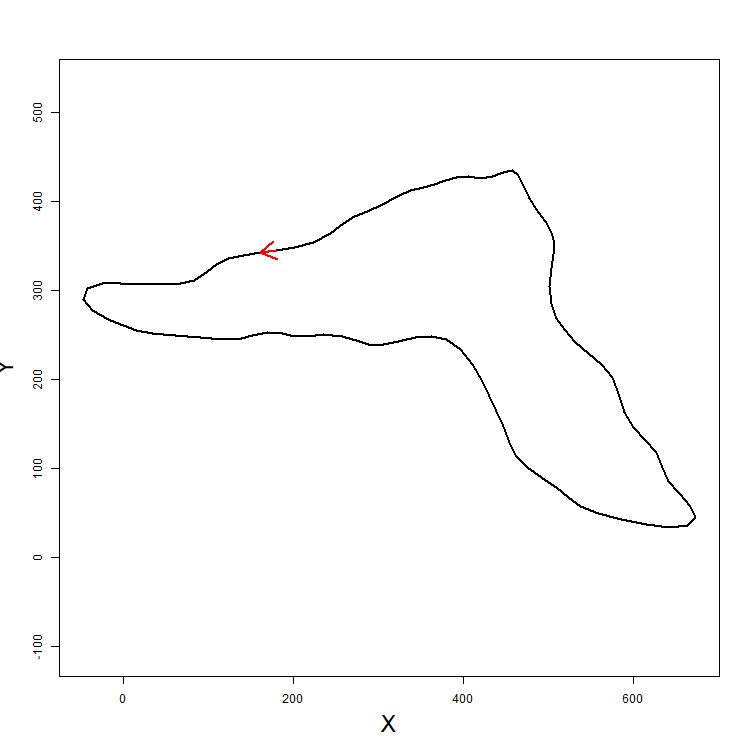} &\includegraphics[align=c, width=.13\textwidth]{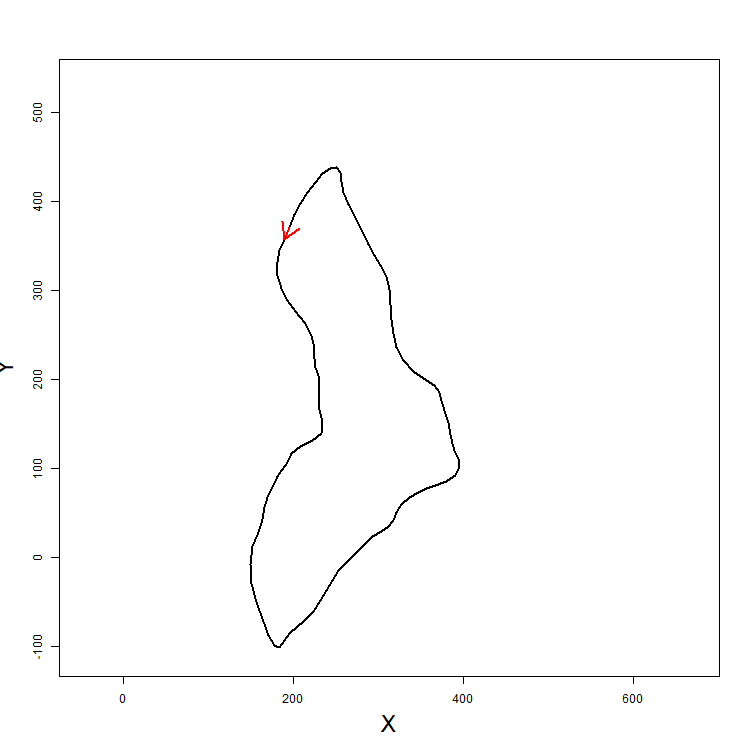}& \includegraphics[align=c, width=.13\textwidth]{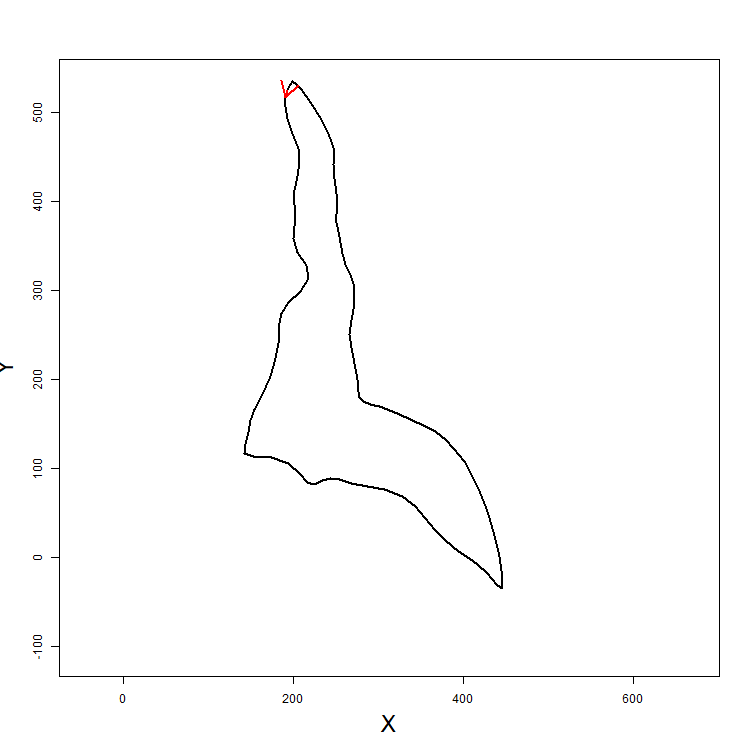} & \includegraphics[align=c, width=.13\textwidth]{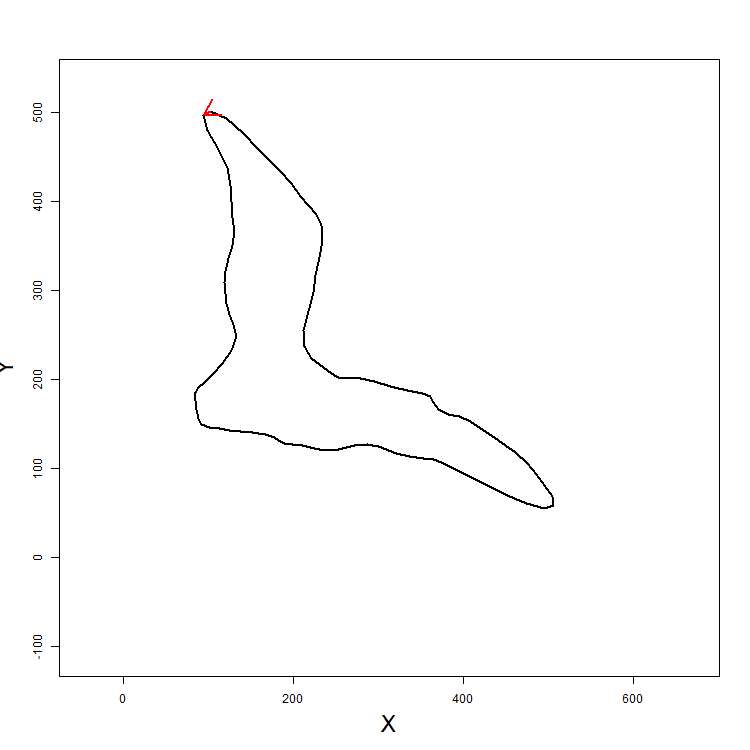} & \includegraphics[align=c, width=.13\textwidth]{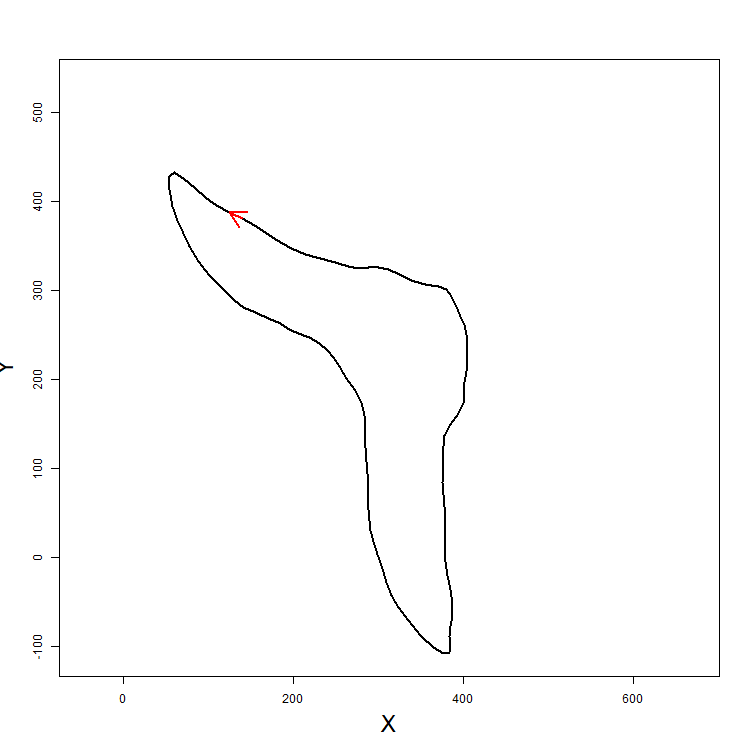}\\
   \hline
 $D_1$: & 0.80 & 0.06 & 0.16 & 0.12 & 0.13 \\ 
  $D_2$: & 0.68 & 0.03 & 0.10 & 0.08 & 0.07 \\  
    (c)&\includegraphics[align=c, width=.13\textwidth]{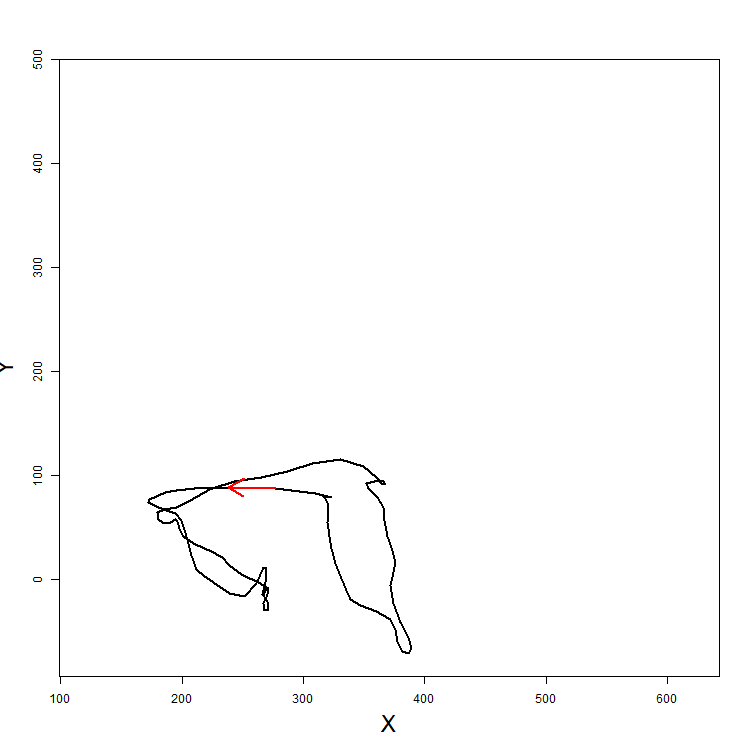} &\includegraphics[align=c, width=.13\textwidth]{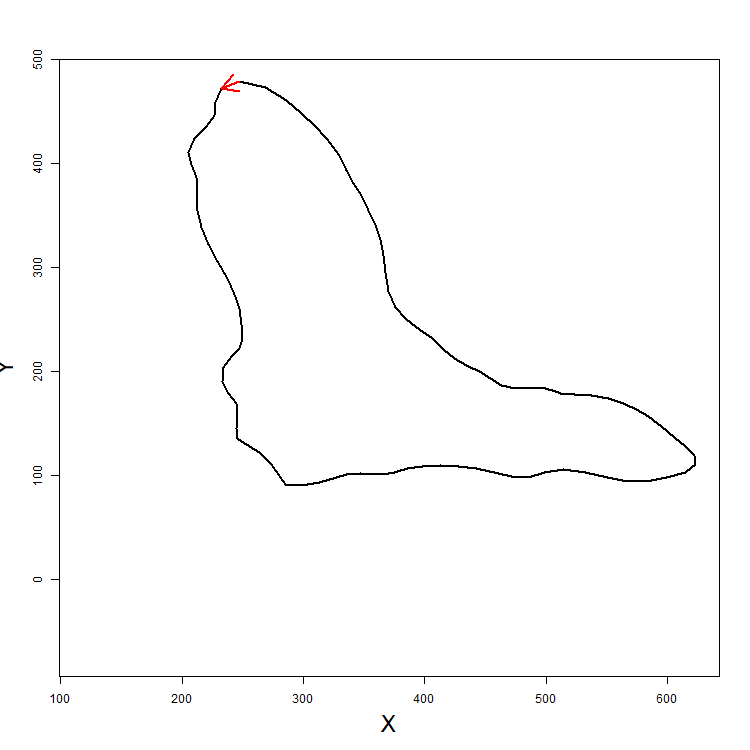}& \includegraphics[align=c, width=.13\textwidth]{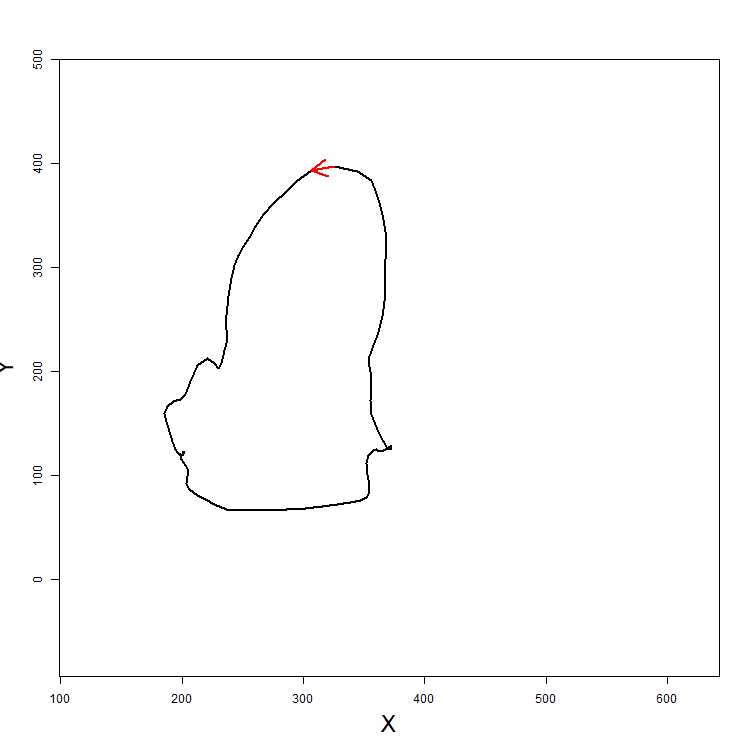} & \includegraphics[align=c, width=.13\textwidth]{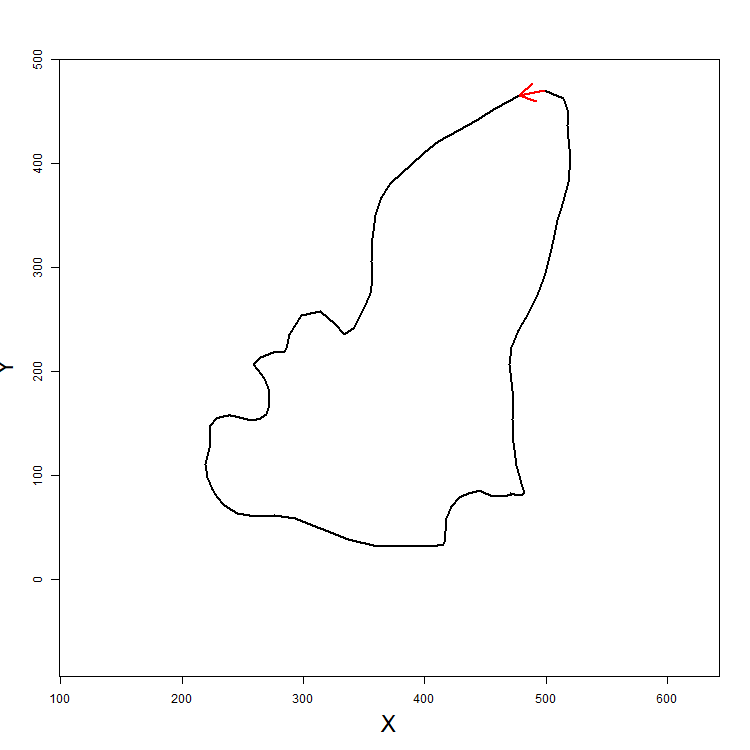} & \includegraphics[align=c, width=.13\textwidth]{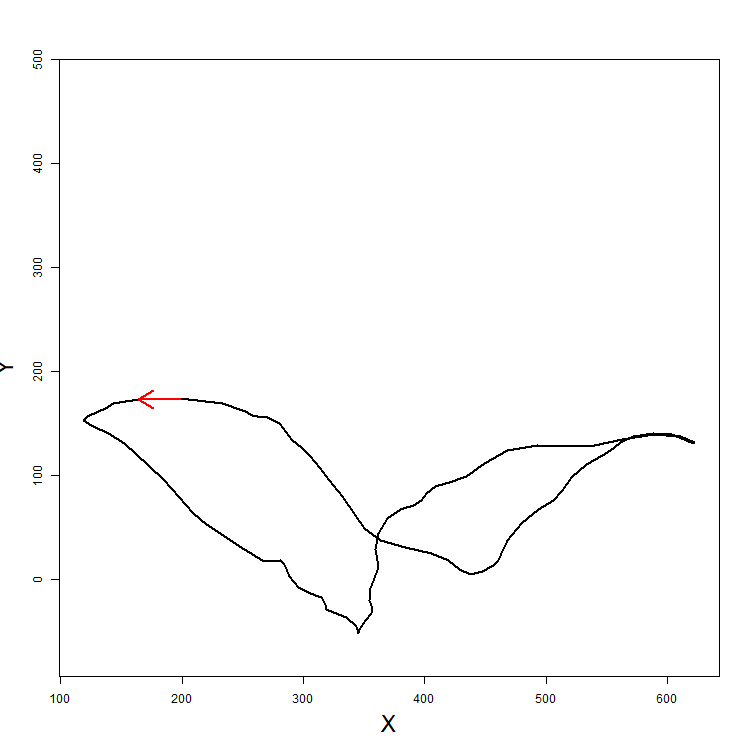} \\
\end{tabular}
\caption{Bat contours generated with (a) our approach without the deformation parameters, (b) our approach with the deformation parameters, and (c) with MFPCA.}
\label{bat-gen}
\end{figure}
\begin{figure}[H]
    \centering
    \begin{tabular}{c c c c c c c  c c c c c c}
$D_1$: & 0.04 & 0.03 & 0.21 & 0.08 & 0.04 \\ 
  $D_2$: & 0.01 & 0.01 & 0.12 & 0.02 & 0.01 \\  
 (a)& \includegraphics[align=c, width=.13\textwidth]{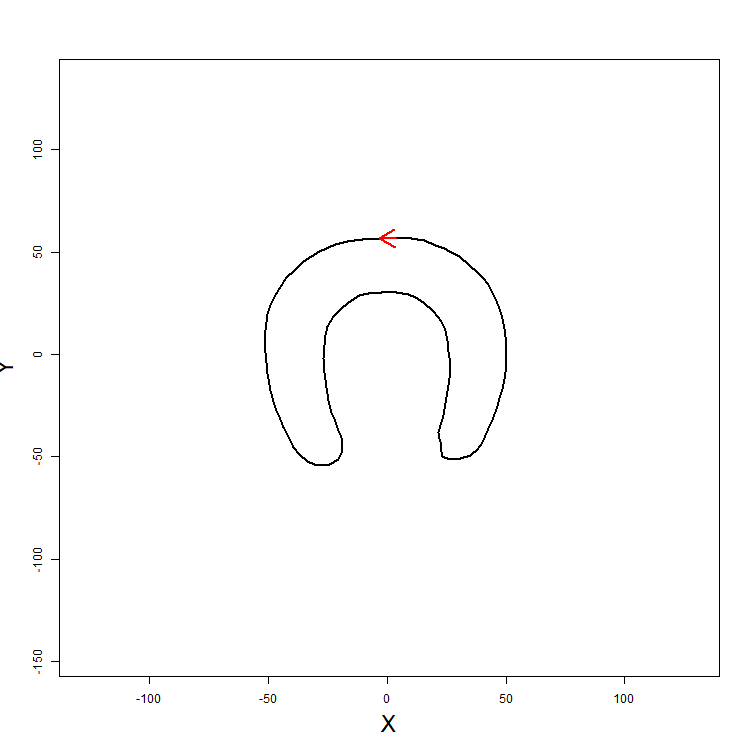} &\includegraphics[align=c, width=.13\textwidth]{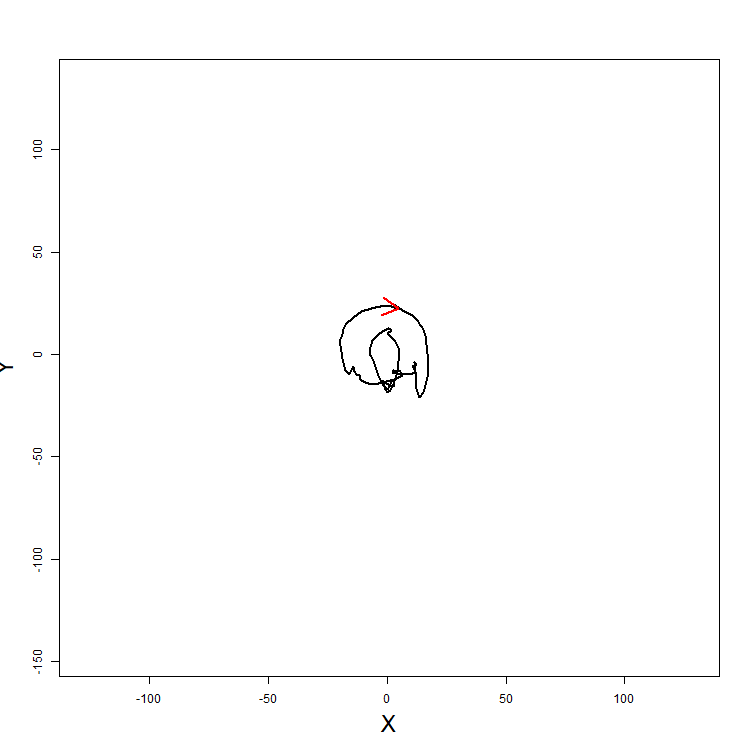}& \includegraphics[align=c, width=.13\textwidth]{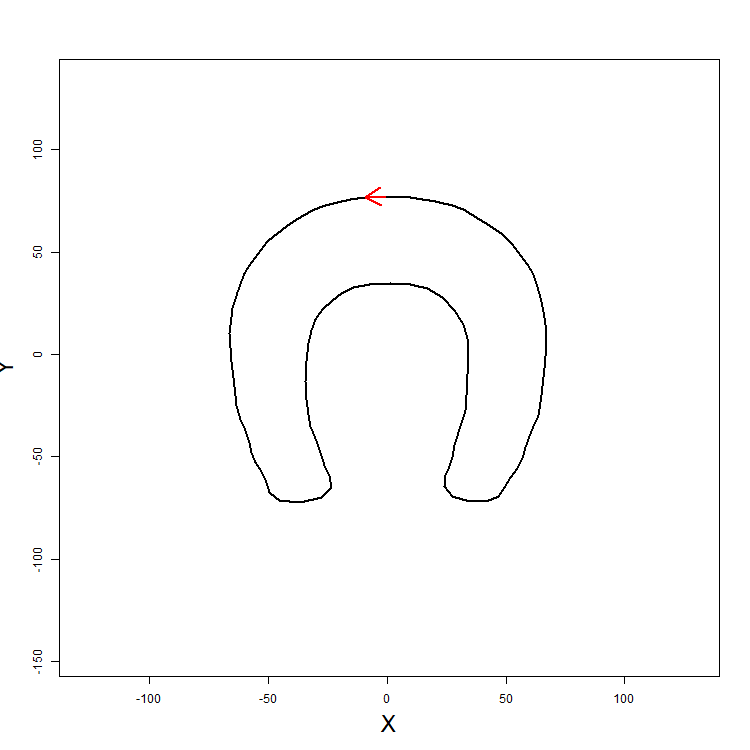} & \includegraphics[align=c, width=.13\textwidth]{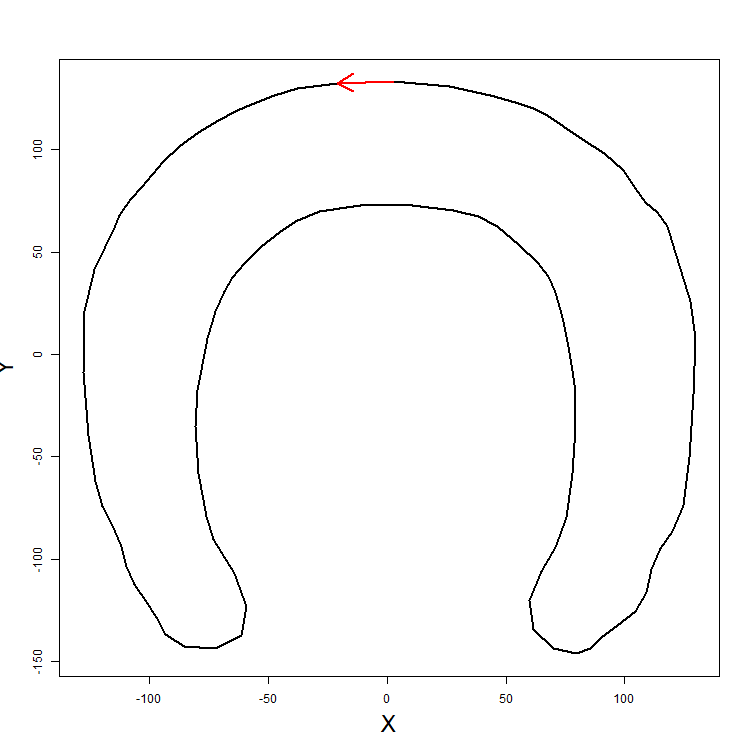} & \includegraphics[align=c, width=.13\textwidth]{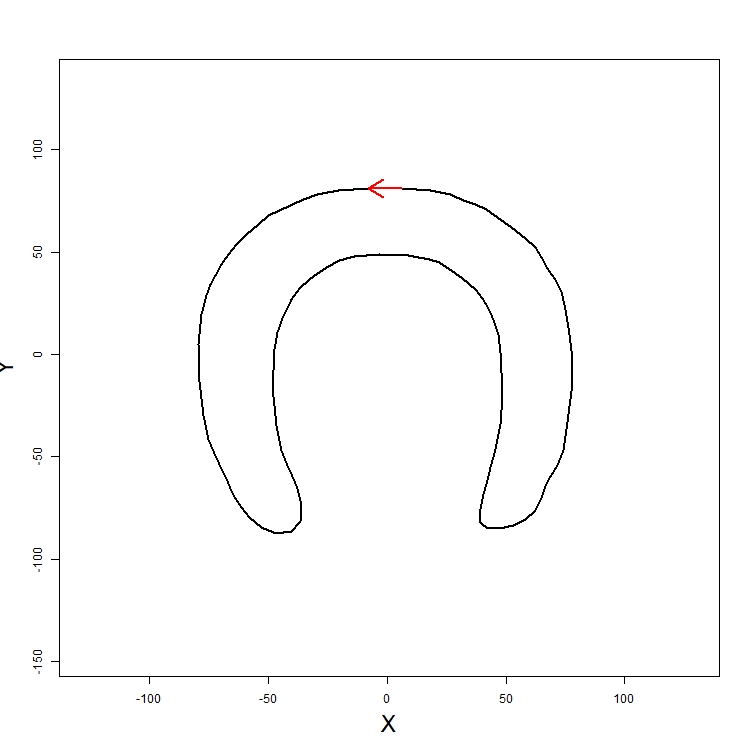}\\
 \hline
  $D_1$: & 0.05 & 0.06 & 0.22 & 0.09 & 0.04 \\ 
  $D_2$: & 0.02 & 0.03 & 0.14 & 0.02 & 0.01 \\ 
   (b)& \includegraphics[align=c, width=.13\textwidth]{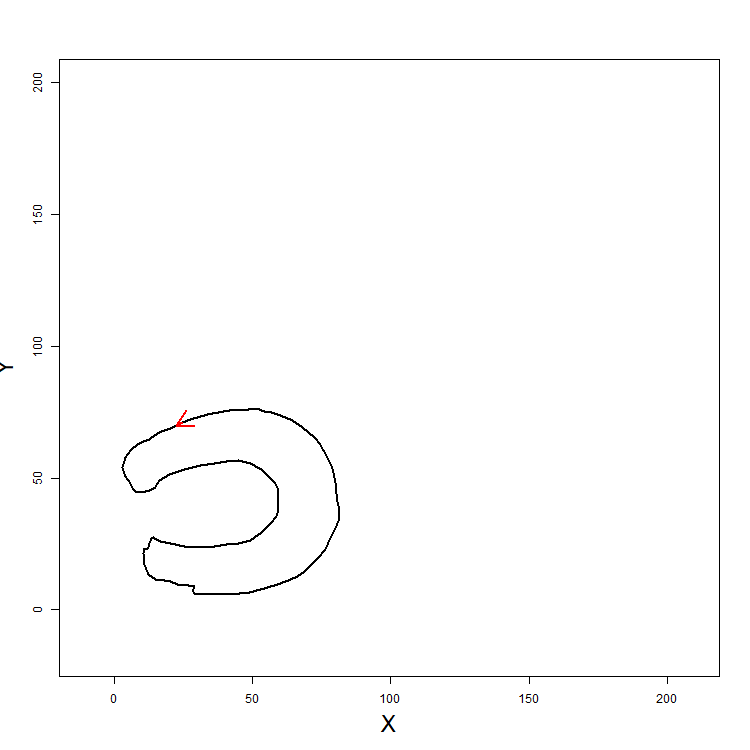} &\includegraphics[align=c, width=.13\textwidth]{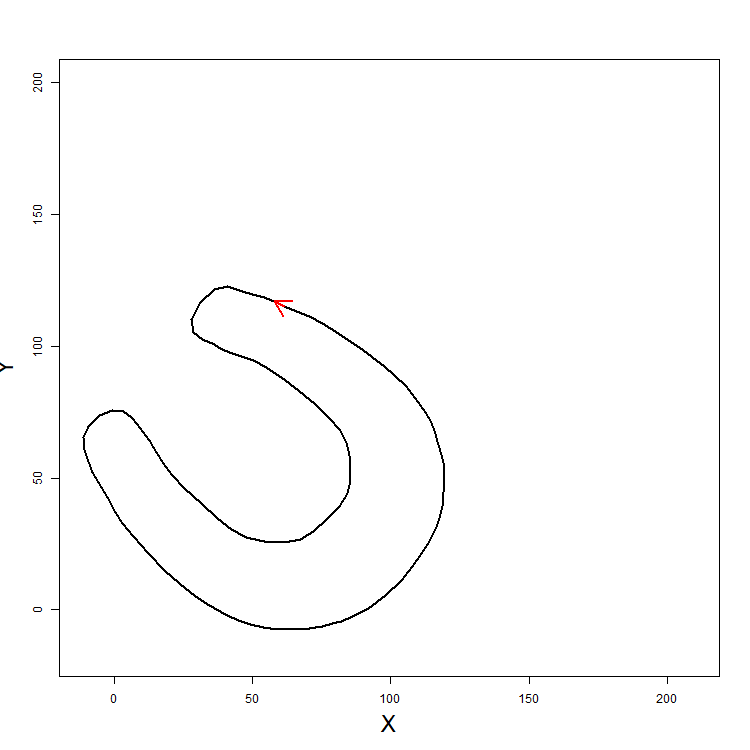}& \includegraphics[align=c, width=.13\textwidth]{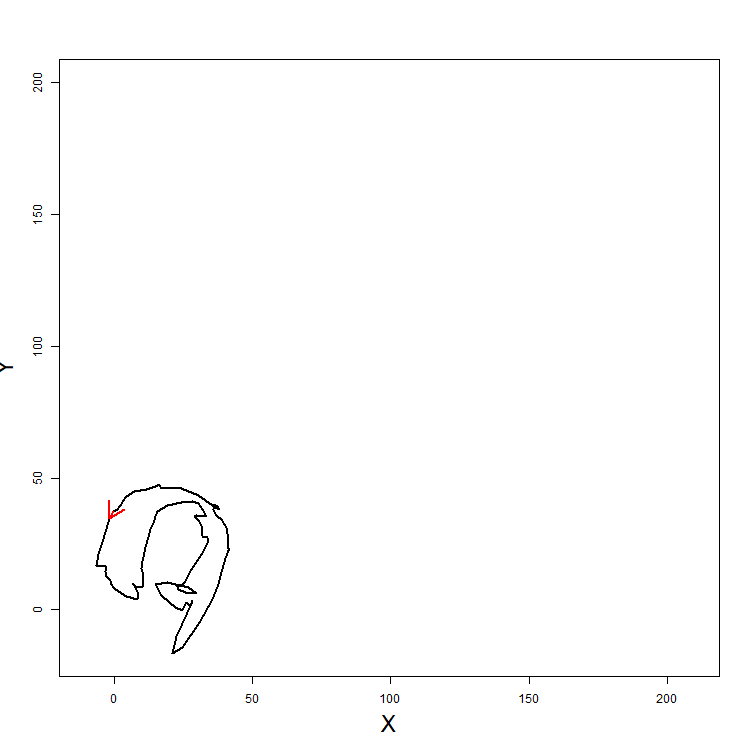} & \includegraphics[align=c, width=.13\textwidth]{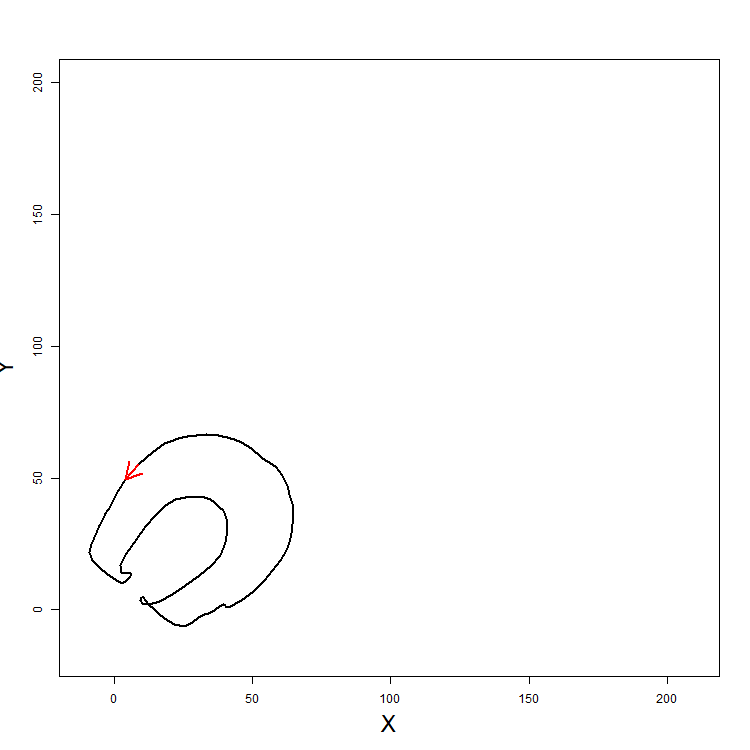} & \includegraphics[align=c, width=.13\textwidth]{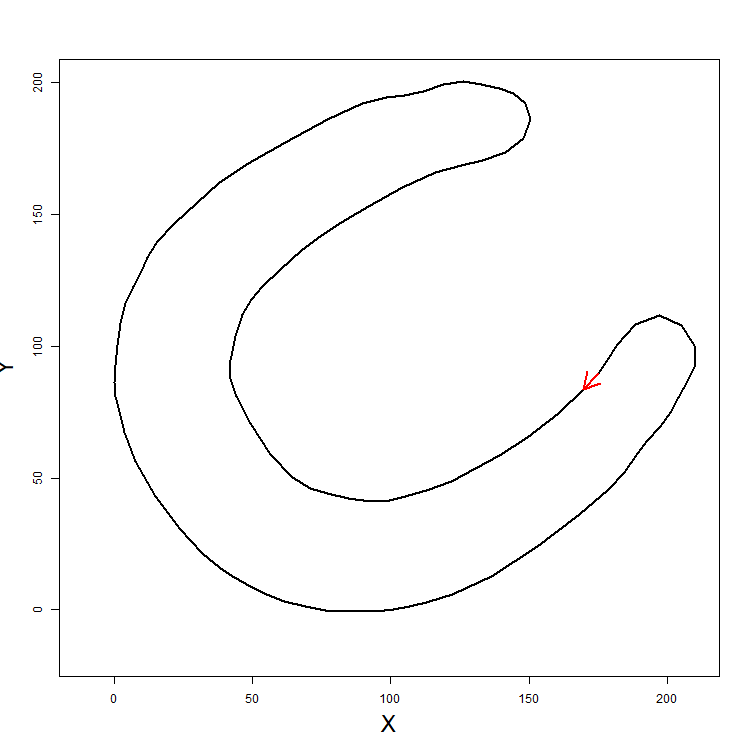}\\
   \hline
  $D_1$: & 0.11 & 0.05 & 0.18 & 0.08 & 0.13 \\ 
  $D_2$: & 0.09 & 0.02 & 0.15 & 0.06 & 0.11 \\ 
    (c)&\includegraphics[align=c, width=.13\textwidth]{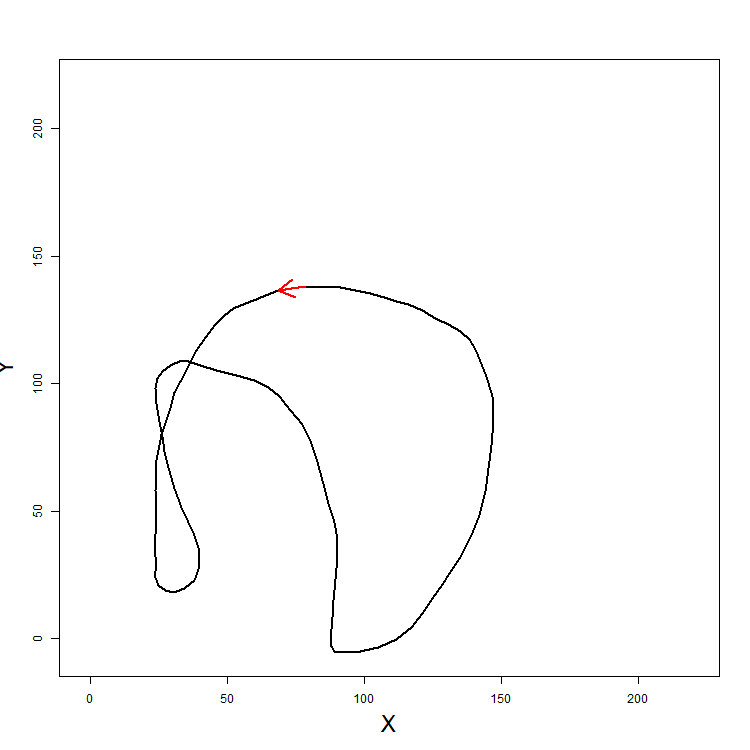} &\includegraphics[align=c, width=.13\textwidth]{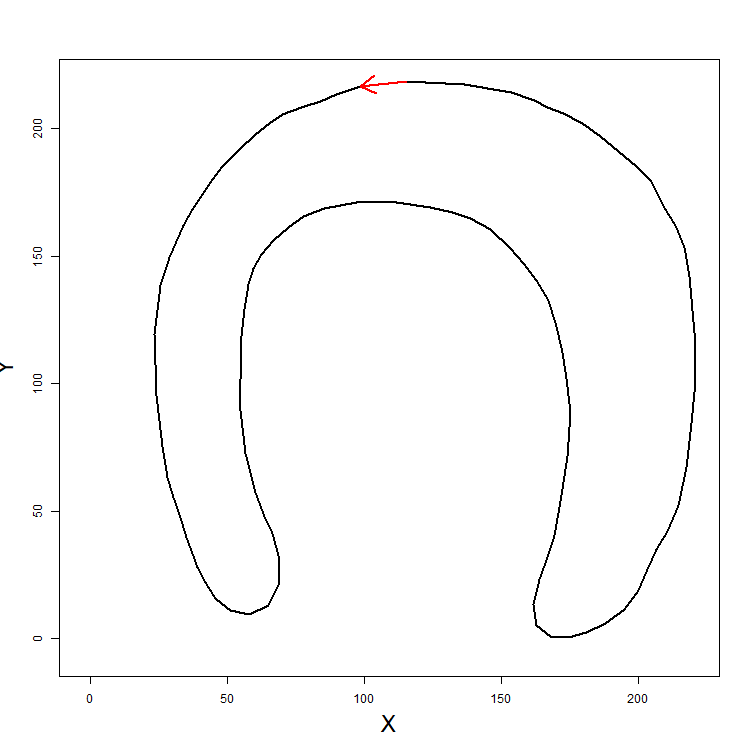}& \includegraphics[align=c, width=.13\textwidth]{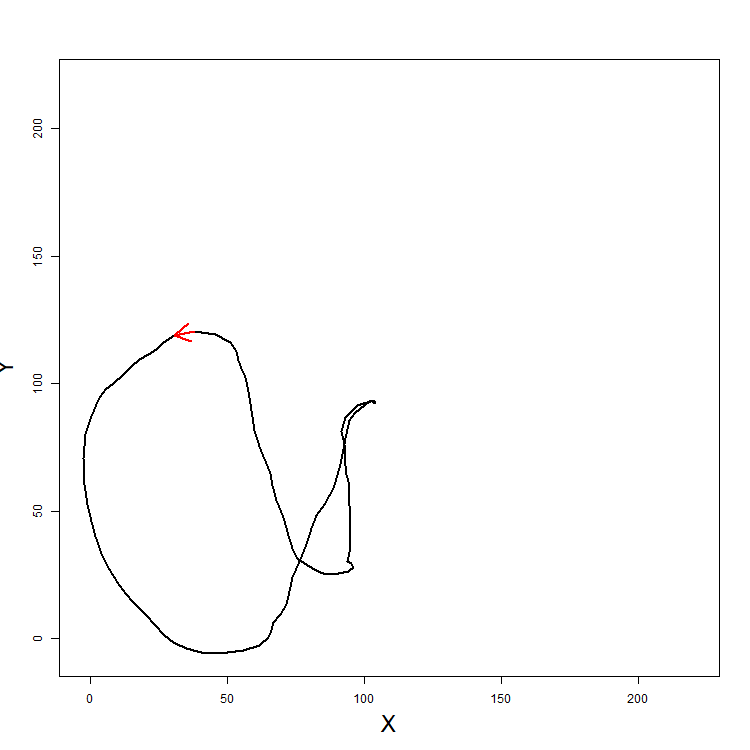} & \includegraphics[align=c, width=.13\textwidth]{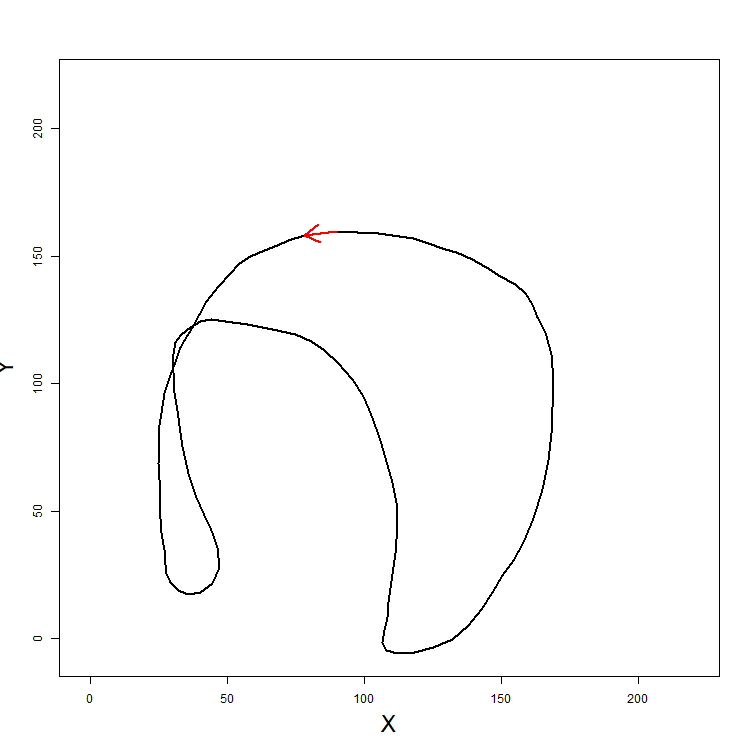} & \includegraphics[align=c, width=.13\textwidth]{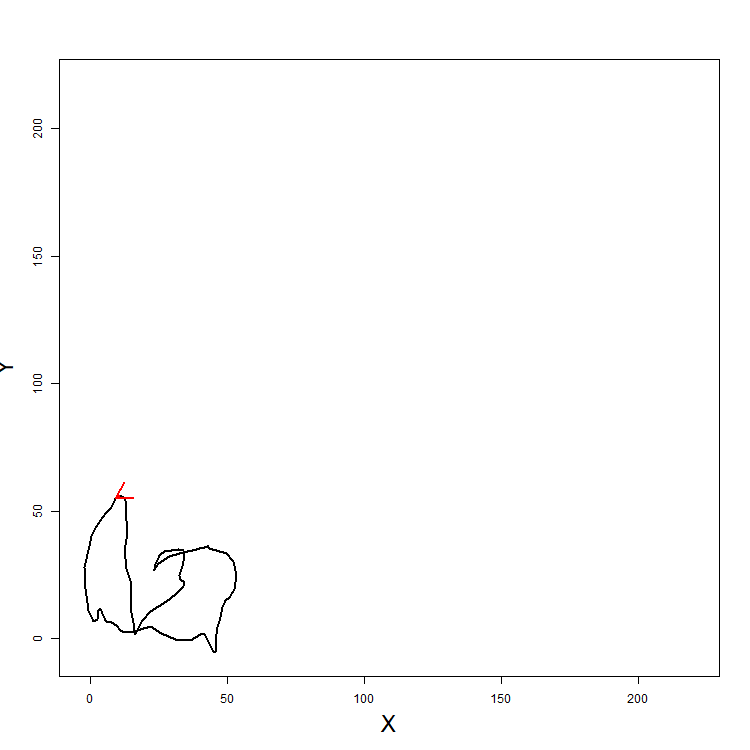} \\
\end{tabular}
    \caption{Horseshoe contours generated with (a) our approach without the deformation parameters, (b) our approach with the deformation parameters, and (c) with MFPCA.}
    \label{horse-gen}
\end{figure}
\begin{figure}[H]
\centering
    \begin{tabular}{c c c c c c c  c c c c c c}
$D_1$: & 0.07 & 0.08 & 0.04 & 0.04 & 0.04 \\ 
  $D_2$: & 0.01 & 0.04 & 0.02 & 0.01 & 0.01 \\ 
 (a)& \includegraphics[align=c, width=.13\textwidth]{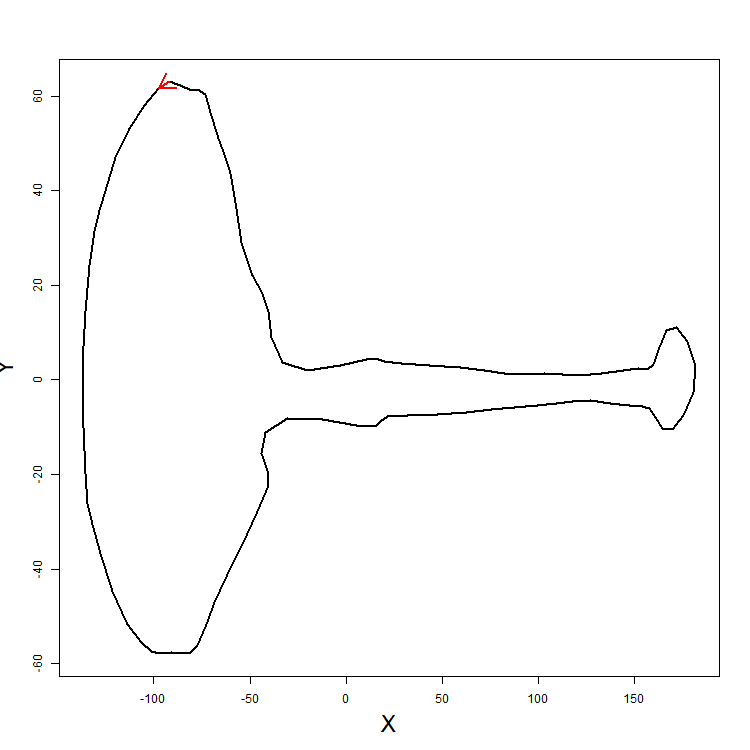} &\includegraphics[align=c, width=.13\textwidth]{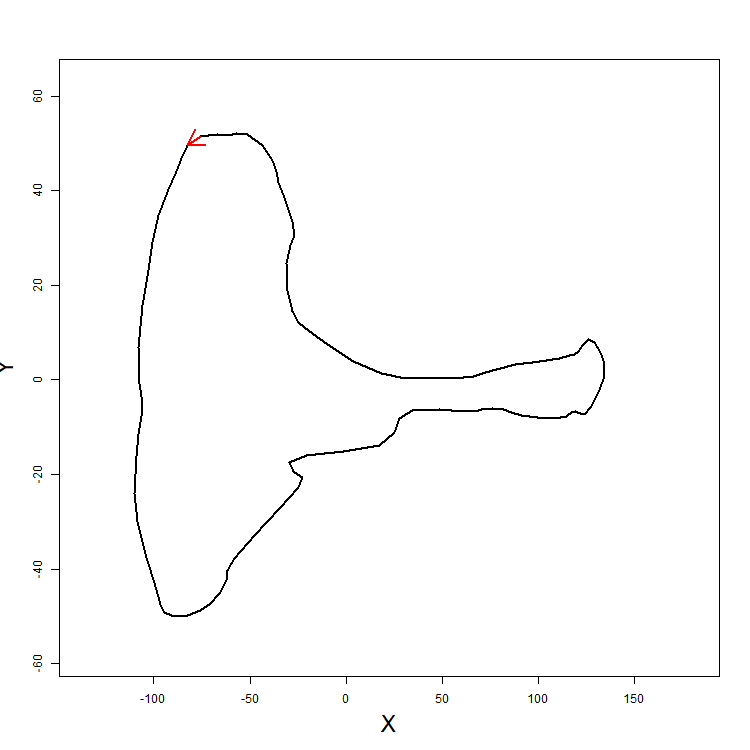}& \includegraphics[align=c, width=.13\textwidth]{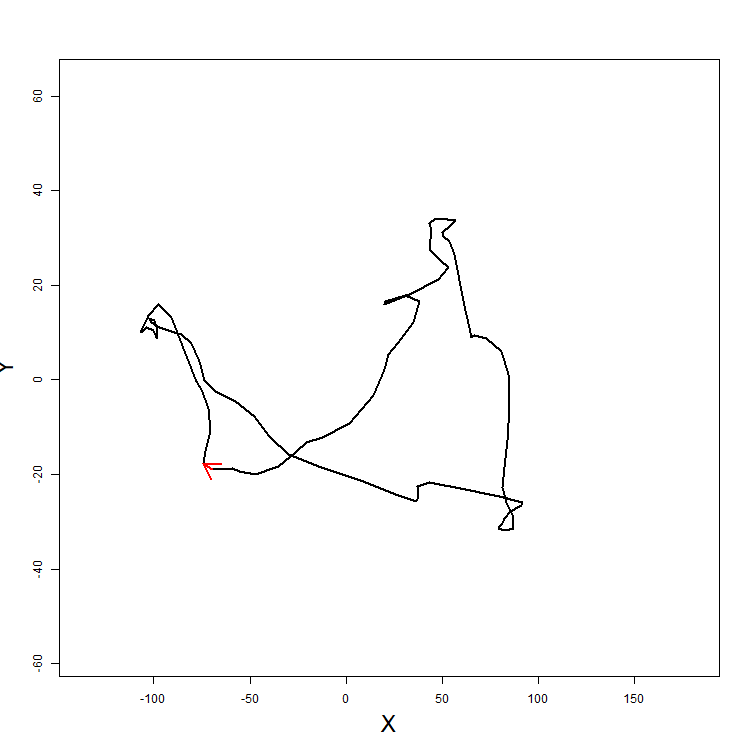} & \includegraphics[align=c, width=.13\textwidth]{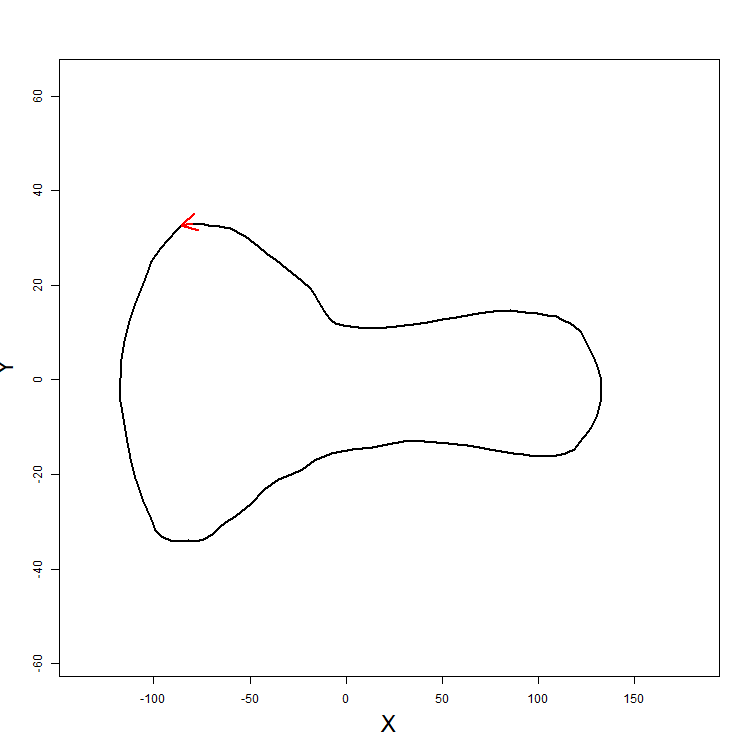} & \includegraphics[align=c, width=.13\textwidth]{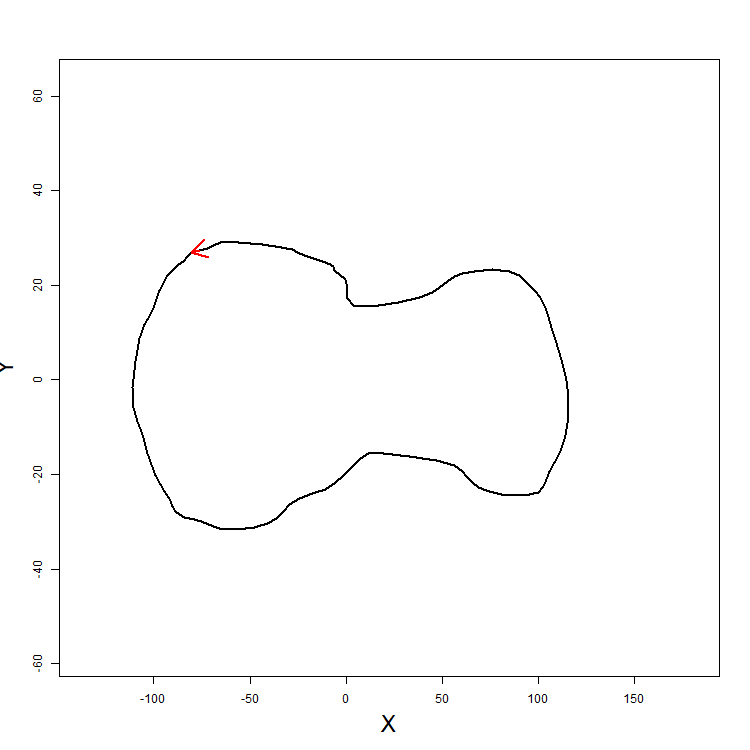}\\
 \hline
  $D_1$: & 0.09 & 0.09 & 0.05 & 0.06 & 0.05 \\ 
  $D_2$: & 0.03 & 0.05 & 0.03 & 0.02 & 0.03 \\ 
   (b)& \includegraphics[align=c, width=.13\textwidth]{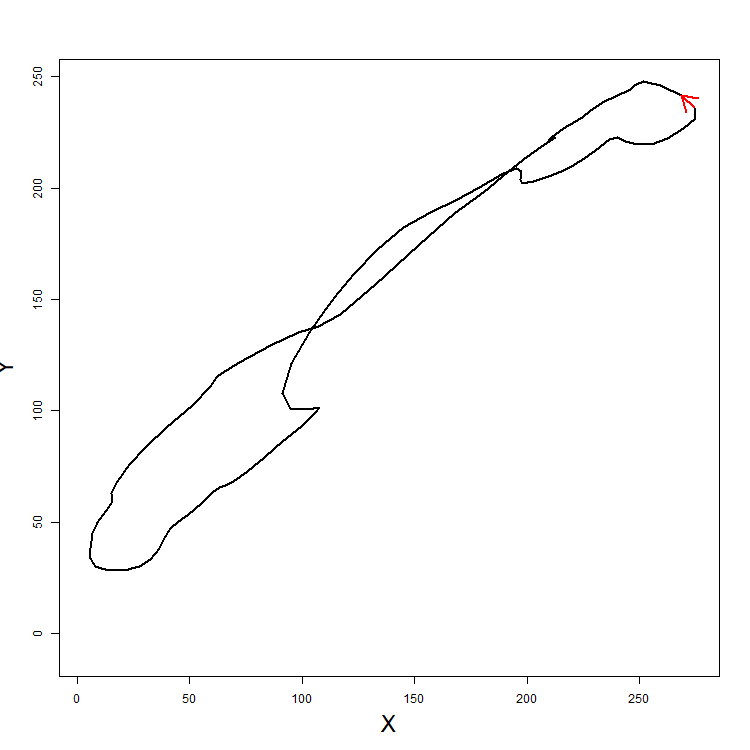} &\includegraphics[align=c, width=.13\textwidth]{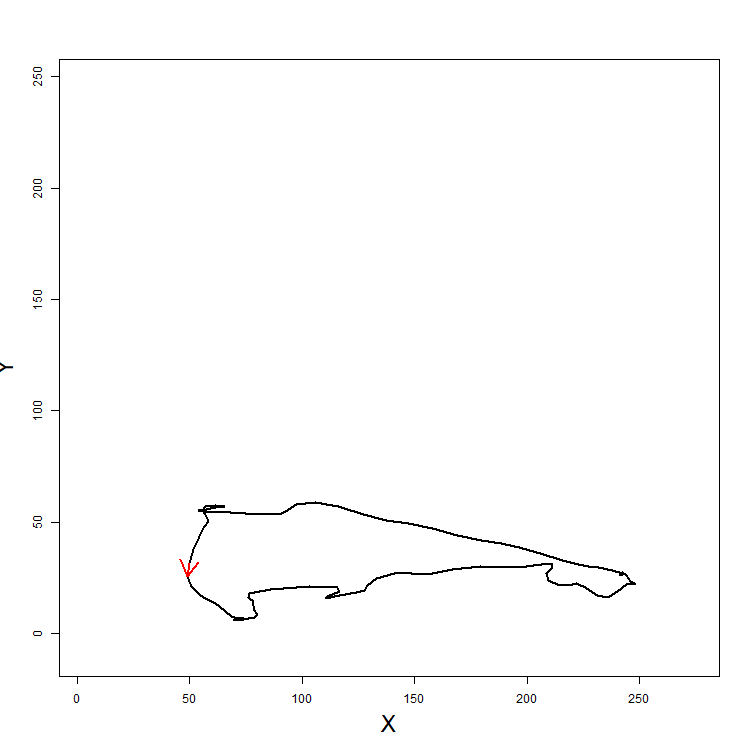}& \includegraphics[align=c, width=.13\textwidth]{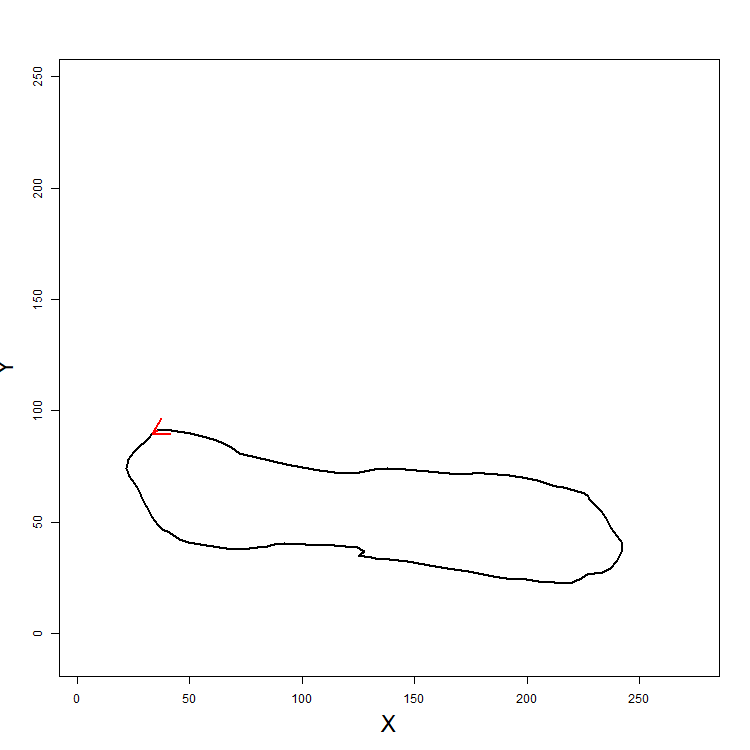} & \includegraphics[align=c, width=.13\textwidth]{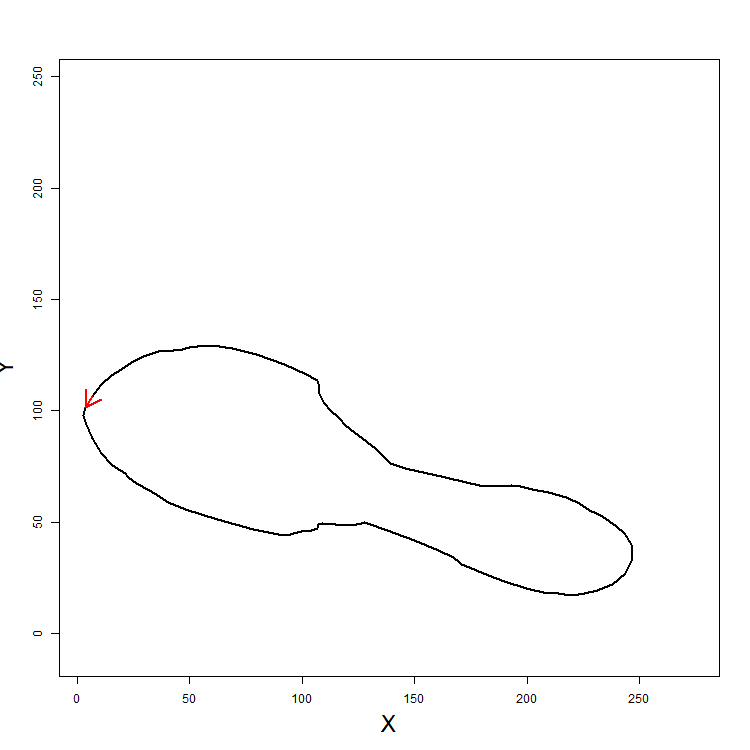} & \includegraphics[align=c, width=.13\textwidth]{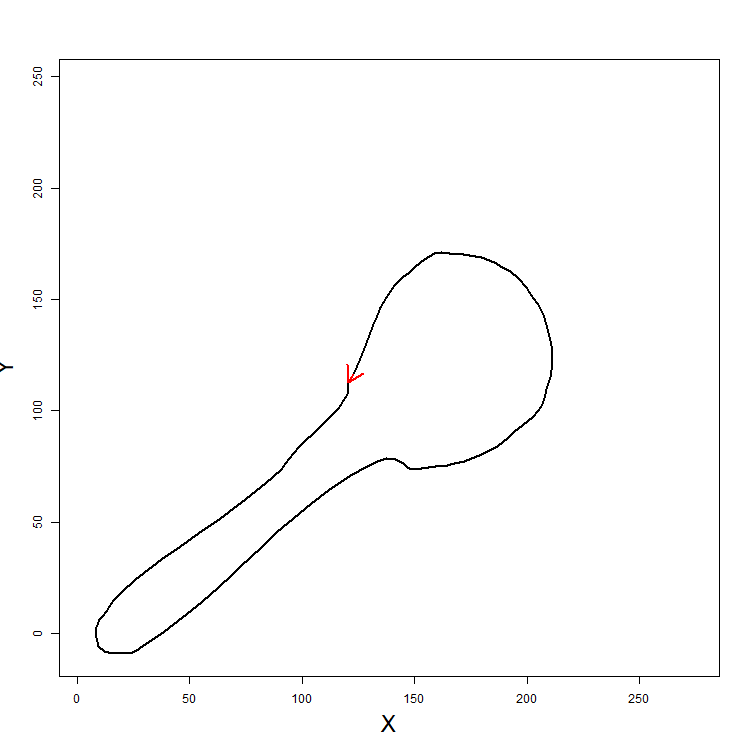}\\
   \hline
  $D_1$: & 0.05 & 0.04 & 0.07 & 0.05 & 0.04 \\ 
  $D_2$: & 0.02 & 0.02 & 0.02 & 0.03 & 0.02 \\ 
    (c)&\includegraphics[align=c, width=.13\textwidth]{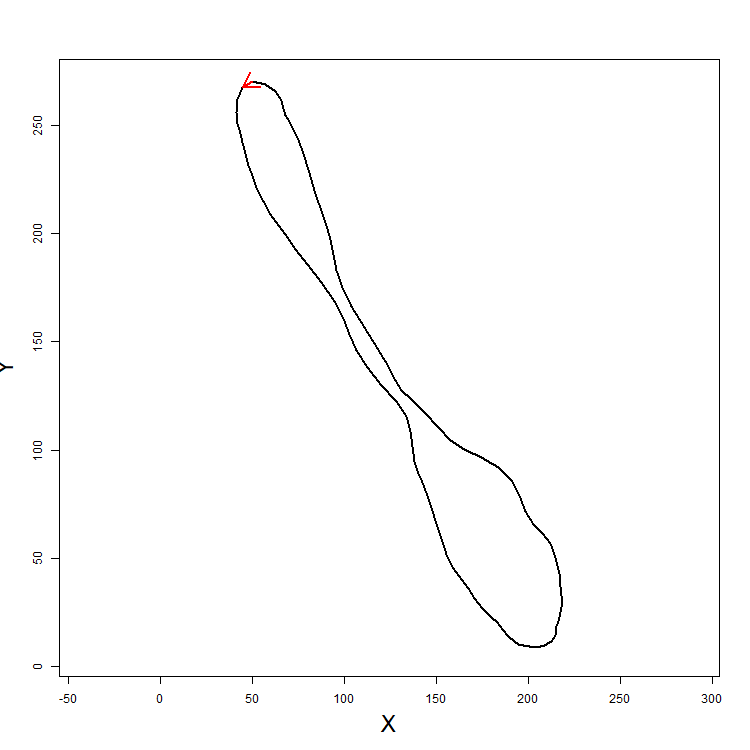} &\includegraphics[align=c, width=.13\textwidth]{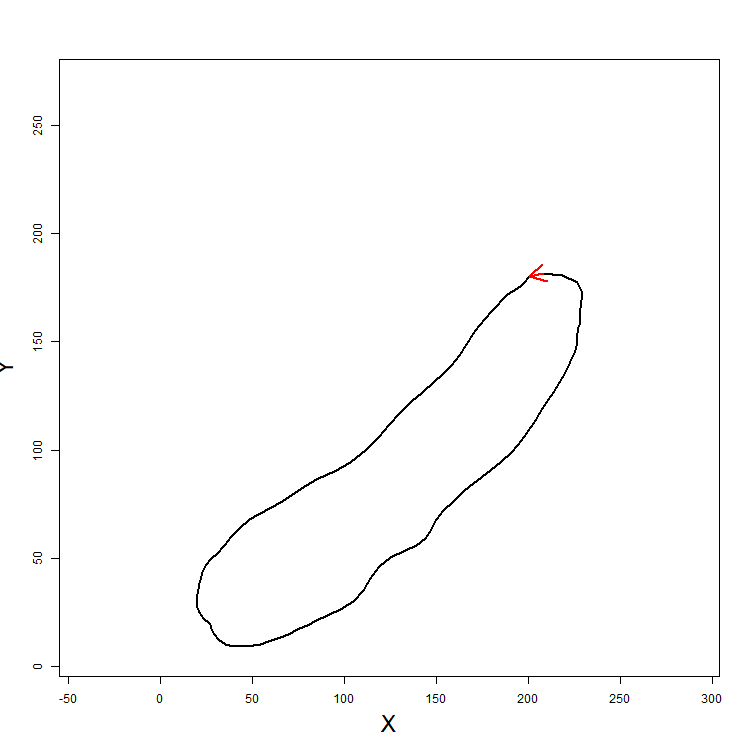}& \includegraphics[align=c, width=.13\textwidth]{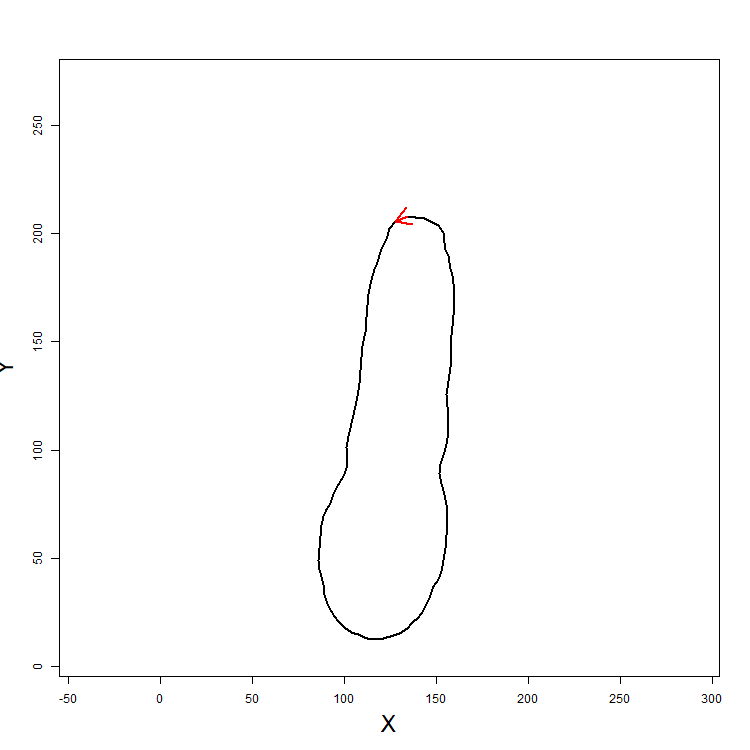} & \includegraphics[align=c, width=.13\textwidth]{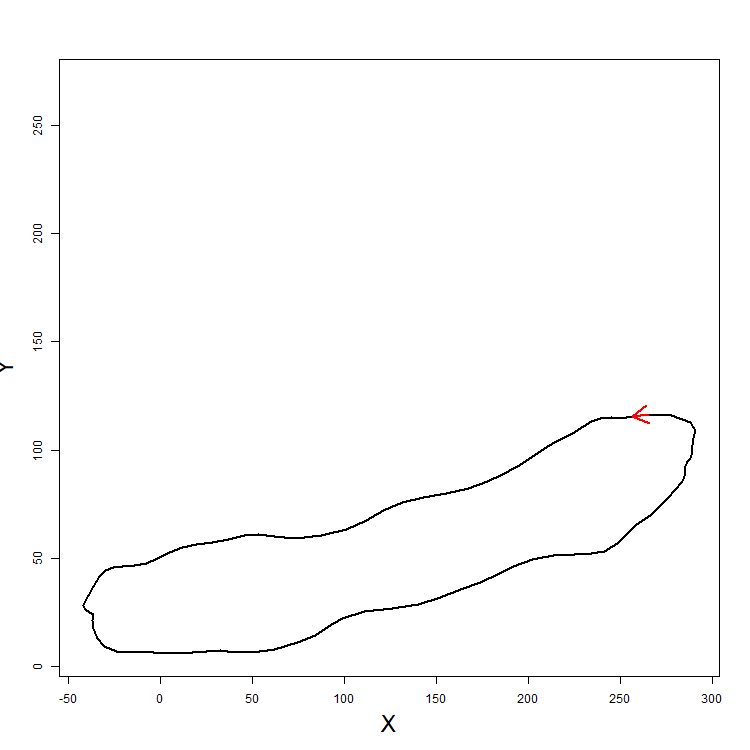} & \includegraphics[align=c, width=.13\textwidth]{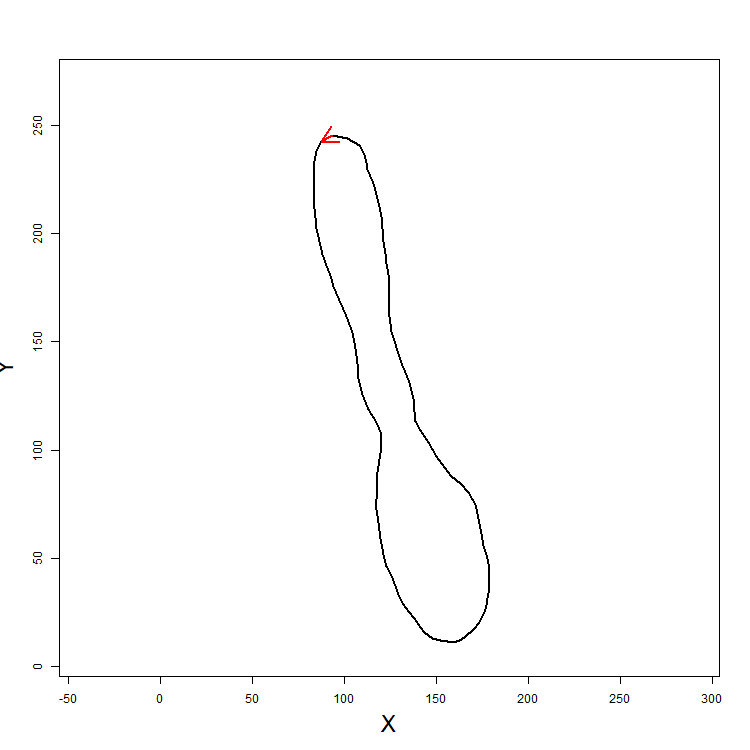} \\
\end{tabular}
\caption{Spoon contours generated with (a) our approach without the deformation parameters, (b) our approach with the deformation parameters, and (c) with MFPCA. }
    \label{spoon-gen}
\end{figure}

\section{Proofs of the formal results}
\begin{proof}[Proof of Proposition \eqref{prop_1}] 
The proposition is composed of two parts: (i) $(\Gamma, \circ)$ is a group, and (ii) the isometry propriety holds.
\paragraph{\normalfont Proof of (i):} First recall that for $r \in \mathbb{R}$, we have that $\text{mod}(r, 1) =r-\lfloor r \rfloor \equiv \{r\}$, i.e., that $\{r\}$ is the fractional part of $r$. We use the following properties of the fractional part in the proof: (a) $\{\{r\}\}= \{r\}$, (b) $\{\{r_1\}+r_2\}= \{r_1+r_2\}$, (c) $\{r+1\}=\{r\}$ for $r, r_1,r_2 \in \mathbb{R}$, (d) $\{r\}=r-1$ for $1\leq r < 2$ and (e) $\{r\}=r+1$ for $-1\leq r < 0$. 

\par The group structure is a consequence of the following lemma.
\begin{lem} \label{lem_annexeB}
Let $\gamma_{\delta_1}, \gamma_{\delta_2} \in \Gamma$, then 
\begin{equation}     
    \gamma_{\delta_1} \circ  \gamma_{\delta_2}= \gamma_{\tilde{\delta}} \in \Gamma, 
    \label{eq_i}
\end{equation}
with $\tilde{\delta}=\text{mod}(\delta_1+\delta_2, 1)$.
\end{lem}

\begin{proof} First note that by definition $\gamma_0=\text{mod}(t,1)=\text{mod}(t-1,1)=\gamma_1$.
\begin{itemize}
\item \textbf{Case where $\delta_1\in\{ 0,1\}$ and/or $\delta_2\in\{ 0,1\}$}. If $\delta_1=0$ and $\delta_2\in[0,1]$, then for any $t\in[0,1]$, we have
\begin{eqnarray*}\gamma_{\delta_1} \circ \gamma_{\delta_2}(t)=\gamma_0 \circ \gamma_{\delta_2}(t)=
        \text{mod}(\gamma_{\delta_2}(t), 1)&=& \text{mod}(\text{mod}(t-\delta_2, 1), 1)\\
        &=&\text{mod}(t-\delta_2, 1)\\
        &=& \left\{\begin{array}{l l}
         \gamma_{\delta_2}&   \text{if }\delta_2 <  1\\
         \gamma_0 & \text{if } \delta_2=1
    \end{array}\right.\\
    &=&\gamma_{\text{mod}(0+\delta_2,1)}.\end{eqnarray*}
        
All the other cases are shown with a similar argument.

\item \textbf{Case where $\delta_1 \in (0,1)$ and $\delta_2\in (0,1)$}. For any $t\in [0,1]$ we have
$$\gamma_{\delta_2}(t)=\left\{\begin{array}{ll}
         1+(t-\delta_2)& \text{if }t-\delta_2 <  0 ,\\
         (t-\delta_2) & \text{if } t-\delta_2\geq 0,
    \end{array}     \right.$$
which leads to
\begin{eqnarray*}
\gamma_{\delta_1} \circ \gamma_{\delta_2}(t)&=& \text{mod}(\gamma_{\delta_2}(t)-\delta_1, 1)\\
&=& \left\{\begin{array}{l l}
         \text{mod}(1+(t-\delta_2)-\delta_1, 1)&   \text{if }t-\delta_2 <  0\\
         \text{mod}((t-\delta_2)-\delta_1, 1) & \text{if } t-\delta_2\geq 0
    \end{array}\right.\\
    &=& \text{mod}(t-(\delta_2+\delta_1), 1)\\
    &=&\text{mod}(t-\tilde \delta, 1).
\end{eqnarray*} 
The last equality comes from the fact that 
$$\delta_1+\delta_2 = \left\{\begin{array}{l l}
         \tilde{\delta}&   \text{if }\delta_1+\delta_2 <  1\\
         \tilde{\delta} + 1 & \text{if } \delta_1+\delta_2 \ge  1
    \end{array}\right. , \ \textrm{ and } \ \text{mod}(t-\tilde \delta, 1)=\text{mod}(t-(\tilde \delta+1), 1).$$ 
Since by definition $\tilde{\delta}\in [0,1)$, we finally have that $\gamma_{\delta_1} \circ \gamma_{\delta_2}(t)=\text{mod}(t-\tilde \delta, 1) = \gamma_{\tilde \delta}(t)$, with $\gamma_{\tilde \delta}\in \Gamma$.
\end{itemize}
\end{proof}

The proof that $(\Gamma, \circ)$ is a group follows directly from Lemma \ref{lem_annexeB} : 
\begin{itemize}
    \item \textbf{Associativity:} Let $\gamma_{\delta_1}, \gamma_{\delta_2},\gamma_{\delta_3} \in \Gamma $, then we have
    \begin{eqnarray*}
    (\gamma_{\delta_1} \circ \gamma_{\delta_2}) \circ \gamma_{\delta_3}&=&  \gamma_{\text{mod}(\delta_1+\delta_2, 1)} \circ \gamma_{\delta_3}\\
    &=& \gamma_{\text{mod}(\text{mod}(\delta_1+\delta_2, 1)+\delta_3, 1)}\\
    &=&\gamma_{\text{mod}(\delta_1+\text{mod}(\delta_2+\delta_3, 1), 1)}\\
    &=&\gamma_{\delta_1} \circ (\gamma_{\delta_2} \circ \gamma_{\delta_3}),
    \end{eqnarray*}
    where we used property (b) of fractional parts to obtain the third equality. 
    %$$\text{mod}(\text{mod}(\delta_1+\delta_2, 1)+\delta_3, 1)=\text{mod}(\delta_1+\delta_2+\delta_3)=\text{mod}(\delta_1 + \text{mod}(\delta_2+\delta_3,1),1).$$
    
    \item \textbf{Identity element}: The function $\gamma_0 \in \Gamma$ is the identity element since for any $\gamma_\delta \in \Gamma$, we have
    \begin{align*}
        \gamma_{\delta} \circ \gamma_{0}= \gamma_{0} \circ \gamma_{\delta}=\gamma_{\text{mod}(\delta +0, 1)}=  \gamma_{\delta}. 
    \end{align*} 
    \item \textbf{Inverse element}: Let $\gamma_\delta \in \Gamma$ then $\gamma_{1-\delta}\in \Gamma$ is the inverse of $\gamma_\delta$ since
    $$
        \gamma_{\delta} \circ \gamma_{1-\delta}= \gamma_{\text{mod}(\delta+(1-\delta), 1)}=  \gamma_{0} \ \textrm{ and } \ \gamma_{1-\delta} \circ \gamma_{\delta}=\gamma_{\text{mod}(1-\delta)+\delta, 1)}= \gamma_{0}.
   $$ %Of course, for the case $\delta=0$, the inverse of $\gamma_0$ is $\gamma_0$. 
\end{itemize}
This concludes the proof of (i).

\paragraph{\normalfont Proof of (ii):} 
Let $\boldsymbol{f}= (f^{(1)} , f^{(2)})^\top \in \mathcal{H}$ and $\gamma_\delta \in \Gamma$, we have that  
\begin{align*}
    \norm{\boldsymbol{f} \circ \gamma_\delta}_2^2& = \int_{0}^1 (f^{(1)} \circ \gamma_\delta(t))^2dt + \int_{0}^1 (f^{(2)} \circ \gamma_\delta(t))^2dt, 
\end{align*}
where for $j\in \{1,2\}$ and $t\in (0,1)$, we have 
$$f^{(j)}\circ\gamma_\delta(t) =  \left\{\begin{array}{l l}
         f^{(j)}(1+t-\delta)&   \text{if }t-\delta <  0,\\
         f^{(j)}(t-\delta) & \text{if } t-\delta\geq 0.
    \end{array}\right.$$
Then
\begin{eqnarray*}
\int_0^1 \left(f^{(j)}\circ\gamma_\delta(t)\right)^2 dt&=& \int_0^\delta \left(f^{(j)}(1+t-\delta)\right)^2dt + \int_{\delta}^1 \left(f^{(j)}(t-\delta)\right)^2dt \\
&=&  \int_{1-\delta}^1 \left(f^{(j)}(u)\right)^2du + \int_0^{1-\delta} \left(f^{(j)}(u)\right)^2du \\
&=& \int_0^1 \left(f^{(j)}(u)\right)^2du,
\end{eqnarray*}
which leads to
$$\norm{\boldsymbol{f} \circ \gamma_\delta}_2^2 = \int_0^1 \left(f^{(1)}(u)\right)^2du+\int_0^1 \left(f^{(2)}(u)\right)^2du = \norm{\boldsymbol{f}}_2^2. $$

\end{proof}
\begin{lem} $d(\cdot, \cdot)$ is a pseudo-distance on $\mathcal{H}$. 
\label{lem}
\end{lem}
\begin{proof}[Proof of Lemma \ref{lem}] 
Let $\boldsymbol{f}, \boldsymbol{g} \in \mathcal{H}$.
\begin{itemize}
    \item \textbf{Positivity:} We have by definition
    $$d(\boldsymbol{f}, \boldsymbol{g}) =\min_{\delta\in [0,1], \theta \in [0, 2\pi]}\norm{\boldsymbol{f}\circ \gamma_\delta- \mathbf{O}_\theta \boldsymbol{g}}_\mathcal{H}  \ge 0,$$ and
    $$d(\boldsymbol{f}, \boldsymbol{f})=\min_{\delta\in [0,1], \theta \in [0, 2\pi]}\norm{\boldsymbol{f}\circ \gamma_\delta- \mathbf{O}_\theta \boldsymbol{f}}_\mathcal{H} = \norm{\boldsymbol{f}\circ \gamma_0- \mathbf{O}_0 \boldsymbol{f}}_\mathcal{H} = 0.$$
    %  is positive or null for $\boldsymbol{f}, \boldsymbol{g} \in \mathcal{H}$. This a direct consequence of $\norm{.}_\mathcal{H}$. 
    \begin{comment}
\item Let $\boldsymbol{f}\in \mathcal{H}$, then  \begin{align*}
    d(\boldsymbol{f}, \boldsymbol{f})&=\min_{\delta\in [0,1), \theta \in [0, 2\pi]}\norm{\boldsymbol{f}\circ \gamma_\delta- \mathbf{O}_\theta \boldsymbol{f}}_\mathcal{H} = \norm{\boldsymbol{f}\circ \gamma_0- \mathbf{O}_0 \boldsymbol{f}}_\mathcal{H} = 0.
\end{align*}
\end{comment}
\item \textbf{Symmetry:}    
    First note that  
    $$d(\boldsymbol{f}, \boldsymbol{g})=\min_{\delta\in [0,1], \theta \in [0, 2\pi]}\norm{\boldsymbol{f}\circ \gamma_\delta- \mathbf{O}_\theta \boldsymbol{g}}_\mathcal{H}=\min_{\delta\in [0,1), \theta \in [0, 2\pi]}\norm{\mathbf{O}_{2\pi-\theta}\boldsymbol{f}- \boldsymbol{g}\circ \gamma_{1-\delta} }_\mathcal{H}.$$
    Letting $\theta'=2\pi-\theta$ and $\delta'=1-\delta$, we have that
    $$d(\boldsymbol{f}, \boldsymbol{g})=\min_{\delta'\in [0,1], \theta' \in [0, 2\pi]}\norm{\mathbf{O}_{\theta'}\boldsymbol{f}- \boldsymbol{g}\circ \gamma_{\delta'} }_\mathcal{H}=d(\boldsymbol{g},\boldsymbol{f}).$$
\item \textbf{Triangle inequality:}  

Since $\norm{\cdot}_\mathcal{H}$ induce a proper distance, we have that 
$$
\norm{\boldsymbol{f}_1-\boldsymbol{f}_2}_\mathcal{H}+\norm{\boldsymbol{f}_2-\boldsymbol{f}_3}_\mathcal{H} \geq \norm{\boldsymbol{f}_1-\boldsymbol{f}_3}_\mathcal{H}
$$

for $\boldsymbol{f}_1,\boldsymbol{f}_2, \boldsymbol{f}_3 \in \mathcal{H}$. 
\par 
For any $\boldsymbol{f},\boldsymbol{r},\boldsymbol{g}\in \mathcal{H}$, by letting
$\boldsymbol{f}_1=\boldsymbol{f}\circ \gamma_\delta, \boldsymbol{f}_2=\mathbf{O}_\theta \boldsymbol{r}$ and $ \boldsymbol{f}_3=\mathbf{O}_{\theta'} \boldsymbol{g}\circ\gamma_{\delta'}
$ in the previous inequality, we obtain
\begin{align}
    &\norm{\boldsymbol{f} \circ \gamma_\delta- \mathbf{O}_\theta \boldsymbol{r}}_\mathcal{H}+\norm{ \mathbf{O}_\theta \boldsymbol{r}- \mathbf{O}_{\theta'} \boldsymbol{g}\circ\gamma_{\delta'}}_\mathcal{H}
 \geq \norm{\boldsymbol{f}\circ \gamma_\delta- \mathbf{O}_{\theta'} \boldsymbol{g}\circ\gamma_{\delta'}}_\mathcal{H} \nonumber \\ 
 &\phantom{aaaaa}\Rightarrow   \min_{\delta,\delta'\in [0,1], \ \theta,\theta' \in [0,2\pi]}\left(\norm{\boldsymbol{f} \circ \gamma_\delta- \mathbf{O}_\theta \boldsymbol{r}}_\mathcal{H}+\norm{ \mathbf{O}_\theta \boldsymbol{r}- \mathbf{O}_{\theta'} \boldsymbol{g}\circ\gamma_{\delta'}}_\mathcal{H}\right) \label{ineq_tri} \\
& \phantom{aaaaaaaaaa}  \geq \min_{\delta,\delta'\in [0,1], \ \theta,\theta' \in [0,2\pi]}\norm{\boldsymbol{f}\circ \gamma_\delta- \mathbf{O}_{\theta'} \boldsymbol{g}\circ\gamma_{\delta'}}_\mathcal{H}.\nonumber
\end{align}

The left-hand side of the inequality \eqref{ineq_tri} admits the following simplification  
\begin{align*}
    &\min_{\delta,\delta'\in [0,1], \ \theta,\theta' \in [0,2\pi]}\left(\norm{\boldsymbol{f} \circ \gamma_\delta- \mathbf{O}_\theta \boldsymbol{r}}_\mathcal{H}+\norm{ \mathbf{O}_\theta \boldsymbol{r}- \mathbf{O}_{\theta'} \boldsymbol{g}\circ\gamma_{\delta'}}_\mathcal{H}\right) \\
    =&\min_{\delta,\delta'\in [0,1], \ \theta,\theta' \in [0,2\pi]}\left(\norm{\boldsymbol{f} \circ \gamma_\delta- \mathbf{O}_\theta \boldsymbol{r}}_\mathcal{H}+\norm{ \boldsymbol{r}- \mathbf{O}_{\theta'-\theta} \boldsymbol{g}\circ\gamma_{\delta'}}_\mathcal{H}\right) \\
    =& 
    \min_{\delta\in [0,1], \ \theta \in [0,2\pi]}\left(\norm{\boldsymbol{f} \circ \gamma_\delta- \mathbf{O}_\theta \boldsymbol{r}}_\mathcal{H}\right) +\min_{\delta'\in [0,1], \ \theta^\star \in [0,2\pi]}\left(\norm{\boldsymbol{r}- \mathbf{O}_{\theta^\star} \boldsymbol{g}\circ\gamma_{\delta'}}_\mathcal{H}\right),   
\end{align*}
where 
    $$ \min_{\delta\in [0,1], \ \theta \in [0,2\pi]}\norm{\boldsymbol{f} \circ \gamma_\delta- \mathbf{O}_\theta \boldsymbol{r}}_\mathcal{H}=d(\boldsymbol{f}, \boldsymbol{r}),$$
    and $$
        \min_{\delta'\in[0,1], \theta^\star\in [0,2\pi]}\norm{\boldsymbol{r}- \mathbf{O}_{\theta^\star} \boldsymbol{g}\circ\gamma_{\delta'}}_\mathcal{H}
        =  
        \min_{\delta'\in[0,1], \theta^\star\in [0,2\pi]}\norm{  \boldsymbol{r}\circ \gamma_{1-\delta'}- \mathbf{O}_{\theta^\star} \boldsymbol{g}}_\mathcal{H} 
        =d(\boldsymbol{r},\boldsymbol{g}).
  $$
    Using a similar reasoning, we can show that the right-hand side of \eqref{ineq_tri} is equal to $d(\boldsymbol{f}, \boldsymbol{g})$, and thus that $d(\boldsymbol{f}, \boldsymbol{r})+d(\boldsymbol{r}, \boldsymbol{g}) \geq d(\boldsymbol{f}, \boldsymbol{g}). $   
\end{itemize}
\end{proof}

\begin{proof}[Proof of Lemma \eqref{lem_f_gamma}] Due to the block structure of the matrix $\boldsymbol{\beta}(\delta)$, the lemma follows directly from the fact that  
\begin{equation}
    \begin{pmatrix} \sin(2m\pi \gamma_\delta(t)) \\ 
        \cos(2m\pi \gamma_\delta(t))
        \end{pmatrix}= \mathbf{O}_{-2m \pi \delta}\begin{pmatrix}
             \sin(2m\pi t ) \\ 
        \cos(2m\pi t)
        \end{pmatrix},
\label{beta_delta}
\end{equation}
for $m \in \mathbb{N}^*$ and $\gamma_\delta \in \Gamma$, which we prove below. \\ 
First note that for $t\in[0,1]$,    
\begin{equation*}
        \sin(2m\pi\gamma_\delta(t))= \left\{ 
        \begin{array}{c c}
            \sin(2m\pi\left(1+(t-\delta) \right) ) & t \leq \delta \\
                        \sin(2m\pi((t-\delta)) ) & t > \delta,  
        \end{array}
        \right.
\end{equation*}
where $\sin(2m\pi(1+(t-\delta)) )= \sin(2\pi(t-\delta)+2m\pi  )= \sin(2\pi(t-\delta))$. Hence $$\sin(2m\pi \gamma_\delta(t))=\sin(2m\pi(t-\delta)). $$
The same reasoning shows that  
$$\cos(2m\pi \gamma_\delta(t))=\cos(2m\pi(t-\delta)). $$
Finally, 
\begin{align*}
    \begin{pmatrix} \sin(2m\pi \gamma_\delta(t)) \\ 
        \cos(2m\pi \gamma_\delta(t))
        \end{pmatrix}=     \begin{pmatrix} \sin(2m\pi t -2m\pi \delta) \\ 
        \cos(2m\pi t -2m\pi \delta)
        \end{pmatrix}= \mathbf{O}_{-2 m\pi \delta }\begin{pmatrix}
             \sin(2m\pi t ) \\ 
        \cos(2m\pi t)
        \end{pmatrix},
\end{align*}
which concludes the proof. 
\end{proof}
\begin{proof}[ Proof of Proposition \eqref{chang_b}]
    Let $
    \boldsymbol{f}_{\theta, \delta}= \mathbf{C}- \mathbf{O}_\theta\boldsymbol{\mu}\circ \gamma_\delta,$ this function admits the following representation    
\begin{align*}
\boldsymbol{f}_{\theta, \delta}&= \begin{pmatrix}
    \boldsymbol{\psi}^\top & 0 \\ 
    0 &  \boldsymbol{\psi}^\top
\end{pmatrix} \Vec{\boldsymbol{\alpha}}_{\theta, \delta},
\end{align*} 
with \begin{align*}
    \Vec{\boldsymbol{\alpha}}_{\theta, \delta} &= \text{Vec}\left(\left(\boldsymbol{\alpha} \boldsymbol{\beta}(\delta)- \mathbf{O}_\theta\boldsymbol{u} \right)^\top\right),
    \end{align*}
and where Vec$(\cdot)$ is the vectorization operator. Hence the norm of $\boldsymbol{f}_{\theta, \delta}$ is given by 
\begin{align*}
    ||\boldsymbol{f}_{\theta, \delta}||_\mathcal{H}^2& = \int_{0}^1 \boldsymbol{f}_{\theta, \delta}^\top(t)\boldsymbol{f}_{\theta, \delta}(t)dt \\  
    &=     \Vec{\boldsymbol{\alpha}}_{\theta, \delta}^\top \underbrace{\left[\int_{0}^1 \begin{pmatrix}
        \boldsymbol{\psi}(t) & 0 \\ 
        0 &\boldsymbol{\psi}(t)  
    \end{pmatrix}  \begin{pmatrix}
        \boldsymbol{\psi}^\top(t) & 0 \\ 
        0 &\boldsymbol{\psi}^\top(t)  
    \end{pmatrix} dt \right]}_{\mathbf{I} }    \Vec{\boldsymbol{\alpha}}_{\theta, \delta} \\ 
    &= \Vec{\boldsymbol{\alpha}}_{\theta, \delta}^\top\Vec{\boldsymbol{\alpha}}_{\theta, \delta} \\
    &= \norm{\Vec{\boldsymbol{\alpha}}_{\theta, \delta}}_2^2 \\
    & = \norm{\text{Vec}\left(\left(\boldsymbol{\alpha} \boldsymbol{\beta}(\delta)- \mathbf{O}_\theta\boldsymbol{u} \right)^\top\right)}_2^2 \\
    &= \norm{\boldsymbol{\alpha} \boldsymbol{\beta}(\delta)- \mathbf{O}_\theta\boldsymbol{u} }_F^2, 
\end{align*}
which concludes the proof.

\end{proof}

\begin{proof}
\label{delta_sol}
% Using Proposition \eqref{chang_b}, we have, for $j\in\{1,\ldots,p \}$ , 
% \begin{align*}
%         \norm{\mathbf{O}_{\theta}\bar{C}_j\circ \gamma_\delta-C_j^*}_\mathcal{H}^2&=\norm{ \bar{A}_j^ \theta \boldsymbol{\phi} \circ\gamma_{\delta}-A_j^* \boldsymbol{\phi}+\bar{B}_j^\theta-B_j^* }_\mathcal{H}^2 \\
%         &= \norm{\bar{A}_j^ \theta P_\delta-A_j^*
%         }_F^2+\norm{\bar{B}_j^\theta- B_j^*}_F^2,
%     \end{align*}
% where $$P_\delta=\begin{pmatrix}
%     \mathbf{O}_{2\pi \delta} & 0 & \ldots & 0 \\ 
%     0 & \mathbf{O}_{4\pi \delta} & \ldots & 0 \\ 
%     \vdots & \ldots  & &\vdots\\ 
%     0 & 0& \ldots & \mathbf{O}_{M\pi \delta}
% \end{pmatrix}.$$ 
% Since the second term on the right-hand side does not depend on $\delta$, \eqref{pb2} is equivalent to
The equation \eqref{pb2} can be seen as a constrained Procrustes problem, since the orthogonal matrix $\boldsymbol{\beta}(-\delta)$ has a fixed form.
% $$
%     \hat{\delta}= \argmin_{\delta\in [0, 1]}\norm{\mathbf{O}_\theta\boldsymbol{u}\boldsymbol{\beta}(-\delta) - \boldsymbol{\alpha}}_F^2.
% 	\label{pb2}
% $$
 Let note $\boldsymbol{u}^\theta= \mathbf{O}_\theta \boldsymbol{u}$,  we have 
\begin{equation}
    \hat{\delta}=\argmin_{\delta\in [0,1]} \norm{\boldsymbol{\alpha}\boldsymbol{\beta}(\delta)-\boldsymbol{u}^\theta }_F^2=\argmin_{\delta\in [0,1]} \norm{\boldsymbol{\beta}(\delta)-\boldsymbol{\alpha}^\top\boldsymbol{u}^\theta  }_F^2
    \label{pb_2_3}.
\end{equation}
Since $\boldsymbol{\beta}(\delta)$ is a sparse orthogonal matrix, Equation \eqref{pb_2_3} reduces to 
\begin{equation}
    \hat{\delta}= \argmin_{\delta\in [0,1]}\left( \sum_{k \in \{1, 3, \ldots, M-1\} } \norm{ \mathbf{\Sigma}_{k, \theta }- \mathbf{O}_{\pi\delta (k+1) } }^2\right)
    \label{pb_sim}, 
\end{equation}
where $$
\boldsymbol{\alpha}^\top \boldsymbol{u}^\theta \odot \begin{pmatrix}
    \mathbf{I}_2 & 0 & \ldots & 0 \\ 
    0& \mathbf{I}_2& \ldots & 0 \\ 
    \vdots & \ldots  & &\vdots\\
    0& 0 & \ldots & \mathbf{I}_2 
\end{pmatrix}=\begin{pmatrix}
    \boldsymbol{\Sigma}_{1, \theta} & 0 & \ldots & 0 \\ 
    0& \boldsymbol{\Sigma}_{3, \theta}& \ldots & 0 \\ 
       \vdots & \ldots  & &\vdots\\
       0& 0 & \ldots & \boldsymbol{\Sigma}_{M-1, \theta} 
\end{pmatrix},
$$
and $\odot$ denotes the Hadamard product. \\ 
Regarding the resolution of \eqref{pb_sim}, note that 
\begin{align*}
   \begin{split} \frac{\partial }{\partial \delta } \norm{ \mathbf{\Sigma}_{k, \theta }- \mathbf{O}_{\pi\delta (k+1) } }^2&= 
       \pi (k+1)\left( \text{Tr}(\boldsymbol{\Sigma}_{k, \theta})\sin\left( \pi(k+1)\delta\right) \right. \\
       & \quad \left. +\text{Tr}(\boldsymbol{\Sigma}_{k, \theta}\mathbf{O}_{\pi/2})\cos\left( \pi(k+1)\delta\right) \right)
   \end{split}, 
\end{align*}
which leads to
\begin{align*}
   \begin{split} \sum_{k\in \{1, 3, \ldots, M-1\} } \frac{\partial }{\partial \delta } \norm{ \mathbf{\Sigma}_{j,k, \theta }- \mathbf{O}_{\pi\delta (k+1) } }^2&= \pi \sum_{k\in \{1, 3, \ldots, M-1\} }
        \left( w_{j, k}^{1,\theta}\sin\left( \pi(k+1)\delta\right) \right. \\
       & \quad \left. -w_{j, k}^{2,\theta}\cos\left( \pi(k+1)\delta\right) \right)
   \end{split},  
\end{align*}
 where $w_{j, k}^{1,\theta}=\text{Tr}((k+1)\boldsymbol{\Sigma}_{k, \theta})$, $w_{j, k}^{2,\theta}=\text{Tr}(-(k+1)\boldsymbol{\Sigma}_{k, \theta}\mathbf{O}_{\pi/2})$. \par 
Therefore, the solution $\hat{\delta}_j$ of \eqref{pb2} belongs to $\mathcal{S}_M$,
where 
\begin{align*}
    \begin{split}\mathcal{S}_M=\left\{\delta \in [0, 1],\ \sum_{k\in \{1, 3, \ldots, M-1\}}
         w_{j, k}^{1,\theta}\sin\left( \pi(k+1)\delta\right) = \sum_{k\in \{1, 3, \ldots, M-1\}}w_{j, k}^{2,\theta}\cos\left( \pi(k+1)\delta\right) \right\}
\end{split}.
\end{align*}

This concludes the proof. 
\end{proof}

\end{document}